\documentclass[12pt,letterpaper,leqno]{amsart}
\usepackage{microtype}
\usepackage{amssymb}
\usepackage{amsfonts}
\usepackage{geometry}
\usepackage{amsxtra}
\usepackage{graphicx}
\usepackage{tikz}
\usepackage{hyperref}
\usepackage{amsthm}
\usepackage[normalem]{ulem}
\usepackage{comment}
\usepackage{cite}
\usepackage{wrapfig}

\usetikzlibrary{trees}
\newcommand\encircle[1]{
	\text{
  		\tikz[baseline=(X.base)] 
    		\node (X) [draw, shape=circle, inner sep=0] {${\strut #1}$};
    	}
 }
\newcommand{\bds}[1]{\boldsymbol{#1}}

\newtheorem{theorem}{Theorem}
\theoremstyle{plain}

\newtheorem{algorithm}{Algorithm}

\newtheorem{corollary}[theorem]{Corollary}

\newtheorem{definition}[theorem]{Definition}
\newtheorem{example}{Example}

\newtheorem{lemma}[theorem]{Lemma}

\newtheorem{proposition}[theorem]{Proposition}
\newtheorem{remark}[theorem]{Remark}

\numberwithin{equation}{section}
\numberwithin{theorem}{section}  
\newtheorem{heuristic}{Calculation}
\input{tcilatex}
\begin{document}
\title[Optimal Derivation of the Boltzmann Equation]{The Derivation of the
Boltzmann Equation from Quantum Many-body Dynamics}
\author{Xuwen Chen}
\address{Department of Mathematics, University of Rochester, Rochester, NY
14627}
\email{xuwenmath@gmail.com}
\urladdr{http://www.math.rochester.edu/people/faculty/xchen84/}
\author{Justin Holmer}
\address{Department of Mathematics, Brown University, 151 Thayer Street,
Providence, RI 02912}
\email{justin\_holmer@brown.edu}
\urladdr{https://www.math.brown.edu/jholmer/}

\begin{abstract}
We consider the quantum many-body dynamics at the weak-coupling scaling. We
derive rigorously the quantum Boltzmann equation, which contains the
classical hard sphere model and, effectively, the inverse power law model,
from the many-body dynamics assuming a physical and optimal regularity
bound. The regularity bound we find, on the one hand, is satisfied by
quasi-free solutions and comes from calculations regarding the local
Maxwellian solution, in which we also prove that 2-body molecular chaos
never happens unless $N=+\infty $; on the other hand, it arises from the
well-posedness threshold of the limiting Boltzmann equation below which we
prove ill-posedness. That is, the regularity cannot be higher at the $N$%
-body level, cannot be lower in the limit, and is hence a double
criticality. To work with this borderline case, we analyze all four sides,
with respect to the Fourier transform, of the BBGKY hierarchy sequence with
new tools and techniques. We prove well-definedness, compactness,
convergence, and uniqueness of hierarchies right at the criticality to
complete an optimal derivation. In particular, we have proved that, for
physical $N$-particle solutions, the Boltzmann equation emerges as the
mean-field limit and time is hence irreversible, from first principles of
quantum mechanics.
\end{abstract}

\maketitle
\tableofcontents


\section{Introduction\label{sec:Introduction}}

In 1872, Boltzmann devised the now so-called Boltzmann (transport) equation,
a fundamental equation in kinetic theory which describes the time-evolution
of the statistical behavior of a thermodynamic system away from a state of
equilibrium in the mesoscopic regime, accounting for both dispersion under
the free flow and dissipation as the result of collisions. Let the
probability distribution for the position and velocity of a typical particle
be denoted by $f$. Under the molecular chaos assumption, he wrote the
collision as%
\begin{equation*}
Q(f,f)=\int_{\mathbb{S}^{2}}\int_{\mathbb{R}^{3}}\left\vert u-v\right\vert
I(\left\vert u-v\right\vert ,\omega )\left[ f(u^{\ast })f(v^{\ast })-f(u)f(v)%
\right] dud\omega
\end{equation*}%
where $I$ is the differential cross section of the collision, $u,v$ are the
incoming velocities for a pair of particles, $\omega \in \mathbb{S}^{2}$ is
a parameter for the deflection angle in the collision of these particles,
and the outgoing velocities are $u^{\ast },v^{\ast }$: 
\begin{equation}
u^{\ast }=u+[\omega \cdot (v-u)]\omega \;\text{ and }\; v^{\ast }=v-[\omega
\cdot (v-u)]\omega \text{.}  \label{def:pre-after collision velocity}
\end{equation}%
Together with the transport part, the Boltzmann equation reads%
\begin{equation}
\left( \partial _{t}+v\cdot \nabla _{x}\right) f=Q(f,f)\ \text{in }\mathbb{R}%
^{1+6}.  \label{eqn:generic Boltzmann}
\end{equation}%
Equation (\ref{eqn:generic Boltzmann}) sits between the law of motion of the
microscopic particles (atoms, molecules, \textellipsis) and the macroscopic
phenomena (and is hence called mesoscopic). For example, the Navier--Stokes
equations and the Euler equations in fluid dynamics are special limits of (%
\ref{eqn:generic Boltzmann}). (See, e.g. \cite{GM1961,LL1987} for a physics
oriented and \cite{SR09} for a mathematics oriented discussion.) Moreover,
it naturally carries the nondecreasing quantity, entropy, written as 
\begin{equation*}
S=-\int f\ln fdv
\end{equation*}%
in the Gibbs entropy form. That is, equation (\ref{eqn:generic Boltzmann})
is time irreversible while the laws of motion of the microscopic particles
are time reversible. Thus, rigorously justifying the emergence of (\ref%
{eqn:generic Boltzmann}) from first principles of the microscopic mechanics
connects the microscopic and macroscopic theories and establishes time
irreversibility, and is hence a fundamental problem.

This problem is also specifically mentioned in Hilbert's explanation to his
6th problem. The explicit statement of the 6th problem of Hilbert, raised in
1900, reads:

\begin{quote}
Mathematical Treatment of the Axioms of Physics. The investigations on the
foundations of geometry suggest the problem: To treat in the same manner, by
means of axioms, those physical sciences in which already today mathematics
plays an important part; in the first rank are the theory of probabilities
and mechanics.
\end{quote}

Hilbert gave the further explanation of this problem and its possible
specific forms as the following:

\begin{quote}
As to the axioms of the theory of probabilities, it seems to me desirable
that their logical investigation should be accompanied by a rigorous and
satisfactory development of the method of mean values in mathematical
physics, and in particular in the kinetic theory of gases. \textellipsis %
Boltzmann's work on the principles of mechanics suggests the problem of
developing mathematically the limiting processes, there merely indicated,
which lead from the atomistic view to the laws of motion of continua.
\end{quote}

One version of the contemporary understanding of the program is the
following.

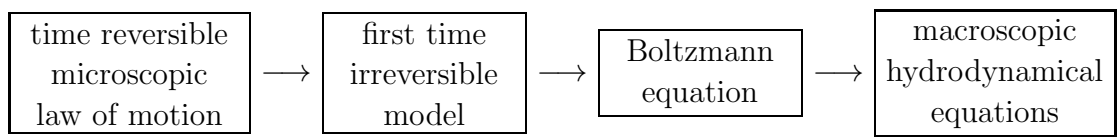
\begin{minipage}{\textwidth}
\bigskip

\fbox{\begin{minipage}{0.18\textwidth}
\begin{center}
time reversible \\
microscopic \\
law of motion
\end{center}
\end{minipage}}
$\longrightarrow$
\fbox{\begin{minipage}{0.15\textwidth}
\begin{center}
first time \\
irreversible model
\end{center}
\end{minipage}}
$\longrightarrow$
\fbox{\begin{minipage}{0.15\textwidth}
\begin{center}
Boltzmann \\
equation
\end{center}
\end{minipage}}
$\longrightarrow$
\fbox{\begin{minipage}{0.18\textwidth}
\begin{center}
macroscopic \\
hydrodynamical \\
equations
\end{center}
\end{minipage}}
\bigskip
\end{minipage}

Here, there is not a definitive answer to the so called ``first time
irreversible model", and that is why some surveys (e.g. \cite{A2017,UV2015})
about time irreversibility state that deriving the Boltzmann equation for a
long time (from $N$-body systems) would prove the time irreversibility but
proving time irreversibility may not require the derivation of the Boltzmann
equation.

In this paper, we consider part of the problem. We consider the rigorous
derivation of a version of the Boltzmann equation (\ref{eqn:generic
Boltzmann}) and hence the time irreversibility from quantum $N$-body
dynamics.

It is evident that the ``mechanics" and the ``atomistic view" in Hilbert's
explanation should be Newtonian mechanics (it may mean specifically the hard
sphere model, but that is much less clear) as quantum mechanics was still
one of the two ``clouds"\footnote{%
There are many versions of the record in many different textbooks and
papers. But all of them says one ``cloud" became relativity and the other
``cloud" became quantum mechanics.} described by Lord Kelvin in April of the
same year 1900. Again at the same year, Planck produced the now called
Planck's law and quantum theory was born. Hence, Hilbert could not have had
accounted for another development in his 1900 statement, that is, the
existence of a more accurate time-reversible but probabilistic microscopic
theory, since the probabilistic interpretation of the wave function was only
raised by Born \cite{B1926}\ in 1926.\footnote{%
The main object of study in our paper, the Schr\"{o}dinger equations, are
deterministic, and the probabilistic feature is an interpretation of the
solutions.} While quantum theory has become the concurrently most accurate
microscopic model, any modern understanding of Hilbert's program should
allow quantum theory, in which the probabilistic feature is an axiom but the
time-irreversibility problem stands still, to be a microscopic starting
point, as well as Newtonian mechanics.

On the other hand, deriving (\ref{eqn:generic Boltzmann}) from quantum
theory sounds like an even more challenging problem as (\ref{eqn:generic
Boltzmann}) is largely regarded as classical. Because it has been proven
many times that there is no obvious gap between quantum theory and classical
mechanics, there is no reason not to use a more accurate model. For example,
one could derive the Euler equations directly from quantum $N$-body dynamics
without passing through any Boltzmann equations \cite{CSZ1,CSZ2,CSWZ,GP1}.%
\footnote{%
See also \cite{EMY,OVY,QY} for derivations of fluid equations directly from
systems with probabilistic features.} Interestingly, in 1877, Boltzmann also
suggested that the energy of a particle could be discrete; in 1896, he
replied to Zermelo \cite[Vol. III, paper 119]{BoltzmannCollection} that
``the Maxwell distribution law (and hence the Boltzmann theory) is not a
theorem from ordinary mechanics\footnote{%
Judging from the time, this ``ordinary mechanics" by Boltzmann should also
mean Newtonian mechanics.} and cannot be proven from mechanical
assumptions". \footnote{%
These replies to Zermelo were recommended as ``superbly clear and right on
the money'' by Lebowitz \cite{L1999}.} In fact, the large number of quantum
``dice'',\footnote{%
Dice as described by Einstein.} might contribute to the time
irreversibility. (See also, the on-going development of quantum
thermodynamics and hence entropy, in for example, \cite{HZ,Z}.)

We denote the microscopic interparticle interaction by the 2-body radial
interaction $\phi $. As assumed in basic kinetic theory, the interaction
should be repelling ($\phi >0$) if the distance between 2 particles is small
while it should be attractive ($\phi <0$) if the distance is large. We
further assume that $\int \phi =0$ to avoid zero momentum exchanges. While $%
\int \phi \neq 0$ is certainly allowed in quantum mechanics\footnote{%
It is even one of the starting points, e.g. the hydrogen atom.}, our desired
limit is classical and hence the model has to be compatible with it.

Put $\mathbf{y}_{N}=\left( y_{1},...,y_{N}\right) \in \mathbb{R}^{3N}$ as
the position vector of $N$ particles in $\mathbb{R}^{3}$. We write the $N$%
-body Hamiltonian at the weak-coupling scaling as%
\begin{equation}
H_{N,\varepsilon }=\sum_{j=1}^{N}-\frac{\varepsilon ^{2}}{2}\triangle
_{y_{j}}+\sum_{1\leqslant i<j\leqslant N}\sqrt{\varepsilon }\phi (\frac{%
y_{i}-y_{j}}{\varepsilon })\text{ with }N=\varepsilon ^{-3}.
\label{hamiltonian: weak coupliing}
\end{equation}%
Denoting the $k$-th marginal densities by $\gamma _{N}^{(k)}$, we consider
the $N$-body dynamic%
\begin{equation}
i\varepsilon \partial _{t}\gamma _{N}^{(N)}=\left[ H_{N,\varepsilon },\gamma
_{N,\varepsilon }^{(N)}\right] ,  \label{eqn:von Neumann equation}
\end{equation}%
under the normalization condition that%
\begin{equation*}
\gamma _{N}^{(k)}(\mathbf{y}_{k};\mathbf{y}_{k})\geqslant 0\text{ and }\int
\gamma _{N}^{(k)}(\mathbf{y}_{k};\mathbf{y}_{k})d\mathbf{y}_{k}=1,
\end{equation*}%
and the symmetric condition that, $\forall \sigma \in S_{k}$, the
permutation group of $k$-elements%
\begin{equation*}
\gamma _{N}^{(k)}(\mathbf{y}_{k};\mathbf{y}_{k}^{\prime })=\gamma
_{N}^{(k)}(y_{\sigma (1)}...y_{\sigma (k)};y_{\sigma (1)}^{\prime
}...y_{\sigma (k)}^{\prime })\text{ and }\gamma _{N}^{(k)}(\mathbf{y}%
_{k}^{\prime };\mathbf{y}_{k})=\overline{\gamma _{N}^{(k)}(\mathbf{y}_{k};%
\mathbf{y}_{k}^{\prime })},
\end{equation*}%
if $\gamma _{N}^{(k)}$ is put in kernel form.

We need the $\left( x,v\right) $ phase space picture (to state the limit in
the common format\footnote{%
We actually work in 4 representations of the problem.}), thus we consider
the Wigner transform $\left\{ f_{N}^{(k)}\right\} $ of $\left\{ \gamma
_{N}^{(k)}\right\} $ defined as%
\begin{eqnarray}
f_{N}^{(k)}(\mathbf{x}_{k},\mathbf{v}_{k}) &=&W_{\varepsilon }(\gamma
_{N}^{(k)})  \label{eqn:Wigner} \\
&=&\left( \frac{1}{2\pi }\right) ^{3k}\int e^{i\mathbf{\xi}_{k}\mathbf{v}%
_{k}}\gamma _{N}^{(k)}(\mathbf{x}_{k}+\frac{\varepsilon }{2}\mathbf{\xi}_{k};%
\mathbf{x}_{k}-\frac{\varepsilon }{2}\mathbf{\xi}_{k})d\mathbf{\xi}_{k}. 
\notag
\end{eqnarray}%
Via direct computation, we have that the family $\left\{ f_{N}^{(k)}\right\} 
$ solves the quantum Bogoliubov-Born-Green-Kirkwood-Yvon (BBGKY) hierarchy

\begin{equation}
(\partial _{t}+\sum_{j=1}^{k}v_{j}\cdot \nabla _{x_{j}})f_{N}^{(k)}=\frac{1}{%
\sqrt{\varepsilon }}A_{\varepsilon }^{(k)}f_{N}^{(k)}+\frac{N}{\sqrt{%
\varepsilon }}B_{\varepsilon }^{(k+1)}f_{N}^{(k+1)},
\label{hierarchy:quantum BBGKY in differential form}
\end{equation}%
where the two inhomogeneous terms are given by 
\begin{equation*}
A_{\varepsilon }^{(k)}=\sum_{1\leqslant i<j\leqslant k}A_{i,j}^{\varepsilon
},\text{if }k\geqslant 2\text{, and }A_{\varepsilon }^{(1)}=0,
\end{equation*}%
\begin{equation*}
B_{\varepsilon }^{(k+1)}=\sum_{j=1}^{k}B_{j,k+1}^{\varepsilon },
\end{equation*}%
with%
\begin{eqnarray}
A_{i,j}^{\varepsilon }f_{N}^{(k)} &=&\frac{-i}{\left( 2\pi \right) ^{3}}%
\sum_{\sigma =\pm 1}\sigma \int_{\mathbb{R}^{3}}dh\text{ }e^{\frac{ih\cdot
(x_{i}-x_{j})}{\varepsilon }}\hat{\phi}(h)  \label{eqn:Aij_epsilon} \\
&&f_{N}^{(k)}\left( t,\mathbf{x}_{k},v_{1}...,v_{i-1},v_{i}-\sigma \frac{h}{2%
},v_{i+1},...,v_{j-1},v_{j}+\sigma \frac{h}{2},v_{j+1},...,v_{k}\right) , 
\notag
\end{eqnarray}%
\begin{eqnarray}
B_{j,k+1}^{\varepsilon }f_{N}^{(k+1)} &=&\frac{-i}{\left( 2\pi \right) ^{3}}%
\sum_{\sigma =\pm 1}\sigma \int_{\mathbb{R}^{3}}dx_{k+1}\int_{\mathbb{R}%
^{3}}dv_{k+1}\int_{\mathbb{R}^{3}}dh\text{ }e^{\frac{ih\cdot (x_{j}-x_{k+1})%
}{\varepsilon }}\hat{\phi}(h)  \label{eqn:Bij_epsilon} \\
&&f_{N}^{(k+1)}\left( t,\mathbf{x}_{k},x_{k+1},v_{1}...,v_{j-1},v_{j}-\sigma 
\frac{h}{2},v_{j+1},...,v_{k+1}+\sigma \frac{h}{2}\right) .  \notag
\end{eqnarray}

This model has been studied by many prominent authors in both the
inhomogeneous and the homogeneous cases. See, for example, \cite%
{BCEP1,BCEP2,BCEP4} by Benedetto, Castella, Esposito, and Pulvirenti and 
\cite{ErdosFock} by Erd\"{o}s, Salmhofer and Yau. One of the original
motivation was the possibility of the emergence of the Uehling-Uhlenbeck
equation \cite{U-U} in the quasi-free case. It was X.Chen and Guo \cite{CY}
who noticed that is not possible formally. Then, for the spatial homogeneous
case, under the quasi-free assumption at the initial condition, T.Chen and
Hott \cite{TCMH} proved that quasi-freeness persists approximately in time
and that the conclusion in \cite{CY} holds. Moreover, they found a density
condition such that the Uehling-Uhlenbeck equation emerges.

We use the casual $A$,$B$ notation for the two inhomogeneous terms in (\ref%
{hierarchy:quantum BBGKY in differential form}) as they individually do not
generate the Boltzmann collision kernel and they tend to zero as $%
\varepsilon \rightarrow 0$ if applied to smooth functions. This is very
different from the usual hierarchy analysis in which a formal limit is first
found in the smooth setting without well-definedness issues, and was raised
as a question in \cite{BCEP1} and more specifically \cite[p.11]{BCEP4}. We
answer this question in detail during the proof. In short, the answer would
be, in order to arise from an $N$-body solution, $f_{N}^{(k)}$ must satisfy
a physically reasonable regularity requirement related to quantum
quasi-freeness, under which the $B$ term will not tend to $0$ (though
without a clear form for the limit due to the effect of the irregular parts
making the hierarchy unbalanced.) Then, another coupling provides a balanced
hierarchy mainly on the core part and a recognizable limit via a special
combination of $A$ and $B$ while the irregular effects vanish. This is also
a reason and benefit that one works in the critical regularity, which also
happens to be physical for this problem.

Another problem of the analysis of hierarchy (\ref{hierarchy:quantum BBGKY
in differential form}) is that\ it is not clear if hierarchy (\ref%
{hierarchy:quantum BBGKY in differential form}) is a well-defined equation
and if $f_{N}^{(k)}$ is nonnegative real-valued immediately from its
definition, being the Wigner transform\footnote{%
If one considers the Husimi transform instead, the marginals would be
nonegative but with more complicated equations. Thus, it is expected that $%
f_{N}^{(k)}$ is nonnegative (or provides nonnegative limits.)} of (\ref%
{eqn:von Neumann equation}). See, for example, the $dv_{k+1}$ integration in
(\ref{eqn:Bij_epsilon}) and the $i$'s in (\ref{eqn:Aij_epsilon}) and (\ref%
{eqn:Bij_epsilon}). It turns out that the functions spaces we are forced to
work in, to be compatible with the quantum quasi-freeness, is also where
hierarchy (\ref{hierarchy:quantum BBGKY in differential form}) is at the
borderline of well-definedness, and this is one aspect of the optimality of
our derivation.

Our target is the following quantum version of (\ref{eqn:generic Boltzmann})%
\begin{equation}
\left( \partial _{t}+v\cdot \nabla _{x}\right) f=Q(f,f)
\label{eqn:QBEquation}
\end{equation}%
in which the quantum collision operator $Q$ is given by%
\begin{eqnarray}
Q(f,g) &\equiv &Q^{+}(f,g)-Q^{-}(f,g)  \label{eqn:collision kernel} \\
&=&\frac{1}{8\pi ^{2}}\int_{\mathbb{R}^{3}}du\int_{\mathbb{S}^{2}}dS_{\omega
}\left\vert \omega \cdot \left( v-u\right) \right\vert \left\vert \hat{\phi}%
\left( \left( \omega \cdot \left( v-u\right) \right) \omega \right)
\right\vert ^{2}  \notag \\
&&\left[ f\left( t,x,v^{\ast }\right) g(t,x,u^{\ast })-f(t,x,v)g(t,x,u)%
\right] .  \notag
\end{eqnarray}%
which would arise as the \textquotedblleft not-so-obvious" mean-field limit
of hierarchy (\ref{hierarchy:quantum BBGKY in differential form}). As
pointed out in \cite{BCEP4}, the classical collision cross section in (\ref%
{eqn:collision kernel}) proves the macroscopic effects of quantum
microscopic interaction, it directly shows a transition rate which is
independent of time and is proportional to the strength of the coupling
between the initial and final densities of states of the system, and it is
hence a representation of the \textquotedblleft Fermi Golden Rule" which is
the quantum version of \textquotedblleft Stosszahlansatz" in the physics
point of view or time irreversibility.

\subsection{The cycle regularity condition\label{s:intro cycle regularity}}

We will work with hierarchy space-time bounds\footnote{%
As pointed out by Cercignani, this is one important part of the hierarchy
analysis.} built to be compatible with the quantum quasi-freeness, which is
required as $f_{N}^{(k)}$ are the Wigner transforms of quantum $N$-body
states, and to be compatible with the single particle norms in the molecular
chaos case $f^{(k)}=f^{\otimes k},$ which happens in the limit. As we are
interested in $N\rightarrow \infty $ behaviors, we look for regularity
conditions uniform in $N$ (or for all large enough $N)$. Though, for each $N$%
, $f_{N}^{(k)}$ could be a very smooth function, the uniform in $N$
regularity of the family is in fact far from being high. Our main working
space and topology, is based on the usual $H_{x}^{r}L_{v}^{2,s}$ norm,
which, when defined for a single particle distribution function, is 
\begin{equation*}
\left\Vert f\right\Vert _{H_{x}^{r}L_{v}^{2,s}}^{2}=\int \left\vert
\left\langle \nabla _{x}\right\rangle ^{r}\left\langle v\right\rangle
^{s}f(x,v)\right\vert ^{2}dxdv.
\end{equation*}%
Considering the limit, it is usually customary to work with space-time
bounds like 
\begin{equation}
\sup_{t\in \left[ 0,T\right] }||(\dprod\limits_{j=1}^{k}\left\langle \nabla
_{x_{j}}\right\rangle ^{r}\left\langle v_{j}\right\rangle
^{s})f_{N}^{(k)}||_{L_{x,v}^{2}}\leqslant C^{k},
\label{bound:customary bound}
\end{equation}%
but the situation we consider in this paper requires us to work with a more
general and less regular setting as assuming (\ref{bound:customary bound})
will lead to a trivial limit if $r>1$ (let alone it is in fact unphysical)
or an ill-posed limiting equation (let along things becoming not
well-defined) if simply putting $r<1$. In fact, the customary bound only
holds at most for $r=\frac{3}{4}$ even when $f_{N}^{(k)}$ tends to a local
Maxwellian.\footnote{$r=\frac{3}{4}$ also happens to be the regularity we
need to prove the emergence of the collision kernel.} We would need to work
under a physical $N$-body regularity implied by the quantum quasi-free
condition. As noted in much literature\footnote{%
See, for example, \cite[p.4]{BCEP4}, \cite[p.2]{ErdosFock}, and \cite{Hu1}.}
the physical situation is quasi-free (or restricted quasi-free).

Denote $\pi $ a permutation in the $k$-permutation group $S_{k}$, we assume
the quantum BBGKY hierarchy sequence (\ref{hierarchy:quantum BBGKY in
differential form}) has the wave packet structure that, at each time $t$, it
decomposes into a sum 
\begin{equation}
f_{N}^{(k)}(t)=\sum_{\pi \in S_{k}}f_{N,\pi }^{(k)}(t)
\label{assumption:wave packet sum}
\end{equation}%
(where each $f_{N,\pi }^{(k)}(t)$ is not necessarily quasi-free) and we
assume the regularity that, there exists a constant $C$ so that 
\begin{equation}
\forall N\geqslant 0,\;t\in \lbrack 0,T]\,,\;\pi \in S_{k}\,,\quad \text{ we
have }\quad \Vert f_{N,\pi }(t)\Vert _{X_{\pi }}\leq C^{k}
\label{energy bound:pi}
\end{equation}%
where the norm \thinspace $X_{\pi }$ is specified below. In the
decomposition (\ref{assumption:wave packet sum}), the term $f_{N,I}$
corresponding to the identity $\pi =I\text{ }$is called the \emph{core }%
term, and the analysis ultimately shows that it is the only term that has a
nontrivial limit as $N\rightarrow \infty $, and the terms $f_{N,\pi }(t)$
with $\pi \neq I$ are called ``irregular" (in terms of the customary bounds (%
\ref{bound:customary bound})) or ``cycle" terms. It is not assumed that the
decomposition (\ref{assumption:wave packet sum}) is unique.

For the core term $f_{N,I}^{(k)}$, we assume the $H_{x}^{1+}(L_{v}^{2,\frac{1%
}{2}+}\cap L_{v}^{\infty ,2+}\cap L_{v}^{1})$ energy condition, that is,
there exists a $C>0$ independent of $k$ or $N$ such that, $\forall k>0$, we
have 
\begin{eqnarray}
\sup_{t\in \left[ 0,T\right] }||(\dprod\limits_{j=1}^{k}\left\langle \nabla
_{x_{j}}\right\rangle ^{1+}\left\langle v_{j}\right\rangle ^{\frac{1}{2}%
+})f_{N,I}^{(k)}||_{L_{x,v}^{2}} &\leqslant &C^{k},
\label{energy bound:Hx^1(v^1/2andL1)} \\
\sup_{t\in \left[ 0,T\right] }||(\dprod\limits_{j=1}^{k}\left\langle \nabla
_{x_{j}}\right\rangle ^{1+}\left\langle v_{j}\right\rangle
^{2+})f_{N,I}^{(k)}||_{L_{x}^{2}\left( L_{v}^{\infty }\right) } &\leqslant
&C^{k}, \\
\sup_{t\in \left[ 0,T\right] }||(\dprod\limits_{j=1}^{k}\left\langle \nabla
_{x_{j}}\right\rangle ^{1+})f_{N,I}^{(k)}||_{L_{x}^{2}L_{v}^{1}} &\leqslant
&C^{k}.
\end{eqnarray}

On the other hand, though they are considered errors to the core term $%
f_{N,I}^{(k)}$ and will vanish in the $N\rightarrow \infty $ limit process,
the cycle terms $f_{N,\pi }(t)$ with $\pi \neq I$, are irregular in the
sense that (\ref{bound:customary bound}) only holds with $r=\frac{3}{4}$
even in the local Maxwellian case. However, they can be up to $H^{1+}$ in
what we call the permutation coordinates depending on the permutation $\pi $%
. That is why we call them the cycle terms.

The cycle coordinates are actually coordinates such that the operator $%
\sum_{j=1}^{k}v_{j}\cdot \nabla _{x_{j}}$ is invariant. To explain it, it is
easier to go back to the Schr\"{o}dinger picture (\ref{hamiltonian: weak
coupliing})-(\ref{eqn:von Neumann equation}) in which the hyperbolic
Laplacian reads 
\begin{equation*}
\varepsilon \sum_{j=1}^{k}\left( \triangle _{y_{j}^{\prime }}-\triangle
_{y_{j}}\right) .
\end{equation*}%
Given $\pi \in S_{k}$, define the cycle coordinates by 
\begin{equation}
p_{j}^{\pi }=y_{j}^{\prime }+y_{\pi (j)}\qquad q_{j}^{\pi }=(y_{j}^{\prime
}-y_{\pi (j)})/\epsilon ,  \label{def:cycle coordinates}
\end{equation}%
Then the hyperbolic Laplacian and the evolution become%
\begin{eqnarray*}
\frac{1}{4}\epsilon \sum_{j=1}^{k}(\Delta _{y_{j}^{\prime }}-\Delta
_{y_{j}}) &=&\sum_{j=1}^{k}\nabla _{p_{j}}\cdot \nabla _{q_{j}} \\
e^{\frac{1}{4}i\epsilon t(\Delta _{\boldsymbol{y}_{k}^{\prime }}-\Delta _{%
\boldsymbol{y}_{k}})} &=&e^{it\nabla _{\boldsymbol{p}_{k}}\cdot \nabla _{%
\boldsymbol{q}_{k}}}.
\end{eqnarray*}%
Recalling the fact that, the Wigner transform (\ref{eqn:Wigner}) is but a
Fourier transform with shifts, if we denote $\xi $ the inverse Fourier
variable to $v$, then $\sum_{j=1}^{k}v_{j}\cdot \nabla _{x_{j}}$ in (\ref%
{hierarchy:quantum BBGKY in differential form}) becomes $\sum_{j=1}^{k}%
\nabla _{\xi _{j}}\cdot \nabla _{x_{j}}$ with 
\begin{equation*}
x_{j}=y_{j}^{\prime }+y_{j}\qquad \xi _{j}=(y_{j}^{\prime }-y_{j})/\epsilon 
\text{.}
\end{equation*}%
That is, the $(x,v)$ coordinate is the special $\pi =I$ case of the cycle
coordinates (\ref{def:cycle coordinates}).

For $\pi \neq I$, we assume the $H^{1+}$ cycle regularity that, there exists 
$C>0$ independent of $\pi $, $k$ or $N$, such that 
\begin{equation}
\sup_{t\in \left[ 0,T\right] }||(\dprod\limits_{j=1}^{k}\left\langle \nabla
_{p_{j}^{\pi }}\right\rangle ^{1+})(W_{\varepsilon }^{-1}f_{N,\pi
}^{(k)})||_{L_{p^{\pi },q^{\pi }}^{2}}\leqslant C^{k}
\label{energy bound:high freq H^1 divergent}
\end{equation}%
where $W_{\varepsilon }^{-1}$ is the inverse Wigner transform which takes
the ($x,v$) picture back to the ($y,y^{\prime }$) picture and the derivative
and integrations are in (\ref{def:cycle coordinates}) coordinates. It is
pretty low and not in the usual coordinates. But there are some extra good
things about the cycle / irregular terms.

The cycle terms are representations of a type of quantum symmetry. We can
quantify such symmetries. Given $\pi \neq I$ in which $(i_{\pi },j_{\pi })$
is a 2-cycle in its cycle decomposition, we define the
interchange/substitution operator $T_{(i_{\pi },j_{\pi })}$ acting on the
function $\tilde{g}^{(k)}$ with $(\mathbf{p}_{k}^{\pi }$,$\mathbf{q}%
_{k}^{\pi })$ variables by%
\begin{equation*}
T_{(i_{\pi },j_{\pi })}\tilde{g}^{(k)}(\mathbf{p}_{k}^{\pi },\mathbf{q}%
_{k}^{\pi })=\tilde{g}^{(k)}((i_{\pi },j_{\pi })\mathbf{p}_{k}^{\pi
},(i_{\pi },j_{\pi })\mathbf{q}_{k}^{\pi }),
\end{equation*}%
that is, $T_{(i_{\pi },j_{\pi })}$ interchanges the $(p_{i_{\pi }}^{\pi
},q_{i_{\pi }}^{\pi })$ and $(p_{j_{\pi }}^{\pi },q_{j_{\pi }}^{\pi })$
variables inside the function $\tilde{g}^{(k)}$. Let 
\begin{equation*}
W_{\varepsilon }^{-1}f_{N,\pi }^{(k)}(t,\mathbf{y}_{k},\mathbf{y}%
_{k}^{\prime })=\tilde{g}_{N}^{(k)}(t,\mathbf{p}_{k}^{\pi },\mathbf{q}%
_{k}^{\pi })\text{,}
\end{equation*}%
then $f_{N,\pi }^{(k)}$ has the symmetry measurement\footnote{%
It is a measurement of symmetry as (\ref{energy bound:cycle symmetry}) would
be zero if $f_{N,\pi }^{(k)}$ is fully symmetric.} that,%
\begin{eqnarray}
&&\sup_{t}\left\Vert \left\langle \nabla _{p_{j_{\pi }}^{\pi }}\right\rangle
^{\frac{3}{4}+}\left( \tilde{g}_{N}^{(k)}(t,\mathbf{p}_{k}^{\pi },\mathbf{q}%
_{k}^{\pi })-T_{(i_{\pi },j_{\pi })}\tilde{g}^{(k)}(t,\mathbf{p}_{k}^{\pi },%
\mathbf{q}_{k}^{\pi })\right) \right\Vert _{L_{p_{j_{\pi }}^{\pi },q_{j_{\pi
}}^{\pi }}^{2}}  \label{energy bound:cycle symmetry} \\
&\lesssim &\varepsilon ^{\frac{1}{2}}\sup_{t}||(\dprod\limits_{j=1}^{k}\left%
\langle \nabla _{p_{j}^{\pi }}\right\rangle ^{1+})(W_{\varepsilon
}^{-1}f_{N,\pi }^{(k)})||_{L_{p^{\pi },q^{\pi }}^{2}}  \notag
\end{eqnarray}

Moreover, due to the high number of collisions of our model and hence high
number of interactions between the cycle terms and the core term, we have
many times in $\left[ 0,T\right] $ such that the symmetry strengthens
(actually ``jitters"\footnote{%
``Jitters" as in electrical engineering (EE).} as explained in \S \ref{S3}.)
To realize such a symmetry strengthening (gain), we would like to assume
that for each $\varepsilon >0$, there is a subset $E_{\varepsilon }\subset %
\left[ 0,T\right] $ such that, given any open interval $I_{\varepsilon }$
with length at least $\varepsilon ^{\frac{1}{2}}$, we have $E_{\varepsilon
}\cap I_{\varepsilon }$ is nonempty, and for each $t_{0}\in E_{\varepsilon }$%
, we have the symmetry difference strengthens to, 
\begin{eqnarray}
&&\left\Vert \left\langle \nabla _{p_{j_{\pi }}^{\pi }}\right\rangle ^{\frac{%
3}{4}+}\left( \tilde{g}_{N}^{(k)}(t_{0},\mathbf{p}_{k}^{\pi },\mathbf{q}%
_{k}^{\pi })-T_{(i_{\pi },j_{\pi })}\tilde{g}^{(k)}(t_{0},\mathbf{p}%
_{k}^{\pi },\mathbf{q}_{k}^{\pi })\right) \right\Vert _{L_{(p_{j_{\pi
}}^{\pi },q_{j_{\pi }}^{\pi })}^{2}}
\label{energy bound:high freq H^1 goodness} \\
&\lesssim &\varepsilon ^{\mu }\sup_{t\in \left[ 0,T\right]
}||(\dprod\limits_{j=1}^{k}\left\langle \nabla _{p_{j}^{\pi }}\right\rangle
^{1+})(W_{\varepsilon }^{-1}f_{N,\pi }^{(k)})||_{L_{p^{\pi },q^{\pi }}^{2}}.
\notag
\end{eqnarray}
for some $\mu >\frac{1}{2}$.

The assumptions above certainly need more explanation. We will justify why
the regularities cannot be higher or lower and how the permutation
coordinates come in via sharp technical estimates, ill-posedness results,
and matching computations regarding constructions of the local Maxwellian
from quantum $N$-body solutions, throughout the whole paper and more
specifically in \S \ref{S:PC}-\ref{S:PCE} and \S \ref{S3}. In particular,
our assumptions hold (and are implied) if $f_{N}^{(k)}(t)$ is generalized
quasi-free, or just quasi-free. That is, our assumptions is general and
covers the physical cases.

\subsection{Statement of the Main Theorem}

\begin{theorem}[Main Theorem]
\label{thm:main}Let $\left\{ f_{N}^{(k)}\right\} $ be the $N$-body dynamic
given by the BBGKY hierarchy (\ref{hierarchy:quantum BBGKY in differential
form}) with a pair interaction $\phi $ in Schwartz class with zero
integration and $\hat{\phi}$ vanishes at zero up to the 1st order. Assume

(i) The initial datum to (\ref{hierarchy:quantum BBGKY in differential form}%
) is of asymptotic molecular chaos, that is, for some one-particle density $%
f_{0}(x,v)\in H_{x}^{1+}(L_{v}^{2,\frac{1}{2}+})$, for all $k$, we have, 
\begin{equation}
f_{N}^{(k)}|_{t=0}\rightharpoonup \dprod\limits_{j=1}^{k}f_{0}(x,v)\text{
weakly in }L_{\boldsymbol{x}_k}^{2}L_{\boldsymbol{v}_k}^{2}\text{ as }%
N\rightarrow \infty .  \label{eqn:initial molecular chaos}
\end{equation}

(ii) In the time interval $\left[ 0,T\right] $, $f_{N}^{(k)}\geqslant 0$.

(iii) In the time interval $\left[ 0,T\right] ,$ the sequence $\left\{
\left\{ f_{N}^{(k)}\right\} _{k=1}^{N}\right\} _{N=1}^{\infty }$ satisfies
the cycle regularity condition specified in \S \ref{s:intro cycle regularity}%
.

Then we have propogation of chaos that $\forall k$ and $\forall t\in \left[
0,T\right] $%
\begin{equation}
f_{N}^{(k)}(t)\rightharpoonup \dprod\limits_{j=1}^{k}f(t,x,v)\text{ weakly
in }L_{\boldsymbol{x}_k}^{2}L_{\boldsymbol{v}_k}^{2}\text{ as }N\rightarrow
\infty ,  \label{limit:propogation of chaos}
\end{equation}%
where $f\in H_{x}^{1+}(L_{v}^{2,\frac{1}{2}+})$ solves (\ref{eqn:QBEquation}%
) with initial condition $f|_{t=0}(x,v)=f_{0}(x,v).$ Moreover, if $f_{0}\in
L_{x,v}^{1}$, then the concluded limit $\lim_{N\rightarrow \infty
}f_{N}^{(k)}(t)\in L_{x,v}^{1}$.
\end{theorem}

We have not assumed the well-posedness of (\ref{eqn:QBEquation}) in $[0,T]$
in the statement of Theorem \ref{thm:main}. On the one hand, we prove a
local $H_{x}^{1+}(L_{v}^{2,0+})\cap H_{x}^{\frac{1}{2}+}(L_{v}^{2,\frac{1}{2}%
+})$ unconditional well-posedness result for (\ref{eqn:QBEquation}). On the
other hand, the limit process actually generates a $H_{x}^{1+}(L_{v}^{2,%
\frac{1}{2}+})$ solution $\lim_{N\rightarrow \infty }f_{N}^{(1)}(t)$ to (\ref%
{eqn:QBEquation}) and is hence the only possible everywhere in time solution
subject to the initial condition at this regularity. This is not surprising
due to two reasons: (1) it is known that if the initial datum takes the form
of (\ref{eqn:initial molecular chaos}) then the first marginal of the
infinite Boltzmann hierarchy (\ref{hierarchy:Boltzmann}) is a formal
solution to (\ref{eqn:QBEquation}); (2) due to our method, if, \emph{a priori%
}, the $H_{x}^{1+}(L_{v}^{2,\frac{1}{2}+})$ norm of solution(s) to (\ref%
{eqn:QBEquation}) (if exists) is known to be finite in $\left[ 0,T\right] $,
one can construct a $C([0,T],H_{x}^{1+}(L_{v}^{2,\frac{1}{2}+}))$ solution
to (\ref{eqn:QBEquation}).\footnote{%
See \cite{CSZ4} for the global well-posedness at such low regularities which
is also sharp as ill-posedness starts to happen below it.}

From the above discussion, our proof of Theorem \ref{thm:main} does not
directly rely on (\ref{eqn:initial molecular chaos}) to conclude a limit
exists for the BBGKY sequence. That is indeed the case.

\begin{corollary}
\label{thm:main1}Let $\left\{ f_{N}^{(k)}\right\} $ be the $N$-body dynamic
given by the BBGKY hierarchy (\ref{hierarchy:quantum BBGKY in differential
form}) with a pair interaction $\phi $ in Schwartz class with zero
integration and $\hat{\phi}$ vanishes at zero up to the 1st order. Assume
(ii) and (iii) in Theorem \ref{thm:main} and

(i') The initial datum to (\ref{hierarchy:quantum BBGKY in differential form}%
) has a $L_{x}^{2}L_{v}^{2}$ weak limit which is in $H_{x}^{1+}(L_{v}^{2,%
\frac{1}{2}+})$, that is there is a family $\left\{ f_{0}^{(k)}\right\} $
such that 
\begin{equation}
f_{N}^{(k)}|_{t=0}\rightharpoonup f_{0}^{(k)}\text{ weakly in }L_{\mathbf{x}%
_{k}}^{2}L_{\mathbf{v}_{k}}^{2}\text{ as }N\rightarrow \infty ,
\end{equation}%
for some $f_{0}^{(k)}$ satisfying the customary bound (\ref{bound:customary
bound}) with $r=1+$ and $s=\frac{1}{2}+$. Then there is a unique family $%
\left\{ f^{(k)}\right\} $ in $\left[ 0,T\right] $ such that $\forall k$ and $%
\forall t\in \left[ 0,T\right] $ 
\begin{equation}
f_{N}^{(k)}(t)\rightharpoonup f^{(k)}\text{ weakly in }L_{\mathbf{x}%
_{k}}^{2}L_{\mathbf{v}_{k}}^{2}\text{ as }N\rightarrow \infty ,
\end{equation}%
and $\left\{ f^{(k)}\right\} $ solves the infinite quantum Boltzmann
hierarchy (\ref{hierarchy:Boltzmann}) with initial condition $f^{(k)}|_{t=0}(%
\mathbf{x}_{k},\mathbf{v}_{k})=f_{0}^{(k)}(\mathbf{x}_{k},\mathbf{v}_{k})$
and satisfies the customary bound (\ref{bound:customary bound}) with $r=1+$
and $s=\frac{1}{2}+$.
\end{corollary}

\subsection{Optimality of the Main Theorem}

The derivation in Theorem \ref{thm:main} is optimal. The low regularity
setting (\ref{assumption:wave packet sum})-(\ref{energy bound:high freq H^1
goodness}) is physically required and a mathematically critical spot needed
for a long time derivation as we explan below. We remark that assuming
instead, finite second moments like energy and variance, will not lower the
requirement. On the other hand, interestingly, assuming much more smoothness
does not help to simplify the argument and might even run into the problem
that it forces (\ref{hierarchy:quantum BBGKY in differential form}) to have
a trivial limit as pointed out in \cite{BCEP1,BCEP4} (and is in fact
nonphysical). As usual, the critical regularity argument tells most of the
story and high regularity theory actually relies on it.

It is self-evident that the justification of any mean-field limits should be
settled in settings pertinent to their physical backgrounds which have made
these problems fundamental. We calculate in \S \ref{S3} that, for the
quantum $N$-body density $f_{N}^{(k)}$ to converge to a local Maxwellian,
the most well-known stable solution to (\ref{eqn:QBEquation}), as $%
N\rightarrow \infty $, the regularity of $f_{N}^{(k)}$ checks our
assumptions in \S \ref{s:intro cycle regularity} and cannot be higher. In
fact, we prove in Lemma \ref{L:basic-B} and Corollary \ref{C:12nonlin}, that
if $f_{N}^{(k)}$ was a bit more regular, then the $B$ term in (\ref%
{eqn:Bij_epsilon}) tends to zero and the limit of (\ref{hierarchy:quantum
BBGKY in differential form}) is trivial. Thus our regularity assumption is
physical and is even at the critical physical regularity. It is usually
expected that the critical physical regularity is set higher than the lowest
regularity\footnote{%
Depending on the systems, some might be known, some might still be unknown.}
reachable by mathematics proofs after one sees the formal limit assuming a
high regularity. That is not the case here and we are in a physically
enforced low regularity situation. Moreover, as suggested by the first
ill-posedness result regarding the Boltzmann equation in \cite{CH10} by the
authors, and as explained below, the problem studied in this paper happens
to have these two regularities coincide and thus creates a double
criticality and makes things extremely delicate and difficult.\footnote{%
If one considers the particle system as a dynamical system, it's known that
the analysis gets more and more rigid as regularity drops.}

After showing the regularity assumptions are physical and cannot be higher,
one needs to have a corresponding and compatible low regularity theory for
the $N$-body analysis and the limiting Boltzmann equation (\ref%
{eqn:QBEquation}) so that, in the end, one could identify the $N$-body limit
with (\ref{eqn:QBEquation}). To start, one needs to prove hierarchy (\ref%
{hierarchy:quantum BBGKY in differential form}) and equation (\ref%
{eqn:QBEquation}) are well-defined as PDEs under this low regularity. A good
parallel problem to the general audience is the well-definedness problem of
the free boundaries in the hard sphere models \cite{Lanford,CIP,SaintRaymond}
when $N$ is finite. To this end, the $L_{v}^{\infty ,2}$ part can be seen
from hierarchy (\ref{hierarchy:quantum BBGKY in differential form}); the $%
L_{v}^{2,\frac{1}{2}+}$ part can be better understood in the analysis of the
limiting hierarchy (\ref{hierarchy:Boltzmann}) and equation (\ref%
{eqn:QBEquation}) at the collision operator (\ref{eqn:collision kernel}),
mainly due to the \textquotedblleft all integrals are well-defined"
confusion caused by the Schwartz assumption on the interaction potential $%
\phi $; away from the quantum $N$-body requirement, the $H_{x}^{1+}$ part is
actually caused by both the well-definedness of the $x$-trace in the
hierarchy definition and the collision operator (\ref{eqn:collision kernel}%
), but it is easier to realize it from the definition of the solution. We
refer readers to \S \ref{S:preparation}-\ref{sec:Convergence} for the
details.

Once the well-definedness of hierarchy (\ref{hierarchy:quantum BBGKY in
differential form}) and equation (\ref{eqn:QBEquation}) is settled, we need
to prove the solutions being discussed to hierarchy (\ref{hierarchy:quantum
BBGKY in differential form}) and equation (\ref{eqn:QBEquation}) solve these
PDEs everywhere in time. That is, an almost everywhere in time with respect
to some measure solution is not acceptable since it would be kind of weird
that the laws in physics only hold with respect to some measure. It is
better to see this from the well-posedness theory of equation (\ref%
{eqn:QBEquation}) in \S \ref{sec:illposedness}: equation (\ref%
{eqn:QBEquation}) is in fact locally well-posed in $H_{x}^{1+}(L_{v}^{2,0+})$
but such solutions would only solve the PDE a.e. in time, while $%
H_{x}^{1+}(L_{v}^{2,0+})\cap H_{x}^{\frac{1}{2}+}(L_{v}^{2,\frac{1}{2}+})$
is the borderline regularity to solve the PDE everywhere in time. ($%
H_{x}^{1+}(L_{v}^{2,0+})\cap H_{x}^{\frac{1}{2}+}(L_{v}^{2,\frac{1}{2}+})$
is also the borderline regularity for equation (\ref{eqn:QBEquation}) to be
unconditionally well-posed.\footnote{%
T. Kato raised the unconditional well-posedness notion in 1995 \cite{Kato}
when strong but a.e. in time solutions became popoluar. So far, all
unconditional well-posedness, even for NLS and NLW, were proved for
everywhere in time solutions.})

Last but not least, the limiting equation must be well-posed in the working
space as well. Here, by well-posedness, we mean existence, uniqueness and
the uniform continuity of the datum to solution map. It is well known that,
for large $k$, limits like (\ref{limit:propogation of chaos}) are not stable
in norms against small perturbations. If the solution map for equation (\ref%
{eqn:QBEquation}) is not uniformly continuous, then the targeted limit (\ref%
{limit:propogation of chaos}) and the believed approximation, hierarchy (\ref%
{hierarchy:quantum BBGKY in differential form}), could change very much as $%
N\rightarrow \infty $ and invalidate the limit process. In \S \ref%
{sec:illposedness}, we prove equation (\ref{eqn:QBEquation}) is locally
well-posed in $H_{x}^{s}(L_{v}^{2,0+})$ for all $s>1$ and ill-posed in $%
H_{x}^{s}(L_{v}^{2,s_{1}})$ for any $s<1$ and all $s_{1}\in \mathbb{R}^{+}$.%
\footnote{%
The ``bad" impolsion solution family we consider has uniformly bounded
second moments, so assuming in addition finite second moments will not
improve.}\footnote{%
Though implosion solutions cause blow ups for the compressible Euler
equations, as proved here and first by \cite{CH10}, they are no problem for (%
\ref{eqn:QBEquation}), except norm deflations causing ill-posedness at low
regularity. We thank Jiajie Chen for discussion related to this matter.}
That is, we are indeed working at the borderline of well-posedness of
equation (\ref{eqn:QBEquation}). We remark that, the low regularity
well-posedness for (\ref{eqn:QBEquation}) here is required to identify the $%
N $-body limit and (\ref{eqn:QBEquation}) because, physically, the $N$-body
limit process has to be done at such a low regularity though the limit
itself could be very smooth like the local Maxwellian. It just also happens
that it is also the sharp well/ill-posedness separation point.\footnote{%
Apparently, the $N$-body analysis requires such low regularity
well-posedness results and thus provides the physical background for these
results as well. T. Chen, Denlinger, and Pavlovi\'{c} \cite%
{CDP19a,CDP19b,CDP21} might be the first to systematically use dispersive
analysis and reach such low regularity well-posedness for Boltzmann type
equations, which is then proved to be the sharp well/ill-posedness
separation points in \cite{CH10} by X.C. and J.H. See \cite{CSZ3,CSZ4} for
further developments along this line.}

The $L^{2}$ weak limit in Theorem \ref{thm:main} cannot be upgraded to
strong either. The cycle terms with $\pi \neq I$, are $O(1)$ in $L^{2}$ (the
Jacobian is $O(1)$). The weak limit is only possible because the $O(1)$
sized quasifree terms are geometrically stretched. More specifically, an $%
O(1)$ ball in $(\boldsymbol{p}_k, \boldsymbol{q}_k)$ space (coordinates
depending on $N$) gets stretched and flattened upon transformation to $(%
\boldsymbol{x}_k,\boldsymbol{\xi}_k )$ in such a way that its intersection
with an $O(1)$ ball in $(\boldsymbol{x}_k,\boldsymbol{\xi}_k)$ space has
volume tending to $0$.

We do admit that, it would be better to prove the energy assumptions (or
(restricted) quasi-freeness)\footnote{%
This is noted to be very difficult in \cite[p.2]{ErdosFock}. Progress has
been made in a similar situation in \cite{TCMH}.} in Theorem \ref{thm:main}
for a general class of initial data and could be considered a drawback. As
explained and proven, these assumptions are the minimal and necessary
requirement and the physical cases fit exactly here. As long as the
assumption remains valid (which is the case for (generalized) quasifree
solutions), the long time derivation of (\ref{eqn:QBEquation}), and hence
the time-irreversibility from the microscopic law of motion, is now
justified. We also recall again Boltzmann's 1896 comment \cite[Vol. III,
paper 119]{BoltzmannCollection} that \textquotedblleft the Maxwell
distribution law (and hence the Boltzmann theory) is not a theorem from
ordinary mechanics and cannot be proven from mechanical assumptions". The
derivation in this paper automatically applies to any future work using
higher regularities while a higher regularity is always more difficult to
prove to hold than a lower one and may not be true physically (and might
lead to trivial limits). Again, we point out in \S \ref{S3} that local
Maxwellians, and their small perturbations, are qualified data for Theorem %
\ref{thm:main} in the limit, and we expect that, counting in the damping
effects,\footnote{%
The equation derived from this physical model is with an angular cut-off, so
we do not expect hypoelliptic structure(s) coming from the non-cut-off case.
How to derive a non-cut-off equation which generated many nice work, for
example, \cite{AMUXY,GS11}, is also open.} we will be able to prove the
energy assumption for this class of datum and that is our next step.

\subsection{Incorporation of the hard-sphere and the inverse power law
models \label{Subsec:CommentOnTheHardSphere}}

Equation (\ref{eqn:QBEquation}) incorporates the celebrated hard-sphere
model and effectively the inverse power law model of power $-1$ (the $\gamma
=-1$ model) at the same time. It extends both models and interconnects them
together as temperature changes, as predicted in theoretical physics and
observed in experiements.

For demonstration purposes, let us assume $\left\vert \hat{\phi}\right\vert
^{2}$ is a bump function supported in $\left[ \frac{\mathfrak{c}_{1}}{2},2%
\mathfrak{c}_{2}\right] $ and is $1$ in $\left[ \mathfrak{c}_{1},\mathfrak{c}%
_{2}\right] $. Then the collision kernel in (\ref{eqn:collision kernel})
equals exactly the hard sphere collision kernel $\left\vert v\cdot \omega
\right\vert $ for $v\in \left[ \mathfrak{c}_{1},\mathfrak{c}_{2}\right] $.
For the $\gamma =-1$ model part, it is not so obvious. It is easier to see
this from the loss term. Consider the angular integral 
\begin{equation*}
\int_{\mathbb{S}^{2}}dS_{\omega }\left\vert \omega \cdot v\right\vert
\left\vert \hat{\phi}\left( \left( \omega \cdot v\right) \omega \right)
\right\vert ^{2},
\end{equation*}%
for large $\left\vert v\right\vert $ in which $\left\vert \hat{\phi}%
\right\vert ^{2}$ is like a bump function. For the integrand to be nonzero,
one would need $\left\vert \omega \cdot \frac{v}{\left\vert v\right\vert }%
\right\vert =\cos \theta \sim 1/\left\vert v\right\vert \ll 1$. By the
geometry, if say $\frac{v}{\left\vert v\right\vert }$ is pointing at the
north pole, and $\omega $ is on $\mathbb{S}^{2}$, then the set on $\mathbb{S}%
^{2}$ almost perpendicular to $\frac{v}{\left\vert v\right\vert }$ is
basically the equator times the width $1/\left\vert v\right\vert $. Hence,
the measure of the $\omega $ integration set is like $1/\left\vert
v\right\vert $, thus 
\begin{equation*}
\int_{\mathbb{S}^{2}}dS_{\omega }\left\vert \omega \cdot v\right\vert
\left\vert \hat{\phi}\left( \left( \omega \cdot v\right) \omega \right)
\right\vert ^{2}\sim 1/\left\vert v\right\vert .
\end{equation*}%
That is, formally, we expect the $\gamma =-1$ model behaviors for large $%
\left\vert v\right\vert $. These two models have been tested countless times
in their regimes of validity. The famous hard-sphere model is in the regime
of moderate/atmospherical temperature, in which, it also implies the
classical ideal gas laws and Newton's cooling law,\footnote{%
Both of them are known to be invalid outside of some temperature range so
there is no contradiction.}\footnote{%
The regime of validity of the ideal gas laws is well-known. The failure of
the exponential to equilibrium law at high temperature might be first
documented by Dalton \cite{R}.} and its rigorous derivation from Newtonian $%
N $-body dynamics has been studied in many work, see, for example, \cite%
{Lanford,CIP,SaintRaymond} using analytic methods (hence does not need \emph{%
a priori} bounds and is up to a sufficiently small time). The $\gamma =-1$
model applies to high temperature situations while, to the best of the
authors' knowledge, this paper offers the first rigorous derivation of an
inverse power law model (effectively and with a cut-off though) from $N$%
-body systems.

Let us recall the fact that the mean speed of the molecules in a gas is
proportional to the temperature and the speed is distributed fairly close to
the mean speed with variance also proportional to the temperature. That is,
equation (\ref{eqn:QBEquation}) interconnects the hard-sphere and $\gamma =-1
$ models in the sense (at least formally) that it behaves like the
hard-sphere model at moderate/atmospherical temperature and the $\gamma =-1$
model at high\footnote{%
The ``high" here is relative as it is well below the noticeable ionization
temperature $\sim 3\times 10^{3}$K at which point some Vlasov theory comes
into play.} temperature.

The paper \cite{CHe} by X.\thinspace C. and L.\thinspace He justifies
mathematically the above observation on $\mathbb{T}^{3}$. For initial
condition near a Maxwellian, they prove that, for large\footnote{%
This \textquotedblleft large" is also relative as particle speeds in (\ref%
{eqn:QBEquation}) and the hard-sphere model are all 4-5 digits smaller than
the speed of light in reality and one could say they both have bounded
collision kernels in practice. However, the unboundedness of the hard-sphere
collision kernel has indeed motivated many innnovations in mathematics
theory and have propelled mathematics forward.} $\mathfrak{c}_{2}$,
solutions to (\ref{eqn:QBEquation}) and the hard-sphere model are close for
a long time depending on $\mathfrak{c}_{2}$, and if $\mathfrak{c}_{1}^{-1},%
\mathfrak{c}_{2}\rightarrow \infty $, (\ref{eqn:QBEquation}) converges to
the hard sphere model. Moreover, they prove that, for fixed $\mathfrak{c}%
_{2} $ and fixed background temperature, solutions to (\ref{eqn:QBEquation})
will tend to equilibrium at a exponential rate (signature of the hard
potentials and Maxwellian particles) for a long time depending on $\mathfrak{%
c}_{2}$ and the background temperature, then at a polynomial rate (signature
of soft potentials) determined by the datum's energy and the $\gamma =-1$
model. This hints at the physical fact that the specific heat capacity (and
hence the adiabatic index) of matter increases as temperature increases
(hence it takes longer to reach equilibrium).\footnote{%
H$_{2}$O could be the most checked example, though this has also been
observed for He.}

That is, (\ref{eqn:QBEquation}), derived from quantum $N$-body dynamics, in
its regime of validity, rigorously interconnects and enhances the
hard-sphere and the inverse power law models. A further goal is to uncover
the physical meanings of $\mathfrak{c}_{1}$ and $\mathfrak{c}_{2}$ or, more
specifically, to determine the empirical correspondence of regions in which $%
| \hat{\phi}|^{2}\sim O(1)$ or $| \hat{\phi}|^{2}\sim o(1)$. In general,
there are many interpretations and models about the microscopic
interactions, but experimental science tends to verify their effects and
implications instead of providing a direct observational window. We plan to
investigate how the heat capacity depends on $\mathfrak{c}_{1}$ and $%
\mathfrak{c}_{2}$ in (\ref{eqn:QBEquation}). To gain more insight into this
topic, a toy problem could be testing numerically if one could match more
digits of the heat capacity using Boltzmann theory by varying $\mathfrak{c}%
_{1}$ and $\mathfrak{c}_{2}$.

\begin{flushleft}
\textbf{Acknowledgements.} X.C. was supported in part by the NSF grant
DMS-2005469 and the Simons Fellowship \#916862. J.H. was partially supported
by the NSF grant DMS-2055072.
\end{flushleft}

\section{Proof of the Main Theorem\label{sec:proof of main}}

\subsection{The quantum set-up and the trivial limit puzzle}

Write the propagator as $S^{(k)}(t)\equiv
\dprod\nolimits_{j=1}^{k}e^{-tv_{j}\cdot \nabla _{x_{k}}}$, we first rewrite
(\ref{hierarchy:quantum BBGKY in differential form}) in Duhamel form: 
\begin{eqnarray}
f_{N}^{(k)}(t_{k}) &=&S^{(k)}(t_{k})f_{N}^{(k)}(0)+\frac{1}{\sqrt{%
\varepsilon }}\int_{0}^{t_{k}}S^{(k)}(t_{k}-t_{k+1})A_{\varepsilon
}^{(k)}f_{N}^{(k)}(t_{k+1})dt_{k+1}  \label{hierachy:PreBBGKYinDuhamel} \\
&&+\frac{N}{\sqrt{\varepsilon }}\int_{0}^{t_{k}}S^{(k)}(t_{k}-t_{k+1})B_{%
\varepsilon }^{(k+1)}f_{N}^{(k+1)}(t_{k+1})dt_{k+1}.  \notag
\end{eqnarray}%
Iterating relation (\ref{hierachy:PreBBGKYinDuhamel}) once, we have%
\begin{eqnarray}
&&f_{N}^{(k)}(t_{k})  \label{hierarchy:RealBBGKY} \\
&=&S^{(k)}(t_{k})f_{N}^{(k)}(0)+\frac{1}{\sqrt{\varepsilon }}%
\int_{0}^{t_{k}}S^{(k)}(t_{k}-t_{k+1})A_{\varepsilon
}^{(k)}f_{N}^{(k)}(t_{k+1})dt_{k+1}  \notag \\
&&+\frac{N}{\sqrt{\varepsilon }}\int_{0}^{t_{k}}S^{(k)}(t_{k}-t_{k+1})B_{%
\varepsilon }^{(k+1)}S^{(k+1)}(t_{k+1})f_{N}^{(k+1)}(0)dt_{k+1}  \notag \\
&&+\frac{N}{\varepsilon }\int_{0}^{t_{k}}S^{(k)}(t_{k}-t_{k+1})B_{%
\varepsilon
}^{(k+1)}(\int_{0}^{t_{k+1}}S^{(k)}(t_{k+1}-t_{k+2})A_{\varepsilon
}^{(k+1)}f_{N}^{(k+1)}(t_{k+2})dt_{k+2})dt_{k+1}  \notag \\
&&+\frac{N^{2}}{\varepsilon }\int_{0}^{t_{k}}S^{(k)}(t_{k}-t_{k+1})B_{%
\varepsilon
}^{(k+1)}(\int_{0}^{t_{k+1}}S^{(k+1)}(t_{k+1}-t_{k+2})B_{\varepsilon
}^{(k+2)}f_{N}^{(k+2)}(t_{k+2})dt_{k+2})dt_{k+1}  \notag \\
&\equiv
&S^{(k)}(t_{k})f_{N}^{(k)}(0)+R_{N}^{2(k)}f_{N}^{(k)}(t_{k})+R_{N}^{3(k+1)}f_{N}^{(k+1)}(t_{k})+(R_{N}^{4(k+1)}+Q_{N}^{(k+1)})f_{N}^{(k+1)}(t_{k})
\notag \\
&&+R_{N}^{5(k+2)}f_{N}^{(k+2)}(t_{k}).  \notag
\end{eqnarray}%
where%
\begin{eqnarray}
R_{N}^{4(k+1)}(t_{k}) &\equiv &\sum_{\ell=1}^{k}\sum_{\substack{ 1\leqslant
i<j\leqslant k+1  \\ (i,j)\neq (\ell,k+1)}}%
R_{N,l,i,j}^{4(k+1)}f_{N}^{(k+1)}(t_{k})  \label{def:R4} \\
&=&\frac{N}{\varepsilon }\sum_{\ell=1}^{k}\sum_{\substack{ 1\leqslant
i<j\leqslant k+1  \\ (i,j)\neq (\ell,k+1)}}%
\int_{0}^{t_{k}}S^{(k)}(t_{k}-t_{k+1})B_{l,k+1}^{\varepsilon }  \notag \\
&&(\int_{0}^{t_{k+1}}S^{(k)}(t_{k+1}-t_{k+2})A_{i,j}^{\varepsilon
}f_{N}^{(k+1)}(t_{k+2})dt_{k+2})dt_{k+1},  \notag
\end{eqnarray}%
and%
\begin{eqnarray}
Q_{N}^{(k+1)}f_{N}^{(k+1)}(t_{k}) &\equiv
&\sum_{j=1}^{k}Q_{j,k+1}^{\varepsilon }f_{N}^{(k+1)}(t_{k})  \label{def:Q_N}
\\
&=&\sum_{j=1}^{k}\frac{N}{\varepsilon }%
\int_{0}^{t_{k}}S^{(k)}(t_{k}-t_{k+1})B_{j,k+1}^{\varepsilon }  \notag \\
&&(\int_{0}^{t_{k+1}}S^{(k)}(t_{k+1}-t_{k+2})A_{j,k+1}^{\varepsilon
}f_{N}^{(k+1)}(t_{k+2})dt_{k+2})dt_{k+1}.  \notag
\end{eqnarray}

Hierarchies (\ref{hierachy:PreBBGKYinDuhamel}) and (\ref{hierarchy:RealBBGKY}%
) are equivalent by definition. On the one hand, (\ref{hierarchy:RealBBGKY})
is longer and more complicated than (\ref{hierachy:PreBBGKYinDuhamel}). On
the other hand, testing the formal limit under smooth condition is usually
the first thing to try in dealing with mean-field limits. Interestingly, for
a very smooth $f_{N}^{(k)}$, the $A,B$ terms in (\ref%
{hierachy:PreBBGKYinDuhamel}) actually tend to zero. In fact, one just needs
the customary bound (\ref{bound:customary bound}) to hold with $r>1$, then $%
N\varepsilon ^{-\frac{1}{2}}B$ and $\varepsilon ^{-\frac{1}{2}}A$ tend to $0$
as $N\rightarrow \infty $, that is, (\ref{hierachy:PreBBGKYinDuhamel})
yields a trivial limit at regularity higher than $H_{x}^{1+}$. At the same
time, (\ref{hierarchy:RealBBGKY}) will produce a nontrivial limit with (\ref%
{eqn:QBEquation}) as the mean-field equation if tested using smooth
functions. While the Boltzmann equation (\ref{eqn:QBEquation}) certainly
does not agree with a trivial transport equation, the iteration of (\ref%
{hierachy:PreBBGKYinDuhamel}) yielding (\ref{hierarchy:RealBBGKY}) also
looks especially suspicious. This is the trivial limit puzzle stated in \cite%
{BCEP1} and more specifically \cite[p.11]{BCEP4}.

This puzzle should be the first thing to solve in the derivation of (\ref%
{eqn:QBEquation}) and we answer it in detail in \S \ref{S3} and \S \ref%
{S:PCE}. It turns out, the quantum $N$-body solutions to (\ref%
{hierachy:PreBBGKYinDuhamel}) coming from (\ref{eqn:von Neumann equation})
is not smooth at all. The simplest and nontrivial example is to check the
regularity using the local Maxwellian. On the one hand, we prove that $%
f_{N}^{(k)}$ can never be a direct tensor product of local Maxwellians
unless $N\rightarrow \infty $. On the other hand, we compute the
expectations of the Sobolev regularity of the quantum $N$-body states
converging to a direct tensor product of local Maxwellians and we found the
regularity conditions in \S \ref{s:intro cycle regularity}. This computation
thus proves that the $N$-body solution cannot be more regular for the
problem considered in this paper.\footnote{%
One might argue that quantum states far away from the local Maxwellian might
be even rougher, but it would result in the limiting equation being
ill-posed and hence unlikely physically.} It also turns out that, if the
solution is quantum quasi-free, then it satisfies the regularity conditions.
That is, the low regularity setting we seek is not only mathematical, but
also physical.

Under such physical regularity conditions, we prove that $N\varepsilon ^{-%
\frac{1}{2}}B$ is actually an $O(1)$ quantity. That is, the limit is not
trivial for solutions under the problem's setting. On the other hand, the
limit of $N\varepsilon ^{-\frac{1}{2}}B$ is unclear, thus we iterate (\ref%
{hierachy:PreBBGKYinDuhamel}) once to its equivalent form (\ref%
{hierarchy:RealBBGKY}) from which we can conclude a cleanly formatted limit.%
\footnote{%
Such iterations of basic hierarchy yielding a managable limit has a similar
scenario in the NLS case, which is only formally realized by X.C. \& J.H. in 
\cite{CH9}, (see also \cite{CSWZ}), was implicitly used in \cite{C3,CH2,CH5}%
, and might be first hinted in \cite{CP5} by T. Chen \& Pavlovic.} Thence
the trivial limit puzzle is solved. (A more quantitative puzzle can be
provided once one finishes \S \ref{S:preparation}.)

\subsection{Four sides of the Boltzmann equation\label{sec: 4 sides}}

As Theorem \ref{thm:main} is optimal, its proof is extremely rigid and
tight. One needs to explore and invoke every $\varepsilon$ room of play to
avoid failing the proof or losing the optimality (and the physicality at the
same time as they are tied together). One thing we do is to explore all four
sides of the Boltzmann equation (not counting the cycle coordinates (\ref%
{def:cycle coordinates})).

The usual side $(t,x,v)$ is associated with the kinetic transport operator $%
\partial _{t}+v\cdot \nabla _{x}$. To be very honest, they are not exactly
used anywhere in our proof. We use this side solely to state the results in
their usual format.

The $(t,x,\xi )$ or the $\sim$ side is associated with the hyperbolic
symmetric Schr\"{o}dinger operator $i\partial _{t}+\nabla _{\xi }\cdot
\nabla _{x}$. We denote functions on this side with a $\sim$. It is obtained
by applying the inverse Fourier transform to $v$ of the usual $(t,x,v)$
side. It is not new and has been used by many authors before. (See \cite%
{CDP19a,CDP19b} for the Wigner transform version.) In our setting, this side
sort of undoes the Wigner transform and gives a more Schr\"{o}dinger-like
equation for which there are Strichartz estimates \cite{KT98}. One
preconception about this side is that it requires the Fourier transform of
the collision kernel, which is only explicit in some cases, to work.
However, the vantage point of the equation $i\partial _{t}+\nabla _{\xi
}\cdot\nabla_{x}$ suggests to consider estimates in the $X_{s,b}$ spaces
which are carried out on the dual side and thus do not require computing the
Fourier transform of the collision kernel. We study the remainder term $%
R_{N}^{2(k)}$ in (\ref{hierarchy:RealBBGKY}) and the well/ill-posedness of (%
\ref{eqn:QBEquation}) on this side. For the quasifree terms, the cycle
coordinates $(p^\pi,q^\pi)$ associated to a permutation $\pi$ are also most
naturally related to the $(x,\xi)$ coordinates.

The $(\tau ,\eta ,v)$ or the $\wedge $ side is associated with the
multiplication operator $\tau +v\cdot \eta $. We denote functions on this
side with a $\wedge $. It is obtained by applying the Fourier transform on $%
(t,x)$ of the usual $(t,x,v)$ side or on the whole $(t,x,\xi)$ of the $%
(t,x,\xi )$ side. That is, we construct the $X_{s,b}$ Fourier restriction
norm spaces as in \cite{BE1,B1,KM,RR1} for (\ref{eqn:QBEquation}), and prove
multilinear estimates regarding $Q$ in $X_{s,b}$, to obtain bilinear
improvements over the Strichartz estimates in the $(t,x,\xi )$ side. As one
works on the characteristic surface on this side, it does not need the
Fourier transform of the collision kernel and gives a direct treatment of
the problem in terms of multilinear estimates without oscillation. This
perspective was used in \cite{CH10} to obtain the first separation of
well/ill-posedness of Boltzmann type equations. A drawback is that direct
analysis can involve numerous cases and technical geometrical
decompositions. We use this side to deal with one of the difficulties, the
so called remainder term $R_{N}^{5(k)} $ in (\ref{hierarchy:RealBBGKY}).

The $(t,\eta ,\xi )$ or the $\vee$ side is associated with the intertwined
kinetic transport operator $\partial _{t}+\eta \cdot \nabla _{\xi }$. We
denote functions on this side with a $\vee $, and it is obtained by applying
the Fourier transform to $x$ of the $(t,x,\xi )$ side. We find the
representation of the $B$ operator on this side more convenient, especially
when $B$ is composed with other operators (Duhamel, an $A$, or another $B$).
This side has oscillatory terms but most of them vanish in the $\epsilon \to
0$ limit, meaning that these oscillations cannot be important for uniform
estimates. We carry out many key $N$-body estimates and estimates regarding
the collision kernel on this side. It is still a bit mysterious why this is
effective while the other transport side $(x,v)$ typically involves
oscillations that cannot be ignored. Even more unexpected is that we found,
in the uniqueness proof, an application of the so-called kinetic transport
Strichartz estimates \cite{BBGL14,BP01,Ov11} on this side although they were
originally conceived for application on the usual $(t,x,v)$ transport side.
This side might need further investigation in the future.

Before we start the proof of the main theorem, an interesting point to
reflect upon is whether these four sides provide new information or if they
can shed light on the classical hard-sphere model. The $(\tau ,\eta ,v)$ and 
$(t,\eta ,\xi )$ sides seem very difficult to define in the hard-sphere
model, due to the freely moving billiards and their free boundaries in
space. That is, even in terms of techniques, the quantum problem here is
indeed very different from the classical problem. Each has its own beauty
although the quantum problem should ultimately incorporate the classical
problem in a limiting regime as explained in \S \ref%
{Subsec:CommentOnTheHardSphere}. Maybe one could understand the hard sphere
model better by changing the model a bit to define the $(\tau ,\eta,v) $ and 
$(t,\eta ,\xi )$ sides.

\subsection{Proof of the Main Theorem}

Having given an overview of the trivial limit puzzle and the existence of
solutions with suitable regularity, we now turn to a discussion of the proof.

Due to the double criticality -- the $N$-body solution cannot carry higher
regularity and the limiting Boltzmann equation cannot admit lower regularity
as one has to stay physical and the other one has to remain well-posed --
there are few off-the-shelf lemmas available to employ. A problem sharing
similar ``endpoint" flavor is the derivation of the $H^{1}$-critical NLS at $%
H^{1}$-regularity \cite{CH7,CH8} by X.C. and J.H., which was constructed on
the scaffold of prior work \cite%
{CHPS,CP2,CP5,C2,C3,CH1,CH2,CH3,CH4,CH5,CH6,CH9,CSZ,CSWZ,CPU,KSS,KM1,GSS,HS1,HS2,S1}
by many authors, on the derivation of the $H^{1}$-subcritical NLS at $H^{1}$%
-regularity with the hierarchy method, pioneered by Erd\"{o}s, Schlein, and
Yau \cite{E-S-Y2,E-S-Y5,E-S-Y3}. However, even at the critical level, the
structures, methodologies, and analysis of the Boltzmann equation and NLS
are totally different. For example, unlike the NLS cases, we do not have a
``subcritical" case to refer to, and the Boltzmann equations' well-posedness
threshold does not lie at the scaling criticality.\footnote{%
Though \cite{CH10} is on the Boltzmann equation using dispersive techniques,
it is simultaneously an example on how different the Boltzmann equations and
the usual dispersive equations, NLS/NLW, are.} Thus, we have to build much
of our analysis from scratch.

\subsubsection{Step 1. Preparation of the $N$-body analysis and estimates of
term sizes in \S \protect\ref{S:preparation}}

Recall that the BBGKY family $\mathcal{F}=\{f_{N}=\{f_{N}^{(k)}\}_{k=1}^{N}%
\}_{N=1}^{\infty }$ satisifies hierarchy (\ref{hierarchy:RealBBGKY}) under
the (cycle) regularity conditions specified in \S \ref{s:intro cycle
regularity}, and the main object of study is the limit of $f_N$ if there is
one. The first step is analyzing the sizes of the terms in (\ref%
{hierarchy:RealBBGKY}), which are all highly oscillatory integrals on the $%
(x,v)$ side, so that one can arrange the proof later. This is where the four
sides in \S \ref{sec: 4 sides}\ natually come in for estimating and
comparing. This is done in \S \ref{S:preparation}.

We warm up by estimating the $\varepsilon ^{-\frac{1}{2}}A$ operator in the $%
(x,\xi )$ side in Lemma \ref{L:basic-A}. We can immediately see that the $%
(x,\xi)$ representation is more convenient than the $(x,v)$ representation.
For $\frac34$ derivatives which is applicable to the whole $f_{N}^{(k)}$, we
prove that the $\varepsilon^{-\frac{1}{2}}A$ operator and $R_{N}^{2(k)}$
tend to zero strongly. We then prove that the $N\varepsilon^{-\frac{1}{2}}B$
operator also tends to zero strongly if one has $H_{x}^{1+}$ derivative, in
Lemma \ref{L:basic-B}, by working on the ($\eta ,\xi $) side. The operator $%
B $ expressed on the $(\eta,\xi)$ side only involves oscillation with an $%
\epsilon$ coefficient that vanishes as $\epsilon \to 0$, and thus the
estimate reduces to managing positive weights, and is readily shown to be
sharp. Lemma \ref{L:basic-B} together with \S \ref{S:PC}-\ref{S:PCE}, and \S %
\ref{S3} settles the trivial limit puzzle and justifies the low regularity
assumption at which $N\varepsilon ^{-\frac{1}{2}}B$ is $O(1)$ but without a
clear limit. (As assumption of uniform smoothness of the densities gives a
zero limit and thus hinders a formal computation of the limit.) We can now
legitimately go to the next iteration (\ref{hierarchy:RealBBGKY}) and find
the limit by tending to the main difficult terms, $Q_{N}^{(k+1)}$, $%
R_{N}^{4(k+1)}$, and $R_{N}^{5(k+2)}$ which are part of the technical
highlights.

We prove in Proposition \ref{P:QEstimates} that the $Q_{N}^{(k+1)}$ operator
is strongly bounded at the price of just $3/4$ derivative, so that we can
take the $N\rightarrow \infty $ limit later on. Even though the proof is
done in the ($\eta ,\xi $) side which admits a convenient representation for 
$B$, the proof is more difficult than the preliminary estimates mentioned
above mainly due to the role of the time integration coming from the Duhamel
operator sandwiched between the $B$ and $A$ operators. Rescaling this time
variable produces an $\epsilon$ gain factor at the expense of leaving a
rescaled time integral that must be carried out effectively over the whole
real line. It turns out when $L^p_{\xi}$ norms are brought to the inside by
Minkowski's integral inequality, scaling produces a power of this rescaled
time that is integrable provided $p=3\pm$. Proposition \ref{P:QEstimates} is
at the same time a foundational estimate for the limit collision operator (%
\ref{eqn:collision kernel for hierarchy}), which we will use for the
analysis of the infinite hierarchy and the limiting equation. We recall that
the needed $3/4$ $x$-derivative in Proposition \ref{P:QEstimates} is the
borderline regularity satisfied by the $N$-body solution as a whole instead
of only the core part. This unexpected match of regularity thus checks again
that our analysis is optimal and physical.

After the proof of Proposition \ref{P:QEstimates} regarding where the
collision kernel should arise, we deal with $R_{N}^{4(k+1)}$ in Proposition %
\ref{P:R4Estimates}. We again work in the ($\eta ,\xi $) side, with a
careful analysis of all the cases, we prove that $R_{N}^{4(k+1)}$ tends to
zero weakly in $L_{x}^{2}$ but strongly in $H_{x}^{-\frac{3}{4}-}$ at the
price of $1/3$ $x$-derivative, hence it applies to both the core term and
the irrgular parts of $f_{N}^{(k)}$. An interesting comment to Propositions %
\ref{P:QEstimates} and \ref{P:R4Estimates} is that they could also be proved
using Strichartz estimates on the $(x,\xi )$ side and this suggests further
investigation of the relation between the ($\eta ,\xi $) and $(x,\xi )$
sides.

For the last term in (\ref{hierarchy:RealBBGKY}), $R_{N}^{5(k+2)}$, we work
in the $(\eta ,v)$ side. $R_{N}^{5(k+2)}$ is the most complicated term in (%
\ref{hierarchy:RealBBGKY}), though it is not particularly bad from the
perspective of the ending estimate, which is in fact better than $%
R_{N}^{4(k+1)}$. The estimate regarding $R_{N}^{5(k+2)}$, despite being a
strong $L^{2}$ estimate, requires the Duhamel operator to hold, that is, it
is actually a dual $X_{s,b}$ type Strichartz estimate in diguise. The direct 
$X_{s,b}$ analysis reveals the exact mechanism of gaining an $\varepsilon $
from the Duhamel iteration sandwiched in the two $B$'s in term $R_{N}^{5(k)}$
and justifies that the $(\eta ,v)$ side is the right place to work it out.
Proposition \ref{P:R5Estimates} records this dual Strichartz estimate%
\footnote{%
To the best of the authors' knowledge, the $X_{s,b}$ analysis for the
Boltzmann equation was started in \cite{CH10} in which there is no dual
Strichartz.} for analysis regarding the Boltzmann equation and proves that $%
R_{N}^{5(k+2)}$ tends to zero for the core term. Then an advantage of $%
X_{s,b}$ techniques surfaces: such a direct frequency argument allows the
derivatives at different variables to be freely redistributed and creates
the flexibility to fit in the irregular part of $f_{N}^{(k)}$.

We can now take the formal limit of (\ref{hierarchy:RealBBGKY}) mainly
regarding $Q_{N}^{(k+1)}$. On the new ($\eta ,\xi $) side, its formal limit
is very obvious and in clean format as no oscillation forms are present. We
can then compute the limit in the ($x,\xi $) side which will be needed
later, and then in the usual $(x,v)$ form, in \S \ref{S:coll-op-deriv}. They
are actually new representations of the collision operator.

Taking the limit also yields Proposition \ref{P:Q0Estimates} which is the $%
\varepsilon =0$ version of Proposition \ref{P:QEstimates}. In fact,
carefully examining the new form of the collision operator in \S \ref%
{S:coll-op-deriv} inspires new estimates which will lead to the optimal
unconditional well-posedness of (\ref{eqn:QBEquation}) in \S \ref%
{sec:uniqueness} and \ref{sec:illposedness}. We record the new estimate as
Proposition \ref{P:general-fixed-time-2}, a fixed-time $L^{p}$ bilinear
estimate. Different from other technical estimates in \S \ref{S:preparation}%
, Proposition \ref{P:general-fixed-time-2} is a $L^{p}$ estimate and hence
we call in some very different and delicate $L^{p}$ harmonic analysis and
the Littlewood-Paley square function\footnote{%
See \cite{Stein}.} to prove it.

We are now left with the size estimates for $R_{N}^{3(k+1)}$, $Q_{N}^{(k+1)}$
and $R_{N}^{5(k+2)}$ applied to the irregular parts, $f_{N,\pi }^{(k)}$ with 
$\pi \neq I$, under the cycle regularity condition, as $R_{N}^{2(k)}$ and $%
R_{N}^{4(k+1)}$ are already compatible with the roughness. We do so in \S %
\ref{S:PC}-\ref{S:PCE}. We 1st set up more suitable notations and provide
examples to the cycle regularity for $\pi \neq I$ in \S \ref{S:PC}. We also
prove (generalized) quantum quasi-freeness implies the regularity
assumptions in \S \ref{s:intro cycle regularity} there. We then compute and
handle $R_{N}^{3(k+1)}f_{N,\pi }^{(k+1)}$ and $Q_{N}^{(k+1)}f_{N,\pi
}^{(k+1)}$ with the cycle regularity in complete detail in \S \ref{S:PCE}.
In particular, we prove (\ref{hierachy:PreBBGKYinDuhamel}) has a $O(1)$ $B$%
-term while (\ref{hierarchy:RealBBGKY}) has a tending to zero $%
R_{N}^{3(k+1)} $ terms as the extra iteration provides the chance of hitting
a minuscule better symmetry spot in the time interval due to the high number
of collisions. We omit the handling of the $R_{N}^{5(k+2)}f_{N,\pi }^{(k+2)}$
term as, it should be clear that it follows a similar pattern we have
demonstrated in other estimates.

So far, in \S \ref{S:preparation}, we have completed the full picture of the
sizes and inner mechanism of (\ref{hierarchy:RealBBGKY}) and ready to put
them to use in \S \ref{sec:CompactnessConvergence}-\ref{sec:uniqueness}. The
calculation in \S \ref{S:preparation} also provides a quantitative answer to
the trivial limit puzzle: (1) physically, $f_{N}^{(k)}$ has 2 parts, the
core term and the ``goes-to-zero" but irregular part; (2) the irregular
part, though it goes to $0$, is not small under the $B$ operator so (\ref%
{hierachy:PreBBGKYinDuhamel}) is unbalanced and has no clear limit; (3) $%
Q_{N}$ applied to the irregular part is small hence (\ref%
{hierarchy:RealBBGKY}) is balanced and effectively a hierarchy of the core
term from which the limit can be obtained and is nontrivial.

\subsubsection{Step 2. Compactness in \S \protect\ref%
{sec:CompactnessConvergence}}

With the preparation done in \S \ref{S:preparation}, we define a metric
space $(\Lambda ^{\ast },\rho )$ to study the limit based on the $%
H_{x}^{1+}L_{v}^{2,\frac{1}{2}+}$ space. Different from usual, due to the $%
\pi \neq I$ cycle terms, the sequence $\mathcal{F}$ is not in a very good
space for compactness. Hence, the compactness and convergence parts of this
paper are also unusual. We study the projected sequence $P\mathcal{F}=
\left\{ P_{N}f_{N}=\left\{ P_{N}^{(k)}f_{N}^{(k)}\right\} _{k=1}^{N}\right\}
_{N=1}^{\infty }$ first, then comeback in Step 3 / \S \ref{sec:Convergence}
to conclude $\left\{ f_{N}^{(k)}\right\} $ has the same limit as $%
N\rightarrow \infty $. We prove in \S \ref{sec:CompactnessConvergence} that $%
P\mathcal{F}$ is compact in our metric space $(\Lambda ^{\ast },\rho )$
based on a relatively crude estimate that 
\begin{equation*}
\Vert P_{N}^{(k)}f_{N}^{(k)}\Vert _{C([0,T];H_{\boldsymbol{x}_{k}}^{1+}L_{%
\boldsymbol{v}_{k}}^{2,\frac{1}{2}+})}\leq 2^{-k}C^{k}k!.
\end{equation*}%
Hence, limit points of $P\mathcal{F}$ are well-defined. The above crude
estimate is also good enough for the convergence part, but not enough for
the uniqueness part. Thus we prove a finer estimate 
\begin{equation*}
\Vert P_{N}^{(k)}f_{N}^{(k)}\Vert _{C([0,T];H_{\boldsymbol{x}_{k}}^{1+}L_{%
\boldsymbol{v}_{k}}^{2,\frac{1}{2}+})}\leq 2^{-k}C^{k}\Big(1+\sum_{\substack{
\pi \in S_{k}  \\ \pi \neq I}}\epsilon ^{(\ell (\pi )-m(\pi ))/2}\Big)
\end{equation*}%
in Lemma \ref{L:PCcutoff} for the $\pi \neq I$ cycle terms so that we can
conclude a proper regularity bound of the limit points.

\subsubsection{Step 3. Convergence and the emergence of the collision kernel
in \S \protect\ref{sec:Convergence}}

With the preparation done in \S \ref{S:preparation}, we prove that every
limit point $\left\{ f^{(k)}\right\} $ of $P\mathcal{F}$ in $(\Lambda ^{\ast
},\rho )$ coming from \S \ref{sec:CompactnessConvergence} satisfies the
infinite Boltzmann hierarchy 
\begin{equation}
\left( \partial _{t}+\mathbf{v}_{k}\cdot \nabla _{\mathbf{x}_{k}}\right)
f^{(k)}=Q^{(k+1)}f^{(k+1)}  \label{hierarchy:Boltzmann}
\end{equation}%
where the collision term can be decomposed into 
\begin{equation*}
Q^{(k+1)}f^{(k+1)}\equiv \sum_{j=1}^{k}Q_{j,k+1}f^{(k+1)}
\end{equation*}%
if we write the collision operator pieces into the gain/loss terms,%
\begin{eqnarray}
Q_{j,k+1}f^{(k+1)} &\equiv &Q_{j,k+1}^{+}f^{(k+1)}-Q_{j,k+1}^{+}f^{(k+1)}
\label{eqn:collision kernel for hierarchy} \\
&=&\frac{1}{8\pi ^{2}}\int_{\mathbb{R}^{3}}dv_{k+1}\int_{\mathbb{S}%
^{2}}dS_{\omega }\left\vert \omega \cdot \left( v-u\right) \right\vert
\left\vert \hat{\phi}\left( \left( \omega \cdot \left( v-u\right) \right)
\omega \right) \right\vert ^{2}  \notag \\
&&f(\mathbf{x}_{k},x_{k+1},v_{1},...,v_{j-1},v_{j}^{\ast
},v_{j+1},...v_{k},v_{k+1}^{\ast })-f(\mathbf{x}_{k},x_{k+1},\mathbf{v}%
_{k+1}).  \notag
\end{eqnarray}%
Then, together with Lemma \ref{L:PCcutoff}, we conclude for some $C>0$ that 
\begin{equation*}
\sup_{t\in \lbrack 0,T]}\Vert f^{(k)}\Vert _{H_{\boldsymbol{x}_{k}}^{1+}L_{%
\boldsymbol{v}_{k}}^{2,\frac{1}{2}+}}\leqslant C^{k}
\end{equation*}%
to be ready for the uniqueness argument in Step 4 / \S \ref{sec:uniqueness}.
We also prove that if there is only 1 limit point $\left\{ f^{(k)}\right\} $
of $P\mathcal{F}$, then $f_{N}^{(k)}$ also converges weakly in $%
L_{x}^{2}L_{v}^{2}$ to $f^{(k)}$.

Looking back from this point, one might ask why whether we could have based $%
(\Lambda ^{\ast },\rho )$ on $L_{x,v}^{2}$ and simplify the argument by
removing the projections. This seems not possible -- we need the test
functions in \S \ref{sec:CompactnessConvergence}-\ref{sec:Convergence} to be
very weak, due to fact that a test function at tier $k+1$ is generated from
a smooth test function at tier $k$ and the adjoint of the collision
operator, and the resulting composed tier $k+1$ test function cannot lie in $%
L^2_{x,v}$. The unusual compactness and convergence argument is a minor
novelty compared to the estimates in \S \ref{S:preparation}.

\subsubsection{Step 4. Uniqueness in \S \protect\ref{sec:uniqueness}}

We prove a $H_{x}^{1+}(L_{v}^{2,0+})\cap H_{x}^{\frac{1}{2}+}(L_{v}^{2,\frac{%
1}{2}+})$ uniqueness theorem regarding the infinite Boltzmann hierarchy (\ref%
{hierarchy:Boltzmann}). As the weak limit coming from the core of the $N$%
-body solution and the solution to (\ref{eqn:QBEquation}) both verify the $%
H_{x}^{1+}(L_{v}^{2,0+})\cap H_{x}^{\frac{1}{2}+}(L_{v}^{2,\frac{1}{2}+})$
regularity, on the one hand we conclude that all limits points from Step 2 / 
\S \ref{sec:CompactnessConvergence} actually agree, that is, there is only
one limit point $\{ f^{(k)}\}_{k=1}^\infty$ for $P\mathcal{F}$ and the
sequence actually converges in $(\Lambda ^{\ast },\rho )$; on the other
hand, the limit is determined by 
\begin{equation*}
\left\{ f^{(k)}\right\} =\left\{
\dprod\limits_{i=1}^{k}f(t,x_{i},v_{i})\right\} \text{ for all }t\in \left[
0,T_{0}\right]
\end{equation*}
where $f$ solves (\ref{eqn:QBEquation}), that is, propagation of (quantum)
molecular chaos and the derivation of the Boltzmann equation.

The proof of the uniqueness of hierarchy (\ref{hierarchy:Boltzmann}) is by
adopting the recently perfected scheme for the NLS case in \cite{CH8}.
Though this is the first time such a scheme is fully carried out for the
Boltzmann case, the grand scheme is not new. Arkeryd, Caprino \& Ianiro has
suggested using the Hewitt-Savage theorem to prove uniqueness for Boltzmann
hierarchies in \cite{ACI}. The scheme in \cite{CH8} actually matured from 
\cite{CHPS} which creatively adds the quantum de Finetti theorem to the
Klainerman-Machedon board game and the dispersive multilinear estimates
originated in \cite{KM1} for the NLS case.

The new find in the uniqueness proof is not in its scheme, but its
enactment. Despite the $L^{2}$ disguise in many aspects, the core of the
estimates is, for the first time, a $L^{p}$ estimate, Proposition \ref%
{P:general-fixed-time-2}. It then unexpectedly enables the utilization of
the Strichartz estimates \cite{Ov11} for the intertwined kinetic transport
operator $\partial _{t}+\eta \cdot \nabla _{\xi }$ in the $(t,\eta ,\xi )$
side for the Boltzmann bilinear estimates for the collision operator in \S %
\ref{Subsec:Multilinear Estimates}. In fact, a more surprising aspect and
yet another sign of optimality is, though being a ($\eta ,\xi $) fixed time
estimate, Proposition \ref{P:general-fixed-time-2} lands right by the
(false) endpoint of the intertwined kinetic Strichartz estimates. On the
other hand, we have obtained the first unconditional uniqueness result for
the Boltzmann equation. We believe the uniqueness result is optimal because
the uniqueness space lies exactly at the borderline at which solutions to (%
\ref{eqn:QBEquation}) satisfy the equation everywhere in time instead of
almost everywhere in time, and the fact that so far, no one has been able to
prove an unconditional uniqueness result for almost everywhere in time
solutions for any equation.

The proof of Theorem \ref{thm:main} is finished at this step but for
completeness, we have two more steps.

\subsubsection{Step 5. Justification of physicality / regularity from the
viewpoint of the local Maxwellian, a quasi-free construction in \S \protect
\ref{S3}}

In \S \ref{S3}, we construct some quantum $N$-body solutions converging to
the local Maxwellian as $N\rightarrow \infty $. We also compute the
regularities for these $N$-body solutions. It turns out that even without
interactions and in a supposedly very smooth case, in order to be a quantum $%
N$-body solution, $f_{N}^{(k)}$ can only have the customary uniform-in-$N$
regularity (\ref{bound:customary bound}) up to $r=\frac{3}{4}$, and its
whole regularity takes the form we assumed in \S \ref{s:intro cycle
regularity}. It is surprising that such a low regularity still happens even
for such a basic and supposedly smooth example. Moreover, based on the
combinatorics, we find the inner symmetry and hence the cycle regularity
condition for the irregular part. Last but not least, we include a
computation in \S \ref{S:CollisionEffects} to calculate the frequency of the
changes of symmetry in cycle terms. Recall that we have proved that our
assumed (cycle) regularity which comes from this local Maxwellian
computation is compatible with the quantum quasi-free condition in \S \ref%
{S:PC}. Thus we see Theorem \ref{thm:main} is physical and optimal in the
sense that the $N$-body solution cannot be more regular.

\subsubsection{Step 6. Proof of optimality / well-posedness and
ill-posedness of (\protect\ref{eqn:QBEquation}) in \S \protect\ref%
{sec:illposedness}}

In \S \ref{sec:illposedness}, we prove that (\ref{eqn:QBEquation}) is
locally well-posed in $H_{x}^{s}L_{v}^{2,0+}$ for $s>1$, and ill-posed in $%
H_{x}^{s}L_{v}^{2,0+}$ for $s<1$. Moreover, the solution constructed in $%
H_{x}^{s}L_{v}^{2,0+}$, for $s>1$, is nonnegative and in $L_{xv}^{1}$ if the
initial datum carries these properties. Together with \S \ref{sec:uniqueness}%
, we conclude that (\ref{eqn:QBEquation}) is locally unconditionally
well-posed in $H_{x}^{1+}(L_{v}^{2,0+})\cap H_{x}^{\frac{1}{2}+}(L_{v}^{2,%
\frac{1}{2}+})$ which is also the borderline regularity to allow everywhere
in time solutions to (\ref{eqn:QBEquation}). The proof of the well-posedness
follows from a dispersive bilinear estimate also used in \S \ref%
{sec:uniqueness}\footnote{%
The scheme in \cite{CH8} developed from \cite{CHPS} has this feature which
seems to imply that, at critical regularity, unconditional uniqueness is
always stronger than Strichartz well-posedness. The first work carrying such
a feature is \cite{HTX} regarding the NLS.} based on the $(t,\eta ,\xi )$
side analysis, while the ill-posedness is adapted from \cite{CH10} (See \cite%
{CSZ3} for a more detailed proof.) The new mechanical contribution is the
proof that such low regularity solutions belong to $L^{1}$ if the initial
datum is in $L^{1}$. This step in \S \ref{sec:illposedness}, on the one
hand, proves that there is a solution to (\ref{eqn:QBEquation}) in the space
of the $N$-body limit, on the other hand proves that the regularity of the
limiting Boltzmann equation cannot be lower, which is yet another aspect of
the optimality of our proof.

Steps 5-6 are not involved\footnote{%
This is why some lengthy details in \S \ref{S3}-\ref{sec:illposedness} are
left in another paper or available upon request.} in the proof of Theorem %
\ref{thm:main} but they are part of the inspiration of Theorem \ref{thm:main}%
. With Steps 5-6 done, the whole overview of Theorem \ref{thm:main} is now
completed.


\section{Preparation for $N$-body Analysis}

\label{S:preparation}

\subsection{BBGKY in the four spaces and basic operator estimates}

Recall (\ref{hierachy:PreBBGKYinDuhamel}), the quantum BBGKY hierarchy is,
for $1\leq k\leq N$, 
\begin{equation}
\partial _{t}f_{N}^{(k)}+{\boldsymbol{v}}_{k}\cdot \nabla _{{\boldsymbol{x}}%
_{k}}f_{N}^{(k)}=\epsilon ^{-1/2}A_{\epsilon }^{(k)}f_{N}^{(k)}+N\epsilon
^{-1/2}B_{\epsilon }^{(k+1)}f_{N}^{(k+1)}  \label{E:S401}
\end{equation}%
where the cumulative interaction operators are 
\begin{equation*}
\epsilon ^{-1/2}A_{\epsilon }^{(k)}=\sum_{1\leq i<j\leq k}\epsilon
^{-1/2}A_{i,j}^{\epsilon }\,,\qquad N\epsilon ^{-1/2}B_{\epsilon
}^{(k+1)}=\sum_{j=1}^{k}N\epsilon ^{-1/2}B_{j,k+1}^{\epsilon }
\end{equation*}%
with components defined by 
\begin{align*}
& [\epsilon ^{-1/2}A_{i,j}^{\epsilon }f_{N}^{(k)}](t,{\boldsymbol{x}}_{k},{%
\boldsymbol{v}}_{k})=-i\epsilon ^{-1/2}\sum_{\sigma =\pm 1}\sigma
\int_{h}e^{ih\cdot (x_{i}-x_{j})/\epsilon }\hat{\phi}(h) \\
& \qquad \qquad f_{N}^{(k)}(t,{\boldsymbol{x}}_{k},v_{1},\ldots
,v_{i}-\sigma \frac{h}{2},\ldots ,v_{j}+\sigma \frac{h}{2},\ldots ,v_{k})\,dh
\end{align*}%
and 
\begin{align*}
& [N\epsilon ^{-1/2}B_{j,k+1}^{\epsilon }f_{N}^{(k+1)}](t,{\boldsymbol{x}}%
_{k},{\boldsymbol{v}}_{k})=-iN\epsilon ^{-1/2}\sum_{\sigma =\pm 1}\sigma
\int_{x_{k+1},v_{k+1},h}e^{ih\cdot (x_{j}-x_{k+1})/\epsilon }\hat{\phi}(h) \\
& \qquad \qquad f_{N}^{(k+1)}(t,{\boldsymbol{x}}_{k+1},v_{1},\ldots
,v_{j}-\sigma \frac{h}{2},\ldots ,v_{k+1}+\sigma \frac{h}{2}%
)\,dx_{k+1}\,dv_{k+1}\,dh
\end{align*}%
The above is the $({\boldsymbol{x}}_{k},{\boldsymbol{v}}_{k})$ formulation.
We shall need the alternative formulations, and let us start with the $({%
\boldsymbol{x}}_{k},{\boldsymbol{\xi }}_{k})$ formulation. The hierarchy
becomes, for $1\leq k\leq N$, 
\begin{equation*}
\partial _{t}\tilde{f}_{N}^{(k)}-i\nabla _{{\boldsymbol{\xi }}_{k}}\cdot
\nabla _{{\boldsymbol{x}}_{k}}\tilde{f}_{N}^{(k)}=\epsilon ^{-1/2}\tilde{A}%
_{\epsilon }^{(k)}\tilde{f}_{N}^{(k)}+N\epsilon ^{-1/2}\tilde{B}_{\epsilon
}^{(k+1)}\tilde{f}_{N}^{(k+1)}
\end{equation*}%
The components of the operators are 
\begin{equation}
\lbrack \epsilon ^{-1/2}\tilde{A}_{i,j}^{\epsilon }\tilde{f}_{N}^{(k)}](t,{%
\boldsymbol{x}}_{k},{\boldsymbol{\xi }}_{k})=-i\epsilon ^{-1/2}\sum_{\sigma
=\pm 1}\sigma \phi \Big(\frac{x_{i}-x_{j}}{\epsilon }+\frac{\sigma }{2}(\xi
_{i}-\xi _{j})\Big)\tilde{f}_{N}^{(k)}(t,{\boldsymbol{x}}_{k},{\boldsymbol{%
\xi }}_{k})  \label{E:S3-10}
\end{equation}%
and 
\begin{equation}
\begin{aligned} &[N\epsilon^{-1/2}\tilde B_{j,k+1}^\epsilon \tilde
f_N^{(k+1)}](t,{\boldsymbol{x}}_k,{\boldsymbol{\xi}}_k) \\ &=
-iN\epsilon^{-1/2} \sum_{\sigma=\pm 1} \sigma \int_{x_{k+1}}
\phi(\frac{x_j-x_{k+1}}{\epsilon}+\frac{\sigma}{2}\xi_j) \tilde
f_N^{(k+1)}(t,{\boldsymbol{x}}_{k+1}, {\boldsymbol{\xi}}_k,0) \, dx_{k+1}
\end{aligned}  \label{E:S3-05}
\end{equation}%
where the last $0$ means that $\xi _{k+1}$ is set $=0$.

In $({\boldsymbol{x}}_k,{\boldsymbol{\xi}}_k)$ form, it is straightforward
to derive a typical estimate for the $A$-operator that sacrifices
derivatives in exchange for the gain of $\epsilon$ factors.

\begin{lemma}[$A$ estimate]
\label{L:basic-A} For any $0\leq s <\frac32$, 
\begin{equation*}
\| [\epsilon^{-1/2}\tilde A_{i,j}^\epsilon] \tilde f_N^{(k)}(t, {\boldsymbol{%
x}}_k, {\boldsymbol{\xi}}_k) \|_{L_{{\boldsymbol{x}}_k {\boldsymbol{\xi}}%
_k}^2} \lesssim \epsilon^{s-\frac12} \| |\nabla_{x_i}|^{s/2}
|\nabla_{x_j}|^{s/2} \tilde f_N^{(k)} (t, {\boldsymbol{x}}_k, {\boldsymbol{%
\xi}}_k) \|_{L_{{\boldsymbol{x}}_k {\boldsymbol{\xi}}_k}^2}
\end{equation*}
\end{lemma}

\begin{proof}
Let $p=\frac3s$. We divide into two cases, depending upon the relative size
of the \emph{frequency} of $\tilde f_N^{(k)}$ in $x_i$ versus the size of
the frequency of $\tilde f_N^{(k)}$ in $x_j$ and thus, by symmetry of the
argument that follows, might as well assume that the frequency in $x_j$
dominates the frequency in $x_i$. This allows us to transfer derivatives in $%
x_i$ to derivatives in $x_j$ at the end of the argument. From \eqref{E:S3-10}%
, by H\"older in $x_i$ between $\phi$ and $\tilde f_N^{(k)}$ 
\begin{equation*}
\| [\epsilon^{-1/2}\tilde A_{i,j}^\epsilon \tilde f_N^{(k)}](t,{\boldsymbol{x%
}}_k,{\boldsymbol{\xi}}_k)\|_{L^2_{x_i}} \lesssim \epsilon^{-1/2}
\sum_{\sigma=\pm 1} \left\| \phi \Big( \frac{x_i-x_j}{\epsilon} + \frac{%
\sigma}{2}(\xi_i-\xi_j) \Big) \right\|_{L^p_{x_i}} \| \tilde f_N^{(k)}(t, {%
\boldsymbol{x}}_k, {\boldsymbol{\xi}}_k) \|_{L^q_{x_i}}
\end{equation*}
where $\frac{1}{p}+\frac{1}{q}=\frac12$. Scaling out in the $\phi$ term
yields the factor $\epsilon^s$. Since $\frac1q = \frac12 - \frac{s}3$,
Sobolev embedding implies 
\begin{equation*}
\lesssim \epsilon^{-1/2} \epsilon^s \|\phi\|_{L^p} \| |\nabla_{x_i}|^s
\tilde f_N^{(k)}(t, {\boldsymbol{x}}_k, {\boldsymbol{\xi}}_k) \|_{L^q_{x_i}}
\end{equation*}
Following through with the $L_{x_j}^2$ norm, transferring half of the $s$%
-derivatives in $x_i$ to $x_j$, then applying the remaining $L^2$ norms,
yields the claim.
\end{proof}

In $({\boldsymbol{\eta}}_k, {\boldsymbol{\xi}}_k) $ coordinates, 
\begin{equation}  \label{E:S406}
[\check A_{i,j}^\epsilon \check f_N^{(k)}](t,\boldsymbol{\eta}_k, 
\boldsymbol{\xi}_k) = \begin{aligned}[t] &-i\epsilon^{-1/2} \sum_{\sigma=\pm
1} \sigma \int_{y} \hat\phi( y) e^{i\sigma (\xi_i-\xi_j) \cdot y/2} \\
&\qquad \check f_N^{(k)}(t,\eta_1, \ldots, \eta_i- \epsilon^{-1}y, \ldots,
\eta_j + \epsilon^{-1}y, \ldots, \eta_k, \boldsymbol{\xi}_k) \, dy
\end{aligned}
\end{equation}
For the $B$-operator, we apply Plancherel $x_{k+1}\mapsto \eta_{k+1}$ in the
integral in \eqref{E:S3-05}, we obtain 
\begin{equation}  \label{E:S407}
\begin{aligned} &[N\epsilon^{-1/2}\check B_{j,k+1}^\epsilon \check
f_N^{(k+1)}](t,{\boldsymbol{\eta}}_k,{\boldsymbol{\xi}}_k) \\ &=
-iN\epsilon^{5/2} \sum_{\sigma=\pm 1} \sigma \int_{\eta_{k+1}}
\hat\phi(\epsilon \eta_{k+1}) e^{i\sigma\epsilon \xi_j\eta_{k+1}/2} \check
f_N^{(k+1)}(t,\eta_1, \ldots, \eta_j-\eta_{k+1}, \ldots, \eta_{k+1},
{\boldsymbol{\xi}}_k,0) \, d\eta_{k+1} \end{aligned}
\end{equation}
Note that $N\epsilon^{5/2}=\epsilon^{-1/2}$ when $N=\epsilon^{-3}$. In $({%
\boldsymbol{\eta}}_k,{\boldsymbol{\xi}}_k)$ form, it is straightforward to
derive a typical estimate for the $B$-operator that sacrifices derivatives
in exchange for the gain of $\epsilon$ factors. Notice that Lemma \ref%
{L:basic-B} is significantly weaker than the estimate in Lemma \ref%
{L:basic-A} for $A$, but nevertheless Lemma \ref{L:basic-B} with $s=\frac12+$
gives a bound that decays as $\epsilon\to 0$ and has right side in the space 
$H_{{\boldsymbol{x}}}^{1+}(L_{{\boldsymbol{v}}}^2\cap L_{{\boldsymbol{v}}%
}^1) $.

\begin{lemma}[$B$ estimate]
\label{L:basic-B} Let $N=\epsilon^{-3}$. If for some $s\geq 0$, there holds $%
|\hat\phi(\zeta)| \lesssim |\zeta|^s$ for $|\zeta|\leq 1$, then 
\begin{align*}
&\| \langle \eta_j \rangle^{-(\frac{s}{2}+\frac34)} [N\epsilon^{-1/2}\check
B_{j,k+1}^\epsilon \check f_N^{(k+1)}](t,{\boldsymbol{\eta}}_k,{\boldsymbol{%
\xi}}_k)\|_{L_{{\boldsymbol{\eta}}_k {\boldsymbol{\xi}}_k}^2} \\
&\qquad \lesssim \epsilon^{s-\frac12} \| \langle \eta_j \rangle^{\frac{s}{2}%
+\frac34+} \langle \eta_{k+1} \rangle^{\frac{s}{2}+\frac34+} \check
f_N^{(k+1)} (t,{\boldsymbol{\eta}}_{k+1}, {\boldsymbol{\xi}}_k,0) \|_{L^2_{{%
\boldsymbol{\eta}}_{k+1} {\boldsymbol{\xi}}_k}}
\end{align*}
In particular, if $s=\frac12+$, then 
\begin{align*}
&\| \langle \eta_j \rangle^{-1-} [N\epsilon^{-1/2}\check B_{j,k+1}^\epsilon
\check f_N^{(k+1)}](t,{\boldsymbol{\eta}}_k,{\boldsymbol{\xi}}_k)\|_{L_{{%
\boldsymbol{\eta}}_k {\boldsymbol{\xi}}_k}^2} \\
&\qquad \lesssim \epsilon^{0+} \| \langle \eta_j \rangle^{1+} \langle
\eta_{k+1} \rangle^{1+} \check f_N^{(k+1)} (t,{\boldsymbol{\eta}}_{k+1}, {%
\boldsymbol{\xi}}_k,0) \|_{L^2_{{\boldsymbol{\eta}}_{k+1} {\boldsymbol{\xi}}%
_k}}
\end{align*}
\end{lemma}

\begin{proof}
Applying $|\hat\phi(\zeta)| \lesssim |\zeta|^s$, 
\begin{align*}
&|[N\epsilon^{-1/2}\check B_{j,k+1}^\epsilon \check f_N^{(k+1)}](t,{%
\boldsymbol{\eta}}_k,{\boldsymbol{\xi}}_k)| \\
&\lesssim \epsilon^{s-\frac12} \int_{\eta_{k+1}} |\eta_{k+1}|^s |\check
f_N^{(k+1)}(t,\eta_1, \ldots, \eta_j-\eta_{k+1}, \ldots, \eta_{k+1}, {%
\boldsymbol{\xi}}_k,0)| \, d\eta_{k+1}
\end{align*}
Apply the $L^2_{{\boldsymbol{\eta}}_k {\boldsymbol{\xi}}_k}$ norm and
Minkowski's integral inequality, 
\begin{align*}
&\| \langle \eta_j \rangle^{-(\frac{s}{2}+\frac34)} [N\epsilon^{-1/2}\check
B_{j,k+1}^\epsilon \check f_N^{(k+1)}](t,{\boldsymbol{\eta}}_k,{\boldsymbol{%
\xi}}_k)\|_{L_{{\boldsymbol{\eta}}_k {\boldsymbol{\xi}}_k}^2} \\
&\lesssim \epsilon^{s-\frac12} \int_{\eta_{k+1}} |\eta_{k+1}|^s \| \langle
\eta_j+\eta_{k+1} \rangle^{-(\frac{s}{2}+\frac34)} \check f_N^{(k+1)}(t, {%
\boldsymbol{\eta}}_k, \eta_{k+1}, {\boldsymbol{\xi}}_k,0) \|_{L_{{%
\boldsymbol{\eta}}_k {\boldsymbol{\xi}}_k}^2} d\eta_{k+1}
\end{align*}
Writing $1=\langle \eta_{k+1}\rangle^{-3/2-}\langle
\eta_{k+1}\rangle^{+3/2+} $ and applying Cauchy-Schwarz in $\eta_{k+1}$, 
\begin{align*}
&\| \langle \eta_j \rangle^{-(\frac{s}{2}+\frac34)} [N\epsilon^{-1/2}\check
B_{j,k+1}^\epsilon \check f_N^{(k+1)}](t,{\boldsymbol{\eta}}_k,{\boldsymbol{%
\xi}}_k)\|_{L_{{\boldsymbol{\eta}}_k {\boldsymbol{\xi}}_k}^2} \\
&\lesssim \epsilon^{s-\frac12} \| \langle \eta_{k+1} \rangle^{s+\frac32+}
\langle \eta_j+\eta_{k+1} \rangle^{-(\frac{s}{2}+\frac34)} \check
f_N^{(k+1)}(t, {\boldsymbol{\eta}}_k, \eta_{k+1}, {\boldsymbol{\xi}}_k,0)
\|_{L_{{\boldsymbol{\eta}}_{k+1} {\boldsymbol{\xi}}_k}^2}
\end{align*}
By dividing into the three cases $|\eta_{k+1}| \ll |\eta_j|$, $|\eta_{k+1}|
\gg |\eta_j|$ and $|\eta_{k+1}|\sim |\eta_j|$, we see that in any case 
\begin{equation*}
\langle \eta_{k+1} \rangle^{s+\frac32+} \langle \eta_j+\eta_{k+1} \rangle^{-(%
\frac{s}{2}+\frac34)} \lesssim \langle \eta_j \rangle^{\frac{s}{2}+\frac{3}{4%
}} \langle \eta_{k+1} \rangle^{ \frac{s}{2}+\frac34+}
\end{equation*}
\end{proof}

\subsection{Duhamel formulations}

To start the analysis of (\ref{hierarchy:RealBBGKY}), we shorten its
notation. Let 
\begin{equation}
\mathcal{D}^{(k)}f^{(k)}(t)=\int_{0}^{t}S^{(k)}(t-t^{\prime
})f^{(k)}(t^{\prime })\,dt^{\prime }  \label{E:IV03}
\end{equation}%
In this notation, the \emph{first Duhamel iterate (\ref%
{hierachy:PreBBGKYinDuhamel}) reads}: 
\begin{equation*}
f_{N}^{(k)}(t)=S^{(k)}(t)f_{N}^{(k)}(0)+\mathcal{D}^{(k)}[\epsilon
^{-1/2}A_{\epsilon }^{(k)}f_{N}^{(k)}](t)+\mathcal{D}^{(k)}[N\epsilon
^{-1/2}B_{\epsilon }^{(k+1)}f_{N}^{(k+1)}](t)
\end{equation*}%
and the \emph{second Duhamel iterate }(\ref{hierarchy:RealBBGKY}) is 
\begin{equation}
\begin{aligned} f_N^{(k)}(t) &= S^{(k)}(t)f_N^{(k)}(0) +
\mathcal{D}^{(k)}[\epsilon^{-1/2} A_\epsilon^{(k)} f_N^{(k)}](t) +
\mathcal{D}^{(k)}[N\epsilon^{-1/2} B_\epsilon^{(k+1)} S^{(k+1)}
f_N^{(k+1)}(0)] \\ &\qquad + \mathcal{D}^{(k)}[N\epsilon^{-1/2}
B_\epsilon^{(k+1)} \mathcal{D}^{(k+1)} \epsilon^{-1/2} A_\epsilon^{(k+1)}
f_N^{(k+1)}] \\ &\qquad + \mathcal{D}^{(k)}[N\epsilon^{-1/2}
B_\epsilon^{(k+1)} \mathcal{D}^{(k+1)} N \epsilon^{-1/2} B_\epsilon^{(k+2)}
f_N^{(k+2)}] \\ &= S^{(k)}(t)f_N^{(k)}(0) + \mathcal{D}^{(k)} R^{2(k)}_N
f_N^{(k)}(t) + \mathcal{D}^{(k)} R^{3(k+1)}_N f_N^{(k+1)}(t) \\ & \qquad +
\mathcal{D}^{(k)} (Q_N^{(k+1)}+R_N^{4(k+1)}) f_N^{(k+1)}(t) +
\mathcal{D}^{(k)} R^{5(k+2)}_N f_N^{(k+2)}(t) \end{aligned}  \label{E:S402}
\end{equation}%
The last two terms contain the composite operators 
\begin{equation}
\begin{aligned} \hspace{0.3in}&\hspace{-0.3in}
(Q_N^{(k+1)}+R_N^{4(k+1)})f_N^{(k+1)} = N\epsilon^{-1/2} B_\epsilon^{(k+1)}
\mathcal{D}^{(k+1)} \epsilon^{-1/2} A_\epsilon^{(k+1)} f_N^{(k+1)} \\ &=
\sum_{\substack{1\leq \ell \leq k \\ 1\leq i<j\leq k+1}} N\epsilon^{-1/2}
B^\epsilon_{\ell,k+1} \mathcal{D}^{(k+1)} \epsilon^{-1/2} A^\epsilon_{i,j}
f_N^{(k+1)} \end{aligned}  \label{E:S411}
\end{equation}%
\begin{equation}
R_{N}^{5(k+2)}f_{N}^{(k+2)}=N\epsilon ^{-1/2}B_{\epsilon }^{(k+1)}\mathcal{D}%
^{(k+1)}N\epsilon ^{-1/2}B_{\epsilon }^{(k+2)}f_{N}^{(k+2)}  \label{E:S412}
\end{equation}%
We define $Q_{N}^{(k+1)}$ and $R_{N}^{4(k+1)}$ as the following components
of the sum indicated in \eqref{E:S411}. The operator $Q_{N}^{(k+1)}$
corresponds to $\ell =i$ and $j=k+1$: 
\begin{equation}
Q_{N}^{(k+1)}f_{N}^{(k+1)}=\sum_{i=1}^{k}Q_{i,k+1}^{\epsilon }f_{N}^{(k+1)}
\label{E:S411a}
\end{equation}%
where 
\begin{equation*}
Q_{i,k+1}^{\epsilon }f_{N}^{(k+1)}=N\epsilon ^{-1/2}B_{i,k+1}^{\epsilon }%
\mathcal{D}^{(k+1)}\epsilon ^{-1/2}A_{i,k+1}^{\epsilon }f_{N}^{(k+1)}
\end{equation*}%
The operator $R_{N}^{4(k+1)}$ corresponds to the remaining terms in the sum %
\eqref{E:S411}: 
\begin{equation}
R_{N}^{4(k+1)}f_{N}^{(k+1)}=\sum_{\substack{ 1\leq \ell \leq k  \\ 1\leq
i<j\leq k+1  \\ \text{except }\ell =i,\;j=k+1}}R_{\ell
,i,j,k+1}^{4,N}f_{N}^{(k+1)}  \label{E:S411b}
\end{equation}%
where 
\begin{equation*}
R_{\ell ,i,j,k+1}^{4,N}f_{N}^{(k+1)}=N\epsilon ^{-1/2}B_{\ell
,k+1}^{\epsilon }\mathcal{D}^{(k+1)}\epsilon ^{-1/2}A_{i,j}^{\epsilon
}f_{N}^{(k+1)}
\end{equation*}

The operators defined in \eqref{E:S411a}, \eqref{E:S411b}, and \eqref{E:S412}
and are studied in the next three subsections.

\subsection{Collision operator $Q_N^{(k+1)}$}

\begin{proposition}[$Q_N^{(k+1)}$ estimates]
\label{P:QEstimates} The operator $Q_N^{(k+1)}$ defined by \eqref{E:S411a}
satisfies the bound 
\begin{equation*}
\| \check Q_N^{(k+1)} \check f_N^{(k+1)}(t) \|_{L_t^\infty L_{\boldsymbol{%
\eta}_k}^2 L_{\boldsymbol{\xi}_k}^2} \lesssim \sum_{i=1}^k \| \langle \eta_i
\rangle^{\frac34+} \langle \eta_{k+1} \rangle^{\frac34+} \langle
\nabla_{\xi_{k+1}} \rangle^{\frac12+} \check f_N^{(k+1)}(t, \boldsymbol{\eta}%
_{k+1}, \boldsymbol{\xi}_{k+1}) \|_{L_t^\infty L_{\boldsymbol{\eta}_{k+1}}^2
L_{\boldsymbol{\xi}_{k+1}}^2}
\end{equation*}
This bound is uniform in $N$ but does not have an $\epsilon^{0+}$ prefactor.
\end{proposition}

\begin{proof}
Now we study $Q_N^{(k+1)}f_N^{(k+1)}(t)$ defined by \eqref{E:S411a}. Recall %
\eqref{E:S407}, 
\begin{equation}  \label{E:S409}
\begin{aligned} N\epsilon^{-1/2}\check B_{\ell,k+1}^\epsilon \check
g_N^{(k+1)}(t,{\boldsymbol{\eta}}_k,{\boldsymbol{\xi}}_k) =
-iN\epsilon^{5/2} \sum_{\alpha=\pm 1} \alpha \int_{\eta_{k+1}}
\hat\phi(\epsilon \eta_{k+1}) e^{i\alpha\epsilon \xi_\ell\eta_{k+1}/2} &\\
\check g_N^{(k+1)}(t,\eta_1, \ldots, \eta_\ell-\eta_{k+1}, \ldots,
\eta_{k+1}, {\boldsymbol{\xi}}_k,0) \, d\eta_{k+1} & \end{aligned}
\end{equation}
By \eqref{E:S406}, 
\begin{equation}  \label{E:S410}
\begin{aligned} & \check
g_N^{(k+1)}(t,\boldsymbol{\eta}_{k+1},\boldsymbol{\xi}_{k+1}) =
\mathcal{D}^{(k+1)} \epsilon^{-1/2} \check A^\epsilon_{i,j} \check
f_N^{(k+1)}(t) \\ &= -i\epsilon^{-1/2} \sum_{\sigma=\pm 1} \sigma
\int_{t^{\prime} =0}^{t} \int_y \hat \phi(y) e^{i\sigma (\xi_i-\xi_j)\cdot
y/2} e^{-i\sigma (t-t^{\prime} )(\eta_i-\eta_j)\cdot y/2} \\ & \qquad \check
f_N^{(k)}(t^{\prime} ,\eta_1, \ldots, \eta_i-y/\epsilon, \, . \,,
\eta_j+y/\epsilon, \, . \,, \eta_{k+1}, \boldsymbol{\xi}_{k+1} -
(t-t^{\prime})\boldsymbol{\eta}_{k+1}) \, dy \,dt^{\prime} \end{aligned}
\end{equation}
We now consider the form of $Q_{i,k+1}^{\epsilon}$, which is the composition
of \eqref{E:S409} and \eqref{E:S410} in the case $\ell=i$ and $j=k+1$. 
\begin{align*}
\hspace{0.3in}&\hspace{-0.3in} \check Q_{i,k+1}^{\epsilon}\check
f_N^{(k+1)}(t) = -\epsilon^{-1} \sum_{\alpha,\sigma \in \{\pm 1\}} \alpha
\sigma \int_{\eta_{k+1}} \int_y \int_{t^{\prime} =0}^t \hat\phi(\epsilon
\eta_{k+1}) \hat \phi(y) \\
&e^{i\alpha \epsilon \xi_i \eta_{k+1}/2} e^{i\sigma \xi_i y/2}
e^{-i\sigma(t-t^{\prime} )(\eta_i-2\eta_{k+1})y/2} \\
&\check f_N^{(k+1)}(t^{\prime} , \eta_1, \, . \,, \eta_i-\eta_{k+1} - \frac{y%
}{\epsilon}, \, . \,, \eta_{k+1} +\frac{y}{\epsilon}, \xi_1-
(t-t^{\prime})\eta_1, \, . \,, \\
&\qquad \xi_i-(t-t^{\prime})(\eta_i-\eta_{k+1}), \, . \, ,\xi_k -
(t-t^{\prime})\eta_k, -(t-t^{\prime} )\eta_{k+1}) \, d\eta_{k+1} \, dy \,
dt^{\prime}
\end{align*}
Now change variable $\eta_{k+1} \mapsto \eta_{k+1} - \frac{y}{\epsilon}$ to
obtain 
\begin{align*}
\hspace{0.3in}&\hspace{-0.3in} \check Q_{i,k+1}^{\epsilon}\check
f_N^{(k+1)}(t) = -\epsilon^{-1} \sum_{\alpha,\sigma \in \{\pm 1\}} \alpha
\sigma \int_{\eta_{k+1}} \int_y \int_{t^{\prime} =0}^t \hat\phi(\epsilon
\eta_{k+1} - y) \hat \phi(y) \\
&e^{i\alpha \xi_i (\epsilon \eta_{k+1} - y)/2} e^{i\sigma \xi_i y/2}
e^{-i\sigma(t-t^{\prime} ) (\epsilon\eta_i-2\epsilon \eta_{k+1} +
2y)y/(2\epsilon)} \\
&\check f_N^{(k+1)}(t^{\prime} , \eta_1, \, . \,, \eta_i-\eta_{k+1}, \ldots,
\eta_{k+1} , \xi_1-(t-t^{\prime})\eta_1, \, . \,, \\
& \quad \xi_i-(t-t^{\prime})(\epsilon \eta_i-
\epsilon\eta_{k+1}+y)/\epsilon, \, . \,, \xi_k - (t-t^{\prime})\eta_k,
-(t-t^{\prime} )(\epsilon \eta_{k+1} -y)/\epsilon) \, d\eta_{k+1} \, dy \,
dt^{\prime}
\end{align*}
Now change $t^{\prime} \mapsto s$ where $s=(t-t^{\prime})/\epsilon$ to
obtain 
\begin{equation}  \label{E:S433}
\begin{aligned} \hspace{0.3in}&\hspace{-0.3in} \check
Q_{i,k+1}^{\epsilon}\check f_N^{(k+1)}(t) = - \sum_{\alpha,\sigma \in \{\pm
1\}} \alpha \sigma \int_{\eta_{k+1}} \int_y \int_{s=0}^{t/\epsilon}
\hat\phi(\epsilon \eta_{k+1} - y) \hat \phi(y) \\ &e^{i\alpha \xi_i
(\epsilon \eta_{k+1} - y)/2} e^{i\sigma \xi_i y/2} e^{-i\sigma s
(\epsilon\eta_i-2\epsilon \eta_{k+1} + 2y)y/2} \\ &\check
f_N^{(k+1)}(t-\epsilon s, \eta_1, \ldots, \eta_i-\eta_{k+1}, \ldots,
\eta_{k+1} , \xi_1 - \epsilon s \eta_1, \, . \,, \\ & \qquad \xi_i - sy -
\epsilon s\eta_i+\epsilon s \eta_{k+1}, \, . \,, \xi_k-\epsilon s
\eta_{k+1}, sy -s\epsilon\eta_{k+1}) \, d\eta_{k+1} \, dy \, ds \end{aligned}
\end{equation}
Let $L^2_{\boldsymbol{\eta}_i^*}$ denote the $L^2_{\boldsymbol{\eta}_q}$
norm for all $1\leq q\leq k$ \emph{except} $q=i$. Apply the $L^2_{%
\boldsymbol{\eta}_k^*}L^2_{\boldsymbol{\xi}_k}$ norm to obtain 
\begin{align*}
\hspace{0.3in}&\hspace{-0.3in} \| \check Q_{i,k+1}^{\epsilon}\check
f_N^{(k+1)}(t)\|_{L^2_{\boldsymbol{\eta}_k^*}L^2_{\boldsymbol{\xi}_k}}
\lesssim \int_{\eta_{k+1}} \int_y \int_{s=0}^{t/\epsilon} |\hat\phi(\epsilon
\eta_{k+1} - y)| |\hat \phi(y)| \\
& \| \check f_N^{(k+1)}(t-\epsilon s, \eta_1, \ldots, \eta_i-\eta_{k+1}, \,
. \,, \eta_{k+1} , \boldsymbol{\xi}_k, sy -s\epsilon \eta_{k+1})\|_{L^2_{%
\boldsymbol{\eta}_k^*}L^2_{\boldsymbol{\xi}_k}} \, d\eta_{k+1} \, dy \, ds
\end{align*}
Bring the $y$ integration to the inside and H\"older between the three
factors $|\hat\phi(\epsilon \eta_{k+1} - y)|$, $|\hat \phi(y)|$, and $\check
f_N^{(k+1)}(\ldots, sy -s\epsilon \eta_{k+1})$. This can be done in two
ways, that generate different factors $s^{-\mu}$ after rescaling $\check
f_N^{(k+1)}$, as in the following table

\begin{center}
\renewcommand{\arraystretch}{1.5}
\setlength{\tabcolsep}{10pt}  
\begin{tabular}{ccc|cc}
$|\hat\phi(\epsilon \eta_{k+1} - y)|$ & $|\hat \phi(y)|$ & $%
|f_N^{(k+1)}(\ldots, sy -s\epsilon \eta_{k+1})|$ & 
\shortstack{rescaling \\
$\check f_N^{(k+1)}$} & \shortstack{use\\when} \\ \hline
$L^{3-}_y$ & $L^{3-}_y$ & $L^{3+}_y$ & $s^{-1+}$ & $s\leq 1$ \\ 
$L^{3+}_y$ & $L^{3+}_y$ & $L^{3-}_y$ & $s^{-1-}$ & $s\geq 1$%
\end{tabular}
\end{center}

This gives 
\begin{align*}
\hspace{0.3in}&\hspace{-0.3in} \| \check Q_{i,k+1}^{\epsilon}\check
f_N^{(k+1)}(t)\|_{L^2_{\boldsymbol{\eta}_k^*}L^2_{\boldsymbol{\xi}_k}}
\lesssim \int_{\eta_{k+1}} \int_{s=0}^{t/\epsilon} s^{-1+} \langle s
\rangle^{0-} \\
& \| \check f_N^{(k+1)}(t-\epsilon s, \eta_1, \ldots, \eta_i-\eta_{k+1}, \,
. \,, \eta_{k+1} , \boldsymbol{\xi}_{k+1})\|_{(L^{3-}_{\xi_{k+1}}\cap
L^{3+}_{\xi_{k+1}})L^2_{\boldsymbol{\eta}_k^*} L^2_{\boldsymbol{\xi}_k}} \,
d\eta_{k+1} \, ds
\end{align*}
Apply Sobolev embedding in $\xi_{k+1}$: 
\begin{align*}
\hspace{0.3in}&\hspace{-0.3in} \| \check Q_{i,k+1}^{\epsilon}\check
f_N^{(k+1)}(t)\|_{L^2_{\boldsymbol{\eta}_k^*}L^2_{\boldsymbol{\xi}_k}}
\lesssim \int_{\eta_{k+1}} \int_{s=0}^{t/\epsilon} s^{-1+} \langle s
\rangle^{0-} \\
& \| \langle \xi_{k+1}\rangle^{\frac12+}\check f_N^{(k+1)}(t-\epsilon s,
\eta_1, \ldots, \eta_i-\eta_{k+1}, \, . \,, \eta_{k+1} ,\boldsymbol{\xi}%
_{k+1})\|_{L^2_{\boldsymbol{\eta}_k^*} L^2_{\boldsymbol{\xi}_{k+1}}} \,
d\eta_{k+1} \, ds
\end{align*}
For fixed $\eta_i$, divide the $\eta_{k+1}$ integration into two pieces
depending upon which of the two quantities $|\eta_i-\eta_{k+1}|$ and $%
|\eta_{k+1}|$ is maximum. Since both cases are similar, we will just assume $%
|\eta_{k+1}|$ is maximum. In this case, 
\begin{equation*}
1 \lesssim \langle \eta_i-\eta_{k+1}\rangle^{-\frac32-} \langle
\eta_i-\eta_{k+1}\rangle^{\frac34+} \langle \eta_{k+1}\rangle^{\frac34+}
\end{equation*}
can be inserted. This allows Cauchy-Schwarz in $\eta_{k+1}$ giving 
\begin{align*}
\hspace{0.3in}&\hspace{-0.3in} \| \check Q_{i,k+1}^{\epsilon}\check
f_N^{(k+1)}(t)\|_{L^2_{\boldsymbol{\eta}_k^*}L^2_{\boldsymbol{\xi}_k}}
\lesssim \int_{s=0}^{t/\epsilon} s^{-1+} \langle s \rangle^{0-} \| \langle
\eta_i-\eta_{k+1} \rangle^{\frac34+} \langle \eta_{k+1} \rangle^{\frac34+} \\
& \langle \xi_{k+1}\rangle^{\frac12+}\check f_N^{(k+1)}(t-\epsilon s,
\eta_1, \ldots, \eta_i-\eta_{k+1}, \, . \,, \eta_{k+1} , \boldsymbol{\xi}%
_{k+1})\|_{L^2_{\boldsymbol{\eta}_{k+1}^*} L^2_{\boldsymbol{\xi}_{k+1}}} \,
ds
\end{align*}
where we note that $L^2_{\boldsymbol{\eta}_k^*}$ norm on $\check f_N^{(k+1)}$
has been replaced by $L^2_{\boldsymbol{\eta}_{k+1}^*}$, which means the $%
L^2_{\boldsymbol{\eta}_q}$ norm for all $1\leq q\leq k+1$ \emph{except} $q=i$%
. Now apply the $L^2_{\eta_i}$ norm to obtain 
\begin{align*}
\| \check Q_{i,k+1}^{\epsilon}\check f_N^{(k+1)}(t)\|_{L^2_{\boldsymbol{\eta}%
_k}L^2_{\boldsymbol{\xi}_k}} \lesssim \int_{s=0}^{t/\epsilon} & s^{-1+}
\langle s \rangle^{0-} \| \langle \eta_i-\eta_{k+1} \rangle^{\frac34+}
\langle \eta_{k+1} \rangle^{\frac34+} \\
& \langle \xi_{k+1}\rangle^{\frac12+}\check f_N^{(k+1)}(t-\epsilon s, 
\boldsymbol{\eta}_{k+1} , \boldsymbol{\xi}_{k+1})\|_{L^2_{\boldsymbol{\eta}%
_{k+1}} L^2_{\boldsymbol{\xi}_{k+1}}} \, ds
\end{align*}
Taking sup in the $t$-component of $\check f_N^{(k+1)}$ and carrying out the 
$s$-integral gives the result.
\end{proof}

\subsection{Remainder operator $R_N^{4(k+1)}$}

\begin{proposition}[$R^{4(k+1)}_N$ estimates]
\label{P:R4Estimates} The operator $R^{4(k+1)}_N$ defined by \eqref{E:S411b}
satisfies the following estimate: 
\begin{align*}
\hspace{0.3in}&\hspace{-0.3in} \|\langle \boldsymbol{\eta}_k
\rangle^{-\frac34-} \check R^{4(k+1)}_N \check f_N^{(k+1)}(t)\|_{L^2_{%
\boldsymbol{\eta}_k} L^2_{\boldsymbol{\xi}_k} } \lesssim \epsilon^{0+}
\sum_{1\leq i<j\leq k} \int_{t^{\prime} =0}^t \\
& \|\langle \eta_i \rangle^{\frac13+} \langle \eta_j \rangle^{\frac13+}
\langle \eta_{k+1}\rangle^{\frac13+} \check f_N^{(k+1)}(t^{\prime} , 
\boldsymbol{\eta}_{k+1}, \boldsymbol{\xi}_{k+1})\|_{L^2_{\boldsymbol{\eta}%
_{k+1}} L^2_{\boldsymbol{\xi}_k} L^\infty_{\xi_{k+1}}} \, dt^{\prime}
\end{align*}
\end{proposition}

\begin{proof}
There are four cases.

\bigskip

\noindent\emph{Case 1. $\ell\neq i$ and $j= k+1$}. Composing \eqref{E:S409}
and \eqref{E:S410} in this case gives 
\begin{align*}
\hspace{0.3in}&\hspace{-0.3in} \check R_{\ell,i,k+1,k+1}^{4,N}\check
f_N^{(k+1)}(t) = -\epsilon^{-1} \sum_{\alpha,\sigma \in \{\pm 1\}} \alpha
\sigma \int_{\eta_{k+1}} \int_y \int_{t^{\prime} =0}^t \hat\phi(\epsilon
\eta_{k+1}) \hat \phi(y) \\
&e^{i\alpha \epsilon \xi_\ell \eta_{k+1}/2} e^{i\sigma (\xi_i-\xi_{k+1})y/2}
e^{-i\sigma(t-t^{\prime} )(\eta_i-\eta_{k+1})y/2} \\
&\check f_N^{(k+1)}(t^{\prime} , \eta_1, \, . \, , \eta_i - \frac{y}{\epsilon%
}, \, . \,,\eta_\ell-\eta_{k+1}, \, . \,, \eta_{k+1}+ \frac{y}{\epsilon}, \\
&\xi_1 - (t-t^{\prime})\eta_1, \, . \,,
\xi_\ell-(t-t^{\prime})(\eta_\ell-\eta_{k+1}), \, . \,,
\xi_k-(t-t^{\prime})\eta_k, -(t-t^{\prime})\eta_{k+1}) \, d\eta_{k+1} \, dy
\, dt^{\prime}
\end{align*}
Change variables $y\mapsto \epsilon y$ to obtain 
\begin{align*}
\hspace{0.3in}&\hspace{-0.3in} \check R_{\ell,i,k+1,k+1}^{4,N}\check
f_N^{(k+1)}(t) = -\epsilon^2 \sum_{\alpha,\sigma \in \{\pm 1\}} \alpha
\sigma \int_{\eta_{k+1}} \int_y \int_{t^{\prime} =0}^t \hat\phi(\epsilon
\eta_{k+1}) \hat \phi(\epsilon y) \\
&e^{i\alpha \epsilon \xi_\ell \eta_{k+1}/2} e^{i\sigma \epsilon
(\xi_i-\xi_{k+1}) y/2} e^{-i\sigma \epsilon (t-t^{\prime}
)(\eta_i-\eta_{k+1})y/2} \\
&\check f_N^{(k+1)}(t^{\prime} , \eta_1, \, . \, , \eta_i - y, \, .
\,,\eta_\ell-\eta_{k+1}, \, . \,, \eta_{k+1} + y, \\
& \xi_1 - (t-t^{\prime})\eta_1, \, . \,,
\xi_\ell-(t-t^{\prime})(\eta_\ell-\eta_{k+1}), \, . \,,
\xi_k-(t-t^{\prime})\eta_k, -(t-t^{\prime} )\eta_{k+1}) \, d\eta_{k+1} \, dy
\, dt^{\prime}
\end{align*}
This gives 
\begin{align*}
\hspace{0.3in}&\hspace{-0.3in} |\check R_{\ell,i,k+1,k+1}^{4,N}\check
f_N^{(k+1)}(t)| \leq \epsilon^2 \int_{\eta_{k+1}} \int_y \int_{t^{\prime}
=0}^t |\hat\phi(\epsilon \eta_{k+1})| |\hat \phi(\epsilon y)| \\
&|\check f_N^{(k+1)}(t^{\prime} , \eta_1, \, . \, , \eta_i - y, \, .
\,,\eta_\ell-\eta_{k+1}, \, . \,, \eta_{k+1} + y, \\
& \xi_1 - (t-t^{\prime})\eta_1, \, . \,,
\xi_\ell-(t-t^{\prime})(\eta_\ell-\eta_{k+1}), \, . \,,
\xi_k-(t-t^{\prime})\eta_k, -(t-t^{\prime} )\eta_{k+1})| \, d\eta_{k+1} \,
dy \, dt^{\prime}
\end{align*}
Start by applying $L^2_{\boldsymbol{\xi}_k}$, applying Minkowski's integral
inequality, and sup-out in the $\xi_{k+1}$ entry: 
\begin{align*}
\hspace{0.3in}&\hspace{-0.3in} \|\check R_{\ell,i,k+1,k+1}^{4,N}\check
f_N^{(k+1)}(t)\|_{L^2_{\boldsymbol{\xi}_k}} \leq \epsilon^2
\int_{\eta_{k+1}} \int_y \int_{t^{\prime} =0}^t |\hat\phi(\epsilon
\eta_{k+1})| |\hat \phi(\epsilon y)| \\
&\|\check f_N^{(k+1)}(t^{\prime} , \eta_1, \, . \, , \eta_i - y, \, .
\,,\eta_\ell-\eta_{k+1}, \, . \,, \eta_{k+1} + y, \boldsymbol{\xi}%
_{k+1})\|_{L^2_{\boldsymbol{\xi}_k} L^\infty_{\xi_{k+1}}} \, d\eta_{k+1} \,
dy \, dt^{\prime}
\end{align*}
Let $L^2_{\boldsymbol{\eta}_k^*}$ indicate the $L^2$ norm over all $\eta_q$
for $1\leq q \leq k$ \emph{except} $q=i$ and $q=\ell$. By Minkowski's
integral inequality, 
\begin{align*}
\hspace{0.3in}&\hspace{-0.3in} \|\check R_{\ell,i,k+1,k+1}^{4,N}\check
f_N^{(k+1)}(t)\|_{L^2_{\boldsymbol{\eta}_k^*} L^2_{\boldsymbol{\xi}_k}} \leq
\epsilon^2\int_{t^{\prime} =0}^t \int_{\eta_{k+1}} \int_y |\hat\phi(\epsilon
\eta_{k+1})| |\hat \phi(\epsilon y)| \\
&\|\check f_N^{(k+1)}(t^{\prime} , \eta_1, \, . \, , \eta_i - y, \, .
\,,\eta_\ell-\eta_{k+1}, \, . \,, \eta_{k+1} + y, \boldsymbol{\xi}%
_{k+1})\|_{L^2_{\boldsymbol{\eta}_k^*} L^2_{\boldsymbol{\xi}_k}
L^\infty_{\xi_{k+1}}} \, d\eta_{k+1} \, dy \, dt^{\prime}
\end{align*}
Divide the $\eta_{k+1}$, $y$ integration space into three regions depending
upon the relative size of $|\eta_i-y|$, $|\eta_\ell-\eta_{k+1}|$, and $%
|\eta_{k+1}+y|$. In the case when the quantity $|\eta_i-y|$ is the largest
of the three, we use 
\begin{equation*}
1 \leq \langle \eta_\ell - \eta_{k+1} \rangle^{-\frac12-} \langle
\eta_{k+1}+y \rangle^{-\frac12-} \langle \eta_i-y \rangle^{\frac13+} \langle
\eta_\ell - \eta_{k+1} \rangle^{\frac13+} \langle \eta_{k+1}+y
\rangle^{\frac13+}
\end{equation*}
From here, it is similar to the conclusion of Case 3.

\bigskip

\noindent\emph{Case 2. $\ell=i$ and $j\leq k$}. Aside from altering
inconsequential phase factors, this case is identical to Case 3 below.

\bigskip

\noindent\emph{Case 3. $\ell=j$ and $j\leq k$}. In this case, we obtain the
bound \eqref{E:S420} below. Composing \eqref{E:S409} and \eqref{E:S410} in
this case gives 
\begin{align*}
\hspace{0.3in}&\hspace{-0.3in} \check R_{j,i,j,k+1}^{4,N}\check
f_N^{(k+1)}(t) = -\epsilon^{-1} \sum_{\alpha,\sigma \in \{\pm 1\}} \alpha
\sigma \int_{\eta_{k+1}} \int_y \int_{t^{\prime} =0}^t \hat\phi(\epsilon
\eta_{k+1}) \hat \phi(y) \\
&e^{i\alpha \epsilon \xi_j \eta_{k+1}/2} e^{i\sigma (\xi_i-\xi_j) y/2}
e^{-i\sigma(t-t^{\prime} )(\eta_i-\eta_j+\eta_{k+1})y/2} \\
&\check f_N^{(k+1)}(t^{\prime} , \eta_1, \, . \, , \eta_i - \frac{y}{\epsilon%
}, \, . \,, \eta_j -\eta_{k+1}+ \frac{y}{\epsilon}, \, . \, , \eta_{k+1}, \\
&\xi_1-(t-t^{\prime})\eta_1, \, . \,,
\xi_j-(t-t^{\prime})(\eta_j-\eta_{k+1}), \, . \,,
\xi_k-(t-t^{\prime})\eta_k, -(t-t^{\prime} )\eta_{k+1}) \, d\eta_{k+1} \, dy
\, dt^{\prime}
\end{align*}
Change variables $y\mapsto \epsilon y$ to obtain 
\begin{align*}
\hspace{0.3in}&\hspace{-0.3in} \check R_{j,i,j,k+1}^{4,N}\check
f_N^{(k+1)}(t) = -\epsilon^2 \sum_{\alpha,\sigma \in \{\pm 1\}} \alpha
\sigma \int_{\eta_{k+1}} \int_y \int_{t^{\prime} =0}^t \hat\phi(\epsilon
\eta_{k+1}) \hat \phi(\epsilon y) \\
&e^{i\alpha \epsilon \xi_j \eta_{k+1}/2} e^{i\sigma \epsilon (\xi_i-\xi_j)
y/2} e^{-i\sigma \epsilon (t-t^{\prime} )(\eta_i-\eta_j+\eta_{k+1})y/2} \\
&\check f_N^{(k+1)}(t^{\prime} , \eta_1, \, . \, , \eta_i - y, \, . \,,
\eta_j -\eta_{k+1} + y, \, . \, , \eta_{k+1}, \\
&\xi_1-(t-t^{\prime})\eta_1, \, . \,,
\xi_j-(t-t^{\prime})(\eta_j-\eta_{k+1}), \, . \,,
\xi_k-(t-t^{\prime})\eta_k, -(t-t^{\prime} )\eta_{k+1}) \, d\eta_{k+1} \, dy
\, dt^{\prime}
\end{align*}
This gives 
\begin{align*}
\hspace{0.3in}&\hspace{-0.3in} |\check R_{j,i,j,k+1}^{4,N}\check
f_N^{(k+1)}(t)| \leq \epsilon^2 \int_{\eta_{k+1}} \int_y \int_{t^{\prime}
=0}^t |\hat\phi(\epsilon \eta_{k+1})| |\hat \phi(\epsilon y)| \\
&|\check f_N^{(k+1)}(t^{\prime} , \eta_1, \, . \, , \eta_i - y, \, . \,,
\eta_j -\eta_{k+1} + y, \, . \, , \eta_{k+1}, \\
&\xi_1-(t-t^{\prime})\eta_1, \, . \,,
\xi_j-(t-t^{\prime})(\eta_j-\eta_{k+1}), \, . \,,
\xi_k-(t-t^{\prime})\eta_k, -(t-t^{\prime} )\eta_{k+1})| \, d\eta_{k+1} \,
dy \, dt^{\prime}
\end{align*}
Let $L^2_{\boldsymbol{\eta}_k^*}$ indicate the $L^2$ norm over all $\eta_q$
for $1\leq q \leq k$ \emph{except} $q=i$ and $q=j$. Start by applying $L^2_{%
\boldsymbol{\eta}_k^*}L^2_{\boldsymbol{\xi}_k}$, applying Minkowski's
integral inequality, and sup-out in the $\xi_{k+1}$ entry: 
\begin{equation}  \label{E:S416}
\begin{aligned} \hspace{0.3in}&\hspace{-0.3in} \|\check
R_{j,i,j,k+1}^{4,N}\check f_N^{(k+1)}(t)\|_{L^2_{\bds \eta_k^*}
L^2_{\bds\xi_k}} \leq \epsilon^2 \int_{t^{\prime} =0}^t \int_{\eta_{k+1}}
\int_y |\hat\phi(\epsilon \eta_{k+1})| |\hat \phi(\epsilon y)| \\ &\|\check
f_N^{(k+1)}(t^{\prime} , \eta_1, \, . \, , \eta_i - y, \, . \,, \eta_j
-\eta_{k+1}+ y, \, . \, , \eta_{k+1}, \boldsymbol{\xi}_{k+1})\|_{L^2_{\bds
\eta_k^*} L^2_{\bds\xi_k} L^\infty_{\xi_{k+1}}} \, d\eta_{k+1} \, dy \,
dt^{\prime} \end{aligned}
\end{equation}
For fixed $\eta_i$ and $\eta_j$, we can divide the $\eta_{y+1}$, $y$
integration space into three pieces depending on the relative size of the
three quantities $|\eta_i-y|$, $|\eta_j-\eta_{k+1}+y|$, and $|\eta_{k+1}|$.
Since all three cases are similar, we will just present one of them. If $%
|\eta_i-y|$ is largest, then we use 
\begin{equation*}
1 \leq \langle \eta_j - \eta_{k+1}+ y \rangle^{-\frac12-} \langle \eta_{k+1}
\rangle^{-\frac12-} \langle \eta_i-y\rangle^{\frac13+}\langle \eta_j -
\eta_{k+1}+ y \rangle^{\frac13+} \langle \eta_{k+1} \rangle^{\frac13+}
\end{equation*}
Apply H\"older in $y$ on the inside using that for fixed $\eta_j$ and $%
\eta_{k+1}$, the quantity $\| \langle \eta_j-\eta_{k+1}+y
\rangle^{-\frac12-} \|_{L_y^{6-}}$ is finite (uniformly in $\eta_j$ and $%
\eta_{k+1}$), then H\"older in $\eta_{k+1}$ using that $\| \langle
\eta_{k+1} \rangle^{-\frac12-} \|_{L_{\eta_{k+1}}^{6-}}$ is finite, to
obtain 
\begin{align*}
\hspace{0.3in}&\hspace{-0.3in} \|\check R_{j,i,j,k+1}^{4,N}\check
f_N^{(k+1)}(t)\|_{L^2_{\boldsymbol{\eta}_k^*} L^2_{\boldsymbol{\xi}_k}} \leq
\epsilon^2 \int_{t^{\prime} =0}^t \| \hat\phi(\epsilon \eta_{k+1})
\|_{L^{3+}_{\eta_{k+1}}} \|\hat \phi(\epsilon y)\|_{L^{3+}_y} \\
&\|\langle \eta_i-y \rangle^{\frac13+} \langle \eta_j-\eta_{k+1}+y
\rangle^{\frac13+} \langle \eta_{k+1}\rangle^{\frac13+} \\
& \qquad \qquad \check f_N^{(k+1)}(t^{\prime} , \eta_1, \, . \, , \eta_i -
y, \, . \,, \eta_j -\eta_{k+1}+ y, \, . \, , \eta_{k+1}, \boldsymbol{\xi}%
_{k+1})\|_{L^2_y L^2_{\eta_{k+1}} L^2_{\boldsymbol{\eta}_k^*} L^2_{%
\boldsymbol{\xi}_k} L^\infty_{\xi_{k+1}}} \, dt^{\prime}
\end{align*}
Scale the norms on $\hat \phi$, which reduces $\epsilon^2$ to $\epsilon^{0+}$%
. Apply the $L^\infty_{\eta_j}L^2_{\eta_i}$ norm, and on the right-side,
bring the $L^2_{\eta_i}$ norm to the inside by Minkowski's integral
inequality. On the inside the norms in the order $L^2_{\eta_{k+1}} L^2_y
L^2_{\eta_i}$ admit translational change of variables that yield: 
\begin{align*}
\hspace{0.3in}&\hspace{-0.3in} \|\check R_{j,i,j,k+1}^{4,N}\check
f_N^{(k+1)}(t)\|_{L^\infty_{\eta_j} L^2_{\eta_i} L^2_{\boldsymbol{\eta}_k^*}
L^2_{\boldsymbol{\xi}_k} } \leq \epsilon^{0+} \int_{t^{\prime} =0}^t \\
&\|\langle \eta_i \rangle^{\frac13+} \langle \eta_j \rangle^{\frac13+}
\langle \eta_{k+1}\rangle^{\frac13+} \check f_N^{(k+1)}(t^{\prime} , 
\boldsymbol{\eta}_{k+1}, \boldsymbol{\xi}_{k+1})\|_{L^2_{\boldsymbol{\eta}%
_{k+1}} L^2_{\boldsymbol{\xi}_k} L^\infty_{\xi_{k+1}}} \, dt^{\prime}
\end{align*}
The same result can be obtained by applying the $L^\infty_{\eta_i}L^2_{%
\eta_j}$ norm instead of the $L^\infty_{\eta_j}L^2_{\eta_i}$ norm. Thus 
\begin{align*}
\hspace{0.3in}&\hspace{-0.3in} \|\check R_{j,i,j,k+1}^{4,N}\check
f_N^{(k+1)}(t)\|_{(L^\infty_{\eta_j} L^2_{\eta_i}\cap
L^\infty_{\eta_i}L^2_{\eta_j}) L^2_{\boldsymbol{\eta}_k^*} L^2_{\boldsymbol{%
\xi}_k} } \leq \epsilon^{0+} \int_{t^{\prime} =0}^t \\
&\|\langle \eta_i \rangle^{\frac13+} \langle \eta_j \rangle^{\frac13+}
\langle \eta_{k+1}\rangle^{\frac13+} \check f_N^{(k+1)}(t^{\prime} , 
\boldsymbol{\eta}_{k+1}, \boldsymbol{\xi}_{k+1})\|_{L^2_{\boldsymbol{\eta}%
_{k+1}} L^2_{\boldsymbol{\xi}_k} L^\infty_{\xi_{k+1}}} \, dt^{\prime}
\end{align*}
Finally, we conclude by applying Schur's test on the left side in the form 
\begin{equation*}
\| \langle u \rangle^{-\frac34-} \langle v \rangle^{-\frac34-} h(u,v)
\|_{L^2_{uv}} \lesssim \|h\|_{L_u^\infty L_v^2}^{1/2} \|h \|_{L_v^\infty
L_u^2}^{1/2}
\end{equation*}
to obtain 
\begin{equation}  \label{E:S420}
\begin{aligned} \hspace{0.3in}&\hspace{-0.3in} \|\langle \eta_i
\rangle^{-\frac34-} \langle \eta_j \rangle^{-\frac34-}\check
R_{j,i,j,k+1}^{4,N}\check f_N^{(k+1)}(t)\|_{L^2_{\bds \eta_k}
L^2_{\bds\xi_k} } \leq \epsilon^{0+} \int_{t^{\prime} =0}^t \\ &\|\langle
\eta_i \rangle^{\frac13+} \langle \eta_j \rangle^{\frac13+} \langle
\eta_{k+1}\rangle^{\frac13+} \check f_N^{(k+1)}(t^{\prime} ,
\boldsymbol{\eta}_{k+1}, \boldsymbol{\xi}_{k+1})\|_{L^2_{\bds \eta_{k+1}}
L^2_{\bds\xi_k} L^\infty_{\xi_{k+1}}} \, dt^{\prime} \end{aligned}
\end{equation}

\bigskip

\noindent\emph{Case 4. $\ell \notin \{i,j\}$ and $j\leq k$}. This case
results in the inequality \eqref{E:S415} below. Composing \eqref{E:S409} and %
\eqref{E:S410} in this case gives 
\begin{align*}
\hspace{0.3in}&\hspace{-0.3in} \check R_{\ell,i,j,k+1}^{4,N}\check
f_N^{(k+1)}(t) = -\epsilon^{-1} \sum_{\alpha,\sigma \in \{\pm 1\}} \alpha
\sigma \int_{t^{\prime} =0}^t \int_{\eta_{k+1}} \int_y \hat\phi(\epsilon
\eta_{k+1}) \hat \phi(y) \\
&e^{i\alpha \epsilon \xi_\ell \eta_{k+1}/2} e^{i\sigma (\xi_i-\xi_j) y/2}
e^{-i\sigma(t-t^{\prime} )(\eta_i-\eta_j)y/2} \\
&\check f_N^{(k+1)}(t^{\prime} , \eta_1, \, . \, , \eta_i - \frac{y}{\epsilon%
}, \, . \,,\eta_\ell-\eta_{k+1}, \, . \,, \eta_j + \frac{y}{\epsilon}, \, .
\,, \eta_{k+1}, \\
& \xi_1-(t-t^{\prime})\eta_1, \, . \,,
\xi_\ell-(t-t^{\prime})(\eta_\ell-\eta_{k+1}), \, . \,,
\xi_k-(t-t^{\prime})\eta_k, -(t-t^{\prime} )\eta_{k+1}) \, d\eta_{k+1} \, dy
\, dt^{\prime}
\end{align*}
Change variables $y\mapsto \epsilon y$ to obtain 
\begin{align*}
\hspace{0.3in}&\hspace{-0.3in} \check R_{\ell,i,j,k+1}^{4,N}\check
f_N^{(k+1)}(t) = -\epsilon^2 \sum_{\alpha,\sigma \in \{\pm 1\}} \alpha
\sigma \int_{t^{\prime} =0}^t \int_{\eta_{k+1}} \int_y \hat\phi(\epsilon
\eta_{k+1}) \hat \phi(\epsilon y) \\
&e^{i\alpha \epsilon \xi_\ell \eta_{k+1}/2} e^{i\sigma \epsilon
(\xi_i-\xi_j) y/2} e^{-i\sigma \epsilon (t-t^{\prime} )(\eta_i-\eta_j)y/2} \\
&\check f_N^{(k+1)}(t^{\prime} , \eta_1, \, . \, , \eta_i - y, \, .
\,,\eta_\ell-\eta_{k+1}, \, . \,, \eta_j + y, \, . \,, \eta_{k+1}, \\
& \xi_1-(t-t^{\prime})\eta_1, \, . \,,
\xi_\ell-(t-t^{\prime})(\eta_\ell-\eta_{k+1}), \, . \,,
\xi_k-(t-t^{\prime})\eta_k, -(t-t^{\prime} )\eta_{k+1}) \, d\eta_{k+1} \, dy
\, dt^{\prime}
\end{align*}
This gives 
\begin{align*}
\hspace{0.3in}&\hspace{-0.3in} |\check R_{\ell,i,j,k+1}^{4,N}\check
f_N^{(k+1)}(t)| \leq \epsilon^2 \int_{t^{\prime} =0}^t \int_{\eta_{k+1}}
\int_y |\hat\phi(\epsilon \eta_{k+1})| |\hat \phi(\epsilon y)| \\
&|\check f_N^{(k+1)}(t^{\prime} , \eta_1, \, . \, , \eta_i - y, \, .
\,,\eta_\ell-\eta_{k+1}, \, . \,, \eta_j + y, \, . \,, \eta_{k+1}, \\
& \xi_1-(t-t^{\prime})\eta_1, \, . \,,
\xi_\ell-(t-t^{\prime})(\eta_\ell-\eta_{k+1}), \, . \,,
\xi_k-(t-t^{\prime})\eta_k, -(t-t^{\prime} )\eta_{k+1})| \, d\eta_{k+1} \,
dy \, dt^{\prime}
\end{align*}
Let $L^2_{\boldsymbol{\eta}_k^*}$ indicate the $L^2$ norm over all $\eta_q$
for $1\leq q \leq k$ \emph{except} $q=i$ and $q=j$. Applying $L^2_{%
\boldsymbol{\eta}_k^*}L^2_{\boldsymbol{\xi}_k}$, applying Minkowski's
integral inequality, and sup-out in the $\xi_{k+1}$ entry: 
\begin{align*}
\hspace{0.3in}&\hspace{-0.3in} \|\check R_{\ell,i,j,k+1}^{4,N}\check
f_N^{(k+1)}(t)\|_{L^2_{\boldsymbol{\eta}_k^*} L^2_{\boldsymbol{\xi}_k}} \leq
\epsilon^2 \int_{t^{\prime} =0}^t \int_{\eta_{k+1}} \int_y
|\hat\phi(\epsilon \eta_{k+1})| |\hat \phi(\epsilon y)| \\
&\|\check f_N^{(k+1)}(t^{\prime} , \eta_1, \, . \, , \eta_i - y, \, . \,,
\eta_j + y, \eta_{k+1}, \boldsymbol{\xi}_{k+1})\|_{L^2_{\boldsymbol{\eta}%
_k^*} L^2_{\boldsymbol{\xi}_k} L^\infty_{\xi_{k+1}}} \, d\eta_{k+1} \, dy \,
dt^{\prime}
\end{align*}
For fixed $\eta_i$, we divide the $y$ integration space into two cases
depending upon which of the two quantities $|\eta_i-y|$ or $|\eta_j+y|$ is
minimum. The two cases are similar so we just present one and assume $%
|\eta_i-y|$ is minimum. In this case we use 
\begin{equation*}
1 \leq \langle \eta_i- y\rangle^{-\frac23-} \langle \eta_i - y
\rangle^{\frac13+} \langle \eta_j+y \rangle^{\frac13+}
\end{equation*}
Insert this bound, and also $1\leq \eta_{k+1} \rangle^{-\frac13-} \langle
\eta_{k+1} \rangle^{\frac13+}$, and Cauchy-Schwarz in $y$ and $\eta_{k+1}$: 
\begin{align*}
\hspace{0.3in}&\hspace{-0.3in} \|\check R_{\ell,i,j,k+1}^{4,N}\check
f_N^{(k+1)}(t)\|_{L^2_{\boldsymbol{\eta}_k^*} L^2_{\boldsymbol{\xi}_k}} \leq
\epsilon^2 \int_{t^{\prime} =0}^t \|\hat\phi(\epsilon \eta_{k+1}) \langle
\eta_{k+1}\rangle^{-\frac13-} \|_{L^2_{\eta_{k+1}}} \|\hat \phi(\epsilon y)
\langle \eta_i - y \rangle^{-\frac23-} \|_{L^2_y} \\
&\|\langle \eta_i-y\rangle^{\frac13+} \langle \eta_j+y
\rangle^{\frac13+}\langle \eta_{k+1}\rangle^{\frac13+} \\
& \hspace{1in} \check f_N^{(k+1)}(t^{\prime} , \eta_1, \, . \, , \eta_i - y,
\, . \,, \eta_j + y, \eta_{k+1}, \boldsymbol{\xi}_{k+1})\|_{L^2_y L^2_{%
\boldsymbol{\eta}_{k+1}^*} L^2_{\boldsymbol{\xi}_k} L^\infty_{\xi_{k+1}}} \,
d\eta_{k+1} \, dy \, dt^{\prime}
\end{align*}
where now $L^2_{\boldsymbol{\eta}_{k+1}^*}$ indicate the $L^2$ norm over all 
$\eta_q$ for $1\leq q \leq k+1$ \emph{except} $q=i$ and $q=j$. By H\"older
and scaling, 
\begin{equation*}
\| \hat \phi(\epsilon \eta)\langle \eta \rangle^{-\frac13-} \|_{L^2}
\lesssim \| \hat \phi \|_{L^{\frac{18}{7}+}} \epsilon^{-\frac76+} \,, \qquad
\| \hat \phi(\epsilon y) \langle y \rangle^{-\frac23-} \|_{L^2} \lesssim \|
\hat \phi \|_{L^{\frac{18}{5}+}} \epsilon^{-\frac56+}
\end{equation*}
which can be inserted above. Following through with the norm $%
L_{\eta_i}^\infty L_{\eta_j}^2 \cap L_{\eta_j}^\infty L_{\eta_i}^2$ gives 
\begin{equation*}
\begin{aligned} \hspace{0.3in}&\hspace{-0.3in} \|\check
R_{\ell,i,j,k+1}^{4,N}\check f_N^{(k+1)}(t)\|_{(L_{\eta_i}^\infty
L_{\eta_j}^2 \cap L_{\eta_j}^\infty L_{\eta_i}^2) L^2_{\bds \eta_k^*}
L^2_{\bds\xi_k}} \leq \epsilon^{0+} \int_{t^{\prime} =0}^t \\ &\|\langle
\eta_i\rangle^{\frac13+} \langle \eta_j \rangle^{\frac13+}\langle
\eta_{k+1}\rangle^{\frac13+} \check f_N^{(k+1)}(t^{\prime} ,
\boldsymbol{\eta}_{k+1}, \boldsymbol{\xi}_{k+1})\|_{L^2_{\bds \eta_{k+1}}
L^2_{\bds\xi_k} L^\infty_{\xi_{k+1}}} \, d\eta_{k+1} \, dy \, dt^{\prime}
\end{aligned}
\end{equation*}
Finally, we conclude by applying Schur's test on the left side in the form 
\begin{equation*}
\| \langle u \rangle^{-\frac34-} \langle v \rangle^{-\frac34-} h(u,v)
\|_{L^2_{uv}} \lesssim \|h\|_{L_u^\infty L_v^2}^{1/2} \|h \|_{L_v^\infty
L_u^2}^{1/2}
\end{equation*}
to obtain 
\begin{equation}  \label{E:S415}
\begin{aligned} \hspace{0.3in}&\hspace{-0.3in} \|\langle \eta_i
\rangle^{-\frac34-} \langle \eta_j \rangle^{-\frac34-}\check
R_{\ell,i,j,k+1}^{4,N}\check f_N^{(k+1)}(t)\|_{ L^2_{\bds \eta_k}
L^2_{\bds\xi_k}} \leq \epsilon^{0+} \int_{t^{\prime} =0}^t \\ &\|\langle
\eta_i\rangle^{\frac13+} \langle \eta_j \rangle^{\frac13+}\langle
\eta_{k+1}\rangle^{\frac13+} \check f_N^{(k+1)}(t^{\prime} ,
\boldsymbol{\eta}_{k+1}, \boldsymbol{\xi}_{k+1})\|_{L^2_{\bds \eta_{k+1}}
L^2_{\bds\xi_k} L^\infty_{\xi_{k+1}}} \, d\eta_{k+1} \, dy \, dt^{\prime}
\end{aligned}
\end{equation}
\end{proof}

\subsection{Remainder operator $R_N^{5(k+2)}$}

\begin{proposition}[$R_{N}^{5(k+2)}$ estimates]
\label{P:R5Estimates}Assume $\left\vert \hat{\phi}\left( \zeta \right)
\right\vert \lesssim \left\vert \zeta \right\vert ^{1-}$ for $\zeta $ near
zero, the operator $R_{N}^{5(k+2)}$ defined by \eqref{E:S412} satisfies the
following estimate 
\begin{align*}
\hspace{0.3in}& \hspace{-0.3in}\Vert \hat{\mathcal{D}}^{(k)}\hat{R}%
_{N}^{5(k+2)}\hat{f}_{N}^{(k+2)}(t)\Vert _{L_{\boldsymbol{\eta }_{k}}^{2}L_{%
\boldsymbol{v}_{k}}^{2}} \\
& \lesssim \epsilon ^{0+}\sum_{i=1}^{k}\int_{t^{\prime \prime }=0}^{t}\Vert
\langle \eta _{i}\rangle ^{1+}\langle \eta _{k+1}\rangle ^{1+}\langle \eta
_{k+2}\rangle ^{1+}\langle v_{k+1}\rangle ^{2+} \\
& \hspace{2in} \hat{f}_{N}^{(k+2)}(t^{\prime \prime },\boldsymbol{\eta }%
_{k+2},\boldsymbol{v}_{k+2})\Vert _{L_{\boldsymbol{\eta }_{k+2}}^{2}L_{%
\boldsymbol{v}_{k}}^{2}L_{v_{k+1}}^{\infty }L_{v_{k+2}}^{1}}\,dt^{\prime
\prime } \\
& +\epsilon ^{0+}\sum_{1\leq i<j\leq k}\int_{t^{\prime \prime
}=0}^{t}(t-t^{\prime \prime })\Vert \langle \eta _{i}\rangle ^{1+}\langle
\eta _{j}\rangle ^{1+}\langle \eta _{k+1}\rangle ^{1+}\langle \eta
_{k+1}\rangle ^{1+} \\
& \hspace{2in} \hat{f}_{N}^{(k+2)}(t^{\prime \prime },\boldsymbol{\eta }%
_{k+2},\boldsymbol{v}_{k+2})\Vert _{L_{\boldsymbol{\eta }_{k+2}}^{2}L_{%
\boldsymbol{v}_{k}}^{2}L_{v_{k+1}}^{1}L_{v_{k+2}}^{1}}\,dt^{\prime \prime }
\end{align*}
\end{proposition}

\begin{proof}
Now we study $R_N^{5(k+2)}f_N^{(k+2)}(t)$ defined by \eqref{E:S412}. This
expands as the sum 
\begin{equation*}
R_N^{5(k+2)}f_N^{(k+2)}(t) = \sum_{\substack{ 1\leq i \leq k  \\ 1\leq j
\leq k+1}} R_{i,j,k+2}^{N,5}f_N^{(k+2)}(t)
\end{equation*}
where 
\begin{equation*}
R_{i,j,k+2}^{N,5}f_N^{(k+2)} = N\epsilon^{-1/2} B^\epsilon_{i,k+1} \mathcal{D%
}^{(k+1)} N \epsilon^{-1/2} B^\epsilon_{j,k+2} f_N^{(k+2)}
\end{equation*}
To prepare for calculating the composition, let us rewrite \eqref{E:S407}
with indices $(i,k+1)$ and then again with indices $(j,k+2)$: 
\begin{equation}  \label{E:S413}
\begin{aligned} N\epsilon^{-1/2}\check B_{i,k+1}^\epsilon \check
g_N^{(k+1)}(t,{\boldsymbol{\eta}}_k,{\boldsymbol{\xi}}_k) =
-iN\epsilon^{5/2} \sum_{\alpha=\pm 1} \alpha \int_{\eta_{k+1}}
\hat\phi(\epsilon \eta_{k+1}) e^{i\alpha\epsilon \xi_i\eta_{k+1}/2} &\\
\check g_N^{(k+1)}(t,\eta_1, \ldots, \eta_i-\eta_{k+1}, \ldots, \eta_{k+1},
{\boldsymbol{\xi}}_k,0) \, d\eta_{k+1} & \end{aligned}
\end{equation}
\begin{equation}  \label{E:S414}
\begin{aligned} & \check{g}_N^{(k+1)}(t, \boldsymbol{\eta}_{k+1},
\boldsymbol{\xi}_{k+1}) = \mathcal{D}^{(k+1)} N\epsilon^{-1/2}\check
B_{j,k+2}^\epsilon \check
f_N^{(k+2)}(t,{\boldsymbol{\eta}}_{k+1},{\boldsymbol{\xi}}_{k+1}) \\ & =
-iN\epsilon^{5/2} \sum_{\sigma=\pm 1} \sigma \int_0^t \int_{\eta_{k+2}}
\hat\phi(\epsilon \eta_{k+2}) e^{i\sigma\epsilon (\xi_j-(t-t^{\prime}
)\eta_j)\eta_{k+2}/2} \\ &\qquad \check f_N^{(k+2)}(t^{\prime} ,\eta_1,
\ldots, \eta_j-\eta_{k+2}, \ldots, \eta_{k+2},
{\boldsymbol{\xi}}_{k+1}-(t-t^{\prime} )\boldsymbol{\eta}_{k+1},0) \,
d\eta_{k+2} \, dt^{\prime} \end{aligned}
\end{equation}

There are three cases

\bigskip

\noindent \emph{Case 1. $j=k+1$}. This case results in the bound %
\eqref{E:S419} below. For this case, we assume $|\hat{\phi}(\zeta )|\lesssim
|\zeta |^{0+}$ for $|\zeta |\leq 1$. Combining \eqref{E:S413} and %
\eqref{E:S414} gives 
\begin{align*}
\hspace{0.3in}&\hspace{-0.3in} \check{R}_{i,k+1,k+2}^{N,5}\check{f}%
_{N}^{(k+2)}(t) \\
& = -\epsilon ^{-1}\sum_{\alpha ,\sigma \in \{-1,1\}}\sigma \alpha
\int_{0}^{t}\int_{\eta _{k+1}}\int_{\eta _{k+2}}\hat{\phi}(\epsilon \eta
_{k+1})\hat{\phi}(\epsilon \eta _{k+2}) e^{i\alpha \epsilon \xi _{i}\eta
_{k+1}/2}e^{-i\sigma \epsilon (t-t^{\prime })\eta _{k+1}\eta _{k+2}/2} \\
& \qquad \check{f}_{N}^{(k+2)}(t^{\prime },\eta _{1},\,.\,,\eta _{i}-\eta
_{k+1},\,.\,,\eta _{k+1}-\eta _{k+2},\eta _{k+2}, \xi_1-(t-t^{\prime
})\eta_1, \, . \,, \\
& \qquad \qquad \qquad \xi_i-(t-t^{\prime })(\eta_i-\eta_{k+1}), \, . \,,
\xi_k - (t-t^{\prime })\eta_k, -(t-t^{\prime })\eta_{k+1}, 0) d\eta
_{k+1}\,d\eta _{k+2}\,dt^{\prime }
\end{align*}%
On $\check{f}_{N}^{(k+2)}$, pass to the Fourier side in $\boldsymbol{\xi }%
_{k+2}\mapsto \boldsymbol{v}_{k+2}$, and on the left side, pass to the
Fourier side in $\boldsymbol{\xi }_{k}\mapsto \boldsymbol{v}_{k}$. The
result is the hat form for this remainder term: 
\begin{align*}
\hspace{0.3in}& \hspace{-0.3in}\hat{R}_{i,k+1,k+2}^{N,5}\hat{f}%
_{N}^{(k+2)}(t)=-\epsilon ^{-1}\sum_{\alpha ,\sigma \in \{-1,1\}}\sigma
\alpha \int_{0}^{t}\int_{\substack{ \eta _{k+1},\eta _{k+2}  \\ %
v_{k+1},v_{k+1}}}\hat{\phi}(\epsilon \eta _{k+1})\hat{\phi}(\epsilon \eta
_{k+2}) \\
& e^{-i\sigma \epsilon (t-t^{\prime })\eta _{k+1}\eta
_{k+2}/2}e^{-i(t-t^{\prime })\boldsymbol{\eta }_{k}\cdot \boldsymbol{v}%
_{k}}e^{-i(t-t^{\prime })\eta _{k+1}v_{k+1}} e^{i(t-t^{\prime }) \eta_{k+1}
v_i} \\
& \hat{f}_{N}^{(k+2)}(t^{\prime },\eta _{1},\,.\,,\eta _{i}-\eta
_{k+1},\,.\,,\eta _{k+1}-\eta _{k+2},\eta _{k+2}, \\
& \qquad v_{1},\,.\,, v_{i}-\alpha \epsilon \eta_{k+1}/2,\,.\,,v_{k+1},
v_{k+2})\,dv_{k+1}\,dv_{k+2}\,d\eta _{k+1}\,d\eta _{k+2}\,dt^{\prime }
\end{align*}%
Now we must add the additional Duhamel operator, for which we replace the
old $t^{\prime }$ with $t^{\prime \prime }$ and the old $t$ with $t^{\prime
} $. The propagator associated with this Duhamel term places a new phase
factor $e^{-i(t-t^{\prime })\boldsymbol{\eta }_{k}\cdot \boldsymbol{v}_{k}}$%
: 
\begin{align*}
\hspace{0.3in}& \hspace{-0.3in}\hat{\mathcal{D}}^{(k)}\hat{R}%
_{i,k+1,k+2}^{N,5}\hat{f}_{N}^{(k+2)}(t)=-\epsilon ^{-1}\sum_{\alpha ,\sigma
\in \{-1,1\}}\sigma \alpha \int_{t^{\prime }=0}^{t}\int_{t^{\prime \prime
}=0}^{t^{\prime }}\int_{\substack{ \eta _{k+1},\eta _{k+2}  \\ %
v_{k+1},v_{k+1} }}\hat{\phi}(\epsilon \eta _{k+1})\hat{\phi}(\epsilon \eta
_{k+2}) \\
& e^{-i(t-t^{\prime })\boldsymbol{\eta }_{k}\cdot \boldsymbol{v}%
_{k}}e^{-i\sigma \epsilon (t^{\prime }-t^{\prime \prime })\eta _{k+1}\eta
_{k+2}/2}e^{-i(t^{\prime }-t^{\prime \prime })\boldsymbol{\eta }_{k}\cdot 
\boldsymbol{v}_{k}}e^{-i(t^{\prime }-t^{\prime \prime })\eta _{k+1}v_{k+1}}
e^{i(t^{\prime \prime\prime})\eta_{k+1}v_i} \\
& \hat{f}_{N}^{(k+2)}(t^{\prime \prime },\eta _{1},\,.\,,\eta _{i}-\eta
_{k+1},\,.\,,\eta _{k+1}-\eta _{k+2},\eta _{k+2}, \\
& \qquad v_{1},\,.\,,v_{i}-\alpha \epsilon \eta
_{k+1}/2,\,.\,,v_{k+1},v_{k+2})\,dv_{k+1}\,dv_{k+2}\,d\eta _{k+1}\,d\eta
_{k+2}\,dt^{\prime \prime }\,dt^{\prime }
\end{align*}%
Now we switch the order of the $t^{\prime }$ and $t^{\prime \prime }$
integrals, which allows us to bring the $t^{\prime }$ integral onto the
phase factors: 
\begin{align*}
\hspace{0.3in}& \hspace{-0.3in}\hat{\mathcal{D}}^{(k)}\hat{R}%
_{i,k+1,k+2}^{N,5}\hat{f}_{N}^{(k+2)}(t)=-\epsilon ^{-1}\sum_{\alpha ,\sigma
\in \{-1,1\}}\sigma \alpha \int_{t^{\prime \prime }=0}^{t}\int_{\substack{ %
\eta _{k+1},\eta _{k+2}  \\ v_{k+1},v_{k+2}}}\hat{\phi}(\epsilon \eta _{k+1})%
\hat{\phi}(\epsilon \eta _{k+2}) \\
& \int_{t^{\prime \prime\prime}}^t e^{-i(t-t^{\prime })\boldsymbol{\eta }%
_{k}\cdot \boldsymbol{v}_{k}}e^{-i\sigma \epsilon (t^{\prime }-t^{\prime
\prime })\eta _{k+1}\eta_{k+2}/2}e^{-i(t^{\prime }-t^{\prime \prime })%
\boldsymbol{\eta }_{k}\cdot \boldsymbol{v}_{k}}e^{-i(t^{\prime }-t^{\prime
\prime })\eta _{k+1}v_{k+1}} e^{i(t^{\prime \prime\prime})\eta_{k+1}v_i} \,
dt^{\prime } \\
& \hat{f}_{N}^{(k+2)}(t^{\prime \prime },\eta _{1},\,.\,,\eta _{i}-\eta
_{k+1},\,.\,,\eta _{k+1}-\eta _{k+2},\eta _{k+2}, \\
& \qquad \qquad \qquad v_{1},\,.\,,v_{i}-\alpha \epsilon \eta
_{k+1}/2,\,.\,,v_{k+1},v_{k+2})\,dv_{k+1}\,dv_{k+2}\,d\eta _{k+1}\,d\eta
_{k+2}\,dt^{\prime \prime }
\end{align*}%
For $\mu, \nu \in \mathbb{R}$, 
\begin{equation*}
\int_{t^{\prime \prime\prime}}^t e^{-i(t-t^{\prime })\mu} e^{-i(t^{\prime
\prime\prime})\nu} \,dt^{\prime }= \frac{ e^{-i(t-t^{\prime\prime})\nu} -
e^{-i(t-t^{\prime\prime})\mu}}{i(\mu-\nu)}
\end{equation*}
which implies 
\begin{equation*}
\left| \int_{t^{\prime \prime\prime}}^t e^{-i(t-t^{\prime })\mu}
e^{-i(t^{\prime \prime\prime})\nu} \,dt^{\prime }\right| \lesssim \frac{1}{%
\langle \mu -\nu\rangle}
\end{equation*}
provided $t\leq 1$. Thus 
\begin{align*}
\hspace{0.3in}& \hspace{-0.3in}|\hat{\mathcal{D}}^{(k)}\hat{R}%
_{i,k+1,k+2}^{N,5}\hat{f}_{N}^{(k+2)}(t)|\leq \epsilon ^{-1}\int_{t^{\prime
\prime }=0}^{t}\int_{\substack{ \eta _{k+1},\eta _{k+2}  \\ v_{k+1},v_{k+2}}}%
|\hat{\phi}(\epsilon \eta _{k+1})||\hat{\phi}(\epsilon \eta _{k+2})| \\
& \langle \eta _{k+1}\cdot (\tfrac12\sigma \epsilon \eta
_{k+2}+v_i+v_{k+1})\rangle^{-1} \\
& |\hat{f}_{N}^{(k+2)}(t^{\prime \prime },\eta _{1},\,.\,,\eta _{i}-\eta
_{k+1},\,.\,,\eta _{k+1}-\eta _{k+2},\eta _{k+2}, \\
& \qquad \qquad v_{1},\,.\,,v_{i}-\alpha \epsilon \eta
_{k+1}/2,\,.\,,v_{k+1},v_{k+2})|\,dv_{k+1}\,dv_{k+2}\,d\eta _{k+1}\,d\eta
_{k+2}\,dt^{\prime \prime }
\end{align*}%
Insert $1\leq \langle v_{k+1}\rangle ^{-2-}\langle v_{k+1}\rangle ^{2+}$,
grouping the $\langle v_{k+1}\rangle ^{2+}$ factor with $\hat{f}_{N}^{(k+2)}$%
, and then sup this factor out in $v_{k+1}$: 
\begin{align*}
\hspace{0.3in}& \hspace{-0.3in}|\hat{\mathcal{D}}^{(k)}\hat{R}%
_{i,k+1,k+2}^{N,5}\hat{f}_{N}^{(k+2)}(t)|\leq \epsilon ^{-1}\int_{t^{\prime
\prime }=0}^{t}\int_{\eta _{k+1},\eta _{k+2}}|\hat{\phi}(\epsilon \eta
_{k+1})||\hat{\phi}(\epsilon \eta _{k+2})|\,|\eta _{k+1}|^{-1+} \\
& \Vert \langle v_{k+1}\rangle ^{2+}\hat{f}_{N}^{(k+2)}(t^{\prime \prime
},\eta _{1},\,.\,,\eta _{i}-\eta _{k+1},\,.\,,\eta _{k+1}-\eta _{k+2},\eta
_{k+2}, \\
& \qquad v_{1},\,.\,,v_{i}-\alpha \epsilon \eta
_{k+1}/2,\,.\,,v_{k+1},v_{k+2})\Vert _{L_{v_{k+1}}^{\infty
}L_{v_{k+2}}^{1}}\,d\eta _{k+1}\,d\eta _{k+2}\,dt^{\prime \prime }
\end{align*}%
where we have used 
\begin{equation*}
\int_{v_{k+1}}\langle \eta _{k+1}\cdot (\tfrac12\sigma \epsilon \eta
_{k+2}+v_i+v_{k+1})\rangle^{-1}\langle v_{k+1}\rangle
^{-2-}\,dv_{k+1}\lesssim |\eta _{k+1}|^{-1+}
\end{equation*}%
Now we proceed depending on which of the three quantities is maximum among $%
|\eta _{i}-\eta _{k+1}|$, $|\eta _{k+1}-\eta _{k+2}|$ and $|\eta _{k+2}|$.
Since all three cases are similar, we will just assume $|\eta _{i}-\eta
_{k+1}|$ is maximum. In this case, we insert 
\begin{equation*}
1\leq \langle \eta _{k+1}-\eta _{k+2}\rangle ^{-\frac{3}{2}-}\langle \eta
_{k+2}\rangle ^{-\frac{3}{2}-}\langle \eta _{i}-\eta _{k+1}\rangle
^{1+}\langle \eta _{k+1}-\eta _{k+2}\rangle ^{1+}\langle \eta _{k+2}\rangle
^{1+}
\end{equation*}%
and apply Cauchy-Schwarz: 
\begin{align*}
\hspace{0.3in}& \hspace{-0.3in}|\hat{\mathcal{D}}^{(k)}\hat{R}%
_{i,k+1,k+2}^{N,5}\hat{f}_{N}^{(k+2)}(t)|\leq \epsilon ^{-1}\int_{t^{\prime
\prime }=0}^{t}\Vert \langle \eta _{k+2}\rangle ^{-\frac{3}{2}-}\hat{\phi}%
(\epsilon \eta _{k+2})\Vert _{L_{\eta _{k+2}}^{2}} \\
& \Vert \hat{\phi}(\epsilon \eta _{k+1})|\eta _{k+1}|^{-1+}\langle \eta
_{k+1}-\eta _{k+2}\rangle ^{-\frac{3}{2}-}\Vert _{L_{\eta _{k+2}}^{\infty
}L_{\eta _{k+1}}^{2}} \\
& \Vert \langle \eta _{i}-\eta _{k+1}\rangle ^{1+}\langle \eta _{k+1}-\eta
_{k+2}\rangle ^{1+}\langle \eta _{k+2}\rangle ^{1+}\langle v_{k+1}\rangle
^{2+} \\
& \qquad \hat{f}_{N}^{(k+2)}(t^{\prime \prime },\eta _{1},\,.\,,\eta
_{i}-\eta _{k+1},\,.\,,\eta _{k+1}-\eta _{k+2},\eta _{k+2}, \\
& \qquad \qquad v_{1},\,.\,,v_{i}-\alpha \epsilon \eta
_{k+1}/2,\,.\,,v_{k+1},v_{k+2})\Vert _{L_{\eta _{k+1}}^{2}L_{\eta
_{k+2}}^{2}L_{v_{k+1}}^{\infty }L_{v_{k+2}}^{1}}\,dt^{\prime \prime }
\end{align*}%
Now use $|\hat{\phi}(\epsilon \eta _{k+2})|\lesssim \epsilon ^{0+}|\eta
_{k+2}|^{0+}$ and $|\hat{\phi}(\epsilon \eta _{k+1})|\lesssim \epsilon
^{1-}|\eta _{k+1}|^{1-}$, where $0+$ and $1-$ are selected to sum to $1+$.
This gives 
\begin{align*}
\hspace{0.3in}& \hspace{-0.3in}|\hat{\mathcal{D}}^{(k)}\hat{R}%
_{i,k+1,k+2}^{N,5}\hat{f}_{N}^{(k+2)}(t)|\lesssim \epsilon
^{0+}\int_{t^{\prime \prime }=0}^{t} \\
& \Vert \langle \eta _{i}-\eta _{k+1}\rangle ^{1+}\langle \eta _{k+1}-\eta
_{k+2}\rangle ^{1+}\langle \eta _{k+2}\rangle ^{1+}\langle v_{k+1}\rangle
^{2+} \\
& \qquad \hat{f}_{N}^{(k+2)}(t^{\prime \prime },\eta _{1},\,.\,,\eta
_{i}-\eta _{k+1},\,.\,,\eta _{k+1}-\eta _{k+2},\eta _{k+2}, \\
& \qquad \qquad v_{1},\,.\,,v_{i}-\alpha \epsilon \eta
_{k+1}/2,\,.\,,v_{k+1},v_{k+2})\Vert _{L_{\eta _{k+1}}^{2}L_{\eta
_{k+2}}^{2}L_{v_{k+1}}^{\infty }L_{v_{k+2}}^{1}}\,dt^{\prime \prime }
\end{align*}%
Now apply $L_{\boldsymbol{\eta }_{k}}^{2}L_{\boldsymbol{v}_{k}}^{2}$ and
Minkowski's integral inequality to obtain 
\begin{equation}
\begin{aligned} \|\hat{\mathcal{D}}^{(k)} \hat R_{i,k+1,k+2}^{N,5} \hat
f_N^{(k+2)}(t)\|_{L^2_{\bds \eta_k} L^2_{\bds v_k}} \lesssim \epsilon^{0+}
\int_{t^{\prime\prime} =0}^t \| \langle \eta_i \rangle^{1+} \langle
\eta_{k+1} \rangle^{1+} \langle \eta_{k+2} \rangle^{1+} \langle v_{k+1}
\rangle^{2+} &\\ \hat f_N^{(k+2)}(t^{\prime\prime} ,
\boldsymbol{\eta}_{k+2}, \boldsymbol{v}_{k+2})\|_{L^2_{\bds \eta_{k+2}}
L^2_{\bds v_k} L^\infty_{v_{k+1}}L^1_{v_{k+2}}} & \, dt^{\prime\prime}
\end{aligned}  \label{E:S419}
\end{equation}

\bigskip

\noindent \emph{Case 2. $j\leq k$ and $i=j$.} This case results in the bound %
\eqref{E:S418} below. For this case, we assume $|\hat{\phi}(\zeta )|\lesssim
|\zeta |^{0+}$ for $|\zeta |\leq 1$. Combining \eqref{E:S413} and %
\eqref{E:S414} gives 
\begin{align*}
\hspace{0.3in}& \hspace{-0.3in}\check{R}_{i,i,k+2}^{N,5}\check{f}%
_{N}^{(k+2)}(t) \\
& =-\epsilon ^{-1}\sum_{\alpha ,\sigma \in \{-1,1\}}\sigma \alpha
\int_{0}^{t}\int_{\eta _{k+1}}\int_{\eta _{k+2}}\hat{\phi}(\epsilon\eta
_{k+1})\hat{\phi}(\epsilon \eta _{k+2}) e^{i\alpha \epsilon
\xi_i\eta_{k+1}/2}e^{i\sigma \epsilon (\xi_i-(t-t^{\prime
})(\eta_i-\eta_{k+1}))\eta_{k+2}/2} \\
& \qquad\check{f}_{N}^{(k+2)}(t^{\prime },\eta _{1},\,.\,,\eta _{i}-\eta
_{k+1}-\eta _{k+2},\,.\,,\eta _{k+1},\eta _{k+2}, \xi_1 - (t-t^{\prime
})\eta_1, \, . \,, \\
& \qquad \qquad \qquad \xi_i - (t-t^{\prime })(\eta_i-\eta_{k+1}), \, . \,,
\xi_k - (t-t^{\prime })\eta_k, -(t-t^{\prime })\eta_{k+1},
0)\,d\eta_{k+1}\,d\eta_{k+2}\,dt^{\prime }
\end{align*}
This is handled similarly to Case 1.

\bigskip

\noindent\emph{Case 3. $j \leq k$ and $i\neq j$}. This case results in the
bound \eqref{E:S418} below. For this case, we assume $|\hat \phi(\zeta)|
\lesssim |\zeta|^{\frac12+}$ for $|\zeta| \leq 1$. Combining \eqref{E:S413}
and \eqref{E:S414} gives 
\begin{align*}
\hspace{0.3in}&\hspace{-0.3in} \check R_{i,j,k+2}^{N,5} \check f_N^{(k+2)}(t)
\\
&= -\epsilon^{-1} \sum_{\alpha,\sigma \in \{-1,1\}} \sigma \alpha \int_0^t
\int_{\eta_{k+1}} \int_{\eta_{k+2}} \hat \phi( \epsilon\eta_{k+1}) \hat
\phi(\epsilon\eta_{k+2}) e^{i\alpha \epsilon \xi_i\eta_{k+1}/2} e^{i\sigma
\epsilon(\xi_j - (t-t^{\prime} )\eta_j)\eta_{k+2}/2} \\
&\qquad \check f_N^{(k+2)}(t^{\prime} , \eta_1, \, . \,,
\eta_i-\eta_{k+1},\, . \,, \eta_j - \eta_{k+2}, \, . \,, \eta_{k+1},
\eta_{k+2}, \xi_1 - (t-t^{\prime }) \eta_1, \, . \,, \\
&\qquad \qquad \xi_i - (t-t^{\prime }) (\eta_i-\eta_{k+1}), \, . \,, \xi_k -
(t-t^{\prime }) \eta_k, -(t-t^{\prime })\eta_{k+1}, 0) \, d\eta_{k+1} \,
d\eta_{k+2} \, dt^{\prime}
\end{align*}
Let $L^2_{\boldsymbol{\eta}_k^*}$ denote all $L^2_{\eta_q}$ \emph{except} $%
q=i$ and $q=j$. By Minkowski's integral inequality, 
\begin{align*}
\hspace{0.3in}&\hspace{-0.3in} \| \check R_{i,j,k+2}^{N,5} \check
f_N^{(k+2)}(t) \|_{L_{\boldsymbol{\eta}_k^*}^2L_{\boldsymbol{\xi}_k}^2} \leq
\epsilon^{-1} \int_0^t \int_{\eta_{k+1}} \int_{\eta_{k+2}} d\eta_{k+1} \,
d\eta_{k+2} \, dt^{\prime} \quad |\hat \phi( \epsilon\eta_{k+1})| |\hat
\phi(\epsilon\eta_{k+2})| \\
&\|\check f_N^{(k+2)}(t^{\prime} , \eta_1, \, . \,, \eta_i-\eta_{k+1},\, .
\,, \eta_j - \eta_{k+2}, \, . \,,\eta_{k+1}, \eta_{k+2}, \boldsymbol{\xi}%
_{k+2})\|_{ L_{\boldsymbol{\eta}_k^*}^2L_{\boldsymbol{\xi}_k}^2
L_{\xi_{k+1}}^\infty L_{\xi_{k+2}}^\infty } \,
\end{align*}
We proceed depending upon the relative size of $|\eta_i-\eta_{k+1}|$ and $%
|\eta_{k+1}|$ on the one hand, and also depending upon the relative size of $%
|\eta_j - \eta_{k+2}|$ and $|\eta_{k+2}|$ on the other hand. Thus, there are
four cases in total, although all are similar, so we just present one.
Suppose that both $|\eta_{k+1}| \geq |\eta_i-\eta_{k+1}|$ and $%
|\eta_{k+2}|\geq |\eta_j-\eta_{k+2}|$. Then we use 
\begin{equation*}
1 \leq \langle \eta_{k+1}\rangle^{-2-} \langle \eta_{k+1} \rangle^{1+}
\langle \eta_i-\eta_{k+1} \rangle^{1+}
\end{equation*}
and 
\begin{equation*}
1 \leq \langle \eta_{k+2} \rangle^{-2-} \langle \eta_{k+2}\rangle^{1+}
\langle \eta_j-\eta_{k+2} \rangle^{1+}
\end{equation*}
Inserting these two inequalities, apply Cauchy-Schwarz in both $\eta_{k+1}$
and $\eta_{k+2}$, and then apply $L^2_{\eta_i}L^2_{\eta_j}$ to the entire
expression to obtain 
\begin{align*}
\hspace{0.3in}&\hspace{-0.3in} \| \check R_{i,j,k+2}^{N,5} \check
f_N^{(k+2)}(t) \|_{L_{\boldsymbol{\eta}_k}^2L_{\boldsymbol{\xi}_k}^2} \leq
\epsilon^{-1} \int_0^t \, \\
& \|\hat \phi( \epsilon\eta_{k+1}) \langle
\eta_{k+1}\rangle^{-2-}\|_{L^2_{\eta_{k+1}}} \|\hat \phi(\epsilon\eta_{k+2})
\langle \eta_{k+2}\rangle^{-2-} \|_{L^2_{\eta_{k+2}}} \\
&\|\langle \eta_i \rangle^{1+} \langle \eta_j \rangle^{1+} \langle
\eta_{k+1} \rangle^{1+} \langle \eta_{k+2} \rangle^{1+} \check
f_N^{(k+2)}(t^{\prime} , \boldsymbol{\eta}_{k+2}, \boldsymbol{\xi}%
_{k+2})\|_{ L_{\boldsymbol{\eta}_{k+2}}^2L_{\boldsymbol{\xi}_k}^2
L_{\xi_{k+1}}^\infty L_{\xi_{k+2}}^\infty } dt^{\prime} \,
\end{align*}
Since we have assumed the pointwise bound $|\hat \phi(y)| \lesssim
|y|^{\frac12+}$ for $|y|\leq 1$, it follows that $\| \hat \phi(\epsilon
y)\langle y \rangle^{-2-} \|_{L^2_y} \leq \epsilon^{\frac12+}$ 
\begin{equation}  \label{E:S418}
\begin{aligned} \hspace{0.3in}&\hspace{-0.3in} \| \check R_{i,j,k+2}^{N,5}
\check f_N^{(k+2)}(t) \|_{L_{\bds\eta_k}^2L_{\bds\xi_k}^2} \leq
\epsilon^{0+} \int_0^t \\ &\|\langle \eta_i \rangle^{1+} \langle \eta_j
\rangle^{1+} \langle \eta_{k+1} \rangle^{1+} \langle \eta_{k+2} \rangle^{1+}
\check f_N^{(k+2)}(t^{\prime} , \boldsymbol{\eta}_{k+2},
\boldsymbol{\xi}_{k+2})\|_{ L_{\bds\eta_{k+2}}^2L_{\bds\xi_k}^2
L_{\xi_{k+1}}^\infty L_{\xi_{k+2}}^\infty }\, dt^{\prime} \end{aligned}
\end{equation}
\end{proof}

\subsection{Limiting collision operator $Q^{(k+1)}$: definition and forms 
\label{S:coll-op-deriv}}

Recall that $Q_N^{(k+1)}$ has been defined by \eqref{E:S411a} as 
\begin{equation*}
Q_N^{(k+1)}f_N^{(k+1)} = \sum_{i=1}^k Q^{\epsilon}_{i,k+1}f_N^{(k+1)}
\end{equation*}
where 
\begin{equation*}
Q^{\epsilon}_{i,k+1} f_N^{(k+1)} = N\epsilon^{-1/2} B^\epsilon_{i,k+1} 
\mathcal{D}^{(k+1)} \epsilon^{-1/2} A^\epsilon_{i,k+1} f_N^{(k+1)}
\end{equation*}
A direct formula for $\check Q^{\epsilon}_{i,k+1}$ has been computed in %
\eqref{E:S433}, that we repeat here: 
\begin{equation}  \label{E:S433b}
\begin{aligned} \hspace{0.3in}&\hspace{-0.3in} \check
Q_{i,k+1}^{\epsilon}\check f_N^{(k+1)}(t) = - \sum_{\alpha,\sigma \in \{\pm
1\}} \alpha \sigma \int_{\eta_{k+1}} \int_y \int_{s=0}^{t/\epsilon}
\hat\phi(\epsilon \eta_{k+1} - y) \hat \phi(y) \\ &e^{i\alpha \xi_i
(\epsilon \eta_{k+1} - y)/2} e^{i\sigma \xi_i y/2} e^{-i\sigma s
(\epsilon\eta_i-2\epsilon \eta_{k+1} + 2y)y/2} \\ &\check
f_N^{(k+1)}(t-\epsilon s, \eta_1, \ldots, \eta_i-\eta_{k+1}, \ldots,
\eta_{k+1} , \xi_1 - \epsilon s \eta_1, \, . \,, \\ & \qquad \xi_i - sy -
\epsilon s\eta_i+\epsilon s \eta_{k+1}, \, . \,, \xi_k-\epsilon s
\eta_{k+1}, sy -s\epsilon\eta_{k+1}) \, d\eta_{k+1} \, dy \, ds \end{aligned}
\end{equation}
We can formally set $\epsilon=0$ in this expression to obtain the \emph{%
defining} expression for collision operator component $\check Q_{i,k+1}$:

\begin{definition}[limit form of the collision operator]
Let 
\begin{equation}  \label{E:S435}
Q^{(k+1)}f^{(k+1)} = \sum_{i=1}^k Q_{i,k+1}f^{(k+1)}
\end{equation}
where the components $Q_{i,k+1}f^{(k+1)}$ take the form of \eqref{E:S433b}
with $\epsilon=0$, using that since $\phi$ is real-valued, $\overline{\hat
\phi(y)}= \hat\phi(-y)$: 
\begin{equation}  \label{E:S434}
\begin{aligned} \check Q_{i,k+1}\check f^{(k+1)}(t, \boldsymbol{\eta}_k,
\boldsymbol{\xi}_k) = - \sum_{\alpha,\sigma \in \{\pm 1\}} \alpha \sigma
\int_{\eta_{k+1}} \int_y \int_{s=0}^{\infty} |\hat \phi(y)|^2
e^{i(\sigma-\alpha)\xi_iy/2} e^{-i \sigma s|y|^2} &\\ \check f^{(k+1)}(t,
\eta_1, \, . \,, \eta_i-\eta_{k+1}, \, . \,, \eta_{k+1} , \xi_1 , \, . \,,
\xi_i-sy, \, . \, , \xi_k, sy ) \, d\eta_{k+1} \, dy \, ds & \end{aligned}
\end{equation}
Alternative forms for \eqref{E:S434} are derived below and given as %
\eqref{E:S444}, \eqref{E:S437a} and the gain minus loss representation of
Proposition \ref{P:gainloss}.
\end{definition}

By taking the formal $\epsilon\to 0$, $N\to \infty$ limit of the quantum
BBGKY hierarchy defined in \eqref{E:S401}, we obtain the Boltzmann
(infinite) hierarchy 
\begin{equation}  \label{E:S432}
\partial_t f^{(k)} + {\boldsymbol{v}}_k \cdot \nabla_{{\boldsymbol{x}}_k}
f^{(k)} = Q^{(k+1)} f^{(k+1)} \,, \qquad k \geq 1
\end{equation}

From \eqref{E:S434}, it is straightforward to take the inverse Fourier
transform $\boldsymbol{\eta}_k \mapsto \boldsymbol{x}_k$ to obtain the $(%
\boldsymbol{x}_k, \boldsymbol{\xi}_k)$ form of the operator 
\begin{equation}  \label{E:S444}
\begin{aligned} \tilde Q_{i,k+1}\tilde f^{(k+1)}(t, \boldsymbol{x}_k,
\boldsymbol{\xi}_k) = - \sum_{\alpha,\sigma \in \{\pm 1\}} \alpha \sigma
\int_y \int_{s=0}^{\infty} |\hat \phi(y)|^2 e^{i(\sigma-\alpha)\xi_iy/2}
e^{-i\sigma s|y|^2} &\\ \tilde f^{(k+1)}(t, \boldsymbol{x}_k , x_i, \xi_1,
\, . \,, \xi_i-sy, \, . \,, \xi_k, sy ) \, dy \, ds & \end{aligned}
\end{equation}

Applying the Fourier transform $\boldsymbol{\xi}_k \mapsto \boldsymbol{v}_k$%
, we obtain the $(\boldsymbol{x}_k, \boldsymbol{v}_k)$ form 
\begin{equation}  \label{E:S437a}
\begin{aligned} Q_{i,k+1}f^{(k+1)}(t, \boldsymbol{x}_k, \boldsymbol{v}_k) =
- \sum_{\alpha,\sigma \in \{\pm 1\}} \alpha \sigma \int_y \int_{v_{k+1}}
\int_{s=0}^{\infty} |\hat \phi(y)|^2 e^{isy(v_{k+1}-v_i)}
e^{-i(\sigma+\alpha) s|y|^2/2} &\\ f^{(k+1)}(t, \boldsymbol{x}_k , x_i, v_1,
\, . \,, v_i+ \frac{\alpha-\sigma}{2}y, \, . \,, v_{k+1} ) \, dy \, dv_{k+1}
\, ds & \end{aligned}
\end{equation}

It is customary to rewrite \eqref{E:S437a} in terms of a gain and loss
operator that involve a collision kernel.

\begin{proposition}[representation of $Q_{i,k+1}$ in terms of gain minus loss%
]
\label{P:gainloss} $Q_{i,k+1}$ decomposes as the difference of a \emph{gain}
and \emph{loss} term 
\begin{equation*}
Q_{i,k+1}=Q_{i,k+1}^{+}-Q_{i,k+1}^{-}
\end{equation*}%
where the loss term is 
\begin{equation*}
Q_{i,k+1}^{-}f^{(k+1)}(t,\boldsymbol{x}_{k},\boldsymbol{v}_{k})=\frac{1}{2}%
\int_{v_{k+1}}\int_{\omega \in S^{2}}|r||\hat{\phi}(r\omega )|^{2}\Big|%
_{r=\omega \cdot (v_{i}-v_{k+1})}f^{(k+1)}(t,\boldsymbol{x}_{k},x_{i},%
\boldsymbol{v}_{k+1})\,d\omega \,dv_{k+1}
\end{equation*}%
and the gain term is 
\begin{align*}
Q_{i,k+1}^{+}f^{(k+1)}(t,\boldsymbol{x}_{k},\boldsymbol{v}_{k})& =\frac{1}{2}%
\int_{v_{k+1}}\int_{\omega \in S^{2}}|r||\hat{\phi}(r\omega )|^{2}\Big|%
_{r=\omega \cdot (v_{i}-v_{k+1})} \\
& f^{(k+1)}(t,\boldsymbol{x}_{k},x_{i},v_{1},\,.\,,v_{i}^{\ast
},\,.\,,v_{k},v_{k+1}^{\ast })\Big|_{\substack{ v_{i}^{\ast }=v_{i}+r\omega 
\\ v_{k+1}^{\ast }=v_{k+1}-r\omega }}\Big|_{r=\omega \cdot
(v_{k+1}-v_{i})}\,d\omega \,dv_{k+1}
\end{align*}%
which are (\ref{eqn:collision kernel for hierarchy}).
\end{proposition}

\begin{proof}
For expository convenience, we will write out the proof only in the case $%
k=2 $. In this case, \eqref{E:S437a} takes the form 
\begin{equation*}
\begin{aligned} Q f^{(2)}(t,x_1,v_1) = - \sum_{\substack{\sigma = \pm 1 \\
\alpha=\pm 1}} \alpha \sigma \int_{v_2} \int_y \int_{s=0}^{+\infty} |\hat
\phi(y)|^2 e^{-i \frac{\sigma+\alpha}{2} s |y|^2} e^{-isy\cdot (v_1-v_2)} &
\\ f^{(2)}(t,x_1,x_1,v_1+\frac{\alpha-\sigma}{2}y,v_2) & \, ds \, dy dv_2
\end{aligned}
\end{equation*}%
We decompose this as 
\begin{equation*}
Q=Q^{+}-Q^{-}
\end{equation*}%
where the gain operator $Q^{+}$ would be: 
\begin{equation*}
\begin{aligned} Q^+ f^{(2)}(t,x_1,v_1) = \sum_{\substack{
(\alpha,\sigma)=(1,-1), \\ (-1,1)} } \int_{v_2} \int_y \int_{s=0}^{+\infty}
|\hat \phi(y)|^2 e^{-i \frac{\sigma+\alpha}{2} s |y|^2} e^{-isy\cdot
(v_1-v_2)} & \\ f^{(2)}(t,x_1,x_1,v_1+\frac{\alpha-\sigma}{2}y,v_2) & \, ds
\, dy dv_2 \end{aligned}
\end{equation*}%
and the loss operator is $Q^{-}$ would be: 
\begin{equation*}
\begin{aligned} Q^- f^{(2)}(t,x_1,v_1) =
\sum_{\substack{(\alpha,\sigma)=(1,1),\\(-1,-1)}} \int_{v_2} \int_y
\int_{s=0}^{+\infty} |\hat \phi(y)|^2 e^{-i \frac{\sigma+\alpha}{2} s |y|^2}
e^{-isy\cdot (v_1-v_2)} & \\
f^{(2)}(t,x_1,x_1,v_1+\frac{\alpha-\sigma}{2}y,v_2) & \, ds \, dy dv_2
\end{aligned}
\end{equation*}%
The above is another version of the gain/loss collision operators in $(x,v)$
coordinates. We now rewrite them in the more standard format.

We reexpress the loss operator $Q^{-}$ as follows: 
\begin{equation*}
\begin{aligned} Q^- f^{(2)}(t,x_1,v_1) = \int_{v_2} \int_y
\int_{s=0}^{+\infty} |\hat \phi(y)|^2 (e^{i s |y|^2} +e^{-i s |y|^2})
e^{-isy\cdot (v_1-v_2)} & \\ f^{(2)}(t,x_1,x_1,v_1,v_2) & \, ds \, dy dv_2
\end{aligned}
\end{equation*}%
Split the sum into two integrals, and in the second integral change
variables $s\mapsto -s$ and $y \mapsto -y $. Since this transformed second
integral is now over $-\infty <s<0$, the two integrals now combine to give a
single integral over $-\infty <s<\infty $: 
\begin{equation*}
Q^{-}f^{(2)}(t,x_{1},v_{1})=\int_{v_{2}}\int_{y }\int_{s=-\infty }^{+\infty
}|\hat{\phi}(y )|^2e^{is|y |^{2}}e^{-isy \cdot
(v_{1}-v_{2})}f^{(2)}(t,x_{1},x_{1},v_{1},v_{2})\,ds\,dy dv_{2}
\end{equation*}%
Carrying out the $s$ integral: 
\begin{equation*}
Q^{-}f^{(2)}(t,x_{1},v_{1})=\int_{v_{2}}\int_{y }|\hat{\phi}(y )|^2\delta
(|y |^{2}-y \cdot (v_{1}-v_{2}))f^{(2)}(t,x_{1},x_{1},v_{1},v_{2})\,dy dv_{2}
\end{equation*}%
Introduce polar coordinates $y =r\omega $ 
\begin{eqnarray*}
&&Q^{-}f^{(2)}(t,x_{1},v_{1}) \\
&=&\int_{v_{2}}\int_{\omega \in S^{2}}\int_{r=0}^{\infty }|\hat{\phi}%
(r\omega )|^2\delta (r^{2}-r\omega \cdot
(v_{1}-v_{2}))f^{(2)}(t,x_{1},x_{1},v_{1},v_{2})\,r^{2}\,dr\,d\omega \,dv_{2}
\end{eqnarray*}%
Use the homogeneity $r\delta (r^{2}-r\omega \cdot (v_{1}-v_{2}))=\delta
(r-\omega \cdot (v_{1}-v_{2}))$: 
\begin{eqnarray*}
&&Q^{-}f^{(2)}(t,x_{1},v_{1}) \\
&=&\int_{v_{2}}\int_{\omega \in S^{2}}\int_{r=0}^{\infty }|\hat{\phi}%
(r\omega )|^2 \delta (r-\omega \cdot
(v_{1}-v_{2}))f^{(2)}(t,x_{1},x_{1},v_{1},v_{2})\,r\,dr\,d\omega \,dv_{2}
\end{eqnarray*}%
The $\delta $ term reduces the $r$ integration to evaluation at $r=\omega
\cdot (v_{1}-v_{2})$ when this quantity is positive: 
\begin{equation*}
Q^{-}f^{(2)}(t,x_{1},v_{1})=\int_{v_{2}}\int_{\substack{ \omega \in S^{2}: 
\\ \omega \cdot (v_{1}-v_{2})>0}}|\hat{\phi}(r\omega )|^2 r\Big|_{r=\omega
\cdot (v_{1}-v_{2})}f^{(2)}(t,x_{1},x_{1},v_{1},v_{2})\, d\omega \,dv_{2}
\end{equation*}%
By even extension: 
\begin{equation*}
Q^{-}f^{(2)}(t,x_{1},v_{1})=\frac{1}{2}\int_{v_{2}}\int_{\omega \in S^{2}}|%
\hat{\phi}(r\omega )|^2 |r|\Big|_{r=\omega \cdot
(v_{1}-v_{2})}f^{(2)}(t,x_{1},x_{1},v_{1},v_{2})\,d\omega \,dv_{2}
\end{equation*}%
which is the standard form for the loss operator.

Now we derive the standard form of the gain operator 
\begin{equation*}
\begin{aligned} Q^+ f^{(2)}(t,x_1,v_1) = \int_{v_2} \int_y
\int_{s=0}^{+\infty} |\hat\phi(y)|^2 e^{-isy\cdot (v_1-v_2)} & \\ (
f^{(2)}(t,x_1,x_1,v_1+ y,v_2) + f^{(2)}(t,x_1,x_1,v_1-y,v_2) ) & \, ds \, dy
dv_2 \end{aligned}
\end{equation*}%
Split the sum into two integrals, and in the second integral change
variables $s\mapsto -s$ and $y \mapsto -y $. Since this transformed second
integral is now over $-\infty <s<0$, the two integrals now combine to give a
single integral over $-\infty <s<\infty $: 
\begin{equation*}
Q^{+}f^{(2)}(t,x_{1},v_{1})=\int_{v_{2}}\int_{y }\int_{s=-\infty }^{+\infty
}|\hat \phi(y)|^2e^{-isy \cdot (v_{1}-v_{2})}f^{(2)}(t,x_{1},x_{1},v_{1}+y
,v_{2})\,ds\,dy dv_{2}
\end{equation*}%
To avoid confusion, we change notation $v_{2}\mapsto v_{2}^{\ast }$ and also
we can set $v_{1}^{\ast }=v_{1}+y $. 
\begin{equation*}
Q^{+}f^{(2)}(t,x_{1},v_{1})=\int_{v_{2}^{\ast }}\int_{y }\int_{s=-\infty
}^{+\infty }|\hat \phi(y)|^2e^{-isy \cdot (v_{1}-v_{2}^{\ast
})}f^{(2)}(t,x_{1},x_{1},v_{1}^{\ast },v_{2}^{\ast })\Big|_{v_{1}^{\ast
}=v_{1}+y }\,ds\,dy dv_{2}^{\ast }
\end{equation*}%
Now move the $v_{2}^{\ast }$ integration to the inside and (for fixed $y $)
change variable $v_{2}^{\ast }\rightarrow v_{2}$ where $v_{2}=v_{2}^{\ast
}+y $ 
\begin{equation*}
Q^{+}f^{(2)}(t,x_{1},v_{1})=\int_{v_{2}}\int_{y }\int_{s=-\infty }^{+\infty
}|\hat \phi(y)|^2e^{-isy \cdot (v_{1}-v_{2}+y
)}f^{(2)}(t,x_{1},x_{1},v_{1}^{\ast },v_{2}^{\ast })\Big| _{\substack{ %
v_{1}^{\ast }=v_{1}+y  \\ v_{2}^{\ast }=v_{2}-y }}\,ds\,dy dv_{2}
\end{equation*}%
Now carry out the $s$ integral 
\begin{equation*}
Q^{+}f^{(2)}(t,x_{1},v_{1})=\int_{v_{2}}\int_{y }\hat{\phi}(-y )\hat{\phi}(y
)\delta (|y |^{2}+y \cdot (v_{1}-v_{2}))f^{(2)}(t,x_{1},x_{1},v_{1}^{\ast
},v_{2}^{\ast })\Big| _{\substack{ v_{1}^{\ast }=v_{1}+y  \\ v_{2}^{\ast
}=v_{2}-y }}\,dy dv_{2}
\end{equation*}%
Change to polar coordinates $y =r\omega $ 
\begin{eqnarray*}
&&Q^{+}f^{(2)}(t,x_{1},v_{1}) \\
&=&\int_{v_{2}}\int_{\omega \in S^{2}}\int_{r=0}^{+\infty }|\hat
\phi(r\omega)|^2\delta (r^{2}+r\omega \cdot
(v_{1}-v_{2}))f^{(2)}(t,x_{1},x_{1},v_{1}^{\ast },v_{2}^{\ast })\Big| 
_{\substack{ v_{1}^{\ast }=v_{1}+y  \\ v_{2}^{\ast }=v_{2}-y }}%
r^{2}\,dr\,d\omega \,dv_{2}
\end{eqnarray*}%
By homogeneity $\delta (r^{2}+r\omega \cdot (v_{1}-v_{2}))r=\delta (r+\omega
\cdot (v_{1}-v_{2}))$, 
\begin{eqnarray*}
&&Q^{+}f^{(2)}(t,x_{1},v_{1}) \\
&=&\int_{v_{2}}\int_{\omega \in S^{2}}\int_{r=0}^{+\infty }|\hat
\phi(r\omega)|^2\delta (r+\omega \cdot
(v_{1}-v_{2}))rf^{(2)}(t,x_{1},x_{1},v_{1}^{\ast },v_{2}^{\ast })\Big| 
_{\substack{ v_{1}^{\ast }=v_{1}+y  \\ v_{2}^{\ast }=v_{2}-y }}\,dr\,d\omega
\,dv_{2}
\end{eqnarray*}%
The $\delta $ term reduces the $r$ integration to evaluation at $r=\omega
\cdot (v_{2}-v_{1})$ when this quantity is positive: 
\begin{eqnarray*}
&&Q^{+}f^{(2)}(t,x_{1},v_{1}) \\
&=&\int_{v_{2}}\int_{\substack{ \omega \in S^{2}  \\ \omega \cdot
(v_{2}-v_{1})>0}}|\hat \phi(r\omega)|^2rf^{(2)}(t,x_{1},x_{1},v_{1}^{\ast
},v_{2}^{\ast })\Big|_{\substack{ v_{1}^{\ast }=v_{1}+r\omega  \\ %
v_{2}^{\ast }=v_{2}-r\omega }}\Big|_{r=\omega \cdot (v_{2}-v_{1})}\,d\omega
\,dv_{2}
\end{eqnarray*}%
By even extension 
\begin{eqnarray*}
&&Q^{+}f^{(2)}(t,x_{1},v_{1}) \\
&=&\frac{1}{2}\int_{v_{2}}\int_{\omega \in S^{2}}|\hat
\phi(r\omega)|^2|r|f^{(2)}(t,x_{1},x_{1},v_{1}^{\ast },v_{2}^{\ast })\Big| 
_{\substack{ v_{1}^{\ast }=v_{1}+r\omega  \\ v_{2}^{\ast }=v_{2}-r\omega }}%
\Big|_{r=\omega \cdot (v_{2}-v_{1})}\,d\omega \,dv_{2}
\end{eqnarray*}%
as needed.
\end{proof}

\subsection{Limiting collision operator estimates}

Analogous to Proposition \ref{P:QEstimates}, we have the following, which is
proved by the same methods as Proposition \ref{P:QEstimates}.

\begin{proposition}[$Q^{(k+1)}$ estimates in $L^2$]
\label{P:Q0Estimates} The operator $Q^{(k+1)}$ defined by \eqref{E:S435}
(with components \eqref{E:S434}) satisfies the bound 
\begin{equation*}
\| \check Q^{(k+1)} \check f^{(k+1)}(t) \|_{ L_{\boldsymbol{\eta}_k}^2 L_{%
\boldsymbol{\xi}_k}^2} \lesssim \sum_{i=1}^k \| \langle \eta_i
\rangle^{\frac34+} \langle \eta_{k+1} \rangle^{\frac34+} \langle
\nabla_{\xi_{k+1}} \rangle^{\frac12+} \check f^{(k+1)}(t, \boldsymbol{\eta}%
_{k+1}, \boldsymbol{\xi}_{k+1}) \|_{L_{\boldsymbol{\eta}_{k+1}}^2 L_{%
\boldsymbol{\xi}_{k+1}}^2}
\end{equation*}
uniformly in $t$.
\end{proposition}

Proposition \ref{P:general-fixed-time-2} below is an $L^p$ based bilinear
estimate that is needed in later sections with $p>2$. Its proof requires use
of the Littlewood-Paley square function.

\begin{lemma}[Littlewood-Paley square function estimate]
\label{L:squarefcn} Let $\chi(\xi)\geq 0$ be a smooth function with support
contained in $\{ \; \xi \; : \; \frac12< |\xi| < 2 \;\}$ such that there
exists $C\geq 1$ for which 
\begin{equation}  \label{E:SF03}
\forall \; \xi \neq 0 \,, \qquad C^{-1} \leq \sum_{M\in 2^{\mathbb{Z}}}
\chi(\xi/M) \leq C
\end{equation}
Let $P_M$ be the Fourier multiplier operator with symbol $\chi(\xi/M)$. Then
there exists a constant $C_p \geq 1$ such that 
\begin{equation*}
C_p^{-1}\| f\|_{L^p} \leq \|Sf\|_{L^p} \leq C_p\|f\|_{L^p}
\end{equation*}
where the \emph{square function} $S$ is 
\begin{equation}  \label{E:SF05}
S f = \left( \sum_{M \in 2^{\mathbb{Z}}} |P_M f|^2 \right)^{1/2}
\end{equation}
\end{lemma}

\begin{proposition}[$Q^{(k+1)}$ estimates in $L^p$]
\label{P:general-fixed-time-2} For any $2\leq p <\infty$ and $0 \leq r \leq
1 $, we have the following bound for fixed $t$, $\boldsymbol{x}_k$, and $%
(\xi_1, \ldots, \xi_{i-1},\xi_{i+1}, \ldots, \xi_k)$.: 
\begin{equation*}
\| L\, \langle \nabla_{\xi_i} \rangle^r \, \tilde Q_{i,k+1} \tilde
f^{(k+1)}(t, \boldsymbol{x}_k, \boldsymbol{\xi}_k) \|_{L^p_{\xi_i}} \lesssim
\| L \, \langle \nabla_{\xi_i} \rangle^r \, \tilde f^{(k+1)}(t, \boldsymbol{x%
}_k, x_i, \boldsymbol{\xi}_{k+1}) \|_{(L_{\xi_{k+1}}^{3+}\cap
L_{\xi_{k+1}}^{3-})L_{\xi_i}^p}
\end{equation*}
where $L$ is any (possibly fractional) derivative operator in $\boldsymbol{x}%
_k$ and/or $(\xi_1, \ldots, \xi_{i-1},\xi_{i+1}, \ldots, \xi_k)$.. The bound
is uniform in $t$, $\boldsymbol{x}_k$, and $(\xi_1, \ldots,
\xi_{i-1},\xi_{i+1}, \ldots, \xi_k)$.
\end{proposition}

\begin{proof}
Since the $L$ operator carries directly onto $\tilde f^{(k+1)}$, we might as
well take $L=I$. From \eqref{E:S444}, we see that if $\alpha=\sigma$, the $%
\langle \nabla_{\xi_i} \rangle^r$ passes directly onto $\tilde f^{(k+1)}$.
Thus let us assume that $(\alpha,\sigma) \in \{(-1,1), (1,-1)\}$. Both cases
are similar, so for convenience we will take $(\alpha,\sigma) = (-1,1)$. 
\begin{equation*}
\tilde Q_{i,k+1} \tilde f^{(k+1)}(t, \boldsymbol{x}_k, \boldsymbol{\xi}_k) =
\int_{s=0}^{\infty} \int_y |\hat \phi(y)|^2 e^{i\xi_i y} e^{-i\sigma
s|y|^2}\tilde f^{(k+1)}(t, \boldsymbol{x}_k , x_i, \xi_1, \, . \,, \xi_i-sy,
\, . \,, \xi_k, sy ) \, dy \, ds
\end{equation*}
By Minkowski's integral inequality, we see that it suffices to consider, for
fixed $s>0$, the operator $\tilde U_s$ that acts on $\tilde g(\xi_1,\xi_2)$
and returns a function of $\xi_1$, given by 
\begin{equation}  \label{E:SF14}
\tilde U_s \tilde g(\xi_1) = \int_y |\hat\phi(y)|^2 e^{i\xi_1 y} \tilde
g(\xi_1-sy,sy) \, dy = \int_{v_1} \int_y |\hat\phi(y)|^2 e^{i\xi_1 (v_1+y)}
\hat g(v_1,sy) \, dy
\end{equation}
(where, for the purposes of this proof, $\hat g$ denotes the Fourier
transform of $\tilde g$ in $\xi_1\mapsto v_1$ only). Our goal is to show
that 
\begin{equation}  \label{E:SF10}
\| \langle \nabla_{\xi_1} \rangle^r \tilde U_s \tilde g(\xi_1)
\|_{L^p_{\xi_1}} \lesssim s^{-1+} \langle s \rangle^{0-} \| \langle
\nabla_{\xi_1}\rangle^r \tilde g(\xi_1,\xi_2) \|_{(L_{\xi_2}^{3+}\cap
L_{\xi_2}^{3-})L_{\xi_1}^p}
\end{equation}
First, we treat the case of $r=0$. By Minkowski, 
\begin{equation*}
\| \tilde U_s \tilde g(\xi_1) \|_{L^p_{\xi_1}}\lesssim \int_y
|\hat\phi(y)|^2 \|\tilde g(\xi_1,sy)\|_{L^p_{\xi_1}} \, dy
\end{equation*}
Then we do H\"older in $y$, depending on the value of $s$:

\begin{center}
\renewcommand{\arraystretch}{1.5}
\setlength{\tabcolsep}{10pt}  
\begin{tabular}{ccc|cc}
$|\hat\phi( y)|^2$ & $\|\tilde g(\xi_1, sy) \|_{L^p_{\xi_1}}$ & %
\shortstack{rescaling \\ generates} & \shortstack{use\\when} &  \\ \hline
$L^{3/2-}_y$ & $L^{3+}_y$ & $s^{-1+}$ & $s\leq 1$ &  \\ 
$L^{3/2+}_y$ & $L^{3-}_y$ & $s^{-1-}$ & $s\geq 1$ & 
\end{tabular}
\end{center}

This gives 
\begin{equation*}
\| \tilde U_s \tilde g(\xi_1) \|_{L^p_{\xi_1}}\lesssim s^{-1+}\langle s
\rangle^{0-} \|\tilde g(\xi_1,\xi_2)\|_{(L_{\xi_2}^{3+}\cap
L_{\xi_2}^{3-})L^p_{\xi_1}}
\end{equation*}
which completes the proof in the $r=0$ case.

Next, we treat the case of $r=1$. 
\begin{equation*}
\| \langle \nabla_{\xi_1} \rangle^1 \tilde U_s \tilde g(\xi_1)
\|_{L^p_{\xi_1}} \sim \| \tilde U_s \tilde g(\xi_1) \|_{L^p_{\xi_1}} + \|
\nabla_{\xi_1} \tilde U_s \tilde g(\xi_1) \|_{L^p_{\xi_1}}
\end{equation*}
The first term is treated by the $r=0$ case. For the second term, we compute
from the definition \eqref{E:SF14} of $\tilde U_s$ that 
\begin{equation*}
\nabla_{\xi_1} \tilde U_s \tilde g(\xi_1) = \int_y |\hat\phi(y)|^2 e^{i\xi_1
y} (iy + \nabla_{\xi_1})\tilde g(\xi_1-sy,sy) \, dy
\end{equation*}
This splits into two terms, the first is like the $r=0$ case with $%
|\hat\phi(y)|^2$ replaced by $y|\hat\phi(y)|^2$, and the second is like the $%
r=0$ case with $\tilde g$ replaced by $\nabla_{\xi_1}\tilde g$. This
completes the proof in the $r=1$ case.

For general $0<r <1$, we first dispose of low frequencies. 
\begin{equation*}
\| \langle \nabla_{\xi_1} \rangle^r \tilde U_s P_{\lesssim 1} \tilde
g(\xi_1) \|_{L^p_{\xi_1}} \lesssim \| \langle \nabla_{\xi_1} \rangle \tilde
U_s P_{\lesssim 1} \tilde g(\xi_1) \|_{L^p_{\xi_1}}
\end{equation*}
and the proof is completed by appealing to the $r=1$ case and using that 
\begin{equation*}
\| \langle \nabla_{\xi_1} \rangle P_{\lesssim 1} \tilde g(\xi_1,\xi_2)
\|_{L^p_{\xi_1}} \lesssim \| \langle \nabla_{\xi_1} \rangle^r \tilde
g(\xi_1,\xi_2) \|_{L^p_{\xi_1}}
\end{equation*}
Next, note that 
\begin{equation*}
\| P_{\lesssim 1} \langle \nabla_{\xi_1} \rangle^r \tilde U_s \tilde
g(\xi_1) \|_{L^p_{\xi_1}} \lesssim \| \tilde U_s \tilde g(\xi_1)
\|_{L^p_{\xi_1}}
\end{equation*}
and thus the proof reduces to the case of $r=0$.

In view of the above considerations, it suffices to prove 
\begin{equation}  \label{E:SF13}
\| P_{\gtrsim 1} \langle \nabla_{\xi_1} \rangle^r \tilde U_s P_{\gtrsim 1}
\tilde g(\xi_1) \|_{L^p_{\xi_1}} \lesssim s^{-1+} \langle s \rangle^{0-} \|
\langle \nabla_{\xi_1}\rangle^r \tilde g(\xi_1,\xi_2)
\|_{(L_{\xi_2}^{3+}\cap L_{\xi_2}^{3-})L_{\xi_1}^p}
\end{equation}
By the $L^p \to L^p$ boundedness of the operators 
\begin{equation*}
P_{\gtrsim 1} \langle \nabla_{\xi_1} \rangle^r |\nabla_{\xi_1}|^{-r} \quad 
\text{ and } \quad P_{\gtrsim 1} \langle \nabla_{\xi_1} \rangle^{-r}
|\nabla_{\xi_1}|^r
\end{equation*}
(which follows by the Mikhlin multiplier theorem) it suffices to prove the
analogous result with homogeneous derivative operators: 
\begin{equation*}
\| P_{\gtrsim 1} |\nabla_{\xi_1}|^r \tilde U_s P_{\gtrsim 1} \tilde g(\xi_1)
\|_{L^p_{\xi_1}} \lesssim s^{-1+} \langle s \rangle^{0-} \|
|\nabla_{\xi_1}|^r \tilde g(\xi_1,\xi_2) \|_{(L_{\xi_2}^{3+}\cap
L_{\xi_2}^{3-})L_{\xi_1}^p}
\end{equation*}

Take any smooth $\chi(\xi)$ as in the statement of Lemma \ref{L:squarefcn}
with the property that\footnote{%
One can construct such a $\chi$ as follows. Take any smooth $0\leq
\chi^0(\xi)\leq 1$ with $\func{supp} \chi^0 \subset \{ \; \xi \; : \;
\frac12 < |\xi| < 2\}$ such that $\chi^0(\xi)=1$ on $A=\{ \; \xi \; : \; 
\frac{1}{\sqrt{2}} \leq |\xi| \leq \sqrt 2 \; \}$. Let $m(\xi) = \sum_{M \in
2^\mathbb{Z}} \chi^0(\xi/M)$. For $\xi\in A$, $m(\xi) = \chi^0(2\xi)+
\chi^0(\xi) +\chi^0(\xi/2)$ since all other terms vanish. Thus $1 \leq
m(\xi) \leq 3$ for $\xi \in A$. From the definition of $m(\xi)$, we have $%
m(2\xi)=m(\xi)$. Thus $1\leq m(\xi)\leq 3$ for all $\xi \neq 0$. Now let $%
\chi(\xi) = \chi^0(\xi)/m(\xi)$. It follows that $\frac13 \leq \chi(\xi)
\leq 1$ on $A$, that $\sum_{M\in 2^{\mathbb{Z}}} \chi(\xi/M) =1$ for $%
\xi\neq 0$, and that $\func{supp} \chi \subset \{ \; \xi \; : \; \frac12
\leq |\xi| \leq 2\}$.} 
\begin{equation*}
\forall \; \xi \neq 0 \,, \qquad \sum_{M\in 2^{\mathbb{Z}}} \chi(\xi/M) =1
\end{equation*}
which is stronger than \eqref{E:SF03} and also implies that 
\begin{equation}  \label{E:SF04}
\sum_{R\in 2^{\mathbb{Z}}} P_R= I
\end{equation}

For each $M\in 2^{\mathbb{Z}}$, 
\begin{equation}  \label{E:SF11}
\begin{aligned} P_M |\nabla_{\xi_1} |^r \tilde U_s P_{\gtrsim 1} &= P_M
|\nabla_{\xi_1} |^r \tilde U_s (P_{M/2}+P_M+P_{2M}) P_{\gtrsim 1} \\ &\qquad
+ P_M |\nabla_{\xi_1} |^r \tilde U_s [I-(P_{M/2}+P_M+P_{2M})]P_{\gtrsim 1}
\end{aligned}
\end{equation}
Let $\tilde\chi(\xi_1) = |\xi|^r\chi(\xi_1)$ and let 
\begin{equation*}
\tilde{\tilde \chi}(\xi_1) = |\xi|^{-r}(\chi(2\xi_1) + \chi(\xi_1) +
\chi(2\xi_1))
\end{equation*}
Let $\tilde P_M$ be the Fourier multiplier with symbol $\tilde \chi(\xi_1/M)$%
, and let $\tilde{\tilde P}_M$ be the Fourier multiplier with symbol $\tilde{%
\tilde \chi}(\xi_1/M)$. Then 
\begin{equation}  \label{E:SF12}
P_M |\nabla_{\xi_1} |^r \tilde U_s (P_{M/2}+P_M+P_{2M}) P_{\gtrsim 1} =
\tilde P_M \tilde U_s \tilde{\tilde P}_M P_{\gtrsim 1}|\nabla_{\xi_1}|^r
\end{equation}
by exchanging $M^r$ from the left $P$ operator to the right $P$ operator.
Note how this effectively moves the $|\nabla_{\xi_1}|^r$ operator past $%
\tilde U_s$ while preserving exact equality. Let 
\begin{equation*}
h_j(\xi_1) = \sum_{M\in 2^{3\mathbb{Z}+j}} \tilde{\tilde \chi}(\xi_1/M) \,,
\qquad j=0,1,2
\end{equation*}
(In other words, $j=0$ sums over $M = \ldots, \frac18, 1, 8, \ldots$, while $%
j=1$ sums over $M=\ldots, \frac14, 2, 16, \ldots$ and $j=2$ sums over $M=
\ldots, \frac{1}{16}, \frac12, 4, 32, \ldots$.). It follows $h_j(\xi) =
h_j(8\xi)$ and there exists a constant $C\geq 1$ such that for $j=0,1,2$, we
have\footnote{%
Indeed, following the construction of $\chi(\xi)$ given in the previous
footnote and the definition of $\tilde{\tilde \chi}(\xi)$, we have $%
2^{-r}\cdot\frac{1}{2} \leq \tilde{\tilde \chi}(\xi)$ for $\frac{1}{2\sqrt{2}%
} \leq |\xi| \leq 2\sqrt 2$, that $\func{supp} \tilde{\tilde \chi} \subset
\{ \; \xi \; : \; \frac14 \leq |\xi| \leq 4 \; \}$, and that for all $\xi$, $%
0\leq \tilde{\tilde \chi}(\xi) \leq 3\cdot 4^r$. Summing over $M\in 2^{3%
\mathbb{Z}+j}$, at most 3 copies overlap for any given $\xi$, so $%
h_j(\xi)\leq 9\cdot 4^r$. By the definition of $h_j(\xi)$, we have the
periodicity $h_j(\xi)=h_j(8\xi)$, and from this, \eqref{E:SF15} follows.} 
\begin{equation}  \label{E:SF15}
C^{-1} \leq h_j(\xi_1) \leq C \qquad\text{ for }\xi_1\neq 0
\end{equation}
Let $H_j$ be the Fourier multiplier operator with symbol $h_j$. Then 
\begin{equation*}
H_j = \sum_{M\in 2^{3\mathbb{Z}+j}} \tilde{\tilde P}_M
\end{equation*}
and each $H_j$ should be thought of as a near-identity operator.
Substituting 
\begin{equation*}
\tilde{\tilde P}_M = H_j - (H_j-\tilde{\tilde P}_M) \,, \qquad j= \log_2M %
\mod 3
\end{equation*}
into \eqref{E:SF12}, we obtain 
\begin{equation*}
P_M |\nabla_{\xi_1} |^r \tilde U_s (P_{M/2}+P_M+P_{2M}) P_{\gtrsim 1} =
\tilde P_M \tilde U_s H_j P_{\gtrsim 1}|\nabla_{\xi_1}|^r - \tilde P_M
\tilde U_s (H_j-\tilde{\tilde P}_M) P_{\gtrsim 1}|\nabla_{\xi_1}|^r
\end{equation*}
Plug this into \eqref{E:SF11} to obtain 
\begin{equation*}
P_{\gtrsim 1} P_M |\nabla_{\xi_1} |^r \tilde U_s P_{\gtrsim 1} = V_M^1 +
V_M^2 + V_M^3
\end{equation*}
where, again with $j=\log_2M \mod 3$, we define the components: 
\begin{align*}
&V_M^1 \overset{\mathrm{def}}{=} P_{\gtrsim 1} \tilde P_M \tilde U_s H_j
P_{\gtrsim 1}|\nabla_{\xi_1}|^r \\
&V_M^2 \overset{\mathrm{def}}{=} - P_{\gtrsim 1} \tilde P_M \tilde U_s (H_j-%
\tilde{\tilde P}_M) P_{\gtrsim 1}|\nabla_{\xi_1}|^r \\
&V_M^3 \overset{\mathrm{def}}{=} P_{\gtrsim 1} P_M |\nabla_{\xi_1} |^r
\tilde U_s [I-(P_{M/2}+P_M+P_{2M})]P_{\gtrsim 1}
\end{align*}
By Lemma \ref{L:squarefcn} for the collection $\{P_M\}_M$ and the triangle
inequality, 
\begin{align*}
\| P_{\gtrsim 1} |\nabla_{\xi_1} |^r \tilde U_s P_{\gtrsim 1} \tilde
g(\xi_1) \|_{L^p_{\xi_1}} &\lesssim \| P_{\gtrsim 1}P_M |\nabla_{\xi_1} |^r
\tilde U_s P_{\gtrsim 1} \tilde g(\xi_1) \|_{L^p_{\xi_1} \ell^2_{M\in 2^{%
\mathbb{Z}}}} \\
&\lesssim \| V_M^1\tilde g(\xi_1) \|_{L^p_{\xi_1} \ell^2_{M\in 2^{\mathbb{Z}%
}}} +\| V_M^2\tilde g(\xi_1) \|_{L^p_{\xi_1} \ell^2_{M\in 2^{\mathbb{Z}}}} +
\| V_M^3\tilde g(\xi_1) \|_{L^p_{\xi_1} \ell^2_{M\in 2^{\mathbb{Z}}}}
\end{align*}
By the triangle inequality, we can split the norm on the main term $%
V_M^1\tilde g(\xi_1)$ according to the partition $\mathbb{Z} = 3\mathbb{Z}
\cup (3\mathbb{Z}+1) \cup (3\mathbb{Z}+2)$ as 
\begin{equation*}
\| V_M^1\tilde g(\xi_1) \|_{L^p_{\xi_1} \ell^2_{M\in 2^{\mathbb{Z}}}} \leq
\| V_M^1\tilde g(\xi_1) \|_{L^p_{\xi_1} \ell^2_{M\in 2^{3\mathbb{Z}}}} + \|
V_M^1\tilde g(\xi_1) \|_{L^p_{\xi_1} \ell^2_{M\in 2^{3\mathbb{Z}+1}}} + \|
V_M^1\tilde g(\xi_1) \|_{L^p_{\xi_1} \ell^2_{M\in 2^{3\mathbb{Z}+2}}}
\end{equation*}
For each of the terms on the right-side, $j=\log_2(M)\mod 3$ is a constant,
and thus Lemma \ref{L:squarefcn} with respect to the collection $\{\tilde
P_M \}_M$ can be applied to each term separately to give 
\begin{equation*}
\| V_M^1\tilde g(\xi_1) \|_{L^p_{\xi_1} \ell^2_{M\in 2^{\mathbb{Z}}}}
\lesssim \sum_{j=0}^2 \| \tilde U_s H_j P_{\gtrsim 1}|\nabla_{\xi_1}|^r
\tilde g \|_{L^p_{\xi_1}}
\end{equation*}
By the $r=0$ case of the estimate and the $L^p\to L^p$ boundedness of $H_j$
(which follows by the Mikhlin multiplier theorem) 
\begin{equation*}
\| V_M^1\tilde g(\xi_1) \|_{L^p_{\xi_1} \ell^2_{M\in 2^{\mathbb{Z}}}}
\lesssim \| P_{\gtrsim 1}|\nabla_{\xi_1}|^r \tilde g \|_{(L^{3+}_{\xi_2}\cap
L^{3-}_{\xi_2})L^p_{\xi_1}}
\end{equation*}
It remains only to treat the error terms $V_M^2\tilde g(\xi_1)$ and $%
V_M^3\tilde g(\xi_1)$. The proof for $V_M^2$ relies on the fact that 
\begin{equation}  \label{E:SF16}
H_j - \tilde{\tilde P}_M = \sum_{\substack{ R\in 2^{3\mathbb{Z}+j}  \\ R\neq
M}} \tilde{\tilde P}_R \,, \quad j = \log_2 M \mod 3
\end{equation}
and the proof for $V_M^3$ relies on the fact that 
\begin{equation}  \label{E:SF17}
I - P_M = \sum_{\substack{ R\in 2^M  \\ R\notin\{M/2,M,2M\}}} P_R
\end{equation}
Letting $\sigma(A)$ denote the symbol associated to the operator $A$, the
key property of \eqref{E:SF16} and \eqref{E:SF17} is the following of
separation of supports: for $R\in 2^\mathbb{Z}$ but $R\notin \{M/2,M,2M\}$: 
\begin{equation*}
\func{dist}(\func{supp} \sigma(P_M), \func{supp} \sigma(P_R)) \sim \max(M,R)
\end{equation*}
and for $R\in 2^{3\mathbb{Z}+j}$, $j=\log_2M \mod 3$, and $R\neq M$, 
\begin{equation*}
\func{dist}(\func{supp} \sigma(\tilde P_M), \func{supp} \sigma(\tilde{\tilde
P}_R)) \sim \max(M,R)
\end{equation*}
Since both proofs are similar, we will just complete the proof of the
estimate for $V_M^3$. By \eqref{E:SF14}, it follows that 
\begin{align*}
P_M \tilde U_s P_R \tilde g(\xi_1) &= P_M \int_y |\hat\phi(y)|^2 e^{i\xi_1
y} (P_R \tilde g)(\xi_1-sy,sy) \, dy \\
&= \int_{v_1} \int_y \chi(v_1/M) \chi((v_1+y)/R) |\hat\phi(y)|^2 e^{i\xi_1
(v_1+y)} \hat g(v_1,sy) \, dy \\
&= \int_{v_1}\int_{y\,, \; |y| \sim \max(M,R)} \chi(v_1/M)\chi((v_1+y)/R)
|\hat\phi(y)|^2 e^{i\xi_1 (v_1+y)} \hat g(v_1,sy) \, dy \\
&= P_M \int_{y\,, |y|\sim \max(M,R)} |\hat\phi(y)|^2 e^{i\xi_1 y} (P_R
\tilde g)(\xi_1-sy,sy) \, dy
\end{align*}
By following the proof of the $r=0$ case, but also using that 
\begin{equation*}
\| \hat \phi(y) \|_{L_{|y|\sim \max(M,R)}^q} \sim \frac{1}{\max(M,R)} \| |y|
\hat \phi(y) \|_{L_{|y|\sim \max(M,R)}^q}
\end{equation*}
we obtain 
\begin{equation}  \label{E:SF18}
\|P_M\tilde U_s P_R \tilde g(\xi_1) \|_{L^p_{\xi_1}} \lesssim \frac{1}{%
\max(M,R)} s^{-1+} \langle s\rangle^{0-} \| P_R \tilde g
\|_{(L^{3-}_{\xi_2}\cap L^{3+}_{\xi_2})L^p_{\xi_1}}
\end{equation}
By Minkowski's integral inequality 
\begin{align*}
\| V_M^3\tilde g \|_{L^p_{\xi_1} \ell^2_{M\in 2^{\mathbb{Z}}}} & = \left\|
\sum_R P_{\gtrsim 1} P_M |\nabla_{\xi_1}|^r \tilde U_s P_R P_{\gtrsim 1}
\tilde g(\xi_1) \right\|_{L^p_{\xi_1} \ell^2_{M\in 2^{\mathbb{Z}}}} \\
&\lesssim \left\| P_{\gtrsim 1} P_M |\nabla_{\xi_1}|^r \tilde U_s P_R
P_{\gtrsim 1} \tilde g(\xi_1) \right\|_{ \ell^2_{M\in 2^{\mathbb{Z}}}
\ell^1_{R\in 2^\mathbb{Z}} L^p_{\xi_1}}
\end{align*}
By \eqref{E:SF18}, 
\begin{align*}
\| V_M^3\tilde g \|_{L^p_{\xi_1} \ell^2_{M\in 2^{\mathbb{Z}}}}&\lesssim
\left\| \frac{M^r}{R^r\max(M,R)} \right\|_{ \ell^2_{M\in 2^{\mathbb{N}_0}}
\ell^1_{R\in 2^{\mathbb{N}_0}} } s^{-1+}\langle s \rangle^{0-}
\||\nabla_{\xi_1}|^r P_{\gtrsim 1}\tilde g\|_{(L^{3-}_{\xi_2}\cap
L^{3+}_{\xi_2})L^p_{\xi_1}}
\end{align*}
where $2^{\mathbb{Z}}$ has been replaced by $2^{\mathbb{N}_0}$ due to the $%
P_{\gtrsim 1}$ projections. The indicated double norm in $M$, $R$ is finite
for $r<1$ completing the proof.
\end{proof}

\subsection{Permutation coordinates and associated norms\label{S:PC}}

For a given $\pi \in S_k$, below we introduce a transformed coordinate
system, and then use it to define an associated norm $X_{\pi}$. In the case $%
\pi=I=\text{Identity}$, the transformation is the identity -- we use the
original coordinates $(\boldsymbol{x}_k, \boldsymbol{v}_k)$, and the
associated norm 
\begin{equation*}
X_{I} = H_{\boldsymbol{x}_k}^{1+} (L_{\boldsymbol{v}_k}^{2,\frac12+} \cap L_{%
\boldsymbol{v}_k}^1 \cap L_{\boldsymbol{v}_k}^{\infty, 2+})
\end{equation*}
It is easiest to describe the transformation starting from the $(\boldsymbol{%
x}_k, \boldsymbol{\xi}_k)$ coordinate system. Given $\pi$, we introduce new
coordinates 
\begin{equation}  \label{E:PC10}
p_j^\pi = \tfrac12(x_j+x_{\pi(j)})+\tfrac12\epsilon(\xi_j-\xi_{\pi(j)})
\,,\qquad q_j^\pi = \tfrac12(\xi_j +
\xi_{\pi(j)})+\tfrac12(x_j-x_{\pi(j)})/\epsilon
\end{equation}
for $1\leq j \leq k$. Notice that, for any particular $j$, 
\begin{equation*}
\pi(j)=j \qquad \implies \qquad p_j^\pi = x_j \,, \; q_j^\pi=\xi_j
\end{equation*}
and thus if $\pi=\text{Identity}$, then $(\boldsymbol{p}_k^\pi, \boldsymbol{q%
}_k^\pi) = (\boldsymbol{x}_k, \boldsymbol{\xi}_k)$. Let 
\begin{equation*}
\tilde g_{N,\pi}^{(k)}(\boldsymbol{p}_k^\pi, \boldsymbol{q}_k^\pi) = \tilde
f_{N,\pi}^{(k)}(\boldsymbol{x}_k, \boldsymbol{\xi}_k)
\end{equation*}
meaning that when the function $\tilde f_{N,\pi}^{(k)}(\boldsymbol{x}_k, 
\boldsymbol{\xi}_k)$ is reexpressed in terms of the coordinates $(%
\boldsymbol{p}_k^\pi, \boldsymbol{q}_k^\pi)$, it will be denoted by $\tilde
g_{N,\pi}^{(k)}(\boldsymbol{p}_k^\pi, \boldsymbol{q}_k^\pi)$ to avoid
confusion. The norm $X_\pi$ is defined to be 
\begin{equation*}
\| \tilde f_{N,\pi}^{(k)} \|_{X_{\pi}} = \| \tilde g_{N,\pi}^{(k)} \|_{H_{%
\boldsymbol{p}_k^\pi}^{1+} H_{\boldsymbol{q}_k^\pi}^{1/2+}}
\end{equation*}
Note that since the transformation \eqref{E:PC10} is $\epsilon$-dependent,
typically 
\begin{equation*}
\| \tilde f_{N,\pi}^{(k)} \|_{X_{\pi}} \not\sim \| \tilde f_{N,\pi}^{(k)}
\|_{X_I}
\end{equation*}
meaning the the comparability bounds are not uniform in $\epsilon$, and in
fact there are examples of functions $\tilde f_{N,\pi}^{(k)}$ for which $\|
\tilde f_{N,\pi}^{(k)} \|_{X_{\pi}}$ remains bounded as $N\to \infty$ while $%
\| \tilde f_{N,\pi}^{(k)} \|_{X_I}\to +\infty$ as $N\to \infty$. An example
is easily given in the case $\pi=(12)$ using quasi-free states.

\begin{definition}
A density $\tilde f_{N}^{(k)}(\boldsymbol{x}_k, \boldsymbol{\xi}_k)$ is
called \emph{quasi-free} if 
\begin{equation}  \label{E:PC11}
\tilde f_N^{(k)} = \sum_{\pi \in S^k} \tilde f_{N,\pi}^{(k)}
\end{equation}
with 
\begin{equation*}
\tilde f_{N,\pi}^{(k)}(\boldsymbol{x}_k, \boldsymbol{\xi}_k) = \prod_{j=1}^k
\tilde g_0(p_j^\pi,q_j^\pi)
\end{equation*}
for some $\tilde g_0$ independent of $N$ (and thus independent of $\epsilon$%
). A density $\tilde f_{N}^{(k)}(\boldsymbol{x}_k, \boldsymbol{\xi}_k)$ is 
\emph{generalized quasi-free} if \eqref{E:PC11} holds with 
\begin{equation}  \label{E:PC12}
\tilde f_{N,\pi}^{(k)}(\boldsymbol{x}_k, \boldsymbol{\xi}_k) = \tilde
g_{\pi}^{(k)} (\boldsymbol{p}_k^\pi,\boldsymbol{q}_k^\pi)
\end{equation}
where $\tilde g_\pi^{(k)}$ is independent of $N$ (and $\epsilon$). In other
words, it is quasi-free but we do not assume factorization in the $(%
\boldsymbol{p}_k, \boldsymbol{q}_k)$ coordinates.
\end{definition}

\begin{example}[2-cycles]
\label{EX:2cycles} An important case is when $\pi\in S_k$ contains a
2-cycle. To simplify matters, let us consider $k=2$ and $\pi=(12)$. Then 
\begin{equation*}
\begin{aligned} & p_1 &= \tfrac12(x_1+x_2) + \tfrac12\epsilon(\xi_1-\xi_2) &
\qquad & q_1 &= \tfrac12(\xi_1+\xi_2) + \tfrac12(x_1-x_2)/\epsilon \\ & p_2
&= \tfrac12(x_2+x_1) + \tfrac12\epsilon(\xi_2-\xi_1) & \qquad & q_2 &=
\tfrac12(\xi_2+\xi_1) + \tfrac12(x_2-x_1)/\epsilon \end{aligned}
\end{equation*}
The conversion $(x_1,x_2,\xi_1,\xi_2) \leftrightarrow (p_1,p_2,q_1,q_2)$ is
clearly $\epsilon$-dependent. However, the conversion $(p_1,p_2,q_1,q_2)
\leftrightarrow (y_{12},y_1, \mu_{12},\mu_1)$ is $\epsilon$ independent,
where 
\begin{equation*}
\begin{aligned} & y_{12} = p_1+p_2 = x_1+x_2 & \qquad & \mu_{12} = q_1+q_2 =
\xi_1+\xi_2 \\ & y_1 = q_1-q_2 = (x_1-x_2)/\epsilon & \qquad & \mu_1 =
p_1-p_2 = \epsilon(\xi_1-\xi_2) \\ \end{aligned}
\end{equation*}
In the coordinates $(y_{12},y_1, \mu_{12}, \mu_1)$ the geometrical
distortion is easier to see. Indeed, if 
\begin{equation*}
|y_{12}|\lesssim 1 \, \qquad |y_1| \lesssim 1 \,, \qquad |\mu_{12}|\lesssim
1 \,, \quad |\mu_1| \lesssim 1
\end{equation*}
then 
\begin{equation*}
| x_1+x_2| \lesssim 1 \,, \qquad |x_1-x_2| \lesssim \epsilon \,, \quad
|\xi_1+\xi_2|\lesssim 1 \,, \qquad |\xi_1-\xi_2| \lesssim \epsilon^{-1}
\end{equation*}
In particular, it is possible that $|\xi_j|\sim \epsilon^{-1}$ for either $%
j=1$ and/or $j=2$. And if $\eta_1, \eta_2$ denote the Fourier dual variables
to $x_1,x_2$, then likewise the induced effective support properties are 
\begin{equation*}
|\eta_1+\eta_2| \lesssim 1 \,, \qquad |\eta_1-\eta_2|\lesssim \epsilon^{-1}
\end{equation*}
Let us write 
\begin{equation*}
\tilde f_{N,(12)}^{(2)}(x_1,x_2,\xi_1,\xi_2) = \tilde
g_{N,(12)}^{(2)}(p_1,p_2, q_1,q_2) = \tilde h_{N,(12)}^{(2)}(y_{12},y_1,
\mu_{12},\mu_1)
\end{equation*}
Now suppose that $\tilde g_{N,(12)}^{(2)}$ is taken to be a smooth compactly
supported function \emph{independent of $N$ (and thus $\epsilon=N^{-1/3}$)}.
It follows that $\tilde h_{N,(12)}^{(2)}$ is also a smooth compactly
supported function \emph{independent of $N$ (and thus $\epsilon=N^{-1/3}$)}.
However, the support of $\tilde f_{N,(12)}^{(2)}$ will vary with $\epsilon$,
and as a result of the tight separation of $x_1,x_2$, derivatives in $x_1$
or $x_2$ will generate $\epsilon^{-1}$ factor losses. Even in an ideal
situation, where one has restricted the support of $\xi_1$, $\xi_2$, we have 
\begin{equation}  \label{E:irregularity}
\| \langle \xi_1 \rangle^{-\frac32-} \langle \xi_2\rangle^{-\frac32-}
|\nabla_{x_1}|^s |\nabla_{x_2}|^s \tilde f_{N,(12)}^{(2)} \|_{L_{\boldsymbol{%
x}_2}^2 L_{\boldsymbol{\xi}_2}^2} \sim \epsilon^{\frac32-2s}
\end{equation}
Thus, one has only that $\tilde f_{N,(12)}^{(2)}$ is uniformly bounded in $%
H_{\boldsymbol{x}_k}^{3/4}$, but grows in $H_{\boldsymbol{x}_k}^s$ for $%
s>\frac34$. In this sense, the term is \emph{irregular}. In fact, when we do
not restrict the $\xi$ support, the situation is even worse: 
\begin{equation*}
\| f_N^{(2)} \|_{X_I} \sim \epsilon^{-2-} \, \qquad \text{while} \qquad \|
f_N^{(2)} \|_{X_{(12)}} \sim 1
\end{equation*}
Thus it is essential to estimate such a density in the $X_{(12)}$ norm
rather than the $X_{I}$ norm.
\end{example}

The hypothesis of our main theorem is that the given densities can be, at
each time $t$, decomposed into a sum 
\begin{equation}  \label{E:PC01}
f_N^{(k)}(t) = \sum_{ \pi \in S_k } f_{N,\pi}^{(k)}(t)
\end{equation}
(where each $f_{N,\pi}^{(k)}(t)$ is not necessarily quasi-free) but there
exists a constant $C$ so that 
\begin{equation}  \label{E:PC01b}
\forall \; N\geq 0 , \; t\in [0,T]\,, \; \pi \in S_k \,, \quad \text{ we
have }\quad \| f_{N,\pi}(t) \|_{X_{\pi}} \leq C^k
\end{equation}
The term $f_{N,I}$ corresponding to $\pi=I =\text{Identity}$ is called the 
\emph{core}, and the analysis ultimately shows that it is the only term that
has a nontrivial limit as $N\to \infty$. It is not assumed that the
decomposition \eqref{E:PC01} is unique. It should be noted that %
\eqref{E:PC01b} holds when $f_N^{(k)}(t)$ is generalized quasi-free, where
the $g_{\pi}^{(k)}$ terms in \eqref{E:PC12} are allowed to be time dependent
but are assumed to be independent of $N$.

Another hypothesis is needed for the main theorem regarding $f_{N,\pi}^{(k)}$
when $\pi$ contains one or more $2$-cycles. It is a symmetry condition that
must hold on an $\epsilon^{1/2+}$-dense set of times. For example, if $%
\pi=(12)$, then 
\begin{equation}  \label{E:PC02}
f_{N,\pi}^{(k)}(t, x_1,x_2, \boldsymbol{x}_{k-3}, \xi_1,\xi_2, \boldsymbol{%
\xi}_{k-3}) = f_{N,\pi}^{(k)}(t, x_2,x_1, \boldsymbol{x}_{k-3}, \xi_2,\xi_1, 
\boldsymbol{\xi}_{k-3})
\end{equation}
In other words, $(x_1,x_2)$ and $(\xi_1,\xi_2)$ are simultaneously flipped
to $(x_2,x_1)$ and $(\xi_2,\xi_1)$ while other coordinates remain unchanged.
We note that this property is only required to hold on an $\epsilon^{1/2+}$
dense set of times. Specifically, there must exist a subset of times $%
\{t_u\} $ on which \eqref{E:PC02} holds with the property that for any $t\in
[0,T]$, there exists a $u$ such that $|t-t_u| \leq \epsilon^{1/2+}$. The
fact that \eqref{E:PC02} is not required to hold for all $t$ resolves the
trivial limit puzzle, as discussed in \S \ref{S3}, and there it is further
argued that in a collisional environment, \eqref{E:PC02} can only be
expected to hold on an $\epsilon$ dense set of times and typically (on a
time set of large measure) \eqref{E:PC02} does not hold.

Consider now the Schr\"odinger coordinates 
\begin{equation*}
y_j = \tfrac12 x_j-\tfrac12 \epsilon \xi_j \qquad y_j^{\prime} = \tfrac12
x_j+\tfrac12\epsilon \xi_j
\end{equation*}
Then 
\begin{equation*}
x_j = y_j^{\prime}+y_j \qquad \xi_j = (y_j^{\prime}-y_j)/\epsilon
\end{equation*}
whereas 
\begin{equation*}
p_j^\pi = y_j^{\prime}+y_{\pi(j)} \qquad q_j^\pi = (y_j^{\prime} -
y_{\pi(j)})/\epsilon
\end{equation*}
By the chain rule, 
\begin{equation*}
\nabla_{y_{\pi(j)}} = \nabla_{p_j} - \epsilon^{-1} \nabla_{q_j} \qquad
\nabla_{y_j^{\prime}} = \nabla_{p_j} + \epsilon^{-1} \nabla_{q_j}
\end{equation*}
By reindexing, 
\begin{align*}
\frac14 \epsilon \sum_{j=1}^k (\Delta_{y_j^{\prime}} - \Delta_{y_j}) &=
\frac14 \epsilon \sum_{j=1}^k (\Delta_{y_j^{\prime}} - \Delta_{y_{\pi(j)}})
\\
&= \frac14 \epsilon \sum_{j=1}^k
(\nabla_{y_j^{\prime}}+\nabla_{y_{\pi(j)}})\cdot(
\nabla_{y_j^{\prime}}-\nabla_{y_{\pi(j)}}) \\
&= \sum_{j=1}^k \nabla_{p_j} \cdot \nabla_{q_j}
\end{align*}
Therefore, the semiclassical propagators have the conversion 
\begin{equation*}
e^{\frac14 i \epsilon t (\Delta_{\boldsymbol{y}_k^{\prime}} - \Delta_{%
\boldsymbol{y}_k})} = e^{it \nabla_{\boldsymbol{p}_k}\cdot \nabla_{%
\boldsymbol{q}_k}}
\end{equation*}

\subsection{Estimates in permutation coordinates for $k=1$\label{S:PCE}}

In the case $\pi=(12)$, we have 
\begin{align*}
&p_1 = y_1^{\prime}+y_2 = \tfrac12(x_1+x_2) + \tfrac12\epsilon(\xi_1-\xi_2)
& \quad &q_1 = (y_1^{\prime}-y_2)/\epsilon = \tfrac12(\xi_1+\xi_2) +
\tfrac12(x_1-x_2)/\epsilon \\
&p_2 = y_2^{\prime}+y_1 = \tfrac12(x_2+x_1) + \tfrac12\epsilon(\xi_2-\xi_1)
& \quad &q_2 = (y_2^{\prime}-y_1)/\epsilon = \tfrac12(\xi_2+\xi_1) +
\tfrac12(x_2-x_1)/\epsilon
\end{align*}
equivalently 
\begin{align*}
&2y_1 = p_2 -\epsilon q_2 = x_1 - \epsilon\xi_1 & \qquad &2y_1^{\prime} =
p_1+\epsilon q_1 =x_1+\epsilon\xi_1 \\
&2y_2 = p_1 - \epsilon q_1 = x_2-\epsilon\xi_2 & \qquad &2y_2^{\prime} =
p_2+\epsilon q_2 = x_2+\epsilon \xi_2
\end{align*}
equivalently 
\begin{align*}
&x_1=y_1^{\prime}+y_1 = \tfrac12(p_1+p_2) + \tfrac12\epsilon(q_1-q_2) & 
\qquad &\xi_1 = (y_1^{\prime}-y_1)/\epsilon = \tfrac12(q_1+q_2) +
\tfrac12(p_1-p_2)/\epsilon \\
&x_2 = y_2^{\prime}+y_2 = \tfrac12(p_2+p_1) + \tfrac12\epsilon(q_2-q_1) & 
\qquad &\xi_2 = (y_2^{\prime}-y_2)/\epsilon = \tfrac12(q_2+q_1) +
\tfrac12(p_2-p_1)/\epsilon
\end{align*}

In the case $k=1$, the second Duhamel iterate \eqref{E:S402} takes the form 
\begin{equation}  \label{E:PCEi01}
\begin{aligned} f_N^{(1)}(t) = & S^{(1)}(t)f_N^{(1)}(0) + \mathcal{D}^{(1)}[
N\epsilon^{-1/2}B_\epsilon^{(2)} S^{(2)} f_N^{(2)}(0)] \\ &+
\mathcal{D}^{(1)}[N\epsilon^{-1/2}
B_\epsilon^{(2)}\mathcal{D}^{(2)}\epsilon^{-1/2}A_\epsilon^{(2)}f_N^{(2)}]
\\ &+ \mathcal{D}^{(1)}[N\epsilon^{-1/2} B_\epsilon^{(2)} \mathcal{D}^{(2)}
N\epsilon^{-1/2} B_\epsilon^{(3)}f_N^{(3)}] \end{aligned}
\end{equation}
In \S \ref{sec:CompactnessConvergence} we need to estimate the weak pairing
of $f_N^{(k)}(t)-f_N^{(k)}(s)$ by $|t-s|^{\alpha}$, and in \S \ref%
{sec:Convergence}, we need to estimate the weak pairing of $%
f_N^{(k)}(t)-f^{(k)}(t)$, where $f^{(k)}(t)$ is the weak limit as $N\to
\infty$, by $\epsilon^{0+}$ (the relevant topologies are defined in \S \ref%
{sec:CompactnessConvergence}). Recall also that 
\begin{equation*}
f_N^{(k)} = \sum_{\pi \in S^k} f_{N,\pi}^{(k)}
\end{equation*}
Thus in \eqref{E:PCEi01}, the right side involves 
\begin{equation*}
f_N^{(2)}(t) = f_{N,\text{Id}}^{(2)}(t) + f_{N,(12)}^{(2)}(t)
\end{equation*}
(2 terms) and 
\begin{equation*}
f_N^{(3)}(t) = f_{N,\text{Id}}^{(3)}(t) + f_{N,(12)}^{(3)}(t) +
f_{N,(13)}^{(3)}(t) + f_{N,(23)}^{(3)}(t) + f_{N,(123)}^{(3)}(t) +
f_{N,(213)}^{(3)}(t)
\end{equation*}
(6 terms). The components $f_{N,\text{Id}}^{(2)}(t)$ and $f_{N,\text{Id}%
}^{(3)}(t)$ are the core terms, and they are estimated, using the estimates
in the earlier subsections of \S \ref{S:preparation}, for general $k$ as
explained in \S \ref{sec:CompactnessConvergence}-\ref{sec:Convergence}. In
this section, we explain that the term $f_{N,(12)}^{(2)}(t)$ gives
negligible contribution, as $N\to \infty$, in the terms of \eqref{E:PCEi01}.
Specifically, for a fixed Schwartz class function $J(x_1,v_1)$, we consider 
\begin{equation*}
\text{III}_{L}(t) = \langle J, \mathcal{D}^{(1)}[
N\epsilon^{-1/2}B_\epsilon^{(2)} S^{(2)} f_{N,(12)}^{(2)}(0)](t) \rangle
\end{equation*}
in Corollary \ref{C:12linear} and 
\begin{equation*}
\text{IV}(t) = \langle J, \mathcal{D}^{(1)}[N\epsilon^{-1/2} B_\epsilon^{(2)}%
\mathcal{D}^{(2)}\epsilon^{-1/2}A_\epsilon^{(2)}f_{N,(12)}^{(2)}])(t) \rangle
\end{equation*}
in Proposition \ref{P:termIV}. Term $\text{III}(t)$ (which is $\text{III}%
_L(t)$ without the inner linear propagator) defined below, is a template
that is used in both Corollary \ref{C:12linear} and Proposition \ref%
{P:termIV}, and also used to explain the trivial limit puzzle in Remark \ref%
{R:TLP}.

First, we consider the form of the operator $\tilde B_{1,2}^\epsilon$ acting
on a $2$-density and returning a $1$-density, but reexpressed in terms of $%
(p,q)$ coordinates when $\pi=(12)$. With, as usual, 
\begin{equation*}
\tilde F_{N,(12)}^{(2)}(t^{\prime},x_1,x_2,\xi_1,\xi_2) = \tilde
G_{N,(12)}^{(2)}(t^{\prime},p_1,p_2,q_1,q_2)
\end{equation*}
we have

\begin{lemma}[$\protect\pi=(12)$ form for $B$]
\label{L:PCE1} 
\begin{align}
\text{III}(t) &= \int_{t^{\prime}=0}^t \int_{x_1,\xi_1} \tilde J(x_1,\xi_1)
e^{i(t-t^{\prime})\nabla_{x_1}\cdot \nabla_{\xi_1}} (N\epsilon^{-1/2} \tilde
B_{1,2}^\epsilon \tilde F_{N,(12)}^{(2)})(t^{\prime},x_1,\xi_1) \, dx_1 \,
d\xi_1 \, dt^{\prime}  \notag \\
& = \epsilon^{-1/2} \int_{t^{\prime}=0}^t \int_{p_1,q_1,q_2} \mathcal{J}%
(t-t^{\prime}, p_1,q_1+q_2) (\phi(q_2)-\phi(q_1))  \label{E:PCE04} \\
& \qquad\qquad \qquad \tilde G_{N,(12)}^{(2)}(t^{\prime},p_1+\epsilon q_2,
p_1-\epsilon q_1, q_1,q_2) \, dp_1 \, dq_1 \, dq_2 \, dt^{\prime}  \notag
\end{align}
where 
\begin{equation}  \label{E:PCE18}
\mathcal{J}(t-t^{\prime},p_1,q_1+q_2) = \int_{v_1} e^{iv_1(q_1+q_2)}
e^{i(t-t^{\prime})\epsilon\Delta_{p_1}} J(p_1+(t-t^{\prime})v_1,v_1)\, dv_1
\end{equation}
For any $n\geq 0$, assuming $|t|\lesssim 1$, 
\begin{equation}  \label{E:PCE05}
\left| \mathcal{J}(t-t^{\prime},p_1,q_1+q_2) \right| \lesssim_n \langle
p_1\rangle^{-n} \langle q_1+q_2 \rangle^{-n}
\end{equation}
\end{lemma}

\begin{proof}
The $\tilde B_{1,2}^\epsilon$ introduces an integral over $x_2$ and also
assigns $\xi_2=0$, equivalently 
\begin{equation}  \label{E:PCE01}
y_2=y_2^{\prime}
\end{equation}
equivalently 
\begin{equation*}
p_2=p_1-\epsilon(q_1+q_2)
\end{equation*}
We convert coordinates 
\begin{equation*}
(x_1,x_2,\xi_1) \leftrightarrow (p_1,q_1,q_2)
\end{equation*}
The differential conversion is 
\begin{equation*}
dx_1 \, dx_2 \, d\xi_1 = 2\epsilon^3 dp_1 \, dq_1 \, dq_2
\end{equation*}
and in this setting, 
\begin{equation*}
4\nabla_{x_1}\cdot \nabla_{\xi_1} = \epsilon \Delta_{p_1} + 2\cdot
\nabla_{p_1}\cdot \nabla_{q_1} + \epsilon^{-1}\Delta_{q_1} -
\epsilon^{-1}\Delta_{q_2}
\end{equation*}
In the $\tilde B_{1,2}^\epsilon$ operator, the inner potential terms are
evaluated at 
\begin{equation*}
\frac{y_1-y_2}{\epsilon} = \frac{y_1-y_2^{\prime}}{\epsilon} = -q_2 \,,
\qquad \frac{y_1^{\prime}-y_2^{\prime}}{\epsilon} = \frac{y_1^{\prime}-y_2}{%
\epsilon} = q_1
\end{equation*}
where the restriction \eqref{E:PCE01} is employed, and thus the potential
term is $\phi(-q_2) - \phi(q_1)$. Since $\phi$ is radial, $%
\phi(-q_2)=\phi(q_2)$. On the outside of the propagator, the test function
is evaluated at $(x_1,\xi_1)$, which converts as 
\begin{equation*}
\tilde J(x_1,\xi_1) = \tilde J(p_1-\epsilon q_2, q_1+q_2)
\end{equation*}
This yields the formula 
\begin{equation}  \label{E:PCE03}
\begin{aligned} \text{III}(t) & = \epsilon^{-1/2}\int_{t^{\prime}=0}^t
\int_{p_1,q_1,q_2} \tilde J(p_1, q_1+q_2) e^{i(t-t^{\prime})(\epsilon
\Delta_{p_1} + 2\cdot \nabla_{p_1}\cdot \nabla_{q_1} +
\epsilon^{-1}\Delta_{q_1} - \epsilon^{-1}\Delta_{q_2}) } \\ &
\qquad\qquad\qquad (\phi(q_2)-\phi(q_1)) \tilde
G_{N,(12)}^{(2)}(t^{\prime},p_1+\epsilon q_2 ,p_1-\epsilon q_1, q_1,q_2) \,
dp_1 \, dq_1 \,dq_2 \, dt^{\prime} \end{aligned}
\end{equation}
after shifting $p_1 \mapsto p_1+\epsilon q_2$. To proceed to \eqref{E:PCE04}%
, we need to write 
\begin{equation*}
\tilde J(p_1-\epsilon q_2, q_1+q_2)= \int J(p_1-\epsilon q_2, v_1)
e^{iv_1(q_1+q_2)} \, dv_1
\end{equation*}
and upon substitution into \eqref{E:PCE03}, we obtain 
\begin{align*}
\text{III}(t) &= \epsilon^{-1/2} \int_{p_1,v_1} J(p_1,v_1) \int_{q_1,q_2}
e^{iv_1(q_1+q_2)} e^{i(t-t^{\prime})(\epsilon\Delta_{p_1} +
2\nabla_{p_1}\cdot \nabla_{q_1} +
\epsilon^{-1}\Delta_{q_1}-\epsilon^{-1}\Delta_{q_2})} \\
&\qquad \qquad (\phi(q_2)-\phi(q_1)) \tilde
G_{N,(12)}^{(2)}(t^{\prime},p_1+\epsilon q_2, p_1-\epsilon q_1, q_1,q_2) \,
dq_1 \, dq_2 \, dp_1\, dv_1 \\
&= \epsilon^{-1/2} \int_{p_1,v_1} J(p_1,v_1) \int_{q_1,q_2}
e^{iv_1(q_1+q_2)} e^{i(t-t^{\prime})\epsilon\Delta_{p_1}}
e^{-2(t-t^{\prime})v_1\cdot \nabla_{p_1}} \\
&\qquad \qquad (\phi(q_2)-\phi(q_1)) \tilde G_{N,(12)}^{(2)}(t^{\prime},
p_1+\epsilon q_2, p_1-\epsilon q_1, q_1,q_2) \, dq_1 \, dq_2 \, dp_1\, dv_1
\\
&= \epsilon^{-1/2} \int_{p_1,v_1} e^{i(t-t^{\prime})\epsilon\Delta_{p_1}}
J(p_1+(t-t^{\prime})v_1,v_1) \int_{q_1,q_2} e^{iv_1(q_1+q_2)} \\
&\qquad \qquad (\phi(q_2)-\phi(q_1)) \tilde
G_{N,(12)}^{(2)}(t^{\prime},p_1+\epsilon q_2, p_1-\epsilon q_1, q_1,q_2) \,
dq_1 \, dq_2 \, dp_1\, dv_1
\end{align*}
which results in \eqref{E:PCE04}. Straightforward estimates resulting from
transferring derivatives to the test function $J$ yield \eqref{E:PCE05}.
\end{proof}

The following is an analogue of Lemma \ref{L:basic-B} (for $k=1$ and in weak
form) for the case $\pi=(12)$.

\begin{corollary}[$\protect\pi=(12)$ estimate for $B$ with symmetry
assumption]
\label{C:12nonlin} Let $E_2$ be the symmetry remainder: 
\begin{equation}  \label{E:PCE06c}
E_2(t,p_1,p_2,q_1,q_2) := \tilde G_{N,(12)}^{(2)}(t,p_1,p_2,q_1,q_2) -
\tilde G_{N,(12)}^{(2)}(t, p_2,p_1,q_2,q_1)
\end{equation}
Suppose that there exists $\mu\geq 0$ and $\alpha>0$ such that, for all $%
T_1<T_2$, 
\begin{equation}  \label{E:PCE06}
\begin{aligned} \hspace{0.3in}&\hspace{-0.3in} \|\langle \boldsymbol{p}_2
\rangle^{-n} \langle \boldsymbol{q}_2 \rangle^{-n} \langle \nabla_{\bds p_2}
\rangle^{\frac34+}
E_2(t,\boldsymbol{p}_2,\boldsymbol{q}_2)\|_{L_{t\in{[T_1,T_2]}}^1 L^2_{\bds
p_2 \bds q_2}} \\ &\lesssim \epsilon^\mu (T_2-T_1)^\alpha \| \langle
\nabla_{p_1 }\rangle^{1+} \langle \nabla_{p_2 }\rangle^{1+} \tilde
G_{N,(12)}^{(2)}(t^{\prime},p_1,p_2, q_1, q_2) \|_{L_t^\infty
L^2_{p_1p_2q_1q_2}} \end{aligned}
\end{equation}
Then the quantity $\text{III}(t)$ from Lemma \ref{L:PCE1} is estimated as 
\begin{equation}  \label{E:PCE07}
\text{III}(t) \lesssim (\epsilon^{\mu-\frac12}t^\alpha + \epsilon^{0+}t) \|
\langle \nabla_{p_1 }\rangle^{1+} \langle \nabla_{p_2 }\rangle^{1+} \tilde
G_{N,(12)}^{(2)}(t^{\prime},p_1,p_2, q_1, q_2) \|_{L_t^\infty
L^2_{p_1p_2q_1q_2}}
\end{equation}
The value of $\mu$ is addressed in the remark below.
\end{corollary}

\begin{proof}
Let $u_j$ be the frequency variable corresponding to $p_j$ under the Fourier
transform. Introduce the partition of $(p_1,p_2)$ space 
\begin{equation}  \label{E:PCE19}
I = P_{1<2,L} + P_{1<2,H} + P_{1>2,L} + P_{1>2,H}
\end{equation}
according to frequencies, where

\begin{itemize}
\item $P_{1<2,L}$ is the projection onto the frequency set $|u_1| < |u_2|$
and $\min(|u_1|, |u_2|) = |u_1|< \epsilon^{-1}$.

\item $P_{1<2,H}$ is the projection onto the frequency set $|u_1| < |u_2|$
and $\min(|u_1|, |u_2|) = |u_1|> \epsilon^{-1}$
\end{itemize}

and analogously define $P_{1>2,L}$ and $P_{1>2,H}$. In \eqref{E:PCE04},
insert the decomposition \eqref{E:PCE19} on $\tilde G_{N,(12)}^{(2)}$ to
obtain 
\begin{equation*}
\text{III}(t) = \text{III}_{1<2,L}(t) + \text{III}_{1<2,H}(t) + \text{III}%
_{2<1,L}(t) + \text{III}_{2<1,H}(t)
\end{equation*}
The treatment of the last two terms is completely analogous to the first two
terms, so we will only address the first two terms, $\text{III}_{1<2,L}(t)$
and $\text{III}_{1<2,H}(t)$.

For $\text{III}_{1<2,H}(t)$, we do not need to use the symmetry assumption %
\eqref{E:PCE06c} since we can effectively use that $1 \leq
\epsilon|\nabla_{p_1}|$ to gain $\epsilon^{\frac12+}$ at the expense of $%
\frac12+$ derivatives, which we now describe. By \eqref{E:PCE05} and the
bound $|\phi(q)| \lesssim \langle q \rangle^{-2n}$, we obtain 
\begin{equation}  \label{E:PCE20}
|\mathcal{J}(t-t^{\prime}, p_1,q_1+q_2) \phi(q_1)| \lesssim_n \langle p_1
\rangle^{-n} \langle q_1+q_2 \rangle^{-n} \langle q_1 \rangle^{-2n} \lesssim
\langle p_1 \rangle^{-n} \langle q_1 \rangle^{-n} \langle q_2 \rangle^{-n}
\end{equation}
and similarly 
\begin{equation}  \label{E:PCE21}
|\mathcal{J}(t-t^{\prime}, p_1,q_1+q_2) \phi(q_2)| \lesssim_n \langle p_1
\rangle^{-n} \langle q_1+q_2 \rangle^{-n} \langle q_2 \rangle^{-2n} \lesssim
\langle p_1 \rangle^{-n} \langle q_1 \rangle^{-n} \langle q_2 \rangle^{-n}
\end{equation}
By Cauchy-Schwarz, 
\begin{equation*}
|\text{III}_{1<2,L}(t)| \leq \epsilon^{-1/2} \int_{t^{\prime}=0}^t
\|P_{1<2,H} \tilde G_{N,(12)}^{(2)}(t^{\prime},p_1+\epsilon q_2,
p_1-\epsilon q_1, q_1,q_2)\|_{ L^2_{p_1q_1q_2}} \,d t^{\prime }
\end{equation*}
Sup out in the $p_1$ coordinate and then apply Sobolev embedding to obtain 
\begin{equation*}
|\text{III}_{1<2,L}(t)| \leq \epsilon^{-1/2} \int_{t^{\prime}=0}^t \|\langle
\nabla_{p_1} \rangle^{\frac32+} P_{1<2,H} \tilde
G_{N,(12)}^{(2)}(t^{\prime},p_1, p_2, q_1,q_2)\|_{ L^2_{p_1p_2q_1q_2}} \,d
t^{\prime }
\end{equation*}
Since $|u_1| \geq \epsilon^{-1}$, we can trade $\frac12+$ derivatives in $%
p_1 $ in change for $\epsilon^{\frac12+}$: 
\begin{equation*}
|\text{III}_{1<2,L}(t)| \leq \epsilon^{0+} \int_{t^{\prime}=0}^t \|\langle
\nabla_{p_1} \rangle^{2+} P_{1<2,H} \tilde G_{N,(12)}^{(2)}(t^{\prime},p_1,
p_2, q_1,q_2)\|_{ L^2_{p_1p_2q_1q_2}} \,d t^{\prime }
\end{equation*}
Since $|u_1|\leq |u_2|$, we can share the $2+$ derivatives, obtaining the
second part on the right-side of \eqref{E:PCE07} in this case.

For $\text{III}_{1<2,L}(t)$, we will need to use the symmetry assumption %
\eqref{E:PCE06}. Take \eqref{E:PCE04}, split into two pieces, and in the
second piece swap $q_1$ and $q_2$ to obtain 
\begin{align*}
\text{III}(t) &= \frac12 \epsilon^{-1/2} \int_{t^{\prime}=0}^t
\int_{p_1,q_1,q_2} \mathcal{J}(t-t^{\prime}, p_1,q_1+q_2)
(\phi(q_2)-\phi(q_1)) \\
& \qquad \qquad [P_{1<2,L}\tilde G_{N,(12)}^{(2)}(t^{\prime},p_1+\epsilon
q_2, p_1-\epsilon q_1, q_1,q_2) \\
& \qquad \qquad \qquad - P_{1<2,L} \tilde
G_{N,(12)}^{(2)}(t^{\prime},p_1+\epsilon q_1, p_1-\epsilon q_2, q_2,q_1)] \,
dp_1 \, dq_1 \, dq_2 \, dt^{\prime}
\end{align*}
Appealing to the definition of $E_2$ above, 
\begin{align*}
\text{III}_{1<2,L}(t) &= \frac12 \epsilon^{-1/2} \int_{t^{\prime}=0}^t
\int_{p_1,q_1,q_2} \mathcal{J}(t-t^{\prime}, p_1,q_1+q_2)
(\phi(q_2)-\phi(q_1)) \\
& \qquad \qquad [P_{1<2,L} \tilde G_{N,(12)}^{(2)}(t^{\prime},p_1+\epsilon
q_2, p_1-\epsilon q_1, q_1,q_2) \\
& \qquad \qquad -P_{1<2,L} \tilde G_{N,(12)}^{(2)}(t^{\prime}, p_1-\epsilon
q_2, p_1+\epsilon q_1, q_1,q_2) \\
& \qquad \qquad + P_{1<2,L} E_2(t^{\prime}, p_1-\epsilon q_2, p_1+\epsilon
q_1, q_1,q_2) ] \, dp_1 \, dq_1 \, dq_2 \, dt^{\prime} \\
&= \frac12\epsilon^{-1/2} \int_{t^{\prime}=0}^t \int_{\theta=-1}^1
\int_{p_1,q_1,q_2} \mathcal{J}(t-t^{\prime}, p_1,q_1+q_2)
(\phi(q_2)-\phi(q_1)) \\
& \qquad \qquad [\frac{d}{d\theta} P_{1<2,L}\tilde
G_{N,(12)}^{(2)}(t^{\prime},p_1+\epsilon \theta q_2, p_1-\epsilon \theta
q_1, q_1,q_2) \\
& \qquad \qquad + P_{1<2,L} E_2(t^{\prime}, p_1-\epsilon q_2, p_1+\epsilon
q_1, q_1,q_2) ] \, dp_1 \, dq_1 \, dq_2 \, dt^{\prime} \, d\theta \\
&= \frac12 \epsilon^{-1/2}\int_{t^{\prime}=0}^t \int_{\theta=-1}^1
\int_{p_1,q_1,q_2} \mathcal{J}(t-t^{\prime}, p_1,q_1+q _2)
(\phi(q_2)-\phi(q_1)) \\
& \qquad \qquad \epsilon [(q_2 \cdot \nabla_{p_1} - q_1 \cdot \nabla_{p_2})
P_{1<2,L}\tilde G_{N,(12)}^{(2)}](t^{\prime},p_1+\epsilon \theta q_2,
p_1-\epsilon \theta q_1, q_1,q_2) \\
& \qquad \qquad + P_{1<2,L} E_2(t^{\prime}, p_1-\epsilon q_2, p_1+\epsilon
q_1, q_1,q_2) ] \, dp_1 \, dq_1 \, dq_2 \, dt^{\prime} \, d\theta
\end{align*}
Using \eqref{E:PCE20}, \eqref{E:PCE21} and applying Cauchy-Schwarz, we
obtain 
\begin{equation*}
\begin{aligned} |\text{III}_{1<2,L}(t)| \lesssim &\epsilon^{1/2}
\int_{\theta=-1}^1 \int_{t'=0}^t \| |\nabla_{p_2}| P_{1<2,L} \tilde
G_{N,(12)}^{(2)}(t', p_1+\epsilon \theta q_2,p_1-\epsilon \theta
q_1,q_1,q_2) \|_{L^2_{p_1q_1q_2}} \, dt' \, d\theta \\ &+ \epsilon^{-1/2}
\int_{t'=0}^t \|P_{1<2,L} E_2(t', p_1-\epsilon q_2, p_1+\epsilon q_2,
q_1,q_2) \|_{L^2_{p_1q_1q_2}} \, dt' \end{aligned}
\end{equation*}
In both terms, sup out in the $p_1$ coordinate, and apply Sobolev embedding
to obtain 
\begin{equation*}
\begin{aligned} |\text{III}_{1<2,L}(t)| \lesssim &\epsilon^{1/2}
\int_{\theta=-1}^1 \int_{t'=0}^t \| \langle \nabla_{p_1}\rangle^{\frac32+}
|\nabla_{p_2}| P_{1<2,L} \tilde G_{N,(12)}^{(2)}(t',p_1,p_2,q_1,q_2)
\|_{L^2_{p_1p_2q_1q_2}} \, dt' \, d\theta \\ &+ \epsilon^{-1/2}
\int_{t'=0}^t \| \langle \nabla_{p_1}\rangle^{\frac32+} P_{1<2,L}
E_2(t',p_1,p_2, q_1,q_2) \|_{L^2_{p_1p_2q_1q_2}} \, dt' \end{aligned}
\end{equation*}
In the first term, we use that $|u_1| \leq \epsilon^{-1}$ to trade $\frac12-$
derivatives for $\epsilon^{-\frac12+}$, giving 
\begin{equation*}
\begin{aligned} |\text{III}_{1<2,L}(t)| \lesssim &\epsilon^{0+}
\int_{\theta=-1}^1 \int_{t'=0}^t \| \langle \nabla_{p_1}\rangle^{1+}
|\nabla_{p_2}| P_{1<2,L} \tilde G_{N,(12)}^{(2)}(t',p_1,p_2,q_1,q_2)
\|_{L^2_{p_1p_2q_1q_2}} \, dt' \, d\theta \\ &+ \epsilon^{-1/2}
\int_{t'=0}^t \| \langle \nabla_{p_1}\rangle^{\frac32+} P_{1<2,L}
E_2(t',p_1,p_2, q_1,q_2) \|_{L^2_{p_1p_2q_1q_2}} \, dt' \end{aligned}
\end{equation*}
In the second term, we transfer $\frac34+$ derivatives from the $p_1$ term
to the $p_2$ term. Applying \eqref{E:PCE06} in the case $n=0$ to the second
term, we obtain the right-side of \eqref{E:PCE07}. In \eqref{E:PCE06} is
only available for some $n\geq 1$, then one can modify the above argument to
capture some additional decay from \eqref{E:PCE20}, \eqref{E:PCE21}.
\end{proof}

\begin{corollary}[Term III for $\protect\pi=(12)$ with linear propagator]
Let $e_2$ be the symmetry remainder: 
\begin{equation}  \label{E:PCE06d}
e_2(t,p_1,p_2,q_1,q_2) := \tilde g_{N,(12)}^{(2)}(t,p_1,p_2,q_1,q_2) -
\tilde g_{N,(12)}^{(2)}(t, p_2,p_1,q_2,q_1)
\end{equation}
Suppose that there exists $\mu\geq 0$ and a time $t_0$ such that 
\begin{equation}  \label{E:PCE06b}
\begin{aligned} \hspace{0.3in}&\hspace{-0.3in} \|\langle \boldsymbol{p}_2
\rangle^{-n} \langle \boldsymbol{q}_2 \rangle^{-n} \langle \nabla_{\bds p_2}
\rangle^{\frac34+} e_2(t_0,\boldsymbol{p}_2,\boldsymbol{q}_2)\|_{L^2_{\bds
p_2 \bds q_2}} \\ &\lesssim \epsilon^\mu \| \langle \nabla_{p_1
}\rangle^{1+} \langle \nabla_{p_2 }\rangle^{1+} \tilde
g_{N,(12)}^{(2)}(t,p_1,p_2, q_1, q_2) \|_{L_t^\infty L^2_{p_1p_2q_1q_2}}
\end{aligned}
\end{equation}
Let 
\begin{equation*}
\text{III}_{L,t_0}(t) = \int_{t^{\prime}=0}^t \int_{x_1,\xi_1} \tilde
J(x_1,\xi_1) e^{i(t-t^{\prime})\nabla_{x_1}\cdot \nabla_{\xi_1}}
(N\epsilon^{-1/2} \tilde B_{1,2}^\epsilon e^{i(t^{\prime}-t_0)\nabla_{%
\boldsymbol{x}_2}\cdot \nabla_{\boldsymbol{\xi}_2}} \tilde
f_{N,(12)}^{(2)})(t_0,x_1,\xi_1) \, dx_1 \, d\xi_1 \, dt^{\prime}
\end{equation*}
Then 
\begin{equation}  \label{E:PCE07b}
\text{III}_{L,t_0}(t) \lesssim (\epsilon^{\mu-\frac12}t^\alpha +
\epsilon^{0+}t) \| \langle \nabla_{p_1 }\rangle^{1+} \langle \nabla_{p_2
}\rangle^{1+} \tilde g_{N,(12)}^{(2)}(t^{\prime},p_1,p_2, q_1, q_2)
\|_{L_t^\infty L^2_{p_1p_2q_1q_2}}
\end{equation}
\label{C:12linear}
\end{corollary}

\begin{proof}
Let $\tilde G_{N,(12)}^{(2)}$ be the \emph{linear} evolution starting from $%
t_0$, i.e. 
\begin{equation*}
\tilde G_{N,(12)}^{(2)}(t,p_1,p_2,q_1,q_2) = e^{i(t-t_0)\nabla_{\boldsymbol{p%
}_2}\cdot \nabla_{\boldsymbol{q}_2}} \tilde g_N(t_0,p_1,p_2,q_1,q_2)
\end{equation*}
Then with $E_2$ as defined in \eqref{E:PCE06c}, it follows that %
\eqref{E:PCE06} holds (and hence so does \eqref{E:PCE07}) with the same
value of $\mu$ as in \eqref{E:PCE06b}.
\end{proof}

\begin{remark}
\label{R:TLP} Now suppose that $\tilde f_{N,(12)}^{(2)}(t)$ represents the
component for $\pi=(12)$ of the BBGKY $k=2$ density $\tilde f_N^{(2)}$ and
correspondingly $\tilde g_{N,(12)}^{(2)}$ is the expression in $(p,q)$
coordinates. Take $\tilde G_{N,(12)}^{(2)}=\tilde g_{N,(12)}^{(2)}$ in
Corollary \ref{C:12nonlin}. If \eqref{E:PCE06} holds with $\mu=\frac12$,
then \eqref{E:PCE07} gives an $O(1)$ bound on $\text{III}(t)$. Thus, we
expect that this bound is indeed inherited from the $N$-body model. If %
\eqref{E:PCE06} holds with $\mu=\frac12+$, then \eqref{E:PCE07} gives an $%
O(\epsilon^{0+})$ bound on $\text{III}(t)$, and thus we do \emph{not} expect
this improved bound in general, since it results in a trivial limit (zero
collisional effects). Now assume that both of the following hold

\begin{enumerate}
\item \eqref{E:PCE06} holds with $\mu=\frac12$ and no higher value of $\mu$.

\item There exists one time $t_0$ such that \eqref{E:PCE06b} holds with $%
\mu=\frac12+$.
\end{enumerate}

Note that (1) and (2) can indeed simultaneously hold, since \eqref{E:PCE06}
involves an integral in $t$ and thus does not ``see'' a better bound that
holds on a small measure set. Now when both (1) and (2) hold, Corollary \ref%
{C:12nonlin}, \ref{C:12linear} imply that both

\begin{enumerate}
\item $\text{III}(t)=O(1)$.

\item $\text{III}_{L,t_0}(t)= O(\epsilon^{0+})$
\end{enumerate}

This resolves the trivial limit puzzle.

Notice that in Lemma \ref{L:basic-B} (pertaining to $\pi=\text{Identity}$,
the core terms) the needed $\epsilon^{\frac12+}$ gain is obtained from the
vanishing of $\hat\phi(\xi)$ at $\xi=0$. In Corollary \ref{C:12nonlin}/\ref%
{C:12linear} the $\epsilon^{\frac12}$/$\epsilon^{\frac12+}$ gain is instead
obtained from the symmetry assumption. Thus the symmetry assumption is the
``faucet'' that determines the extent of $\epsilon$ gain in the $\pi=(12)$
terms, and this faucet must be tuned precisely as described above.
\end{remark}

Replacing $\tilde G_{N,(12)}^{(2)}$ in \eqref{E:PCE04} with $\mathcal{D}%
^{(2)}\epsilon^{-1/2}A_{1,2}^\epsilon \tilde g_{N,(12)}^{(2)}$, we can
obtain the expression for Term IV in $\pi=(12)$ coordinates.

\begin{proposition}[Term IV for $\protect\pi=(12)$]
Consider the weak form of Term IV: 
\begin{align*}
\text{IV}(t) &= \int_{x_1,\xi_1} \tilde J(x_1,\xi_1) \int_{t^{\prime}=0}^t
e^{i(t-t^{\prime})\nabla_{x_1}\cdot \nabla_{\xi_1}} N \epsilon^{-1/2} \tilde
B_{1,2}^\epsilon \\
& \qquad \left[ \int_{t^{\prime\prime}=0}^{t^{\prime}}
e^{i(t^{\prime}-t^{\prime\prime}) (\nabla_{x_1}\cdot \nabla_{\xi_1} +
\nabla_{x_2}\cdot \nabla_{\xi_2})} \epsilon^{-1/2} \tilde A_{1,2}^\epsilon
\tilde f_{N,(12)}(t^{\prime\prime}) \, dt^{\prime\prime} \right] \,
dt^{\prime }\, dx_1 \, d\xi_1
\end{align*}
Implement the coordinate conversion to the $\pi=(12)$ frame, 
\begin{equation}  \label{E:PCE08}
\tilde f_{N,(12)}^{(2)}(t^{\prime\prime}, x_1,x_2, \xi_1, \xi_2) = \tilde
g_{N,(12)}^{(2)}(t^{\prime\prime}, p_1,p_2, q_1, q_2)
\end{equation}
Let $e_2$ be the symmetry remainder: 
\begin{equation}  \label{E:PCE06e}
e_2(t,p_1,p_2,q_1,q_2) := \tilde g_{N,(12)}^{(2)}(t,p_1,p_2,q_1,q_2) -
\tilde g_{N,(12)}^{(2)}(t, p_2,p_1,q_2,q_1)
\end{equation}
Assume that 
\begin{equation}  \label{E:PCE06f}
\begin{aligned} \hspace{0.3in}&\hspace{-0.3in} \|\langle \boldsymbol{p}_2
\rangle^{-n} \langle \boldsymbol{q}_2 \rangle^{-n} \langle \nabla_{\bds p_2}
\rangle^{\frac34+}
e_2(t,\boldsymbol{p}_2,\boldsymbol{q}_2)\|_{L_{t\in{[T_1,T_2]}}^1 L^2_{\bds
p_2 \bds q_2}} \\ &\lesssim \epsilon^{0+} (T_2-T_1)^\alpha \| \langle
\nabla_{p_1 }\rangle^{1+} \langle \nabla_{p_2 }\rangle^{1+} \tilde
g_{N,(12)}^{(2)}(t^{\prime},p_1,p_2, q_1, q_2) \|_{L_t^\infty
L^2_{p_1p_2q_1q_2}} \end{aligned}
\end{equation}
Then 
\begin{equation*}
|\text{IV}(t)| \lesssim t^{1-} \epsilon^{0+} \| \langle
\nabla_{p_1}\rangle^{\frac34+}\langle \nabla_{p_2}\rangle^{\frac34+}\langle
q_1 \rangle^{0+}\langle q_2 \rangle^{0+} \tilde g_{N,(12)}^{(2)}(t^{\prime
\prime }, p_1,p_2,q_1,q_2) \|_{L_{t^{\prime \prime }}^\infty
L^2_{p_1p_2q_1q_2}}
\end{equation*}
\label{P:termIV}
\end{proposition}

\begin{proof}
Let 
\begin{equation*}
\tilde{F}_{N,(12)}^{(2)}(t^{\prime })=\int_{t^{\prime \prime }=0}^{t^{\prime
}}e^{i(t^{\prime }-t^{\prime \prime })(\nabla _{x_{1}}\cdot \nabla _{\xi
_{1}}+\nabla _{x_{2}}\cdot \nabla _{\xi _{2}})}\epsilon ^{-1/2}\tilde{A}%
_{1,2}^{\epsilon }\tilde{f}_{N,(12)}(t^{\prime \prime })\,dt^{\prime \prime }
\end{equation*}%
We convert coordinates 
\begin{equation*}
\tilde{F}_{N,(12)}^{(2)}(t^{\prime },x_{1},x_{2},\xi _{1},\xi _{2})=\tilde{G}%
_{N,(12)}^{(2)}(t^{\prime },p_{1},p_{2},q_{1},q_{2})
\end{equation*}%
The converted expression, using \eqref{E:PCE08}, is 
\begin{equation}
\begin{aligned} \hspace{0.3in}&\hspace{-0.3in} \tilde
G_{N,(12)}^{(2)}(t^{\prime}, p_1,p_2,q_1,q_2) =
\epsilon^{-1/2}\int_{t^{\prime\prime}=0}^{t^{\prime}}
e^{i(t^{\prime}-t^{\prime\prime})( \nabla_{p_1}\cdot \nabla_{q_1} +
\nabla_{p_2}\cdot \nabla_{q_2})} \\ &\sum_{\sigma=\pm 1} \sigma \phi\left(
\sigma \frac{p_2-p_1}{\epsilon} + q_1 -q_2 \right) \tilde
g_{N,(12)}^{(2)}(t^{\prime\prime}, p_1,p_2, q_1,q_2) \, dt^{\prime\prime}
\end{aligned}  \label{E:PCE11}
\end{equation}%
Substituting into \eqref{E:PCE04}, 
\begin{equation}
\begin{aligned} \text{IV}(t) & = \epsilon^{-1/2}\sum_{\sigma = \pm 1}
\int_{t^{\prime}=0}^t \int_{p_1,q_1,q_2}
\mathcal{J}(t-t^{\prime},p_1,q_1+q_2) (\phi(q_2)-\phi(q_1)) \\ & \qquad
\qquad \tilde G_{N,(12)}^{(2)}(t^{\prime},p_1+\epsilon q_2, p_1 - \epsilon
q_1, q_1, q_2) \, dp_1 \, dq_1 \, dq_2 \, dt^{\prime} \end{aligned}
\label{E:PCE12}
\end{equation}%
where $\tilde{G}_{N,(12)}^{(2)}$ is given by \eqref{E:PCE11}. Note that 
\begin{equation}
\begin{aligned} \hspace{0.3in}&\hspace{-0.3in} \int_{p_1} e^{-ip_1u_1}
\tilde G_{N,(12)}^{(2)}(t^{\prime}, p_1+\epsilon q_2, p_1 - \epsilon q_1,
q_1, q_2) \, dp_1 \\ & = \int_{u_2} e^{i\epsilon q_2 \cdot u_1}
e^{-i\epsilon (q_1+q_2)\cdot u_2} \check G_{N,(12)}^{(2)}(t^{\prime},
u_1-u_2, u_2, q_1, q_2) \, du_2 \end{aligned}  \label{E:PCE14}
\end{equation}%
In \eqref{E:PCE12}, move the $p_{1}$ integration to the inside, apply
Plancherel $p_{1}\mapsto u$, and insert \eqref{E:PCE14}, to obtain 
\begin{equation*}
\begin{aligned} \text{IV}(t) & = \epsilon^{-1/2}\sum_{\sigma = \pm 1}
\int_{t^{\prime}=0}^t \int_{u_1,u_2, q_1,q_2}
\check{\mathcal{J}}(t-t^{\prime},u_1,q_1+q_2) (\phi(q_2)-\phi(q_1)) \\ &
\qquad \qquad e^{i\epsilon q_2 \cdot u_1} e^{-i\epsilon (q_1+q_2)\cdot u_2}
\check G_{N,(12)}^{(2)}(t^{\prime},u_1-u_2, u_2, q_1, q_2) \, du_1 \, du_2
\, dq_1 \, dq_2 \, dt^{\prime} \end{aligned}
\end{equation*}%
Shift $u_{2}\mapsto u_{2}+\frac{1}{2}u_{1}$ to obtain the more symmetric
expression 
\begin{equation}
\begin{aligned} \text{IV}(t) & = \epsilon^{-1/2}\sum_{\sigma = \pm 1}
\int_{t^{\prime}=0}^t \int_{u_1,u_2, q_1,q_2}
\check{\mathcal{J}}(t-t^{\prime},u_1,q_1+q_2) (\phi(q_2)-\phi(q_1))
e^{\frac12 i\epsilon (q_2-q_1) \cdot u_1} \\ & \qquad \qquad e^{-i\epsilon
(q_1+q_2)\cdot u_2} \check G_{N,(12)}^{(2)}(t^{\prime},\tfrac12 u_1-u_2,
\tfrac12u_1 + u_2, q_1, q_2) \, du_1 \, du_2 \, dq_1 \, dq_2 \, dt^{\prime}
\end{aligned}  \label{E:PCE15}
\end{equation}%
The check-space representation of (\ref{E:PCE11}) is 
\begin{equation}
\begin{aligned} \hspace{0.3in}&\hspace{-0.3in} \check
G_{N,(12)}^{(2)}(t^{\prime}, u_1,u_2,q_1,q_2) =
\epsilon^{-1/2}\int_{t^{\prime\prime}=0}^{t^{\prime}} \int_w e^{i w\cdot
(q_1-q_2)} e^{-i(t^{\prime}-t^{\prime\prime})w\cdot (u_1-u_2)} \hat \phi(w)
\\ & \qquad \sigma \int_w \check g_{N,(12)}^{(2)}(t^{\prime\prime}, u_1+
\frac{\sigma w}{\epsilon}, u_2 - \frac{\sigma w}{\epsilon}, q_1 -
(t^{\prime}-t^{\prime\prime})u_1, q_2 - (t^{\prime}-t^{\prime\prime})u_2) \,
dw \, dt^{\prime\prime} \end{aligned}  \label{E:PCE13}
\end{equation}%
Substituting \eqref{E:PCE13} into \eqref{E:PCE15}, we obtain 
\begin{equation*}
\begin{aligned} \text{IV}(t) & = \epsilon^{-1}\sum_{\sigma = \pm 1}\sigma
\int_{t^{\prime}=0}^t \int_{t^{\prime\prime}=0}^{t^{\prime}} \int_{u_1,u_2,
w, q_1,q_2} \check{\mathcal{J}}(t-t^{\prime},u_1,q_1+q_2)
(\phi(q_2)-\phi(q_1)) \hat \phi(w) \\ & \qquad \qquad e^{\frac12i\epsilon
(q_2-q_1) \cdot u_1} e^{-i\epsilon (q_1+q_2)\cdot u_2} e^{i w\cdot
(q_1-q_2)} e^{2i(t^{\prime}-t^{\prime\prime})w\cdot u_2} \\ & \qquad \qquad
\check g_{N,(12)}^{(2)}(t^{\prime\prime}, \tfrac12 u_1-u_2+ \frac{\sigma
w}{\epsilon}, \tfrac12 u_1 + u_2 - \frac{\sigma w}{\epsilon}, \\ & \qquad
\qquad \qquad \qquad q_1 - (t^{\prime}-t^{\prime\prime})(\tfrac12 u_1-u_2),
q_2 - (t^{\prime}-t^{\prime\prime})(\tfrac12u_1 + u_2)) \, \\ & \qquad
\qquad dw \, du_1 \, du_2 \, dq_1 \, dq_2 \, dt^{\prime\prime} \,
dt^{\prime} \end{aligned}
\end{equation*}%
Replace $u_{2}\mapsto u_{2}+\frac{\sigma w}{\epsilon }$. 
\begin{equation*}
\begin{aligned} \text{IV}(t) & = \epsilon^{-1}\sum_{\sigma = \pm 1}\sigma
\int_{t^{\prime}=0}^t \int_{t^{\prime\prime}=0}^{t^{\prime}} \int_{u_1,u_2,
w, q_1,q_2} \check{\mathcal{J}}(t-t^{\prime},u_1,q_1+q_2)
(\phi(q_2)-\phi(q_1)) \hat \phi(w) \\ & \qquad \qquad e^{\frac12i\epsilon
u_1\cdot (q_2-q_1)} e^{-i\epsilon u_2\cdot (q_1+q_2) } e^{i w\cdot[
(1-\sigma)q_1 - (1+\sigma)q_2]} e^{2i(t'-t^{\prime\prime}) w\cdot (u_2+
\sigma w/\epsilon)} \\ & \qquad \qquad \check
g_{N,(12)}^{(2)}(t^{\prime\prime}, \tfrac12 u_1-u_2, \tfrac12u_1 + u_2 , \\
& \qquad \qquad \qquad \qquad q_1 - (t^{\prime}-t^{\prime\prime})(\tfrac12
u_1 - u_2 - \frac{\sigma w}{\epsilon}) , q_2 -
(t^{\prime}-t^{\prime\prime})(\tfrac12 u_1 +u_2 + \frac{\sigma
w}{\epsilon})) \, \\ & \qquad \qquad dw \, du_1 \, du_2 \, dq_1 \, dq_2 \,
dt^{\prime\prime} \, dt^{\prime} \end{aligned}
\end{equation*}%
Convert to the $\wedge $ side by replacing the $\check{g}_{N,(12)}^{(2)}$
term with 
\begin{equation*}
\begin{aligned} \int_{w_1,w_2} & \hat g_{N,(12)}^{(2)}(t^{\prime\prime},
\tfrac12 u_1-u_2, \tfrac12u_1 + u_2 , w_1,w_2) \\ & e^{iw_1\cdot [ q_1 -
(t^{\prime}-t^{\prime\prime})(\frac12 u_1 - u_2 - \frac{\sigma
w}{\epsilon})] } e^{iw_2\cdot[q_2 - (t^{\prime}-t^{\prime\prime})(\frac12
u_1 +u_2 + \frac{\sigma w}{\epsilon}))]} \, dw_1 \, dw_2 \end{aligned}
\end{equation*}%
Let 
\begin{equation*}
\begin{aligned} \hspace{0.3in}&\hspace{-0.3in}
H(t-t^{\prime\prime},u_1,u_2,w_1,w_2) \\ &= \epsilon^{-1}\sum_{\sigma = \pm
1}\sigma \int_{t'=t^{\prime\prime}}^t \int_{q_1,q_2, w}
\check{\mathcal{J}}(t-t', u_1, q_1+q_2) (\phi(q_1)-\phi(q_2)) \hat \phi(w)
\\ & \hspace{3cm} e^{\frac12i\epsilon u_1\cdot(q_2-q_1)} e^{-i\epsilon
u_2\cdot(q_1+q_2)} e^{iw\cdot[ (1-\sigma)q_1 - (1+\sigma)q_2 ]} e^{iw_1\cdot
q_1} e^{iw_2\cdot q_2} \\ & \hspace{3cm} e^{i(t'-t^{\prime\prime})[ 2w\cdot
u_2 - w_1(\frac12u_1-u_2) - w_2(\frac12u_1+u_2)]}
e^{i(t'-t^{\prime\prime})(2w+w_1-w_2) \cdot \frac{\sigma w}{\epsilon}} \, dw
\, dq_1 \, dq_2 \, dt' \end{aligned}
\end{equation*}%
Bring the $t^{\prime }$ integration to the inside, we obtain 
\begin{equation}
\begin{aligned} \text{IV}(t) & = \int_{t^{\prime\prime}=0}^t \int_{u_1,u_2,
w_1,w_2} H(t,t^{\prime\prime}, u_1,u_2, w_1,w_2) \\ & \qquad \qquad \hat
g_{N,(12)}^{(2)} (t^{\prime\prime}, \tfrac12 u_1-u_2, \tfrac12u_1 + u_2,
w_1, w_2) \, du_1 \, du_2 \, dw_1 \, dw_2 \, dt^{\prime\prime} \,
dt^{\prime} \end{aligned}  \label{E:PCE22}
\end{equation}%
To evaluate $H$, substitute 
\begin{equation*}
\check{\mathcal{J}}(t-t^{\prime
},u_{1},q_{1}+q_{2})=\int_{v_{1}}e^{iv_{1}\cdot
(q_{1}+q_{2})}e^{i(t-t^{\prime })u_{1}\cdot v_{1}}e^{i(t-t^{\prime
})\epsilon |u_{1}|^{2}}\hat{J}(u_{1},v_{1})\,dv_{1}
\end{equation*}%
which then allows the evaluation of the $q_{1}$ and $q_{2}$ integrals: 
\begin{align*}
H(t-t^{\prime \prime },u_{1},u_{2},w_{1},w_{2})=& \epsilon ^{-1}\sum_{\sigma
=\pm 1}\sigma \int_{t^{\prime }=t^{\prime \prime }}^{t}\int_{w,v_{1}}\hat{J}%
(u_{1},v_{1})\hat{\phi}(w) \\
& (\hat{\phi}\otimes \delta -\delta \otimes \hat{\phi})(-v_{1}+\tfrac{1}{2}%
\epsilon u_{1}+\epsilon u_{2}-w(1-\sigma )-w_{1}, \\
& \qquad \qquad -v_{1}-\tfrac{1}{2}\epsilon u_{1}+\epsilon u_{2}+w(1+\sigma
)-w_{2}) \\
& e^{i(t-t^{\prime })u_{1}\cdot v_{1}}e^{i(t-t^{\prime })\epsilon
|u_{1}|^{2}}e^{i(t^{\prime }-t^{\prime \prime })[2w\cdot u_{2}-w_{1}(\frac{1%
}{2}u_{1}-u_{2})-w_{2}(\frac{1}{2}u_{1}+u_{2})]} \\
& e^{i(t^{\prime }-t^{\prime \prime })(2w+w_{1}-w_{2})\cdot \frac{\sigma w}{%
\epsilon }}\,dw\,dv_{1}
\end{align*}%
Using the delta functions to evaluate in $v_{1}$ and appealing to radiality
of $\phi $ gives 
\begin{align*}
H(t-t^{\prime \prime },u_{1},u_{2},w_{1},w_{2})=& \epsilon ^{-1}\sum_{\sigma
=\pm 1}\sigma \int_{t^{\prime }=t^{\prime \prime }}^{t}\int_{w}\Big[\hat{J}%
(u_{1},v_{1})e^{i(t-t^{\prime })u_{1}\cdot v_{1}}\Big]_{v_{1}=\frac{1}{2}%
\epsilon u_{1}+\epsilon u_{2}-(1-\sigma )w-w_{1}}^{v_{1}=-\frac{1}{2}%
\epsilon u_{1}+\epsilon u_{2}+(1+\sigma )w-w_{2}} \\
& \hat{\phi}(\epsilon u_{1}+2w+w_{2}-w_{1})\hat{\phi}(w)e^{i(t-t^{\prime
})\epsilon |u_{1}|^{2}} \\
& e^{i(t^{\prime }-t^{\prime \prime })[2w\cdot u_{2}-w_{1}(\frac{1}{2}%
u_{1}-u_{2})-w_{2}(\frac{1}{2}u_{1}+u_{2})]}e^{i(t^{\prime }-t^{\prime
\prime })(2w+w_{1}-w_{2})\cdot \frac{\sigma w}{\epsilon }}\,dw
\end{align*}%
where the notation $h(v_{1})\Big]_{v_{1}=a}^{v_{1}=b}=h(b)-h(a)$.
Integration in $t^{\prime }$ is just integration of imaginary exponentials: 
\begin{equation*}
\epsilon ^{-1}\int_{t^{\prime }=t^{\prime \prime }}^{t}e^{i(t-t^{\prime
})r_{1}}e^{i(t^{\prime }-t^{\prime \prime })r_{2}}\,dt^{\prime }=\frac{%
e^{i(t-t^{\prime \prime })r_{2}}-e^{i(t-t^{\prime \prime })r_{1}}}{i\epsilon
(r_{2}-r_{1})}
\end{equation*}%
Applying this, with 
\begin{align*}
& r_{1}=u_{1}\cdot v_{1}+\epsilon |u_{1}|^{2} \\
& r_{2}=2w\cdot u_{2}-w_{1}(\tfrac{1}{2}u_{1}-u_{2})-w_{2}(\tfrac{1}{2}%
u_{1}+u_{2})+(2w+w_{1}-w_{2})\cdot \frac{\sigma w}{\epsilon }
\end{align*}%
\begin{align*}
\hspace{0.3in}& \hspace{-0.3in}H(t-t^{\prime \prime
},u_{1},u_{2},w_{1},w_{2}) \\
& =\sum_{\sigma =\pm 1}\sigma \int_{w}\Big[\hat{J}(u_{1},v_{1})\frac{%
e^{i(t-t^{\prime \prime })r_{2}}-e^{i(t-t^{\prime \prime })r_{1}}}{i\epsilon
(r_{2}-r_{1})}\Big]_{v_{1}=\frac{1}{2}\epsilon u_{1}+\epsilon
u_{2}-(1-\sigma )w-w_{1}}^{v_{1}=-\frac{1}{2}\epsilon u_{1}+\epsilon
u_{2}+(1+\sigma )w-w_{2}} \\
& \hspace{3cm}\hat{\phi}(\epsilon u_{1}+2w+w_{2}-w_{1})\hat{\phi}(w)\,dw
\end{align*}%
This can be written more compactly as 
\begin{align*}
\hspace{0.3in}& \hspace{-0.3in}H(t-t^{\prime \prime
},u_{1},u_{2},w_{1},w_{2}) \\
& =\sum_{\sigma ,\alpha \in \{\pm 1\}}\alpha \int_{w}\hat{J}(u_{1},v_{1})%
\frac{e^{i(t-t^{\prime \prime })r_{2}}-e^{i(t-t^{\prime \prime })r_{1}}}{%
i\epsilon \sigma (r_{2}-r_{1})}\hat{\phi}(\epsilon u_{1}+2w+w_{2}-w_{1})\hat{%
\phi}(w)\,dw
\end{align*}%
where 
\begin{align*}
& v_{1}=-\tfrac{1}{2}\alpha \epsilon u_{1}+\epsilon u_{2}+w(\alpha +\sigma )-%
\tfrac{1}{2}(1+\alpha )w_{2}-\tfrac{1}{2}(1-\alpha )w_{1} \\
& r_{1}=u_{1}\cdot v_{1}+\epsilon |u_{1}|^{2} \\
& r_{2}=2w\cdot u_{2}-w_{1}(\tfrac{1}{2}u_{1}-u_{2})-w_{2}(\tfrac{1}{2}%
u_{1}+u_{2})+(2w+w_{1}-w_{2})\cdot \frac{\sigma w}{\epsilon }
\end{align*}%
Note that the denominator is 
\begin{equation*}
i\epsilon \sigma (r_{2}-r_{1})=i(2w+w_{1}-w_{2})\cdot w+O(\epsilon )
\end{equation*}%
The two vectors in the dot product on the main term, $2w+w_{1}-w_{2}$ and $w$%
, also appear inside $\hat{\phi}$, and moreover we assume that $\hat{\phi}%
(0)=0$. Thus vanishing denominators can be suitably compensated, and overall
the size of $H$ is $O(1)$. This can be proved by working in spherical
coordinates for $w + \frac{w_1-w_2}{4}$. In order to gain $\epsilon ^{0+}$,
we will need to split up $H$, and, correspondingly $\text{IV}(t)$ given by %
\eqref{E:PCE22} as 
\begin{equation*}
H=H_{2}-H_{1}\,,\qquad \text{IV}(t)=\text{IV}_{1}(t)-\text{IV}_{2}(t)
\end{equation*}%
where 
\begin{equation*}
\begin{aligned} \hspace{0.3in}&\hspace{-0.3in} H_j(t-t^{\prime \prime },
u_1,u_2,w_1,w_2) \\ &= \sum_{\sigma, \alpha \in \{\pm 1\}} \alpha \int_{w}
\hat J(u_1,v_1) \frac{ e^{i(t-t^{\prime \prime})r_j}}{i\epsilon \sigma
(r_2-r_1)} \hat\phi(\epsilon u_1+2 w + w_2-w_1) \hat\phi(w) \, dw
\end{aligned}
\end{equation*}%
Since $r_{2}$ is oscillatory, an additional integration by parts can be
employed for $H_{2}$ inside $\text{IV}_{2}(t)$. For $H_{1}$, we will employ
a near-symmetry. Let 
\begin{equation*}
H_{3}(t-t^{\prime \prime },u_{1},u_{2},w_{1},w_{2})=H_{1}(t-t^{\prime \prime
},u_{1},-u_{2},w_{2},w_{1})+H_{1}(t-t^{\prime \prime
},u_{1},u_{2},w_{1},w_{2})
\end{equation*}%
Then $H_{3}=O(\epsilon )$ (in a certain precise sense). This allows us to
reexpress $\text{IV}_{1}(t)$ in order to invoke a symmetry assumption on $%
\tilde{g}_{N,(12)}^{(2)}$. Writing $H_{1}=\frac{1}{2}H_{1}+\frac{1}{2}H_{1}$%
, and for the second copy of $H_{1}$, substituting the symmetry: 
\begin{align*}
\hspace{0.3in}& \hspace{-0.3in}H_{1}(t-t^{\prime \prime
},u_{1},u_{2},w_{1},w_{2})=\tfrac{1}{2}H_{1}(t-t^{\prime \prime
},u_{1},u_{2},w_{1},w_{2})+\tfrac{1}{2}H_{1}(t-t^{\prime \prime
},u_{1},u_{2},w_{1},w_{2}) \\
& =\tfrac{1}{2}H_{1}(t-t^{\prime \prime },u_{1},u_{2},w_{1},w_{2})-\tfrac{1}{%
2}H_{1}(t-t^{\prime \prime },u_{1},-u_{2},w_{2},w_{1}) \\
& \hspace{2cm}+\tfrac{1}{2}H_{3}(t-t^{\prime \prime
},u_{1},u_{2},w_{1},w_{2})
\end{align*}%
When substituted into the expression for $\text{IV}_{1}(t)$, in the second
term, we change variable $(u_{2},w_{1},w_{2})\mapsto (-u_{2},w_{2},w_{1})$
to obtain 
\begin{equation*}
\begin{aligned} \text{IV}_1(t) & = \frac12 \int_{t^{\prime\prime}=0}^t
\int_{u_1,u_2, w_1,w_2} H_1(t,t^{\prime\prime}, u_1,u_2, w_1,w_2) [\hat
g_{N,(12)}^{(2)} (t^{\prime\prime}, \tfrac12 u_1-u_2, \tfrac12u_1 + u_2,
w_1, w_2) \\ & \hspace{3cm} - \hat g_{N,(12)}^{(2)} (t^{\prime\prime},
\tfrac12 u_1+u_2, \tfrac12u_1 - u_2, w_2, w_1)] \, du_1 \, du_2 \, dw_1 \,
dw_2 \, dt^{\prime\prime} \, dt^{\prime} \\ & \qquad + \frac12
\int_{t^{\prime\prime}=0}^t \int_{u_1,u_2, w_1,w_2} H_3(t,t^{\prime\prime},
u_1,u_2, w_1,w_2) \\ & \hspace{3cm} \hat g_{N,(12)}^{(2)} (t^{\prime\prime},
\tfrac12 u_1-u_2, \tfrac12u_1 + u_2, w_1, w_2)\, du_1 \, du_2 \, dw_1 \,
dw_2 \, dt^{\prime\prime} \, dt^{\prime} \end{aligned}
\end{equation*}%
In this first term, we use the symmetry assumption, and in the second term,
the smallness of $H_{3}$. All $H_{j}$ satisfy 
\begin{equation*}
|H_{j}(t-t^{\prime \prime },u_{1},u_{2},w_{1},w_{2})|\lesssim _{n}\langle
u_{1}\rangle ^{-n}\langle w_{1}\rangle ^{-n}\langle w_{2}\rangle ^{-n}
\end{equation*}%
(note the absence of $\langle w_{1}\rangle ^{-n}$).
\end{proof}

\section{Compactness of the BBGKY Family\label{sec:CompactnessConvergence}}

In this section, we use the estimates in \S \ref{S:preparation} to prove a
compactness property of solutions to the quantum BBGKY hierarchy.

Fix $\delta>0$, sufficiently small. Let the operator $P_N^{(k)}$, acting on $%
k$-densities, cutoff all components of $\boldsymbol{x}_k$ and $\boldsymbol{v}%
_k$ to both spatial radius and frequency radius $\epsilon^{-\delta}$ (in
other words, it cuts off all of $\boldsymbol{x}_k$, $\boldsymbol{\eta}_k$, $%
\boldsymbol{v}_k$, $\boldsymbol{\xi}_k$ to inside radius $\epsilon^{-\delta}$%
) Note that the radius is expanding (both in space and frequency) as $N\to
\infty$, so it in the limit it becomes the identity. Given a hierarchy $%
f_N=\{ f_N^{(k)} \}_{k=1}^N$, let $P_Nf_N = \{ P_N^{(k)}f_N^{(k)} \}_{k=1}^N$
be the cut-off hierarchy. Given a collection $\mathcal{F} = \{f_N\}_N$ of
hierarchies $f_N=\{f_N^{(k)}\}_{k=1}^N$, let $P\mathcal{F} = \{ Pf_N \}_N$
be the corresponding collection of cut-off hierarchies.

\begin{theorem}
\label{T:compactness} Let $C\geq 2$. Suppose that $\mathcal{F} = \{f_N\}_N$
is a collection of hierarchies $f_N=\{f_N^{(k)}\}_{k=1}^N$ such that each $%
f_N^{(k)}$ admits a decomposition 
\begin{equation*}
f_N^{(k)} = \sum_{\pi \in S_k} f_{N,\pi}^{(k)}
\end{equation*}
and

\noindent (a) Each component $f_{N,\pi}^{(k)}$ satisfies the uniform-in-$N$
bound 
\begin{equation}  \label{E:S3-06}
\forall \; k \geq 1 \,, \qquad \| f_N^{(k)}(t,{\boldsymbol{x}}_k,{%
\boldsymbol{v}}_k) \|_{C([0,T];X_{\pi}^{(k)})} \leq 2^{-k}C^k
\end{equation}

\noindent (b) Each $f_N= \{f_N^{(k)} \}_{k=1}^N$ satisfies the quantum BBGKY
hierarchy \eqref{E:S401}.

\noindent Then $P\mathcal{F}$ is precompact in the metric space metric space 
$C([0,T]; (\Lambda^*,\rho) )$, where the metric space $(\Lambda^*,\rho)$ is
defined below.
\end{theorem}

The metric space $\Lambda^*$, defined below, is $H_{\boldsymbol{x}_k}^{1+}
L_{\boldsymbol{v}_k}^{2,\frac12+}$ on each $k$ density with the weak-*
topology. The family $\mathcal{F} = \{ f_N \}_N$ is not bounded in $%
\Lambda^* $. In fact, the assumption \eqref{E:S3-06} only establishes that 
\begin{equation*}
\| f_N^{(k)} \|_{C([0,T]; L_{\boldsymbol{x}_k}^2 L_{\boldsymbol{v}_k}^2)}
\leq 2^{-k}C^kk!
\end{equation*}
Due to the fact that the $X_\pi^{(k)}$ norms involve an $\epsilon$-dependent
conversion of variables, as soon as derivatives are added, this upper bound
gains factors of $\epsilon^{-1}$ and thus diverges as $N\to \infty$. With
the projections $P_N^{(k)}$, however, we are able to recover 
\begin{equation*}
\Vert P_{N}^{(k)}f_{N}^{(k)}\Vert _{C([0,T];H_{\boldsymbol{x}_{k}}^{1+}L_{%
\boldsymbol{v}_{k}}^{2,\frac{1}{2}+})}\leq 2^{-k}C^{k}\Big(1+\sum_{\substack{
\pi \in S_{k}  \\ \pi \neq I}}\epsilon ^{(\ell (\pi )-m(\pi ))/2}\Big)
\end{equation*}
by Lemma \ref{L:PCcutoff} below and thus any limit point $f$ of $P\mathcal{F}
$ satisfies $f^{(k)} \in C([0,T]; (H_{\boldsymbol{x}_k}^{1+} L_{\boldsymbol{v%
}_k}^{2,\frac12+})_{\text{wk*}})$ and 
\begin{equation*}
\| f^{(k)} \|_{L^\infty_{[0,T]} H_{\boldsymbol{x}_k}^{1+} L_{\boldsymbol{v}%
_k}^{2,\frac12+}} \leq 2^{-k}C^k
\end{equation*}
In \S \ref{sec:Convergence}, we will prove that every limit point of $P%
\mathcal{F}$ in $C([0,T]; (\Lambda^*,\rho))$ satisfies the quantum Boltzmann
hierarchy. By the uniqueness of solutions to the quantum Boltzmann
hierarchy, proved in \S \ref{sec:uniqueness}, it follows that there is only
one limit point of $P\mathcal{F}$ in $C([0,T]; (\Lambda^*,\rho))$, which
then implies that $P_Nf_N \to f$ in $C([0,T]; (\Lambda^*,\rho))$ as $N\to
\infty$, where $f$ is the solution to the quantum Boltzmann hierarchy. At
the end of \S \ref{sec:Convergence}, we explain that $P_Nf_N \to f$ in $%
C([0,T]; (\Lambda^*,\rho))$ implies that $f_N^{(k)} \to f^{(k)}$ in $%
C([0,T]; (L^2_{\boldsymbol{x}_k, \boldsymbol{v}_k})_{\text{wk*}})$. This of
course raises the question of why we do not instead define $\Lambda^*$ to be 
$L_{\boldsymbol{x}_k}^2 L_{\boldsymbol{v}_k}^2$ on each $k$ density with the
weak-* topology. The reason is that, in the proof of Theorem \ref%
{T:convergence}, starting from a test function $J(\boldsymbol{x}_k, 
\boldsymbol{v}_k)$ (in the space of $k$-densities), a test function $H(%
\boldsymbol{x}_{k+1}, \boldsymbol{v}_{k+1})$ (in the space of $k+1$
densities) emerges \eqref{E:CV02}, built from $J$ and the adjoint of the
kernel of the collision operators. These test functions must lie in $\Lambda$%
. The test function $H$ only belongs to $H_{\boldsymbol{x}%
_{k+1}}^{-\frac34-}L_{\boldsymbol{v}_{k+1}}^{-2,\frac12-}$, as proved below %
\eqref{E:CV02}. Since $H$ does not belong to $L^2_{\boldsymbol{x}_{k+1}}
L^2_{\boldsymbol{v}_{k+1}}$, we cannot define $\Lambda^*$ to be $L_{%
\boldsymbol{x}_k}^2 L_{\boldsymbol{v}_k}^2$, and in fact the weakest space $%
\Lambda^*$ that could be used is $H_{\boldsymbol{x}_k}^{\frac34+} L_{%
\boldsymbol{v}_k}^{2,\frac12+}$ for each $k$ density. Fortunately, the
inclusion of the cutoff operator $P_N^{(k)}$ does not complicate any of the
weak limit analysis, since $P_N^{(k)}$ can be transferred to the test
function.

Now we define the space $(\Lambda^*, \rho)$. (For the moment, think of $t$
as absent or fixed). We start by defining the space $\Lambda$ of all $J=
(J^{(k)})_{k=1}^\infty$ satisfying $J^{(k)} \in H^{-1-}_{{\boldsymbol{x}}%
_k}L^{2,-\frac12-}_{{\boldsymbol{v}}_k}$ and 
\begin{equation*}
\lim_{k\to \infty} C^k \|J^{(k)} \|_{H^{-1-}_{{\boldsymbol{x}}%
_k}L^{2,-\frac12-}_{{\boldsymbol{v}}_k}} = 0
\end{equation*}
where $C$ is as in \eqref{E:S3-06}.\footnote{%
Although the definition of our space $\Lambda$ depends on $C$, we have
suppressed this in the notation, and just write $\Lambda$ instead of the
clumsier $\Lambda_C$.} The set $\Lambda$ is a Banach space with the norm 
\begin{equation*}
\|J\|_\Lambda = \sup_k C^k \|J^{(k)} \|_{H^{-1-}_{{\boldsymbol{x}}%
_k}L^{2,-\frac12-}_{{\boldsymbol{v}}_k}}
\end{equation*}
The dual space $\Lambda^*$ consists of all $f=(f^{(k)})_{k=1}^\infty$
satisfying $f^{(k)} \in H^{1+}_{{\boldsymbol{x}}_k}L^{2,\frac12+}_{{%
\boldsymbol{v}}_k}$ with norm 
\begin{equation}  \label{E:S3-08}
\| f\|_{\Lambda^*} = \sum_{k=1}^\infty C^{-k} \| f^{(k)}\|_{H^{1+}_{{%
\boldsymbol{x}}_k}L^{2,\frac12+}_{{\boldsymbol{v}}_k}}
\end{equation}
The space $\Lambda$ is separable; in fact, there is a countable dense set $%
\{ J_i \}_{i=1}^\infty \subset \Lambda$ that can be selected so that for
each $i$ and each $k$, $J_i^{(k)}$ is Schwartz class in ${\boldsymbol{x}}_k$%
, ${\boldsymbol{v}}_k$, and for each $i$, the function $J_i^{(k)}=0$ except
for $1\leq k \leq K(i)$. Put a metric on the space $\Lambda^*$, as follows: 
\begin{equation}  \label{E:S3-07}
\rho( f, g) = \sum_{i=1}^\infty 2^{-i} \min \left( 1, \;
\left|\sum_{k=1}^{K(i)} \int_{{\boldsymbol{x}}_k, {\boldsymbol{v}}_k} 
\overline{J_i^{(k)}}( f^{(k)}- g^{(k)}) \, d{\boldsymbol{x}}_k d{\boldsymbol{%
v}}_k \right|\right)
\end{equation}
Then, as topological spaces, $(\Lambda^*, \rho) = (\Lambda^*, \text{wk}*)$.

Consider now time dependent hierarchies and the space $C([0,T];
(\Lambda^*,\rho))$ which has metric $\sup_{0\leq t\leq T} \rho( f(t), g(t))$%
. We would like to show that our collection of BBGKY hierarchy solutions $%
\mathcal{F} = \{ f_N \}_N$ is a precompact set in $C([0,T];
(\Lambda^*,\rho)) $. Clearly for each $t\in [0,T]$, $\mathcal{F}_t =\{
f_N(t) \}_N$ is contained in the unit ball of $\Lambda^*$ by assumption %
\eqref{E:S3-06} and \eqref{E:S3-08}. Ascoli's theorem states that if $%
\mathcal{F}$ is equicontinuous under $\rho$ and for each $t\in [0,T]$, the
set $\mathcal{F}_t = \{ f_N(t) \}_N$ has compact closure in $(\Lambda^*,
\rho)$, then $\mathcal{F}$ is contained in a compact subset of $C([0,T];
(\Lambda^*,\rho))$. The fact that for each $t\in [0,T]$, the set $\mathcal{F}%
_t = \{ f_N(t) \}_N$ has compact closure in $(\Lambda^*, \rho)$ follows from
the weak-* compactness of the closed unit ball in $\Lambda^*$, and the fact
that $\rho$ induces the weak-* topology. The following elementary lemma
gives the equicontinuity criterion that we will employ.

\begin{lemma}
\label{L:equi} Suppose that there exists $\alpha>0$ such that the following
holds: for all $i\geq 1$, for all $1\leq k \leq K(i)$, and for all $N$,
there exists constants $C_{i,k}>0$ such that 
\begin{equation}  \label{E:S3-01}
\left| \int_{{\boldsymbol{x}}_k, {\boldsymbol{v}}_k} \overline{J_i^{(k)}}(
f_N^{(k)}(t)- f_N^{(k)}(s)) \, d{\boldsymbol{x}}_k d{\boldsymbol{v}}_k
\right| \leq C_{i,k} |t-s|^\alpha
\end{equation}
where $C_{i,k}$ is independent of $N$. Then $\mathcal{F}= \{ f_N \}_N$ is
equicontinuous in $C([0,T];(\Lambda^*,\rho))$, meaning that for each $\mu>0$%
, there exists $\delta>0$ such that for all $N$ 
\begin{equation*}
\forall \; 0\leq s \leq t \leq T \,, \qquad t-s \leq \delta \implies \rho(
f_N^{(k)}(t), f_N^{(k)}(s)) \leq \mu
\end{equation*}
\end{lemma}

\begin{proof}
Let $I(\mu)$ be chosen sufficiently large so that $2^{-I+1} \leq \mu$. Then
from \eqref{E:S3-07} and \eqref{E:S3-01}, 
\begin{equation*}
\rho_N(f_N^{(k)}(t),f_N^{(k)}(s)) \leq \frac12\mu + \left(
\sum_{i=1}^{I(\mu)} \sum_{k=1}^{K(i)} C_{i,k} \right) |t-s|^\alpha
\end{equation*}
For $\mu>0$, the function $\mu \mapsto \sum_{i=1}^{I(\mu)} \sum_{k=1}^{K(i)}
C_{i,k}$ is finite (although its rate of growth as $\mu\searrow 0$ is
unknown). Therefore, it suffices to take 
\begin{equation*}
\delta = \left( \frac{\mu}{2\sum_{i=1}^{I(\mu)} \sum_{k=1}^{K(i)} C_{i,k}}
\right)^{1/\alpha}
\end{equation*}
\end{proof}

The cutoff $P_N^{(k)}$ allows for control of the $\pi \neq \text{Id}$ terms
using the following lemma.

\begin{lemma}
\label{L:PCcutoff} Suppose that $\pi\neq \text{Id}$ and $f_{N,\pi}^{(k)}(%
\boldsymbol{x}_k, \boldsymbol{\xi}_k) = g_{N,\pi}^{(k)}(\boldsymbol{p}_k,%
\boldsymbol{q}_k)$, where $(\boldsymbol{p}_k, \boldsymbol{q}_k)= (%
\boldsymbol{p}_k^\pi, \boldsymbol{q}_k^\pi)$ are the transformed coordinates %
\eqref{E:PC10}. Then 
\begin{equation}  \label{E:PC103}
\| \tilde P_N^{(k)} \tilde f_{N,\pi}^{(k)} \|_{H^{1+}_{\boldsymbol{x}_k}
H^{\frac12+}_{\boldsymbol{\xi}_k}} \lesssim \epsilon^{\frac12(\ell-m)-}\|
\tilde g_{N,\pi}^{(k)} \|_{H^{\frac12+}_{\boldsymbol{p}_k}H^{\frac12+}_{%
\boldsymbol{q}_k}}
\end{equation}
where $\ell$ is the sum of the lengths of all nontrivial disjoint cycles in $%
\pi$ and $m$ is the number of such cycles.
\end{lemma}

\begin{proof}
Assume that $k=r$ and $\pi=(12\cdots r)$. Due the frequency cutoffs included
in $\tilde{P}_N^{(r)}$, 
\begin{equation}  \label{E:PC104}
\| \tilde{P}_N^{(r)} \tilde f_{N,\pi}^{(r)} \|_{H_{\boldsymbol{x}%
_r}^{1+\delta} H_{\boldsymbol{\xi}_r}^{\frac12+\delta}} \lesssim
\epsilon^{-\delta r (1+\delta)} \epsilon^{-\delta r(\frac12+\delta)} \| 
\tilde{P}_N^{(r)} \tilde f_{N,\pi}^{(r)} \|_{L_{\boldsymbol{x}_r}^2 L_{%
\boldsymbol{\xi}_r}^2}
\end{equation}
We will obtain two bounds: 
\begin{equation}  \label{E:PC100}
\| \tilde{P}_N^{(r)} \tilde f_{N,\pi}^{(r)}(\boldsymbol{x}_r, \boldsymbol{\xi%
}_r) \|_{L_{\boldsymbol{x}_r}^2 L_{\boldsymbol{\xi}_r}^2 } \lesssim \|
\tilde g_{N,\pi}^{(r)}(\boldsymbol{p}_r, \boldsymbol{q}_r) \|_{L_{%
\boldsymbol{p}_r}^2 L_{\boldsymbol{q}_r}^2}
\end{equation}
and 
\begin{equation}  \label{E:PC101}
\| \tilde{P}_N^{(r)} \tilde f_{N,\pi}^{(r)}(\boldsymbol{x}_r, \boldsymbol{\xi%
}_r) \|_{L_{\boldsymbol{x}_r}^2 L_{\boldsymbol{\xi}_r}^2 } \lesssim
\epsilon^{3(r-1)/2}\epsilon^{-\frac32 \delta r} \| \langle \nabla_{%
\boldsymbol{p}_r}\rangle^{\frac32+} \langle \nabla_{\boldsymbol{q}_r}
\rangle^{\frac32+} \tilde g_{N,\pi}^{(r)}(\boldsymbol{p}_r, \boldsymbol{q}%
_r) \|_{L_{\boldsymbol{p}_r}^2 L_{\boldsymbol{q}_r}^2}
\end{equation}
The first estimate \eqref{E:PC100} just follows by dropping $\tilde{P}%
_N^{(r)}$ and using the fact that the Jacobian for the variable conversion $(%
\boldsymbol{x}_r, \boldsymbol{\xi}_r) \mapsto (\boldsymbol{p}_r, \boldsymbol{%
q}_r)$ is $O(1)$ (independent of $\epsilon$).

The second estimate \eqref{E:PC101} is proved as follows. Due to the spatial
cutoff in $\boldsymbol{\xi}_r$ included in $\tilde{P}_N^{(r)}$ 
\begin{equation}  \label{E:PC102}
\| \tilde{P}_N^{(r)} \tilde f_{N,\pi}^{(r)} \|_{L_{\boldsymbol{x}_r}^2 L_{%
\boldsymbol{\xi}_r}^2} \lesssim \epsilon^{-\frac32 \delta r} \| \tilde
f_{N,\pi}^{(r)} \|_{L_{\boldsymbol{\xi}_r}^ \infty L_{\boldsymbol{x}_r}^2 }
\end{equation}
where at this point we discarded $\tilde{P}_N^{(r)}$. On the inside, sup out
over $p_2, \ldots, p_r, q_1$. Let 
\begin{equation*}
\tilde q_j = \epsilon q_j \,, \qquad j=1,\ldots, r
\end{equation*}
Then, for fixed $\boldsymbol{\xi}_r$, the mapping 
\begin{equation*}
\boldsymbol{x}_r \mapsto (p_1, \tilde q_2, \ldots, \tilde q_r)
\end{equation*}
is invertible with $O(1)$ Jacobian (independent of $\epsilon$). For fixed $%
\boldsymbol{\xi}_r$, implementing this change of variable gives 
\begin{equation*}
\| \tilde f_{N,\pi}^{(r)}(\boldsymbol{x}_r, \boldsymbol{\xi}_r) \|_{L_{%
\boldsymbol{x}_r}^2 } \lesssim \| \tilde g_{N,\pi}^{(r)}(\boldsymbol{p}_r,
q_1 , \epsilon^{-1}\tilde q_2, \ldots, \epsilon^{-1}\tilde q_r)
\|_{L_{p_1}^2 L_{\tilde q_2 \cdots \tilde q_r}^2 L_{p_2 \cdots p_r}^\infty
L_{q_1}^\infty }
\end{equation*}
By Sobolev embedding, 
\begin{equation*}
\| \tilde f_{N,\pi}^{(r)}(\boldsymbol{x}_r, \boldsymbol{\xi}_r) \|_{L_{%
\boldsymbol{x}_r}^2 } \lesssim \| \langle \nabla_{p_2} \rangle^{\frac32+}
\cdots \langle \nabla_{p_r} \rangle^{\frac32+}\langle \nabla_{q_1}
\rangle^{\frac32+} \tilde g_{N,\pi}^{(r)}(\boldsymbol{p}_r, q_1 ,
\epsilon^{-1}\tilde q_2, \ldots, \epsilon^{-1}\tilde q_r) \|_{L_{\boldsymbol{%
p}_k}^2 L_{q_1}^2 L_{\tilde q_2 \cdots \tilde q_r}^2 }
\end{equation*}
Changing variable $\tilde q_j \to q_j$ for $j=2,\ldots, r$, we obtain 
\begin{equation*}
\| \tilde f_{N,\pi}^{(r)}(\boldsymbol{x}_r, \boldsymbol{\xi}_r) \|_{L_{%
\boldsymbol{x}_r}^2 } \lesssim \epsilon^{3(r-1)/2} \| \langle \nabla_{%
\boldsymbol{p}_r}\rangle^{\frac32+} \langle \nabla_{\boldsymbol{q}_r}
\rangle^{\frac32+} \tilde g_{N,\pi}^{(r)}(\boldsymbol{p}_r, \boldsymbol{q}%
_r) \|_{L_{\boldsymbol{p}_r}^2 L_{\boldsymbol{q}_r}^2}
\end{equation*}
By interpolating \eqref{E:PC100} and \eqref{E:PC101}. 
\begin{equation*}
\| \tilde f_{N,\pi}^{(r)}(\boldsymbol{x}_r, \boldsymbol{\xi}_r) \|_{L_{%
\boldsymbol{x}_r}^2 L_{\boldsymbol{\xi}_r}^2 } \lesssim \epsilon^{(r-1)/2-}
\| \langle \nabla_{\boldsymbol{p}_r}\rangle^{\frac12+} \langle \nabla_{%
\boldsymbol{q}_r} \rangle^{\frac12+} \tilde g_{N,\pi}^{(r)}(\boldsymbol{p}%
_r, \boldsymbol{q}_r) \|_{L_{\boldsymbol{p}_r}^2 L_{\boldsymbol{q}_r}^2}
\end{equation*}
Combining with \eqref{E:PC104} gives the claimed bound \eqref{E:PC103} in
the case of one cycle of length $r$, so $\ell=r$ and $m=1$. The general case
follows by separately treating each collection of coordinates in a disjoint
cycle.
\end{proof}

With these preliminaries out of the way, we can now prove Theorem \ref%
{T:compactness}.

\begin{proof}[Proof of Theorem \protect\ref{T:compactness}]
We verify the condition \eqref{E:S3-01} in Lemma \ref{L:equi}. To this end,
fix $k\geq 1$, and for notational convenience we will take $s=0$. Since the
projection $P_N^{(k)}$ can just be transferred to the test function $J^{(k)}$%
, where it has no effect on estimates, we will drop it from the exposition.
We need to show that for fixed Schwartz class $J(\boldsymbol{x}_k, 
\boldsymbol{v}_k)$, 
\begin{equation}  \label{E:S430}
\left| \int_{{\boldsymbol{x}}_k, {\boldsymbol{v}}_k} \overline{J}^{(k)}\;
(f_N^{(k)}(t)- f_N^{(k)}(0)) \, d{\boldsymbol{x}}_k d{\boldsymbol{v}}_k
\right| \lesssim |t|
\end{equation}
where the implicit constant can depend on $k$ and $J$, but not on $N$.

Appealing to the second Duhamel iterate \eqref{E:S402}, 
\begin{equation}  \label{E:S431}
\begin{aligned} f_N^{(k)}(t)-f_N^{(k)}(0) &= [S^{(k)}(t)-I]f_N^{(k)}(0) +
\mathcal{D}^{(k)} R^{2(k)}_N f_N^{(k)}(t) + \mathcal{D}^{(k)} R^{3(k+1)}_N
f_N^{(k+1)}(t) \\ & \qquad + \mathcal{D}^{(k)} (Q_N^{(k+1)}+R_N^{4(k+1)})
f_N^{(k+1)}(t) + \mathcal{D}^{(k)} R^{5(k+2)}_N f_N^{(k+2)}(t) \end{aligned}
\end{equation}
where we adopt the notation of \S \ref{S:preparation} for the Duhamel
operators $\mathcal{D}^{(k)}$, the collision operators $Q_N^{(k+1)}$, and
the remainder operators $R_N^{2(k)}$, $R_N^{3(k+1)}$, $R_N^{4(k+1)}$, and $%
R_N^{5(k+2)}$. Specifically, 
\begin{equation*}
R_N^{2(k)}f_N^{(k)} = \epsilon^{-1/2}A_\epsilon^{(k)}f_N^{(k)} \,,
\end{equation*}
\begin{equation*}
R_N^{3(k+1)}f_N^{(k+1)} = N\epsilon^{-1/2}B_\epsilon^{(k+1)}S^{(k+1)}
f_N^{(k+1)}(0) \,,
\end{equation*}
$Q_N^{(k+1)}$ is defined in \eqref{E:S411a}, $R_N^{4(k+1)}$ is defined in %
\eqref{E:S411b}, and $R_N^{5(k+2)}$ is defined in \eqref{E:S412}.

Each term in \eqref{E:S431} is substituted into \eqref{E:S430}. Each of the
five pieces is then estimated by Cauchy Schwarz with $J$ in the Sobolev norm
dual to the norm on the left side of needed estimate (Lemma \ref{L:basic-A}--%
\ref{L:basic-B}, Proposition \ref{P:QEstimates}--\ref{P:R5Estimates}).

For the $S^{(k)}(t)-I$ term, we use the trivial estimate 
\begin{equation*}
\| \langle \boldsymbol{\eta}_k \rangle^{-1} \langle \boldsymbol{v}_k
\rangle^{-1} |e^{it\boldsymbol{\eta}_k\cdot \boldsymbol{v}_k}-1| \hat
f^{(k)}(0) \|_{L^2_{\boldsymbol{\eta}_k} L^2_{\boldsymbol{v}_k}} \lesssim
|t| \| \hat f^{(k)}(0) \|_{L^2_{\boldsymbol{\eta}_k} L^2_{\boldsymbol{v}_k}}
\end{equation*}
For the $R^{2(k)}$ term, the needed estimate follows from Lemma \ref%
{L:basic-A} with $s=\frac12+$, and the $|t|$ factor on the right-side comes
from the outer Duhamel operator. For the $R^{3(k+1)}$ term, the needed
estimate follows from Lemma \ref{L:basic-B} with $s=\frac12+$, combined with
the straightforward estimate 
\begin{align*}
\hspace{0.3in}&\hspace{-0.3in} \| e^{-t \boldsymbol{\eta}_{k+1}\cdot
\nabla_{\xi_{k+1}}} \check f_N^{(k+1)}(0, \boldsymbol{\eta}_{k+1}, 
\boldsymbol{\xi}_k, 0) \|_{L_t^\infty L^2_{\boldsymbol{\eta}_{k+1}} L^2_{%
\boldsymbol{\xi}_k}} \\
&\lesssim \| e^{-\boldsymbol{\eta}_{k+1}\cdot \nabla_{\xi_{k+1}}} \check
f_N^{(k+1)}(0, \boldsymbol{\eta}_{k+1}, \boldsymbol{\xi}_{k+1}) \|_{L^2_{%
\boldsymbol{\eta}_{k+1}} L^2_{\boldsymbol{\xi}_k} L^\infty_{\xi_{k+1}}}
\end{align*}
For the $Q_N^{(k+1)}$ the needed estimate follows from Proposition \ref%
{P:QEstimates}. For the $R_N^{4(k+1)}$ the needed estimate follows from
Proposition \ref{P:R4Estimates}. For the $R_N^{5(k+2)}$ the needed estimate
follows from Proposition \ref{P:R5Estimates}. Note that these estimates
incorporate the outer Duhamel operator, but still generate a factor of $|t|$.
\end{proof}

\section{Convergence to the Boltzmann Hierarchy\label{sec:Convergence}}

Recall the definition of the Boltzmann hierarchy (\ref{hierarchy:Boltzmann}) 
\begin{equation*}
\partial _{t}f^{(k)}+{\boldsymbol{v}}_{k}\cdot \nabla _{{\boldsymbol{x}}%
_{k}}f^{(k)}=Q^{(k+1)}f^{(k+1)}\,,\qquad k\geq 1
\end{equation*}%
with the collision operator given as (\ref{eqn:collision kernel for
hierarchy}). As we mostly work in the $\vee $ side in this section, we also
recall the $\vee $ side collision operator given by \eqref{E:S435} and %
\eqref{E:S434}.

\begin{theorem}
\label{T:convergence} Suppose that $f_\infty=\{f^{(k)}\}_{k=1}^{\infty }$ is
any limit point (convergence in $C([0,T];(\Lambda ^{\ast },\rho ))$) of $%
\mathcal{P}\mathcal{F}$, where $\mathcal{F}=\{f_{N}\}_{N}$ is a collection
of hierarchies satisfying the hypotheses of Theorem \ref{T:compactness}.
Then $f_\infty$ satisfies the Boltzmann hierarchy (\ref{hierarchy:Boltzmann}%
) with initial condition $f_\infty(0)=\lim_{N\rightarrow \infty }f_{N}(0)$
(convergence in $(\Lambda ^{\ast },\rho )$) and satisfies $f^{(k)} \in
C([0,T];(H_{{\boldsymbol{x}}_{k}}^{1+}L_{{\boldsymbol{v}}_{k}}^{2,\frac{1}{2}%
+})_{\text{wk*}})$ along with the bounds 
\begin{equation*}
\forall \;k\geq 1\,,\qquad \Vert f^{(k)}(t,{\boldsymbol{x}}_{k},{\boldsymbol{%
v}}_{k})\Vert _{L^\infty_{[0,T]} H_{{\boldsymbol{x}}_{k}}^{1+}L_{{%
\boldsymbol{v}}_{k}}^{2,\frac{1}{2}+}}\leq 2^{-k}C^{k}
\end{equation*}
\end{theorem}

\begin{proof}
It suffices to \emph{assume} that $f_\infty$ satisfies the Boltzmann
hierarchy (\ref{hierarchy:Boltzmann}) and then prove that 
\begin{equation*}
\sup_{t\in \lbrack 0,T]}\rho (P_Nf_{N}(t),f_\infty(t))\rightarrow 0\text{ as 
}N\rightarrow \infty \text{ (along some subsequence)}
\end{equation*}%
By the definition of $\rho $ given in \eqref{E:S3-07}, it suffices to show
for that for any Schwartz $J$, 
\begin{equation}
\left\vert \sup_{t\in \lbrack 0,T]}\int_{\boldsymbol{x}_{k},\boldsymbol{v}%
_{k}}\bar{J}^{(k)}\,(P_N^{(k)}f_{N}^{(k)}(t)-f^{(k)}(t))\,d\boldsymbol{x}%
_{k}d\boldsymbol{v}_{k}\right\vert \lesssim \epsilon ^{0+}  \label{E:S440}
\end{equation}
where the implicit constant can depend on $k$ and $J$.

Appealing to the second Duhamel iterate \eqref{E:S402}, 
\begin{equation}  \label{E:S437}
\begin{aligned} f_N^{(k)}(t) &= S^{(k)}(t)f_N^{(k)}(0) + \mathcal{D}^{(k)}
R^{2(k)}_N f_N^{(k)}(t) + \mathcal{D}^{(k)} R^{3(k+1)}_N f_N^{(k+1)}(t) \\ &
\qquad + \mathcal{D}^{(k)} (Q_N^{(k+1)}+R_N^{4(k+1)}) f_N^{(k+1)}(t) +
\mathcal{D}^{(k)} R^{5(k+2)}_N f_N^{(k+2)}(t) \end{aligned}
\end{equation}
where we adopt the notation of \S \ref{S:preparation} for the Duhamel
operators $\mathcal{D}^{(k)}$, the collision operators $Q_N^{(k+1)}$, and
the remainder operators $R_N^{2(k)}$, $R_N^{3(k+1)}$, $R_N^{4(k+1)}$, and $%
R_N^{5(k+2)}$. Specifically, 
\begin{equation*}
R_N^{2(k)}f_N^{(k)} = \epsilon^{-1/2}A_\epsilon^{(k)}f_N^{(k)} \,, \quad
R_N^{3(k+1)}f_N^{(k+1)} = N\epsilon^{-1/2}B_\epsilon^{(k+1)}S^{(k+1)}
f_N^{(k+1)}(0) \,,
\end{equation*}
$Q_N^{(k+1)}$ is defined in \eqref{E:S411a}, $R_N^{4(k+1)}$ is defined in %
\eqref{E:S411b}, and $R_N^{5(k+2)}$ is defined in \eqref{E:S412}.

The Duhamel representation of \eqref{E:S432} is 
\begin{equation}  \label{E:S438}
f^{(k)}(t) = S^{(k)}(t)f^{(k)}(0) +\mathcal{D}^{(k)} Q^{(k+1)} f^{(k+1)}(t)
\end{equation}
Taking the difference of \eqref{E:S437} and \eqref{E:S438} gives 
\begin{equation}  \label{E:S439}
\begin{aligned} & P_N^{(k)}f_N^{(k)}(t)-f^{(k)}(t) \\ &=
S^{(k)}(t)[f_N^{(k)}(0)-f^{(k)}(0)] + \mathcal{D}^{(k)} R^{2(k)}_N
f_N^{(k)}(t) + \mathcal{D}^{(k)} R^{3(k+1)}_N f_N^{(k+1)}(t) \\ & \qquad +
\mathcal{D}^{(k)}[Q_N^{(k+1)}P_N^{(k+1)}f_N^{(k+1)}(t)-Q^{(k+1)}f^{(k+1)}(t)] + \mathcal{D}^{(k)}R_N^{4(k+1)} f_N^{(k+1)}(t) \\ & \qquad + \mathcal{D}^{(k)} R^{5(k+2)}_N f_N^{(k+2)}(t) + (P_N^{(k)}-I)f_N^{(k)}(t) + \mathcal{D}^{(k)}Q_N^{(k+1)}(I-P_N^{(k+1)})f_N^{(k+1)}(t) \end{aligned}
\end{equation}
We substitute each of the terms in \eqref{E:S439} into \eqref{E:S440}. The
terms involving $R^{2(k)}_N$, $R^{3(k+1)}_N$, $R_N^{4(k+1)}$, $R^{5(k+2)}_N$
are estimated exactly as in the proof of Theorem \ref{T:compactness} by
appealing to Lemma \ref{L:basic-A}--\ref{L:basic-B}, Proposition \ref%
{P:QEstimates}--\ref{P:R5Estimates} after Cauchy-Schwarz placing $J$ in the
corresponding dual Sobolev norm. Note that each of the estimates here yields
a factor $\epsilon^{0+}$, which are needed here although were not needed for
the proof of Theorem \ref{T:compactness}. It is also trivial to dispose of
the term $S^{(k)}(t)[f_N^{(k)}(0)-f^{(k)}(0)]$ since it is assumed that $%
f_N(0)\to f_\infty(0)$ in $(\Lambda^*,\rho)$. The term involving $%
(P_N^{(k)}-I)f_N^{(k)}(t)$ goes to zero due to the smoothness and decay of $%
J^{(k)}$. All cycle terms of $\mathcal{D}%
^{(k)}Q_N^{(k+1)}(I-P_N^{(k+1)})f_N^{(k+1)}(t)$ go to zero by the estimates
of \S \ref{S:PCE}, while the core component of $f_N^{(k+1)}(t)$ is handled
by transferring the collision operator onto the test function, generating a
new test function in $H_{\boldsymbol{x}_{k+1}}^{-\frac34-}L_{\boldsymbol{v}%
_{k+1}}^{2,-\frac12-}$, as in the proof of \eqref{E:S443} below (see %
\eqref{E:CV02}).

Thus, the crux of the proof is to handle $\mathcal{D}%
^{(k)}[Q_N^{(k+1)}P_N^{(k+1)}f_N^{(k+1)}(t)-Q^{(k+1)}f^{(k+1)}(t)]$. To
shorten formulae will drop the $P_N^{(k+1)}$ operator from here on out. It
is helpful to recall \eqref{E:S433} and \eqref{E:S434}: 
\begin{equation*}
\begin{aligned} \hspace{0.3in}&\hspace{-0.3in} \check
Q_{i,k+1}^{\epsilon}\check f_N^{(k+1)}(t) = - \sum_{\alpha,\sigma \in \{\pm
1\}} \alpha \sigma \int_{\eta_{k+1}} \int_y \int_{s=0}^{t/\epsilon}
\hat\phi(\epsilon \eta_{k+1} - y) \hat \phi(y) \\ &e^{i\alpha \xi_i
(\epsilon \eta_{k+1} - y)/2} e^{i\sigma \xi_i y/2} e^{-i\sigma s
(\epsilon\eta_i-2\epsilon \eta_{k+1} + 2y)y/2} \\ &\check
f_N^{(k+1)}(t-\epsilon s, \eta_1, \ldots, \eta_i-\eta_{k+1}, \ldots,
\eta_{k+1} , \xi_1 - \epsilon s \eta_1, \, . \,, \\ & \qquad \xi_i - sy -
\epsilon s\eta_i+\epsilon s \eta_{k+1}, \, . \,, \xi_k-\epsilon s
\eta_{k+1}, sy -s\epsilon\eta_{k+1}) \, d\eta_{k+1} \, dy \, ds \end{aligned}
\end{equation*}
and 
\begin{equation*}
\begin{aligned} \check Q_{i,k+1}\check f^{(k+1)}(t, \boldsymbol{\eta}_k,
\boldsymbol{\xi}_k) = - \sum_{\alpha,\sigma \in \{\pm 1\}} \alpha \sigma
\int_{\eta_{k+1}} \int_y \int_{s=0}^{\infty} |\hat \phi(y)|^2
e^{i(\sigma-\alpha)\xi_iy/2} e^{-i \sigma s|y|^2} &\\ \check f^{(k+1)}(t,
\eta_1, \, . \,, \eta_i-\eta_{k+1}, \, . \,, \eta_{k+1} , \xi_1 , \, . \,,
\xi_i-sy, \, . \, , \xi_k, sy ) \, d\eta_{k+1} \, dy \, ds & \end{aligned}
\end{equation*}
The relationship $N=\epsilon^{-3}$ is always fixed; for this proof we adjust
our notation for $f_N$ to $f_\epsilon$. Fix the test function $J$ and let 
\begin{equation*}
I_{i,k+1}^{\epsilon, \delta} (t) = \int_{\boldsymbol{\eta}_k, \boldsymbol{\xi%
}_k} \bar{\check J} \, \check{\mathcal{D}}^{(k)} \check
Q^\epsilon_{i,k+1}\check f_\delta^{(k+1)}(t) \,d\boldsymbol{\eta_k} \, d%
\boldsymbol{\xi_k}
\end{equation*}
In this notation, $\epsilon=0$ means $Q_{i,k+1}$ (the limiting collision
operator) and $\delta=0$ means $f_\delta^{(k+1)}$ (limiting value of $%
f_N^{(k+1)})$.

We must show that 
\begin{equation}  \label{E:S441}
\sup_{0\leq t \leq T} | I_{i,k+1}^{\epsilon,\epsilon}(t) -
I_{i,k+1}^{0,0}(t) | \to 0 \text{ as } \epsilon \to 0
\end{equation}

Before addressing \eqref{E:S441}, which will of course exploit cancelation
between the two terms, we examine $I_{i,k+1}^{\epsilon,\delta}(t)$
individually. Note that we can move the propagator in $\check{\mathcal{D}}%
^{(k)}$ onto $\check J$: 
\begin{equation*}
I_{i,k+1}^{\epsilon, \delta} (t) = \int_{t^{\prime} =0}^t \int_{\boldsymbol{%
\eta}_k, \boldsymbol{\xi}_k} \bar{\check J}(\boldsymbol{\eta}_k,\boldsymbol{%
\xi}_k+(t-t^{\prime} )\boldsymbol{\eta}_k) \, \check
Q^\epsilon_{i,k+1}\check f_\delta^{(k+1)}(t^{\prime} , \boldsymbol{\eta}_k, 
\boldsymbol{\xi}_k) \,d\boldsymbol{\eta_k} \, d\boldsymbol{\xi_k} \,
dt^{\prime}
\end{equation*}
With $\check Q_{i,k+1}^\epsilon$ written out, 
\begin{align*}
I_{i,k+1}^{\epsilon, \delta} (t) &= - \sum_{\alpha,\sigma \in \{\pm 1\}}
\alpha \sigma \int_{t^{\prime} =0}^t \int_{\boldsymbol{\eta}_k, \boldsymbol{%
\xi}_k} \bar{\check J}(\boldsymbol{\eta}_k,\boldsymbol{\xi}_k+(t-t^{\prime} )%
\boldsymbol{\eta}_k) \int_{\eta_{k+1}} \int_y
\int_{s=0}^{t^{\prime}/\epsilon} \\
&\hat\phi(\epsilon \eta_{k+1} - y) \hat \phi(y) e^{i\alpha \xi_i (\epsilon
\eta_{k+1} - y)/2} e^{i\sigma \xi_i y/2} e^{-i\sigma s
(\epsilon\eta_i-2\epsilon \eta_{k+1} + 2y)y/2} \\
&\check f_\delta^{(k+1)}(t^{\prime}-\epsilon s, \eta_1, \ldots,
\eta_i-\eta_{k+1}, \ldots, \eta_{k+1} , \xi_1 - \epsilon s \eta_1, \, . \,,
\\
& \qquad \xi_i - sy - \epsilon s\eta_i+\epsilon s \eta_{k+1}, \, . \,,
\xi_k-\epsilon s \eta_{k+1}, sy -s\epsilon\eta_{k+1}) \, d\eta_{k+1} \, dy
\, ds
\end{align*}
This formula is also valid when $\epsilon=0$ provided $t^{\prime}/\epsilon$
is replaced by $\infty$.

Our first observation is that we can discard $s\geq \epsilon^{-1/2}$.
Specifically, define (note the new superscript $\beta$): 
\begin{align*}
I_{i,k+1}^{\epsilon, \delta, \beta} (t) &= - \sum_{\alpha,\sigma \in \{\pm
1\}} \alpha \sigma \int_{t^{\prime} =0}^t \int_{\boldsymbol{\eta}_k, 
\boldsymbol{\xi}_k} \bar{\check J}(\boldsymbol{\eta}_k,\boldsymbol{\xi}%
_k+(t-t^{\prime} )\boldsymbol{\eta}_k) \int_{\eta_{k+1}} \int_y
\int_{s=0}^{\beta^{-1}} \\
&\hat\phi(\epsilon \eta_{k+1} - y) \hat \phi(y) e^{i\alpha \xi_i (\epsilon
\eta_{k+1} - y)/2} e^{i\sigma \xi_i y/2} e^{-i\sigma s
(\epsilon\eta_i-2\epsilon \eta_{k+1} + 2y)y/2} \\
&\check f_\delta^{(k+1)}(t^{\prime}-\epsilon s, \eta_1, \ldots,
\eta_i-\eta_{k+1}, \ldots, \eta_{k+1} , \xi_1 - \epsilon s \eta_1, \, . \,,
\\
& \qquad \xi_i - sy - \epsilon s\eta_i+\epsilon s \eta_{k+1}, \, . \,,
\xi_k-\epsilon s \eta_{k+1}, sy -s\epsilon\eta_{k+1}) \, d\eta_{k+1} \, dy
\, ds
\end{align*}
that restricts the $s$ integration range to $0\leq s \leq \beta^{-1}$.

\begin{lemma}
\label{L:I-cut} Uniformly in $\epsilon\geq 0$, $\delta\geq 0$, we have 
\begin{equation*}
\sup_{0\leq t \leq T} |I_{i,k+1}^{\epsilon, \delta} (t) -
I_{i,k+1}^{\epsilon, \delta, \beta} (t)| \lesssim (\epsilon \beta^{-1+} +
t\beta^{0+}) \| \langle \eta_i \rangle^{\frac34+} \langle \eta_{k+1}
\rangle^{\frac34+} \langle \nabla_{\xi_{k+1}}\rangle^{\frac12+} \check
f_\delta^{(k+1)} \|_{L_{[0,T]}^\infty L_{\boldsymbol{\eta}_{k+1}}^2 L_{%
\boldsymbol{\xi}_{k+1}}^2}
\end{equation*}
In particular, if we take $\beta=\epsilon^{1/2}$, we obtain an $%
\epsilon^{0+} $ prefactor.
\end{lemma}

\begin{proof}
From the definitions, 
\begin{align*}
& I_{i,k+1}^{\epsilon, \delta} (t) - I_{i,k+1}^{\epsilon, \delta, \beta} (t)
= - \sum_{\alpha,\sigma \in \{\pm 1\}} \alpha \sigma \int_{t^{\prime} =0}^t
\int_{\boldsymbol{\eta}_k, \boldsymbol{\xi}_k} \bar{\check J}(\boldsymbol{%
\eta}_k,\boldsymbol{\xi}_k+(t-t^{\prime} )\boldsymbol{\eta}_k)
\int_{\eta_{k+1}} \int_y \int_{s=\beta^{-1}}^{t^{\prime}/\epsilon} \\
&\hat\phi(\epsilon \eta_{k+1} - y) \hat \phi(y) e^{i\alpha \xi_i (\epsilon
\eta_{k+1} - y)/2} e^{i\sigma \xi_i y/2} e^{-i\sigma s
(\epsilon\eta_i-2\epsilon \eta_{k+1} + 2y)y/2} \\
&\check f_\delta^{(k+1)}(t^{\prime}-\epsilon s, \eta_1, \ldots,
\eta_i-\eta_{k+1}, \ldots, \eta_{k+1} , \xi_1 - \epsilon s \eta_1, \, . \,,
\\
& \qquad \xi_i - sy - \epsilon s\eta_i+\epsilon s \eta_{k+1}, \, . \,,
\xi_k-\epsilon s \eta_{k+1}, sy -s\epsilon\eta_{k+1}) \, d\eta_{k+1} \, dy
\, ds
\end{align*}
By Cauchy-Schwarz 
\begin{align*}
\hspace{0.3in}&\hspace{-0.3in} |I_{i,k+1}^{\epsilon, \delta} (t) -
I_{i,k+1}^{\epsilon, \delta, \beta} (t)| \lesssim \|\check J \|_{L^2_{%
\boldsymbol{\xi}_k} L^2_{\boldsymbol{\eta}_k}} \int_{t^{\prime}=0}^t
\int_{s=\beta^{-1}}^{t^{\prime}/\epsilon} \Big\| \int_{\eta_{k+1}} \int_y
|\hat \phi(\epsilon \eta_{k+1} - y)| |\hat\phi(y)| \\
& |\check f_\delta^{(k+1)}(t^{\prime}-\epsilon s, \eta_1, \ldots,
\eta_i-\eta_{k+1}, \, . \,, \eta_{k+1} , \xi_1 - \epsilon s \eta_1, \, . \,,
\\
& \qquad \xi_i - sy - \epsilon s\eta_i+\epsilon s \eta_{k+1}, \, . \,,
\xi_k-\epsilon s \eta_{k+1}, sy -s\epsilon\eta_{k+1})| \, dy \, d\eta_{k+1} %
\Big\|_{L^2_{\boldsymbol{\eta}_k}L^2_{\boldsymbol{\xi}_k}} \, ds
\end{align*}
By following the proof of Proposition \ref{P:QEstimates}, we obtain 
\begin{equation*}
\begin{aligned} &|I_{i,k+1}^{\epsilon, \delta} (t) - I_{i,k+1}^{\epsilon,
\delta, \beta} (t)| \\ &\lesssim \int_{t^{\prime}=0}^t
\int_{s=\beta^{-1}}^{t^{\prime}/\epsilon} s^{-1+}\langle s \rangle^{0-} \,
ds \, dt^{\prime} \; \| \langle \eta_i \rangle^{\frac34+} \langle \eta_{k+1}
\rangle^{\frac34+} \langle \nabla_{\xi_{k+1}}\rangle^{\frac12+} \check
f_\delta^{(k+1)} \|_{L_{[0,T]}^\infty L_{\boldsymbol{\eta}_{k+1}}^2
L_{\boldsymbol{\xi}_{k+1}}^2} \end{aligned}
\end{equation*}
Carrying out the $s$ integral, we obtain 
\begin{equation*}
\int_{t^{\prime}=0}^t \int_{s=\beta^{-1}}^{t^{\prime}/\epsilon}
s^{-1+}\langle s \rangle^{0-} \, ds \, dt^{\prime} \sim
\int_{t^{\prime}=0}^t \langle
\min(\beta^{-1},t^{\prime}/\epsilon)\rangle^{0-} \, dt^{\prime} \leq
\int_{t^{\prime}=0}^t [\max(\beta, \epsilon/t^{\prime})]^{0+} \, dt^{\prime}
\end{equation*}
The $t^{\prime}$ integral is carried out in two pieces. First, $%
t^{\prime}\leq \epsilon \beta^{-1}$, in which case $\max(\beta^{-1},
t^{\prime}/\epsilon) = t^{\prime}/\epsilon$. Second, $t^{\prime}\geq
\epsilon \beta^{-1}$, in which case $\max(\beta^{-1}, t^{\prime}/\epsilon) =
\beta^{-1}$. The evaluation of these two integrals gives the result.
\end{proof}

In view of Lemma \ref{L:I-cut}, to prove \eqref{E:S441}, it suffices to show 
\begin{equation*}
\sup_{0\leq t \leq T} | I_{i,k+1}^{\epsilon,\epsilon,\epsilon^{1/2}}(t) -
I^{0,0,\epsilon^{1/2}}(t) | \to 0 \text{ as }\epsilon \to 0
\end{equation*}
And to prove this, it suffices to prove 
\begin{equation}  \label{E:S442}
\sup_{0\leq t \leq T} | I_{i,k+1}^{\epsilon,\epsilon,\epsilon^{1/2}}(t) -
I^{0,\epsilon,\epsilon^{1/2}}(t) | \lesssim \epsilon^{0+} \| \langle \eta_i
\rangle^{\frac34+} \langle \eta_{k+1} \rangle^{\frac34+} \langle
\nabla_{\xi_i} \rangle^{0+} \langle \nabla_{\xi_{k+1}}\rangle^{\frac12+}
\check f_\delta^{(k+1)} \|_{L_{[0,T]}^\infty L_{\boldsymbol{\eta}_{k+1}}^2
L_{\boldsymbol{\xi}_{k+1}}^2}
\end{equation}
and 
\begin{equation}  \label{E:S443}
\sup_{0\leq t \leq T} | I_{i,k+1}^{0,\epsilon,\epsilon^{1/2}}(t) -
I^{0,0,\epsilon^{1/2}}(t) | \to 0 \quad \text{as} \quad \epsilon \to 0
\end{equation}

We begin with the proof of \eqref{E:S442}. For this, we start by switching
the order of $t^{\prime} $ and $s$ integrals and shifting the $t^{\prime} $
integration 
\begin{equation}  \label{E:S445}
\begin{aligned} I_{i,k+1}^{\epsilon, \delta, \beta} (t) &= -
\sum_{\alpha,\sigma \in \{\pm 1\}} \alpha \sigma \int_{s=0}^{\beta^{-1}}
\int_{t^{\prime} =-\epsilon s}^{t-\epsilon} \int_{\boldsymbol{\eta}_k,
\boldsymbol{\xi}_k} \bar{\check
J}(\boldsymbol{\eta}_k,\boldsymbol{\xi}_k+(t-t^{\prime} -\epsilon
s)\boldsymbol{\eta}_k) \int_{\eta_{k+1}} \int_y \\ &\hat\phi(\epsilon
\eta_{k+1} - y) \hat \phi(y) e^{i\alpha \xi_i (\epsilon \eta_{k+1} - y)/2}
e^{i\sigma \xi_i y/2} e^{-i\sigma s (\epsilon\eta_i-2\epsilon \eta_{k+1} +
2y)y/2} \\ &\check f_\delta^{(k+1)}(t^{\prime}, \eta_1, \ldots,
\eta_i-\eta_{k+1}, \ldots, \eta_{k+1} , \xi_1 - \epsilon s \eta_1, \, . \,,
\\ & \qquad \xi_i - sy - \epsilon s\eta_i+\epsilon s \eta_{k+1}, \, . \,,
\xi_k-\epsilon s \eta_{k+1}, sy -s\epsilon\eta_{k+1}) \, d\eta_{k+1} \, dy
\, ds\end{aligned}
\end{equation}
This moves the $s$ translation from the time argument of $f$ to the test
function. The rest of the proof of \eqref{E:S442} is just a matter of
applying the fundamental theorem of calculus: 
\begin{equation*}
I^{\epsilon, \delta, \beta}_{i,k+1}(t) - I^{0,\delta, \beta}_{i,k+1}(t) =
\int_{\theta =0}^{1} \frac{d}{d\theta}[ I^{\epsilon \theta, \delta,
\beta}_{i,k+1}(t)] \,d\theta
\end{equation*}
and then carrying out, via the chain rule, the $\theta$-derivative of %
\eqref{E:S445} with $\epsilon$ replaced by $\epsilon \theta$ (note that $%
\delta$ and $\beta$ are held fixed, although after the calculation is
completed, we set $\delta=\epsilon$ and $\beta=\epsilon^{1/2}$). Rather than
write one very long formula, we provide a table giving the result of each
term generated. We have enumerated the terms in the left column for ease of
reference below.

\begin{center}
\renewcommand{\arraystretch}{1.5}
\setlength{\tabcolsep}{5pt}  
\begin{tabular}{c|c|c}
Term \# & $\theta$ derivative lands on & generates \\ \hline
1 & $\bar{\check J}(\boldsymbol{\eta}_k, \boldsymbol{\xi}_k +
(t-t^{\prime}-\epsilon \theta s)\boldsymbol{\eta}_k)$ & $-\epsilon s \nabla_{%
\boldsymbol{\xi}_k} \bar{\check J}(\boldsymbol{\eta}_k, \boldsymbol{\xi}_k +
(t-t^{\prime}-\epsilon \theta s)\boldsymbol{\eta}_k) \cdot \boldsymbol{\eta}%
_k$ \\ 
2 & $\hat\phi(\epsilon \theta \eta_{k+1} - y)$ & $\epsilon \nabla \hat\phi(
\epsilon \theta \eta_{k+1}-y) \cdot \eta_{k+1}$ \\ 
3 & $\exp\left[ \tfrac12 i \theta\epsilon (\alpha \xi_i + 2\sigma s y) \cdot
\eta_{k+1} \right]$ & $\tfrac12 i (\alpha \epsilon \xi_i + 2\sigma \epsilon
s y) \cdot \eta_{k+1} \exp[ \tfrac12 i \cdots]$ \\ 
4 & $\exp\left[ - \tfrac12 i \theta\epsilon \sigma s \eta_i \cdot y \right]$
& $- \tfrac12 i \epsilon s \sigma \eta_i \cdot y \exp[ - \tfrac12 i \cdots]$
\\ 
5 & $\check f_\delta^{(k+1)}(\cdots)$ & $\epsilon s [ -\boldsymbol{\eta}%
_k\cdot \nabla_{\boldsymbol{\xi}_k} + \eta_{k+1}\cdot (\nabla_{\xi_i} -
\nabla_{\xi_{k+1}})] \check f_\delta^{(k+1)}(\cdots) $%
\end{tabular}
\end{center}

Note that in each case, a factor $\epsilon$ emerges, but in some cases this
comes along with an $s$, which is not small, but is bounded by $\beta^{-1}$,
and when $\beta=\epsilon^{-1/2}$, we have $|s\epsilon| \leq \epsilon^{1/2}$.
Factors $\boldsymbol{\eta}_k$ are absorbed onto $\check J$, as are any
factors $\boldsymbol{\xi}_k$ since we can write 
\begin{equation*}
\boldsymbol{\xi}_k = [\boldsymbol{\xi}_k + (t-t^{\prime}-\epsilon \theta s) 
\boldsymbol{\eta}_k] - [(t-t^{\prime}-\epsilon \theta s) \boldsymbol{\eta}_k]
\end{equation*}
and the prefactor $|(t-t^{\prime}-\epsilon \theta s)| \lesssim 1$. The extra 
$y$ can be absorbed by $\hat\phi(y)$. The factor $\eta_{k+1}$ must be added
to the right-side, as must the $\xi$-derivatives on $\check
f_\delta^{k+1}(\cdots)$ that emerge in Term 5.

The very long expression for $\frac{d}{d\theta}[ I^{\epsilon \theta, \delta,
\beta}_{i,k+1}(t)]$, which has five copies of \eqref{E:S445} with each of
the term replacements as given in the table above, can be estimated by the
method of proof of Proposition \ref{P:QEstimates} to yield 
\begin{equation*}
\sup_{0\leq t \leq T} | I_{i,k+1}^{\epsilon,\epsilon,\epsilon^{1/2}}(t) -
I^{0,\epsilon,\epsilon^{1/2}}(t) | \lesssim \epsilon^{1/2} \| \langle \eta_i
\rangle^{\frac34+} \langle \eta_{k+1} \rangle^{\frac74+} \langle
\nabla_{\xi_i} \rangle^1 \langle \nabla_{\xi_{k+1}}\rangle^{\frac32+} \check
f_\delta^{(k+1)} \|_{L_{[0,T]}^\infty L_{\boldsymbol{\eta}_{k+1}}^2 L_{%
\boldsymbol{\xi}_{k+1}}^2}
\end{equation*}
This can be interpolated with the trivial bound that ignores cancelation and
just estimates $I_{i,k+1}^{\epsilon,\epsilon,\epsilon^{1/2}}$ and $%
I^{0,\epsilon,\epsilon^{1/2}}(t)$ separately via Proposition \ref%
{P:QEstimates}, with no $\epsilon$ gain. The result of this interpolation
leaves an $\epsilon^{0+}$ at the expense of 
\begin{equation*}
\langle \eta_{k+1}\rangle^{0+} \langle \nabla_{\xi_i} \rangle^{0+} \langle
\nabla_{\xi_{k+1}}\rangle^{0+}
\end{equation*}
This completes the proof of \eqref{E:S442}.

Next, we prove \eqref{E:S443}. Both terms are expressed in terms of the
limiting collision operator $Q_{i,k+1}$, and the difference can be expressed
as $Q_{i,k+1}$ acting on $f_\epsilon^{(k+1)} - f^{(k+1)}$. Thus, the
argument hinges upon whether the kernel of $Q_{i,k+1}$ and the outside test
function $J$ can together serve as a test function in $(k+1)$ variables, so
that we can appeal to the fact that $f_\epsilon^{(k+1)} - f^{(k+1)} \to 0$
in $(H^{1+}_{\boldsymbol{x}_{k+1}} L^{2,\frac12+}_{\boldsymbol{v}_{k+1}})_{%
\text{wk*}}$.

By Lemma \ref{L:I-cut}, we can remove the $\beta^{-1}=\epsilon^{-1/2}$
cutoff on the $s$-integral and thus it suffices to examine 
\begin{equation*}
\begin{aligned} I_{i,k+1}^{0, \epsilon} (t)- I_{i,k+1}^{0, 0} (t) &= -
\sum_{\alpha,\sigma \in \{\pm 1\}} \alpha \sigma \int_{s=0}^{\infty}
\int_{t^{\prime} =0}^{t} \int_{\boldsymbol{\eta}_k, \boldsymbol{\xi}_k}
\bar{\check
J}(\boldsymbol{\eta}_k,\boldsymbol{\xi}_k+(t-t^{\prime})\boldsymbol{\eta}_k)
\\ & \int_{\eta_{k+1}} \int_y |\hat \phi(y)|^2 e^{i(\sigma-\alpha)\xi_iy/2}
e^{-i\sigma s |y|^2} (\check f_\epsilon^{(k+1)}- \check f^{(k+1)}) \\ &
(t^{\prime}, \eta_1, \, . \,, \eta_i-\eta_{k+1}, \, . \,, \eta_{k+1} , \xi_1
, \, . \,, \xi_i - sy , \, . \,, \xi_k, sy ) \, d\eta_{k+1} \, dy \,
ds\end{aligned}
\end{equation*}
Change variable $y\mapsto \xi_{k+1}/s$ 
\begin{equation*}
\begin{aligned} I_{i,k+1}^{0, \epsilon} (t)- I_{i,k+1}^{0, 0} (t) &= -
\sum_{\alpha,\sigma \in \{\pm 1\}} \alpha \sigma \int_{s=0}^{\infty}
\int_{t^{\prime} =0}^{t} \int_{\boldsymbol{\eta}_k, \boldsymbol{\xi}_k}
\bar{\check
J}(\boldsymbol{\eta}_k,\boldsymbol{\xi}_k+(t-t^{\prime})\boldsymbol{\eta}_k)
\\ & \int_{\eta_{k+1}} \int_{\xi_{k+1}} s^{-3}|\hat \phi( \xi_{k+1} /s)|^2
e^{i(\sigma-\alpha)\xi_i \xi_{k+1} /(2s)} e^{-i\sigma | \xi_{k+1} |^2/s}
(\check f_\epsilon^{(k+1)}- \check f^{(k+1)}) \\ & (t^{\prime}, \eta_1, \, .
\,, \eta_i-\eta_{k+1}, \, . \,, \eta_{k+1} , \xi_1 , \, . \,, \xi_i -
\xi_{k+1} , \, . \,, \xi_k, \xi_{k+1} ) \, d\eta_{k+1} \, d \xi_{k+1} \,
ds\end{aligned}
\end{equation*}
Let 
\begin{equation}  \label{E:CV01}
h(\xi_i, \xi_{k+1} ) = \int_{s=0}^\infty s^{-3}|\hat\phi( \xi_{k+1} /s)|^2
e^{i(\sigma-\alpha)\xi_i \xi_{k+1} /(2s)} e^{-i\sigma| \xi_{k+1} |^2/s} \, ds
\end{equation}
Then 
\begin{equation*}
\begin{aligned} I_{i,k+1}^{0, \epsilon} (t)- I_{i,k+1}^{0, 0} (t) &= -
\sum_{\alpha,\sigma \in \{\pm 1\}} \alpha \sigma \int_{t^{\prime} =0}^{t}
\int_{\boldsymbol{\eta}_k, \boldsymbol{\xi}_k} \int_{\eta_{k+1}}
\int_{\xi_{k+1}} \bar{\check
J}(\boldsymbol{\eta}_k,\boldsymbol{\xi}_k+(t-t^{\prime})\boldsymbol{\eta}_k)
\\ & \qquad h(\xi_i, \xi_{k+1} ) (\check f_\epsilon^{(k+1)}- \check
f^{(k+1)})(t^{\prime}, \eta_1, \, . \,, \eta_i-\eta_{k+1}, \, . \,,
\eta_{k+1} , \\ &\hspace{2in} \xi_1 , \, . \,, \xi_i - \xi_{k+1} , \, . \,,
\xi_k, \xi_{k+1} ) \, d\eta_{k+1} \, d \xi_{k+1} \end{aligned}
\end{equation*}
Replace $\eta_i \mapsto \eta_i+\eta_{k+1}$ and $\xi_i\mapsto \xi_i
+\xi_{k+1} $. The result is 
\begin{equation}  \label{E:CV04}
\begin{aligned} I_{i,k+1}^{0, \epsilon} (t)- I_{i,k+1}^{0, 0} (t) &= -
\sum_{\alpha,\sigma \in \{\pm 1\}} \alpha \sigma \int_{t^{\prime} =0}^{t}
\int_{\boldsymbol{\eta}_{k+1}, \boldsymbol{\xi}_{k+1}} \\ & \qquad
H(\boldsymbol{\eta}_{k+1}, \boldsymbol{\xi}_{k+1}) (\check
f_\epsilon^{(k+1)}- \check f^{(k+1)})(t^{\prime}, \boldsymbol{\eta}_{k+1},
\boldsymbol{\xi}_{k+1} ) \, d\eta_{k+1} \, d \xi_{k+1} \end{aligned}
\end{equation}
where 
\begin{equation}  \label{E:CV02a}
\begin{aligned} H(\boldsymbol{\eta}_{k+1}, \boldsymbol{\xi}_{k+1}) =
&h(\xi_i+\xi_{k+1}, \xi_{k+1} )\bar{\check J}(\eta_1 \, \, . \,,
\eta_i+\eta_{k+1}, \, . \,, \eta_k, \\ &\quad \xi_1 +(t-t^{\prime})\eta_1 \,
\, . \,, \xi_i+\xi_{k+1}+(t-t^{\prime})(\eta_i+\eta_{k+1}), \, . \,,
\xi_k+(t-t^{\prime})\eta_k) \end{aligned}
\end{equation}
We claim that 
\begin{equation}  \label{E:CV02}
H \in L_{\boldsymbol{\eta}_{k+1}}^{2,-\frac34-}H_{\boldsymbol{\xi}%
_{k+1}}^{-\frac12-}
\end{equation}
uniformly in $t$, $t^{\prime}$.

First, we note that with $h(\xi_i,\xi_{k+1})$ defined by \eqref{E:CV01}, we
have 
\begin{equation}  \label{E:CV03}
|h(\xi_i,\xi_{k+1})| \lesssim |\xi_{k+1}|^{-2}
\end{equation}
To see that \eqref{E:CV03} holds, note that 
\begin{equation*}
|h(\xi_i, \xi_{k+1} )| \leq \int_{s=0}^\infty s^{-3}|\hat\phi( \xi_{k+1}
/s)|^2 \, ds
\end{equation*}
Break the $s$ integral into $s\leq |\xi_{k+1}|$ and $s>|\xi_{k+1}|$. For $%
s<|\xi_{k+1}|$, since $|\xi_{k+1}|/s \geq 1$, we use the assumption that $%
\hat \phi(y) \lesssim |y|^{-1-}$ for $|y|\geq 1$. Then 
\begin{equation*}
\int_{s=0}^{|\xi_{k+1}|} \frac{1}{s^3} |\hat\phi( \frac{\xi_{k+1}}{s} )|^2
\, ds \leq \int_{s=0}^{|\xi_{k+1}|} \frac{1}{s^3} \frac{s^{2+}}{%
|\xi_{k+1}|^{2+}} \, ds = \frac{1}{|\xi_{k+1}|^2}
\end{equation*}
For $s>|\xi_{k+1}|$, since $|\xi_{k+1}|/s <1$, we just use that $%
|\hat\phi(y) \lesssim 1$ for $|y| \leq 1$. Then 
\begin{equation*}
\int_{s=|\xi_{k+1}|}^{+\infty} \frac{1}{s^3} |\hat\phi( \frac{\xi_{k+1}}{s}
)|^2 \, ds \lesssim \int_{s=|\xi_{k+1}|}^{+\infty} \frac{1}{s^3} \, ds = 
\frac{1}{|\xi_{k+1}|^2}
\end{equation*}
This completes the proof of \eqref{E:CV03}.

By \eqref{E:CV03} and the fact that $J \in \mathcal{S}$, 
\begin{align*}
|H(\boldsymbol{\eta}_{k+1}, \boldsymbol{\xi}_{k+1})| \lesssim &
|\xi_{k+1}|^{-2} \langle \eta_1 \rangle^{-2} \, \cdots \, \langle \eta_i +
\eta_{k+1}\rangle^{-2} \cdots \langle \eta_{k} \rangle^{-2} \cdots \\
&\langle \xi_1+(t-t^{\prime})\eta_1 \rangle^{-2} \cdots \langle
\xi_i+\xi_{k+1}+ (t-t^{\prime})(\eta_i+\eta_{k+1}) \cdots \langle
\xi_k+(t-t^{\prime})\eta_k\rangle^{-2}
\end{align*}
Hence 
\begin{equation*}
\| H(\boldsymbol{\eta}_{k+1}, \boldsymbol{\xi}_{k+1}) \|_{L^2_{\boldsymbol{%
\xi}_k}} \lesssim |\xi_{k+1}|^{-2} \langle \eta_1 \rangle^{-2} \, \cdots \,
\langle \eta_i + \eta_{k+1}\rangle^{-2} \cdots \langle \eta_{k} \rangle^{-2}
\end{equation*}
By dual Sobolev embedding, 
\begin{align*}
\| H(\boldsymbol{\eta}_{k+1}, \boldsymbol{\xi}_{k+1}) \boldsymbol{1}%
_{|\xi_{k+1}| \leq 1}\|_{H_{\xi_{k+1}}^{-\frac12-}L^2_{\boldsymbol{\xi}_k}}
&\lesssim \| H(\boldsymbol{\eta}_{k+1}, \boldsymbol{\xi}_{k+1}) \boldsymbol{1%
}_{|\xi_{k+1}| \leq 1}\|_{L_{\xi_{k+1}}^{\frac32-} L^2_{\boldsymbol{\xi}_k}}
\\
&\lesssim \langle \eta_1 \rangle^{-2} \, \cdots \, \langle \eta_i +
\eta_{k+1}\rangle^{-2} \cdots \langle \eta_{k} \rangle^{-2}
\end{align*}
Also, 
\begin{align*}
\| H(\boldsymbol{\eta}_{k+1}, \boldsymbol{\xi}_{k+1}) \boldsymbol{1}%
_{|\xi_{k+1}| \leq 1}\|_{H_{\xi_{k+1}}^{-\frac12-}L^2_{\boldsymbol{\xi}_k}}
&\lesssim \| H(\boldsymbol{\eta}_{k+1}, \boldsymbol{\xi}_{k+1}) \boldsymbol{1%
}_{|\xi_{k+1}| \leq 1}\|_{L_{\xi_{k+1}}^{2} L^2_{\boldsymbol{\xi}_k}} \\
&\lesssim \langle \eta_1 \rangle^{-2} \, \cdots \, \langle \eta_i +
\eta_{k+1}\rangle^{-2} \cdots \langle \eta_{k} \rangle^{-2}
\end{align*}
Consequently, 
\begin{equation*}
\| H(\boldsymbol{\eta}_{k+1},\boldsymbol{\xi}_{k+1}) \|_{L_{\boldsymbol{\eta}%
_{k+1}}^{2,-\frac34-} H_{\xi_{k+1}}^{-\frac34-}} \lesssim \left(
\int_{\eta_i, \eta_{k+1}} \langle \eta_i\rangle^{-\frac32-} \langle
\eta_{k+1} \rangle^{-\frac32-} \langle \eta_i+\eta_{k+1}\rangle^{-4} \,
d\eta_i \, d\eta_{k+1} \right)^{1/2} < \infty
\end{equation*}
completing the proof of \eqref{E:CV02}.

This establishes that the test function in \eqref{E:CV04} belongs to $%
\Lambda $. Since (recalling $\rho$ is the metric assigned to $\Lambda^*$) 
\begin{equation*}
\rho( f_\epsilon^{(k+1)}(t^{\prime}), f^{(k+1)}(t^{\prime})) \to 0
\end{equation*}
uniformly in $t^{\prime}$, it follows that \eqref{E:S443} holds.
\end{proof}

\begin{corollary}
Suppose that $\mathcal{F}= \{f_N\}_N$ is a collection of hierarchies
satisfying the hypotheses of Theorem \ref{T:compactness} and $P_Nf_N \to
f_\infty$ in $C([0,T]; (\Lambda^*,\rho))$, where $f_\infty = \{ f^{(k)}
\}_{k=1}^\infty$. Then for all $k\geq 1$, $f_N^{(k)} \to f^{(k)}$ in $%
C([0,T]; (L^2_{\boldsymbol{x}_k \boldsymbol{v}_k})_{\text{wk*}})$.
\end{corollary}

\begin{proof}
The hypotheses of Theorem \ref{T:compactness} suffice to imply that 
\begin{equation*}
\| f_N^{(k)}(t) \|_{L^2_{\boldsymbol{x}_k\boldsymbol{v}_k}} \leq C^k k!
\end{equation*}
By density it suffices to show that for any of the test functions $J^{(k)}$
described in the construction of $\Lambda$, we have 
\begin{equation*}
\int_{\boldsymbol{x}_k \boldsymbol{v}_k} J^{(k)}(\boldsymbol{x}_k, 
\boldsymbol{v}_k) (f_N^{(k)}-f^{(k)})(\boldsymbol{x}_k, \boldsymbol{v}_k) \,
d\boldsymbol{x}_k \, d\boldsymbol{v}_k \to 0
\end{equation*}
It suffices to show that both 
\begin{equation*}
\int_{\boldsymbol{x}_k \boldsymbol{v}_k} J^{(k)}(\boldsymbol{x}_k, 
\boldsymbol{v}_k) (f_N^{(k)}-P_N^{(k)}f_N^{(k)})(\boldsymbol{x}_k, 
\boldsymbol{v}_k) \, d\boldsymbol{x}_k \, d\boldsymbol{v}_k \to 0
\end{equation*}
\begin{equation*}
\int_{\boldsymbol{x}_k \boldsymbol{v}_k}J^{(k)}(\boldsymbol{x}_k, 
\boldsymbol{v}_k) ( P_N^{(k)} f_N^{(k)}-f^{(k)})(\boldsymbol{x}_k, 
\boldsymbol{v}_k) \, d\boldsymbol{x}_k \, d\boldsymbol{v}_k \to 0
\end{equation*}
The first of these holds since $\|(I-P_N^{(k)}) J^{(k)}\|_{L^2_{\boldsymbol{x%
}_k \boldsymbol{v}_k}} \to 0$ as $N\to \infty$, while, uniformly in $N$, $%
\|f_N^{(k)}\|_{L^2_{\boldsymbol{x}_k\boldsymbol{v}_k}} \leq C^kk!$. The
second statement holds since it was assumed that $P_N f_N \to f_\infty$ in $%
C([0,T]; (\Lambda^*,\rho))$.
\end{proof}


\section{Unconditional Uniqueness of the Boltzmann Hierarchy \label%
{sec:uniqueness}}

We now turn our attention to proving Theorem \ref{Thm:Uniqueness of
Hierarchy} which concludes that there is only one limit point resulting from
the process in \S \ref{sec:CompactnessConvergence}-\ref{sec:Convergence}.

\begin{definition}
We say the the family $\{f^{(k)}\}$ is admissible if it satisfies one of the
following:

(i) It is a weak $N\rightarrow \infty $ limit point of $\{f_{N}^{(k)}\}$
which is some family of marginals of a symmetric $N$-body system on $\Omega
. $

(ii) It is a family of symmetric probability marginals.
\end{definition}

\begin{theorem}
\label{Thm:Uniqueness of Hierarchy}There is at most one admissible solution
to the quantum Boltzmann hierarchy (\ref{hierarchy:Boltzmann}) in $\left[ 0,T%
\right] $ subject to the condition that, there is a $C_{0}>0$, such that%
\begin{eqnarray}
\sup_{t\in \left[ 0,T\right] }\left\Vert \left(
\dprod\limits_{j=1}^{k}\left\langle \nabla _{x_{j}}\right\rangle
^{1+}\left\langle v_{j}\right\rangle ^{0+}\right) f^{(k)}(t,\mathbf{x}_{k},%
\mathbf{v}_{k})\right\Vert _{L_{x,v}^{2}} &\leqslant &C_{0}^{k}
\label{eqn:condition for uniqueness} \\
\sup_{t\in \left[ 0,T\right] }\left\Vert \left(
\dprod\limits_{j=1}^{k}\left\langle \nabla _{x_{j}}\right\rangle ^{\frac{1}{2%
}+}\left\langle v_{j}\right\rangle ^{\frac{1}{2}+}\right) f^{(k)}(t,\mathbf{x%
}_{k},\mathbf{v}_{k})\right\Vert _{L_{x,v}^{2}} &\leqslant &C_{0}^{k}
\label{eqn:condition for uniqueness 2}
\end{eqnarray}
\end{theorem}

\begin{corollary}
\label{Cor:Uniqueness of Boltzmann}There is at most one $C\left( \left[ 0,T%
\right] ,H_{x}^{1+}L_{v}^{2,0+}\cap H_{x}^{\frac{1}{2}+}L_{v}^{2,\frac{1}{2}%
+}\right) $ solution to the Boltzmann equation (\ref{eqn:QBEquation}).
\end{corollary}

\begin{proof}
This corollary follows from the proof of Theorem \ref{Thm:Uniqueness of
Hierarchy} without using Lemma \ref{lem:Hewitt-Savaege}, and hence does not
need the admissibile condition.
\end{proof}

The proof of Theorem \ref{Thm:Uniqueness of Hierarchy} consists of 3 main
ingredients: a Klainerman-Machedon (KM) combinatoric which is stated as
Lemma \ref{lem:extended board game} and proved in \S \ref{sec:boardgame
proof} to combine the factorial many terms into exponentially many terms; a
Hewitt-Savage theorem,\footnote{\cite{ACI} also suggests the usage of this
theorem.} which gives a representation as a superposition of molecular chaos
for the solution and hence simplifies the proofs of the needed bilinear
estimates\footnote{%
As we are not at scaling critical regularity, trace type multilinear
estimate, which implies the product type we use here, can be proved. But,
away from requiring a even more technical analysis, it would result in a
conditional uniqueness theorem which needs a rerun of Section \ref%
{sec:Convergence} to verify the condition.}; and finally, the bilinear
estimates in \S \ref{Subsec:Multilinear Estimates}, which will be iterated
to conclude the difference of the solutions is actually zero.

We use the following version of the Hewitt-Savage theorem.

\begin{lemma}[Hewitt-Savage]
\label{lem:Hewitt-Savaege}Let $\{f_{N}^{(k)}\}$ be the family of marginals
of a symmetric $N$-body system on $\Omega $ and let the family $\{f^{(k)}\}$
be a weak $N\rightarrow \infty $ limit point of $\{f_{N}^{(k)}\}$, then
there exists a probability measure $d\mu (\rho )$ on $\mathcal{P}(\Omega )$,
the space of probability measures on $\Omega $, such that%
\begin{equation}
f^{(k)}=\int_{\mathcal{P}(\Omega )}\rho ^{\otimes k}d\mu (\rho ).
\label{eqn:H-WRep}
\end{equation}
\end{lemma}

\begin{proof}
There are many versions and related references for this theorem. See, for
example, \cite{ACI}, in which a version was used to deal with the
homogeneous case. The version we are using here is actually from the lecture
note \cite[\S 2]{R1}. As written in \cite[(2.3) or (2.10)]{R1}, in the $N$%
-body context, the version one would like to use is mostly%
\begin{equation*}
f_{N}^{(k)}\rightharpoonup _{\ast }\int_{\mathcal{P}(\Omega )}\rho ^{\otimes
k}d\mu (\rho )\text{ as }N\rightarrow \infty .
\end{equation*}%
The quick argument on \cite[p.29]{R1} directly near \cite[(2.10)]{R1} needs
compactness of $\Omega $. It is then further investigated and a proof for
the non-compact $\Omega $ case is given in \cite[p.34]{R1}.
\end{proof}

Under representation (\ref{eqn:H-WRep}), we can restate the requirement (\ref%
{eqn:condition for uniqueness})-(\ref{eqn:condition for uniqueness 2}) using
the Chebyshev's inequality. In fact, if we take the $\left\langle \nabla
_{x}\right\rangle ^{1+}\left\langle v\right\rangle ^{0+}$ part as an
example, like in \cite{CHPS}, we have for all $K>C_{0}$ that%
\begin{eqnarray*}
&&\mu \left\{ \rho \in \mathcal{P}(\Omega ):\left\Vert \left\langle \nabla
_{x}\right\rangle ^{1+}\left\langle v\right\rangle ^{0+}\rho \right\Vert
_{L_{x,v}^{2}}>K\right\} \\
&\leqslant &\frac{\left\Vert \left( \dprod\limits_{j=1}^{k}\left\langle
\nabla _{x_{j}}\right\rangle ^{1+}\left\langle v_{j}\right\rangle
^{0+}\right) f^{(k)}(t,\mathbf{x}_{k},\mathbf{v}_{k})\right\Vert
_{L_{x,v}^{2}}}{K^{k}}\leqslant \frac{C_{0}^{k}}{K^{k}}\text{, for all }k,
\end{eqnarray*}%
that is,%
\begin{equation*}
\mu \left\{ \rho \in \mathcal{P}(\Omega ):\left\Vert \left\langle \nabla
_{x}\right\rangle ^{1+}\left\langle v\right\rangle ^{0+}\rho \right\Vert
_{L_{x,v}^{2}}>K\right\} =0.
\end{equation*}%
That is, $\mu $ is supported in the set of functions (not only probabilty
measures): 
\begin{equation}
E(C_{0})=\left\{ \rho \in \mathcal{P}(\Omega ):\max \left\{ \left\Vert
\left\langle \nabla _{x}\right\rangle ^{1+}\left\langle v\right\rangle
^{0+}\rho \right\Vert _{L_{x,v}^{2}},\left\Vert \left\langle \nabla
_{x}\right\rangle ^{\frac{1}{2}+}\left\langle v\right\rangle ^{\frac{1}{2}%
+}\rho \right\Vert _{L_{x,v}^{2}}\right\} \leqslant C_{0}\right\} .
\label{eqn:measure spt property}
\end{equation}

Let $f_{1}^{(k)}$ and $f_{2}^{(k)}$ be two solutions of the quantum
Boltzmann hierarchy subject to the same initial condition and (\ref%
{eqn:condition for uniqueness})-(\ref{eqn:condition for uniqueness 2}), and $%
\mu _{1,t}$ and $\mu _{2,t}$ be their corresponding Hewitt-Savage measures,
we would like to deduce Theorem \ref{Thm:Uniqueness of Hierarchy} by proving 
\begin{equation*}
f^{(k)}(t,\mathbf{x}_{k},\mathbf{v}_{k})=\int \rho ^{\otimes k}d\mu
_{t}(\rho )\equiv \int \rho ^{\otimes k}d\left( \mu _{1,t}-\mu _{2,t}\right)
(\rho )=0.
\end{equation*}%
Here, $\mu _{t}$ is a signed measure, but we only need the properties that $%
d\left\vert \mu _{t}\right\vert =d\mu _{1,t}+d\mu _{2,t}$ is finite and it
is supported in $E(C_{0})$ defined in (\ref{eqn:measure spt property}). It
suffices to prove $f^{(1)}=0$ as the general case follows from the same
proof but with longer superscripts. Using the linearity of (\ref%
{hierarchy:Boltzmann}), we know 
\begin{equation*}
f^{(k)}(t_{k},\mathbf{x}_{k},\mathbf{v}_{k})=%
\int_{0}^{t_{k}}S_{k,k+1}^{(k)}Q^{(k+1)}(f^{(k+1)})dt_{k+1},
\end{equation*}%
where we have taken up the shorthand 
\begin{equation*}
S_{i,l}^{(k)}\equiv \dprod\limits_{j=1}^{k}e^{-(t_{i}-t_{l})v_{j}\cdot
\nabla _{x_{k}}}\text{ and }S_{i,l}=e^{-(t_{i}-t_{l})v\cdot \nabla _{x}}
\end{equation*}%
Iterating the hierarchy relation, we obtain the Dyson series-like
interaction expansion\footnote{%
There are many names attached to such expansions. But as we are in the
quantum setting, we use Dyson or Duhamel-Born here.} of $f^{(1)},$ 
\begin{equation}
f^{(1)}(t_{1},x_{1},v_{1})=\int_{0}^{t_{1}}\int_{0}^{t_{2}}...%
\int_{0}^{t_{k}}D^{(k+1)}(f^{(k+1)}(t_{k+1}))d\text{\b{t}}_{k+1}
\label{eqn:f_1'sBornExpansion}
\end{equation}%
where \b{t}$_{k+1}=\left( t_{2},t_{3},...,t_{k+1}\right) $ and%
\begin{equation*}
D^{(k+1)}(f^{(k+1)}(t_{k+1}))=S_{1,2}^{(1)}Q^{(2)}S_{2,3}^{(2)}Q^{(3)}...S_{k,k+1}^{(k)}Q^{(k+1)}.
\end{equation*}%
As $Q^{(k)}$ has $k$ terms inside (without splitting into $Q^{+},Q^{-}$), (%
\ref{eqn:f_1'sBornExpansion}) contains $(k+1)!$ many summands. In the
Lanford method, such a factorial factor is countered by a simplex integral
of the time domain. In the quantum setting, there are some known
combinatorics based on Feymann diagrams. But we will not use any Feymann
diagrams, we use our own combinatoric, a KM board game, to reduce the number
of terms by combining them.\footnote{%
We are using binary trees for our algorithm, but they are not Feymann
diagrams. Feymann diagrams make up a proper subset of binary trees.}

\begin{lemma}[Klainerman-Machedon board game]
\label{lem:extended board game}One can group the $(k+1)!$ many summands
inside (\ref{eqn:f_1'sBornExpansion}) into at most $4^{k}$ classes indexed
by $\mu \in m_{k}$, where $m_{k}$ is a set of suitable permutations in the
permutation group $S_{k+1}$ satisfying $\mu (j)<j$ for $j=2,...,k+1$.%
\footnote{$m_{k}$ is the set of upper echelon trees as we will see in the
proof in \S \ref{sec:boardgame proof}.} For each class $\mu $, all summands
inside that class, can be summed (combined) into one integral%
\begin{equation}
I_{\mu }^{(k+1)}\left( f^{(k+1)}\right) (t_{1})=\int_{T(\mu )}D_{\mu
}^{(k+1)}(f^{(k+1)}(t_{k+1}))d\text{\b{t}}_{k+1}
\label{eqn:reference integration}
\end{equation}%
where the time integration domain $T(\mu )\subset \left[ 0,t_{1}\right] ^{k}$
is a union of simplexes and is explicitly determined by $\mu $ and the
integrand is given by%
\begin{equation*}
D_{\mu }^{(k+1)}(f^{(k+1)}(t_{k+1}))=S_{1,2}^{(1)}Q_{\mu
(2),2}S_{2,3}^{(2)}Q_{\mu (3),3}...S_{k,k+1}^{(k)}Q_{\mu (k+1),k+1}.
\end{equation*}
\end{lemma}

\begin{proof}
See \S \ref{sec:boardgame proof}.
\end{proof}

With Lemma \ref{lem:extended board game}, we turn our attention to the
estimate of $I_{\mu }$.

\begin{proposition}
\label{prop:estimate with boardgame}For $I_{\mu }$ coming from Lemma \ref%
{lem:extended board game}, and $f^{(k)}=\int \rho ^{\otimes k}d\mu _{t}(\rho
)$, we have%
\begin{eqnarray*}
&&\sup_{t_{1}\in \left[ 0,T\right] }\left\Vert \langle \nabla _{\xi
_{1}}\rangle ^{0+}\left( I_{\mu }^{(k+1)}\left( f^{(k+1)}\right)
(t_{1})\right) ^{\vee }\right\Vert _{L_{\eta _{1}\xi _{1}}^{3-}} \\
&\leqslant &2\left( CT^{\frac{1}{2}-}\right) ^{k-1}T\int \left\Vert \rho
\right\Vert _{H_{x_{1}}^{1+}L_{v_{1}}^{2,0+}}^{k-2}\left\Vert \rho
\right\Vert _{H_{x_{1}}^{\frac{1}{2}+}L_{v_{1}}^{2,\frac{1}{2}%
+}}^{2}d\left\vert \mu _{t_{k+1}}\right\vert (\rho )
\end{eqnarray*}%
where $C$ is a constant from Sobolev and Strichartz type inequalities and
does not depend on $f$.
\end{proposition}

\begin{proof}
See \S \ref{sec:Iteration Estimates}.
\end{proof}

With Lemma \ref{lem:extended board game} and Proposition \ref{prop:estimate
with boardgame}, we provide the proof of Theorem \ref{Thm:Uniqueness of
Hierarchy} as the following.

\begin{proof}[Proof of Theorem \protect\ref{Thm:Uniqueness of Hierarchy}]
By Lemma \ref{lem:extended board game}, we have%
\begin{equation*}
\sup_{t_{1}\in \left[ 0,T\right] }\left\Vert \langle \nabla _{\xi
_{1}}\rangle ^{0+}\check{f}^{(1)}\right\Vert _{L_{\eta _{1}\xi
_{1}}^{3-}}\leqslant 4^{k}\sum_{\mu \in m_{k}}\sup_{t_{1}\in \left[ 0,T%
\right] }\left\Vert \langle \nabla _{\xi _{1}}\rangle ^{0+}\left( I_{\mu
}^{(k+1)}\left( f^{(k+1)}\right) (t_{1})\right) ^{\vee }\right\Vert
_{L_{\eta _{1}\xi _{1}}^{3-}}\text{ for all }k
\end{equation*}%
Let $T<1$ to be determined, use Proposition \ref{prop:estimate with
boardgame}, 
\begin{equation*}
\sup_{t_{1}\in \left[ 0,T\right] }\left\Vert \langle \nabla _{\xi
_{1}}\rangle ^{0+}\check{f}^{(1)}\right\Vert _{L_{\eta _{1}\xi
_{1}}^{3-}}\leqslant 4^{k}C^{k}T^{\frac{k}{4}}\int \left\Vert \rho
\right\Vert _{H_{x_{1}}^{1+}L_{v_{1}}^{2,0+}}^{k}\left\Vert \rho \right\Vert
_{H_{x_{1}}^{\frac{1}{2}+}L_{v_{1}}^{2,\frac{1}{2}+}}d\left\vert \mu
_{t_{k+1}}\right\vert (\rho )\ \text{for all }k
\end{equation*}%
Apply the support property (\ref{eqn:measure spt property}),%
\begin{eqnarray*}
\sup_{t_{1}\in \left[ 0,T\right] }\left\Vert \check{f}^{(1)}\right\Vert
_{L_{\eta _{1}\xi _{1}}^{3-}} &\leqslant &4^{k}C^{k}T^{\frac{k}{4}%
}C_{0}^{k+1}\cdot 2\ \text{for all }k \\
&\rightarrow &0\text{ as }k\rightarrow \infty
\end{eqnarray*}%
if we select $T$ such that $\left( 4CC_{0}T^{\frac{1}{4}}\right) <\frac{1}{2}
$.
\end{proof}

\subsection{Proof of the Klainerman-Machedon board game\label{sec:boardgame
proof}}

The Klainerman-Machedon (KM) board game \cite{KM1} and its extensions \cite%
{CH8,CSZ}, since invented, has been used in every paper in which the
analysis of the Gross-Pitaevskii hierarchy is involved. Its original version
in which the time integration domain was unknown, has been used without
proof in \cite{CDP19a} for the Boltzmann hierarchy. We hereby provide its
full proof, for the Boltzmann hierarchy, with the time integration domain
computed using the newest techniques \cite{CH8}. Most of the materials in
this section are a different version of \cite{CH8} as well.

There are 2 key observations. One is the fact that after some suitable
substitution, many summands inside $D^{(k+1)}(f^{(k+1)})$ actually have the
same integrands and hence they can be combined into the so-called upper
echelon classes if we follow some rules. The other one is that, if put in
tree representations, all possible cases inside an upper echelon class are
actually all the possibilities in which children must carry a higher index
than parents.

Recall the notation of $\mu $ in Lemma \ref{lem:extended board game}: $%
\left\{ \mu \right\} $ is a set of maps from $\{2,\ldots ,k+1\}$ to $%
\{1,\ldots ,k\}$ satisfying $\mu (2)=1$ and $\mu (l)<l$ for all $l,$ and 
\begin{equation*}
D_{\mu }^{(k+1)}(f^{(k+1)}(t_{k+1}))=S_{1,2}^{(1)}Q_{\mu
(2),2}S_{2,3}^{(2)}Q_{\mu (3),3}...S_{k,k+1}^{(k)}Q_{\mu (k+1),k+1}.
\end{equation*}%
Throughout this section, we only work with $k\geqslant 4$, that is coupling
to at least $f^{(5)}$, as it is the minimal length for the argument to have
enough room to work. (We actually want $k\rightarrow \infty $ anyway.)

\begin{example}
An example of $\mu$ when $k=5$ is 
\begin{equation*}
\begin{tabular}{c|ccccc}
$j$ & $2$ & $3$ & $4$ & $5$ & $6$ \\ \hline
$\mu $ & $1$ & $1$ & $3$ & $2$ & $1$%
\end{tabular}
.
\end{equation*}
\end{example}

If $\mu $ satisfies $\mu (j)\leq \mu (j+1)$ for $2\leq j\leq k$ in addition
to $\mu (j)<j$ for all $2\leq j\leq k+1$, then it is in \emph{upper-echelon
form} as they are called in \cite{KM}. (The word ``upper echelon" certainly
makes more sense when one uses the matrix / board game representation of $%
D_{\mu }^{(k+1)}(f^{(k+1)})$ in \cite{KM}.) Let $\mu $ be a collapsing map
as defined above and $\sigma $ a permutation of $\{2,\ldots ,k+1\}$. A \emph{%
KM acceptable move}, which we denote $\text{KM}(j,j+1)$, is allowed when $%
\mu (j)\neq \mu (j+1)$ and $\mu (j+1)<j$, and is the following action: $(\mu
^{\prime },\sigma ^{\prime })=\text{KM}(j,j+1)(\mu ,\sigma )$: 
\begin{align*}
\mu ^{\prime }& =(j,j+1)\circ \mu \circ (j,j+1) \\
\sigma ^{\prime }& =(j,j+1)\circ \sigma
\end{align*}

The first key observation is that if $(\mu ^{\prime },\sigma ^{\prime })=%
\text{KM}(j,j+1)(\mu ,\sigma )$ and $f^{(k+1)}$ is a symmetric density, then 
\begin{equation}
\int D_{\mu ^{\prime }}^{(k+1)}(f^{(k+1)})(t_{1},{\sigma ^{\prime }}^{-1}(%
\underline{t}_{k+1}))d\underline{t}_{k+1}=\int D_{\mu
}^{(k+1)}(f^{(k+1)})(t_{1},{\sigma }^{-1}(\underline{t}_{k+1}))d\underline{t}%
_{k+1}  \label{E:KM-den1}
\end{equation}%
where, 
\begin{equation*}
\text{for }\underline{t}_{k+1}=(t_{2},\ldots ,t_{k+1})\text{ we define }%
\sigma ^{-1}(\underline{t}_{k+1})=(t_{\sigma ^{-1}(2)},\ldots ,t_{\sigma
^{-1}(k+1)})
\end{equation*}%
A simple example to see (\ref{E:KM-den1}) is the following.

\begin{example}
The integrals 
\begin{eqnarray*}
I_{1}
&=&%
\int_{D}S_{1,2}^{(1)}Q_{1,2}S_{2,3}^{(2)}Q_{2,3}S_{3,4}^{(3)}Q_{1,4}S_{4,5}^{(4)}Q_{4,5}(f^{(5)})d%
\underline{t}_{4}\text{, } \\
I_{2}
&=&%
\int_{D}S_{1,2}^{(1)}Q_{1,2}S_{2,3}^{(2)}Q_{1,3}S_{3,4}^{(3)}Q_{2,4}S_{4,5}^{(4)}Q_{3,5}(f^{(5)})d%
\underline{t}_{4}
\end{eqnarray*}%
with $D=\{t_{1}\geqslant t_{2}\geqslant t_{3}\geqslant t_{4}\geqslant
t_{5}\} $, actually have the same integrand. For simplicity, pluging in $%
f^{(5)}=\rho ^{\otimes 5}$ (the general case is the same but longer), we have%
\begin{eqnarray*}
I_{1} &=&\int_{D}S_{1,2}Q\left( S_{2,4}Q\left( S_{4,5}\rho ,S_{4,5}Q(\rho
,\rho )\right) ,S_{2,3}Q\left( S_{3,5}\rho ,S_{3,5}\rho \right) \right) d%
\underline{t}_{4} \\
I_{2} &=&\int_{D}S_{1,2}Q\left( S_{2,3}Q(S_{3,5}\rho ,S_{3,5}Q(\rho ,\rho
)),S_{2,4}Q\left( S_{4,5}\rho ,S_{4,5}\rho \right) )\right) d\underline{t}%
_{4}
\end{eqnarray*}%
Doing a $t_{3}\leftrightarrow t_{4}$ swap in $I_{1}$, we have%
\begin{equation*}
I_{1}=\int_{D^{\prime }}S_{1,2}Q\left( S_{2,3}Q(S_{3,5}\rho ,S_{3,5}Q(\rho
,\rho )),S_{2,4}Q\left( S_{4,5}\rho ,S_{4,5}\rho \right) )\right) d%
\underline{t}_{4}
\end{equation*}%
where $D^{\prime }=\{t_{1}\geqslant t_{2}\geqslant t_{4}\geqslant
t_{3}\geqslant t_{5}\}$. That is, $I_{1}$ and $I_{2}$ can be combined.
\end{example}

For each $\mu $ and $\sigma $, we define the Duhamel integrals 
\begin{equation}
I(\mu ,\sigma ,f^{(k+1)})(t_{1})=\int_{t_{1}\geq t_{\sigma (2)}\geq \cdots
\geq t_{\sigma (k+1)}}D_{\mu }^{(k+1)}(f^{(k+1)})(t_{1},\underline{t}%
_{k+1})\,d\underline{t}_{k+1}  \label{E:KM-den2}
\end{equation}%
It follows from \eqref{E:KM-den1} that 
\begin{equation*}
I(\mu ^{\prime },\sigma ^{\prime (k+1)},f^{(k+1)})=I(\mu ,\sigma ,f^{(k+1)})
\end{equation*}%
We combine KM acceptable moves as follows: if $\rho $ is a permutation of $%
\{2,\ldots ,k+1\}$ such that it is possible to write $\rho $ as a
composition of transpositions 
\begin{equation*}
\rho =\tau _{1}\circ \cdots \circ \tau _{r}
\end{equation*}%
for which each operator $\text{KM}(\tau _{j})$ on the right side of the
following is an acceptable action 
\begin{equation*}
\text{KM}(\rho )\overset{\mathrm{def}}{=}\text{KM}(\tau _{1})\circ \cdots
\circ \text{KM}(\tau _{r})
\end{equation*}%
then $\text{KM}(\rho )$, defined by this composition, is acceptable as well.
In this case $(\mu ^{\prime },\sigma ^{\prime })=\text{KM}(\rho )(\mu
,\sigma )$ and 
\begin{align*}
& \mu ^{\prime }=\rho \circ \mu \circ \rho ^{-1} \\
& \sigma ^{\prime }=\rho \circ \sigma
\end{align*}%
\eqref{E:KM-den1} and \eqref{E:KM-den2} hold as well. If $\mu $ and $\mu
^{\prime }$ are such that there exists $\rho $ as above for which $(\mu
^{\prime },\sigma ^{\prime })=\text{KM}(\rho )(\mu ,\sigma )$ then we say
that $\mu ^{\prime }$ and $\mu $ are \emph{KM-relatable}. This is an
equivalence relation that partitions the set of collapsing maps into
equivalence classes.

In the following, we represent these actions via tree diagrams in which the
effect of the actions and the ``not obvious at all" time integration domain $%
T(\mu )$ emerge clearly. Given a $\mu $ which is also a summand inside $%
D^{(k+1)}(f^{(k+1)})$, we construct a binary tree via Algorithm \ref{alg:u
to tree}.

\begin{algorithm}
\label{alg:u to tree} \leavevmode

\begin{enumerate}
\item Set counter $j=2$

\item Given $j$, find the next pair of indices $a$ and $b$ so that $a>j$, $%
b>j$ and 
\begin{equation*}
\mu (a)=\mu (j)\text{ and }\mu (b)=j
\end{equation*}
and moreover $a$ and $b$ are the minimal indices for which the above
equalities hold. It is possible that there is no such $a$ and/or no such $b$.

\item At the node $j$, put $a$ as the left child and $b$ as the right child
(if there is no $a$, then the $j$ node will be missing a left child, and if
there is no $b$, then the $j$ node will be missing a right child.)

\item If $j=k+1,$ then stop, otherwise set $j=j+1$ and go to step 2.
\end{enumerate}
\end{algorithm}

\begin{example}
\label{example:Tree-1} \quad

\begin{minipage}{0.30\linewidth}
\begin{center}
\begin{tikzpicture} 
\node{$1$}
	child[missing]
	child{node{$2$}
		child{node{$3$}} 
		child{node{$5$}}
	};
\end{tikzpicture}
\end{center}
\end{minipage}
\begin{minipage}{0.65\linewidth}
Let us work with the following example 

\medskip

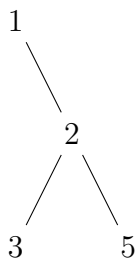
\begin{center}
\begin{tabular}{c|ccccc}
$j$ & $2$ & $3$ & $4$ & $5$ & $6$ \\ \hline
$\mu _{\text{out}}$ & $1$ & $1$ & $1$ & $2$ & $3$
\end{tabular}
\end{center}

\medskip

We start with $j=2$, and note that $\mu _{\text{out}}(2)=1$ so need to find
minimal $a>2$, $b>2$ such that $\mu (a)=1$ and $\mu (b)=2$. In this case, it
is $a=3$ and $b=5$, so we put those as left and right children of $2$,
respectively, in the tree (shown at left)
\end{minipage}

\bigskip

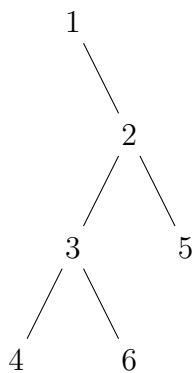
\begin{minipage}{0.30\linewidth}
\begin{center}
\begin{tikzpicture} 
\node{$1$}
	child[missing]
	child{node{$2$}
		child{node{$3$} 
			child{node{$4$}} 
			child{node{$6$}} 
		}
		child{node{$5$}}
	}; 
\end{tikzpicture}
\end{center}
\end{minipage}
\begin{minipage}{0.65\linewidth}
Now we move to $j=3$. Since $\mu _{\text{out}}(3)=1$, we find minimal $a$
and $b$ so that $a>3$, $b>3$ and $\mu (a)=1$ and $\mu (b)=3$. We find that $
a=4$ and $b=6$, so we put these as left and right children of $3$,
respectively, in the tree (shown at left).  Since all indices appear in the tree, it is complete.
\end{minipage}
\end{example}

\begin{definition}
A binary tree is called an \emph{admissible} tree if every child node's
label is strictly larger than its parent node's label.\footnote{%
This is certainly a natural requirement coming from the hierarchy.} For an
admissible tree, we call the graph of the tree without any labels in its
nodes, the \emph{skeleton} of the tree.
\end{definition}

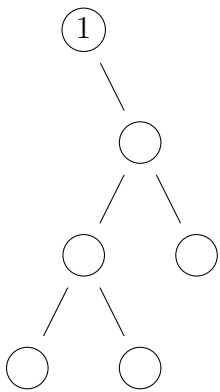
\begin{minipage}{0.30\linewidth}
\begin{center}
\begin{tikzpicture} 
\node{$\encircle{1}$}
	child[missing]
	child{node{$\encircle{\;}$}
		child{node{$\encircle{\;}$} 
			child{node{$\encircle{\;}$}} 
			child{node{$\encircle{\;}$}} 
		}
		child{node{$\encircle{\;}$}}
	}; 
\end{tikzpicture}
\end{center}
\end{minipage}%
\begin{minipage}{0.65\linewidth}
For example, the skeleton of the tree in Example \ref{example:Tree-1} is shown at left.

\bigskip

As the trees are coming from the hierarchy, Algorithm \ref{alg:u to tree}, produces only admissible trees. The procedure
is reversible -- given an admissible binary tree, we can uniquely
reconstruct the $\mu $ that generated it.
\end{minipage}

\begin{algorithm}
\leavevmode
\label{alg:tree to u}

\begin{enumerate}
\item For every right child, $\mu $ maps the child value to the parent value
(i.e. if $f$ is a right child of $d$, then $\mu (f)=d$). Start by filling
these into the $\mu $ table.

\item Fill in the table using that for every left child, $\mu $ maps the
child value to $\mu (\text{parent value})$.
\end{enumerate}
\end{algorithm}

\begin{example}
Suppose we are given the tree

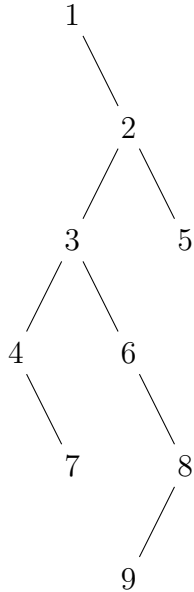
\begin{minipage}{0.375\linewidth}
\begin{center}
\begin{tikzpicture} 
\node{$1$}
child[missing]
child{node{$2$}
	child{node{$3$} 
		child{node{$4$} 
			child[missing]
			child{node{$7$}} 
		}
		child{node{$6$} 
			child[missing] 
			child{node{$8$} 
				child{node{$9$}}
				child[missing] 
			} 
		} 
	} 
	child{node{$5$}}
}; 
\end{tikzpicture}
\end{center}
\end{minipage}
\begin{minipage}{0.575\linewidth}

Using that for every right child, $\mu$ maps the child value to the parent
value, we fill in the following values in the $\mu $ table:

\bigskip

\begin{center}
\begin{tabular}{c|cccccccc}
$j$ & $2$ & $3$ & $4$ & $5$ & $6$ & $7$ & $8$ & $9$ \\ \hline
$\mu $ & $1$ &  &  & $2$ & $3$ & $4$ & $6$ & 
\end{tabular}
\end{center}

\bigskip

Now we employ the left child rule, and note that since $3$ is a left child
of $2$ and $\mu (2)=1$, we must have $\mu (3)=1$, and etc. to recover

\bigskip

\begin{center}
\begin{tabular}{c|cccccccc}
$j$ & $2$ & $3$ & $4$ & $5$ & $6$ & $7$ & $8$ & $9$ \\ \hline
$\mu $ & $1$ & $1$ & $1$ & $2$ & $3$ & $4$ & $6$ & $6$
\end{tabular}
\end{center}
\end{minipage}
\end{example}

One can show that, in the tree representation of $\mu $, a KM acceptable
move, is the operation which switches the labels of two nodes with
consecutive labels on an admissible tree provided that the outcome is still
an admissible tree by writing out the related trees. For example,
interchanging the labeling of 5 and 6 in the tree in Example \ref%
{example:Tree-1} is an acceptable move. That is, KM acceptable moves
preserve the tree structures but permute the labeling under the
admissibility requirement. Two collapsing maps $\mu $ and $\mu ^{\prime }$
are KM-relatable if and only the trees corresponding to $\mu $ and $\mu
^{\prime }$ have the same skeleton.

Given $k$, we would like to have the number of different binary tree
structures of $k$ nodes. This number is exactly one of the Catalan number as
defined and is controlled by $4^{k}$. Hence, we just provided a proof of
Lemma \ref{lem:extended board game}, dropping the computation of $T(\mu )$.
To this end, we need to define what is an upper echelon form. Though the
requirement $\mu (j)\leq \mu (j+1)$ for $2\leq j\leq k$ is good enough, we
give an algorithm which produces the upper echelon tree given the tree
structure, as the tree representation of an upper echelon form is in fact
labeled in sequential order. See, for example, the tree in Example \ref%
{example:Tree-1}.

\begin{algorithm}
\label{alg:TreeToUpperEchelon} \leavevmode

\begin{enumerate}
\item Given a tree structure with $k$ nodes, label the top node with $2$ and
set counter $j=2.$

\item If $j=k+1$, then stop, otherwise continue.

\item If the node labeled $j$ has a left child, then label that left child
node with $j+1$, set counter $j=j+1$ and go to step (2). If not, continue.

\item In the already labeled nodes which has an empty right child, search
for the node with the smallest label. If such a node can be found, label
that node's empty right child as $j+1$, set counter $j=j+1,$ and go to step
(2). If none of the labeled nodes has an empty right child, then stop.
\end{enumerate}
\end{algorithm}

\begin{definition}
We say $\mu $ is in upper echelon form if $\mu (j)\leq \mu (j+1)$ for $2\leq
j\leq k$ or its corresponding tree given by Algorithm \ref{alg:u to tree}
agrees with the tree with the same skeleton given by Algorithm \ref%
{alg:TreeToUpperEchelon}.
\end{definition}

We define a map $T$ which maps an upper echelon tree to a time integration
domain / a set of inequality relations by 
\begin{equation}
\begin{aligned} T(\alpha ) = \{\; t_{j}\geqslant t_{k} \; : \; & j,k\text{
are labels on nodes of }\alpha \\ &\text{such that the }k\text{ node is a
child of the }j\text{ node} \; \} \end{aligned}  \label{E:TD-def}
\end{equation}%
where $\alpha $ is an upper echelon tree. We then have the integration
domain as follows.

\begin{proposition}
\label{Prop:KMIntegrationLimits}Given a $\mu $ in upper echelon form, we
have 
\begin{equation*}
\sum_{\mu _{m}\sim \mu }\int_{t_{1}\geqslant t_{2}\geqslant t_{3}\geqslant
...\geqslant t_{k+1}}D_{\mu _{m}}^{(k+1)}(f^{(k+1)})(t_{1},\underline{t}%
_{k+1})d\underline{t}_{k+1}=\int_{T(\mu )}D_{\mu }^{(k+1)}(f^{(k+1)})(t_{1},%
\underline{t}_{k+1})d\underline{t}_{k+1}.
\end{equation*}%
Here, $\mu _{m}\sim \mu $ means that $\mu _{m}$ is equivalent to $\mu $
under acceptable moves / the trees representing $\mu $ and $\mu _{m}$ have
the same structure and $T(\mu )$ is the domain defined in \eqref{E:TD-def}.
\end{proposition}

\begin{proof}

We prove by an example. For the general case, one merely needs to rewrite $%
\Sigma _{1}$ and $\Sigma _{2}$, to be defined in this proof. The key is the
admissible condition or the simple requirement that the child must carry a
larger lable than the parent.

Recall the upper echelon tree in Example \ref{example:Tree-1}, and denote it
with $\alpha $. Here are all the admissible trees equivalent to $\alpha .$


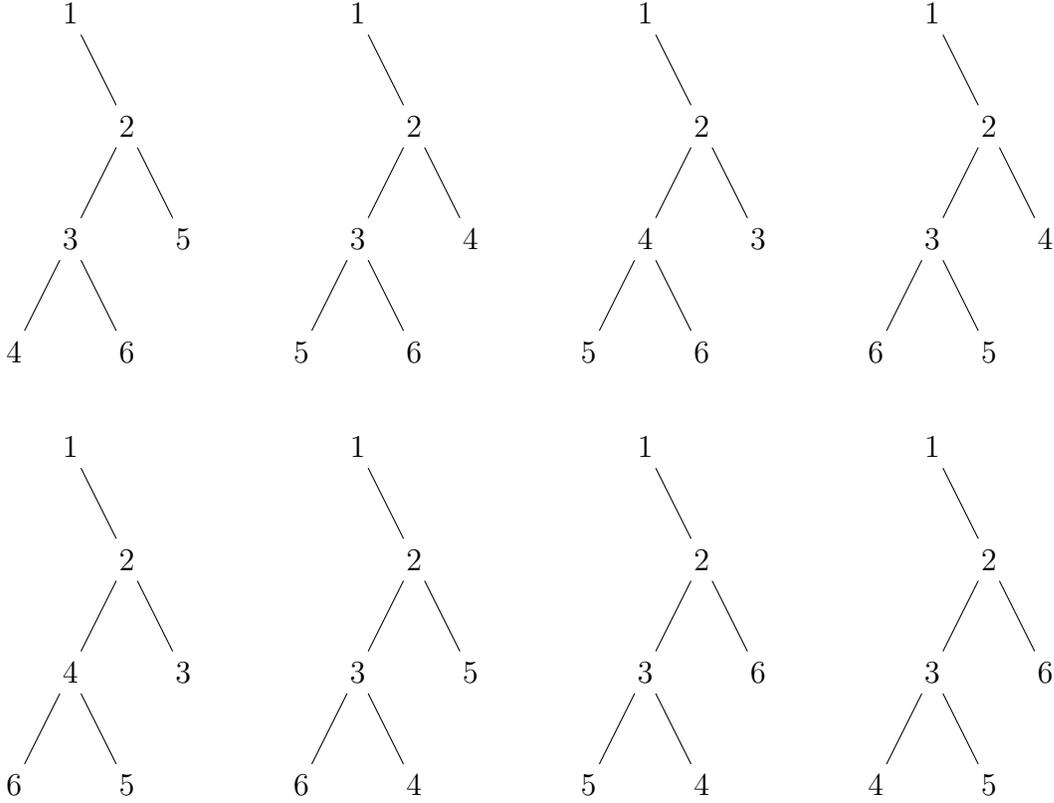
\begin{minipage}{1.45in}
\begin{tikzpicture}
\node{$1$} 
	child[missing] 
	child{node{$2$}
		child{node{$3$}
			child{node{$4$}}
			child{node{$6$}}
		}
		child{node{$5$}}
	};
\end{tikzpicture}
\end{minipage}
\begin{minipage}{1.45in}
\begin{tikzpicture}
\node{$1$} 
	child[missing] 
	child{node{$2$}
		child{node{$3$}
			child{node{$5$}}
			child{node{$6$}}
		}
		child{node{$4$}}
	};
\end{tikzpicture}
\end{minipage}
\begin{minipage}{1.45in}
\begin{tikzpicture}
\node{$1$} 
	child[missing] 
	child{node{$2$}
		child{node{$4$}
			child{node{$5$}}
			child{node{$6$}}
		}
		child{node{$3$}}
	};
\end{tikzpicture}
\end{minipage}
\begin{minipage}{1.45in}
\begin{tikzpicture}
\node{$1$} 
	child[missing] 
	child{node{$2$}
		child{node{$3$}
			child{node{$6$}}
			child{node{$5$}}
		}
		child{node{$4$}}
	};
\end{tikzpicture}
\end{minipage}

\bigskip

\begin{minipage}{1.45in}
\begin{tikzpicture}
\node{$1$} 
	child[missing] 
	child{node{$2$}
		child{node{$4$}
			child{node{$6$}}
			child{node{$5$}}
		}
		child{node{$3$}}
	};
\end{tikzpicture}
\end{minipage}
\begin{minipage}{1.45in}
\begin{tikzpicture}
\node{$1$} 
	child[missing] 
	child{node{$2$}
		child{node{$3$}
			child{node{$6$}}
			child{node{$4$}}
		}
		child{node{$5$}}
	};
\end{tikzpicture}
\end{minipage}
\begin{minipage}{1.45in}
\begin{tikzpicture}
\node{$1$} 
	child[missing] 
	child{node{$2$}
		child{node{$3$}
			child{node{$5$}}
			child{node{$4$}}
		}
		child{node{$6$}}
	};
\end{tikzpicture}
\end{minipage}
\begin{minipage}{1.45in}
\begin{tikzpicture}
\node{$1$} 
	child[missing] 
	child{node{$2$}
		child{node{$3$}
			child{node{$4$}}
			child{node{$5$}}
		}
		child{node{$6$}}
	};
\end{tikzpicture}
\end{minipage}

\bigskip


We first read by definition that 
\begin{equation*}
T(\alpha )=\{t_{1}\geqslant t_{2},t_{2}\geqslant t_{3},t_{3}\geqslant
t_{4},t_{3}\geqslant t_{6},t_{2}\geqslant t_{5}\}\text{.}
\end{equation*}%
Let $\sigma $ denote some composition of acceptable moves, we then notice
the equivalence of the two sets 
\begin{eqnarray*}
\Sigma _{1} &=&\left\{ \sigma :\sigma ^{-1}(1)<\sigma ^{-1}(2)<\sigma
^{-1}(3)<\sigma ^{-1}(4),\sigma ^{-1}(2)<\sigma ^{-1}(5),\sigma
^{-1}(3)<\sigma ^{-1}(6)\right\} , \\
\Sigma _{2} &=&\left\{ \sigma :\sigma \text{ takes input tree to }\alpha 
\text{ where the input tree is admissibile}\right\} ,
\end{eqnarray*}%
both generated by the requirement that the child must carry a larger label
than the parent. That is, both $\Sigma _{1}$ and $\Sigma _{2}$ classifies
the whole upper echelon class represented by $\alpha $.

Hence, 
\begin{equation*}
\bigcup_{\sigma \in \Sigma _{1}}\left\{ t_{1}\geqslant t_{\sigma
(2)}\geqslant t_{\sigma (3)}...\geqslant t_{\sigma \left( 6\right) }\right\}
=\{t_{1}\geqslant t_{2}\geqslant t_{3}\geqslant t_{4},t_{2}\geqslant
t_{5},t_{3}\geqslant t_{6}\}=T(\alpha )
\end{equation*}%
and we are done.
\end{proof}

\subsection{Bilinear estimates\label{Subsec:Multilinear Estimates}}

If 
\begin{equation*}
\tilde{f}^{(2)}(t,x_{1},x_{2},\xi _{1},\xi _{2})=\tilde{f}(t,x_{1},\xi _{1})%
\tilde{g}(t,x_{2},\xi _{2})\,,
\end{equation*}%
we can deduce several consequences of Proposition \ref%
{P:general-fixed-time-2}.

\begin{lemma}[Well-posedness and uniqueness estimate I]
\begin{eqnarray*}
&&\Vert \langle \nabla _{x}\rangle ^{1+}\langle \nabla _{\xi }\rangle ^{0+}%
\tilde{Q}^{\pm }(e^{it\nabla _{x}\cdot \nabla _{\xi }}\tilde{g},e^{it\nabla
_{x}\cdot \nabla _{\xi }}\tilde{h})\Vert _{L_{(-T,T)}^{2+}L_{x\xi }^{2}} \\
&\lesssim &\Vert \langle \nabla _{x}\rangle ^{1+}\langle \nabla _{\xi
}\rangle ^{0+}\tilde{g}\Vert _{L_{x\xi }^{2}}\Vert \langle \nabla
_{x}\rangle ^{1+}\langle \nabla _{\xi }\rangle ^{0+}\tilde{h}\Vert _{L_{x\xi
}^{2}}
\end{eqnarray*}
\end{lemma}

\begin{proof}
We prove this estimate inside Lemma \ref{L:10estimates} in the middle of the
well-posedness argument.
\end{proof}

\begin{lemma}[Uniqueness estimate II]
\begin{align*}
\hspace{0.3in}& \hspace{-0.3in}\Vert \langle \nabla _{\xi }\rangle ^{0+}%
\check{Q}^{\pm }(\check{g},\check{h})(t,\eta ,\xi )\Vert
_{L_{(-T,T)}^{2+}L_{\xi \eta }^{3-}} \\
& \lesssim \left\{ \begin{aligned} & \| \langle \nabla_\xi \rangle^{0+}
\check g(t,\eta, \xi) \|_{L_{(-T,T)}^\infty L_{\eta\xi}^{3-}} \| \langle
\nabla_x \rangle^{1+} \langle \nabla_\xi \rangle^{0+}\tilde h(0,x, \xi)
\|_{L_{x\xi}^2} \text{ (if $h$ is a linear sol)}\\ &\| \langle \nabla_x
\rangle^{1+} \langle \nabla_\xi \rangle^{0+} \tilde g(0,x, \xi)
\|_{L_{x\xi}^2} \| \langle \nabla_\xi \rangle^{0+} \check h(t,\eta, \xi)
\|_{L_{(-T,T)}^\infty L_{\eta\xi}^{3-}} \text{ (if $g$ is a linear sol)}
\end{aligned}\right.
\end{align*}
\end{lemma}

\begin{proof}
Recall 
\begin{align*}
\tilde{Q}_{\alpha ,\sigma }(\tilde{g},\tilde{h})(t,x,\xi
)=\int_{s=0}^{\infty }\int_{\zeta }e^{i(\sigma -\alpha )\xi \cdot \zeta
/2}e^{-2is\sigma |\zeta |^{2}/2}\hat{\phi}(-\zeta )\hat{\phi}(\zeta )& \\
\tilde{g}(t,x,\xi -s\zeta )\tilde{h}(t,x,s\zeta )\,ds\,d\zeta &
\end{align*}%
Taking the Fourier transform $x\mapsto \eta $ gives 
\begin{align*}
\check{Q}_{\alpha ,\sigma }(\check{g},\check{h})(t,\eta ,\xi
)=\int_{s=0}^{\infty }\int_{\zeta }\int_{u}e^{i(\sigma -\alpha )\xi \cdot
\zeta /2}e^{-2is\sigma |\zeta |^{2}/2}\hat{\phi}(-\zeta )\hat{\phi}(\zeta )&
\\
\check{g}(t,\eta -u,\xi -s\zeta )\check{h}(t,u,s\zeta )\,ds\,d\zeta \,du&
\end{align*}%
We can estimate in the norm $L_{\eta }^{3-}$ first, bringing it to the
inside by Minkowski's integral inequality, and applying Young's inequality
on the inner convolution (putting $L_{\eta }^{3-}$ on either $\check{g}$ or $%
\check{h}$, as desired). Then, continuing as in the proof of Proposition \ref%
{P:general-fixed-time-2}), we obtain 
\begin{equation*}
\Vert \langle \nabla _{\xi }\rangle ^{0+}\check{Q}^{\pm }(\check{g},\check{h}%
)(t,\eta ,\xi )\Vert _{L_{\xi \eta }^{3-}}\lesssim \left\{ \begin{aligned} &
\| \langle \nabla_\xi \rangle^{0+} \check g(t,\eta, \xi)
\|_{L_\xi^{3-}L_\eta^{3-}} \| \check h(t,\eta, \xi) \|_{(L_\xi^{3+}\cap
L_\xi^{3-})L_\eta^1} \\ &\| \langle \nabla_\xi \rangle^{0+} \check g(t,\eta,
\xi) \|_{L_\xi^{3-}L_\eta^1} \| \check h(t,\eta, \xi) \|_{(L_\xi^{3+}\cap
L_\xi^{3-})L_\eta^{3-}} \end{aligned}\right.
\end{equation*}%
where it is meant that either the top or the bottom expression on the right
side can be used. Applying Sobolev in $\xi $ on the $\check{h}$ terms to
convert $L_{\xi }^{3+}$ to $L_{\xi }^{3-}$ at the expense of adding $\langle
\nabla _{\xi }\rangle ^{0+}$ gives 
\begin{equation*}
\Vert \langle \nabla _{\xi }\rangle ^{0+}\check{Q}^{\pm }(\check{g},\check{h}%
)(t,\eta ,\xi )\Vert _{L_{\xi \eta }^{3-}}\lesssim \left\{ \begin{aligned} &
\| \langle \nabla_\xi \rangle^{0+} \check g(t,\eta, \xi)
\|_{L_{\eta\xi}^{3-}} \| \langle \nabla_\xi \rangle^{0+}\check h(t,\eta,
\xi) \|_{L_\xi^{3-}L_\eta^1} \\ &\| \langle \nabla_\xi \rangle^{0+} \check
g(t,\eta, \xi) \|_{L_\xi^{3-}L_\eta^1} \| \langle \nabla_\xi \rangle^{0+}
\check h(t,\eta, \xi) \|_{L_{\eta\xi}^{3-}} \end{aligned}\right.
\end{equation*}

The Strichartz estimate for the ``kinetic" transport equation \cite[%
Definition 2.1 \& Theorem 2.4]{Ov11} with $a=2$ applies with $%
L_{t}^{q}L_{\xi }^{r}L_{\eta }^{p}$ with 
\begin{equation*}
\frac{1}{2}=\frac{1}{2}\left( \frac{1}{r}+\frac{1}{p}\right) \,,\qquad \frac{%
3}{2}<p\leq 2\,,\qquad 2\leq r<3
\end{equation*}%
where $q$ is defined via 
\begin{equation*}
\frac{1}{q}=\frac{3}{2}\left( \frac{1}{p}-\frac{1}{r}\right)
\end{equation*}%
for such a pair $(p,r)$. In the endpoint case (which is not valid, see \cite%
{BBGL14}), $p=\frac{3}{2}$, $r=3$ and $q=2$. For $(p,r)$ meeting the
requirements above, $q>2$. We will work with a triple $L_{t}^{2+}L_{\xi
}^{3-}L_{\eta }^{\frac{3}{2}+}$. Now writing $1=\langle \eta \rangle
^{-1-}\langle \eta \rangle ^{1+}$ and applying H\"{o}lder in $\eta $, 
\begin{equation*}
\Vert \langle \nabla _{\xi }\rangle ^{0+}\check{h}(t,\eta ,\xi )\Vert
_{L_{\xi }^{3-}L_{\eta }^{1}}\lesssim \Vert \langle \eta \rangle ^{-1-}\Vert
_{L_{\eta }^{3-}}\Vert \langle \eta \rangle ^{1+}\langle \nabla _{\xi
}\rangle ^{0+}\check{h}(t,\eta ,\xi )\Vert _{L_{\xi }^{3-}L_{\eta }^{3/2+}}
\end{equation*}%
The $L_{\xi }^{3-}$ forces a specific $L_{\eta }^{\frac{3}{2}+}$ according
to the Strichartz theory reviewed above. Since we are forced to work with a
particular $\frac{3}{2}+$ in the norm $L_{\eta }^{\frac{3}{2}+}$, we choose
the $1+$ sufficiently above $1$ in the exponent $\langle \eta \rangle ^{1+}$
so that the reciprocal $\langle \eta \rangle ^{-1-}$ is sufficiently below $%
-1$ so that $\Vert \langle \eta \rangle ^{-1-}\Vert _{L_{\eta }^{3-}}<\infty 
$. If $\check{h}(t,\eta ,\xi )$ is a linear solution, then we can apply the
Strichartz estimates to obtain

\begin{equation*}
\Vert \langle \nabla _{\xi }\rangle ^{0+}\check{h}(t,\eta ,\xi )\Vert
_{L_{(-T,T)}^{2+}L_{\xi }^{3-}L_{\eta }^{1}}\lesssim \Vert \langle \eta
\rangle ^{1+}\langle \nabla _{\xi }\rangle ^{0+}\check{h}(0,\eta ,\xi )\Vert
_{L_{\eta \xi }^{2}}
\end{equation*}%
Thus the claimed estimate follows.
\end{proof}

\begin{lemma}[Uniqueness estimate III - final estimate]
\label{Lem:Uniqueness III}%
\begin{equation*}
\Vert \langle \nabla _{\xi }\rangle ^{0+}\check{Q}^{\pm }(\check{g},\check{h}%
)\Vert _{L_{(-T,T)}^{\infty }L_{\eta \xi }^{3-}}\lesssim \Vert \langle
\nabla _{x}\rangle ^{\frac{1}{2}+}\langle \nabla _{\xi }\rangle ^{\frac{1}{2}%
+}\tilde{g}\Vert _{L_{(-T,T)}^{\infty }L_{x\xi }^{2}}\Vert \langle \nabla
_{x}\rangle ^{\frac{1}{2}+}\langle \nabla _{\xi }\rangle ^{\frac{1}{2}+}%
\tilde{h}\Vert _{L_{(-T,T)}^{\infty }L_{x\xi }^{2}}
\end{equation*}%
We note that the estimate is done at fixed time; the $L_{(-T,T)}^{\infty }$
norm is included since that is the form in which the estimate is applied.
\end{lemma}

\begin{proof}
We start by applying the estimate $\Vert \hat{F}\Vert _{L^{p}}\lesssim \Vert
F\Vert _{L^{p^{\prime }}}$ where $\frac{1}{p}+\frac{1}{p^{\prime }}=1$ and $%
p\geq 2$. This estimate is applied in $\eta $, so the left side is in the
\textquotedblleft check space\textquotedblright\ and the right side is in
\textquotedblleft tilde space\textquotedblright . 
\begin{equation*}
\Vert \langle \nabla _{\xi }\rangle ^{0+}\check{Q}^{\pm }(\check{g},\check{h}%
)\Vert _{L_{(-T,T)}^{\infty }L_{\eta \xi }^{3-}}\lesssim \Vert \langle
\nabla _{\xi }\rangle ^{0+}\tilde{Q}^{\pm }(\tilde{g},\tilde{h})\Vert
_{L_{(-T,T)}^{\infty }L_{\xi }^{3-}L_{x}^{\frac{3}{2}+}}
\end{equation*}%
Recall 
\begin{align*}
\tilde{Q}_{\alpha ,\sigma }(\tilde{g},\tilde{h})(t,x,\xi
)=\int_{s=0}^{\infty }\int_{\zeta }e^{i(\sigma -\alpha )\xi \cdot \zeta
/2}e^{-2is\sigma |\zeta |^{2}/2}\hat{\phi}(-\zeta )\hat{\phi}(\zeta )& \\
\tilde{g}(t,x,\xi -s\zeta )\tilde{h}(t,x,s\zeta )\,ds\,d\zeta &
\end{align*}%
As in the proof of Proposition \ref{P:general-fixed-time-2}, we can
effectively move the $\langle \nabla _{\xi }\rangle ^{0+}$ operator to act
directly on $\tilde{g}$, although for the gain term this also generates a
power of $\zeta $ (which is easily absorbed by the $\hat{\phi}$ terms). We
indicate this with the $\approx $ symbol, since it must be properly
justified with Littlewood-Paley theory: 
\begin{align*}
\langle \nabla _{\xi }\rangle ^{0+}\tilde{Q}_{\alpha ,\sigma }(\tilde{g},%
\tilde{h})(t,x,\xi )\approx \int_{s=0}^{\infty }\int_{\zeta }e^{i(\sigma
-\alpha )\xi \cdot \zeta /2}e^{-2is\sigma |\zeta |^{2}/2}\hat{\phi}(-\zeta )%
\hat{\phi}(\zeta )& \\
\langle \nabla _{\xi }\rangle ^{0+}\tilde{g}(t,x,\xi -s\zeta )\tilde{h}%
(t,x,s\zeta )\,ds\,d\zeta &
\end{align*}%
Bring the $L_{x}^{3/2+}$ norm inside by the Minkowski integral inequality,
and H\"{o}lder between the $\tilde{g}$ and $\tilde{h}$ terms: 
\begin{align*}
\Vert \langle \nabla _{\xi }\rangle ^{0+}\tilde{Q}_{\alpha ,\sigma }(\tilde{g%
},\tilde{h})(t,x,\xi )\Vert _{L_{x}^{3/2+}}\lesssim \int_{s=0}^{\infty
}\int_{\zeta }|\hat{\phi}(-\zeta )||\hat{\phi}(\zeta )|& \\
\Vert \langle \nabla _{\xi }\rangle ^{0+}\tilde{g}(t,x,\xi -s\zeta )\Vert
_{L_{x}^{3+}}\Vert \tilde{h}(t,x,s\zeta )\Vert _{L_{x}^{3+}}& \,ds\,d\zeta
\end{align*}%
Now apply the $L_{\xi }^{3-}$ norm and bring it inside the right side by the
Minkowski integral inequality: 
\begin{align}
\Vert \langle \nabla _{\xi }\rangle ^{0+}\tilde{Q}_{\alpha ,\sigma }(\tilde{g%
},\tilde{h})(t,x,\xi )\Vert _{L_{\xi }^{3-}L_{x}^{3/2+}}\lesssim
\int_{s=0}^{\infty }\int_{\zeta }|\hat{\phi}(-\zeta )||\hat{\phi}(\zeta )|&
\label{eqn:ineq in uniqueness III pf} \\
\Vert \langle \nabla _{\xi }\rangle ^{0+}\tilde{g}(t,x,\xi )\Vert _{L_{\xi
}^{3-}L_{x}^{3+}}\Vert \tilde{h}(t,x,s\zeta )\Vert _{L_{x}^{3+}}&
\,ds\,d\zeta  \notag
\end{align}

Split the $s$ integration in (\ref{eqn:ineq in uniqueness III pf}) into $%
0<s<1$ and $1<s<+\infty $. For $0<s<1$, apply H\"{o}lder in $\zeta $ as
follows%
\begin{eqnarray*}
&&\int_{s=0}^{1}\int_{\zeta }|\hat{\phi}(-\zeta )||\hat{\phi}(\zeta )|\Vert
\langle \nabla _{\xi }\rangle ^{0+}\tilde{g}(t,x,\xi )\Vert _{L_{\xi
}^{3-}L_{x}^{3+}}\Vert \tilde{h}(t,x,s\zeta )\Vert _{L_{x}^{3+}}\,ds\,d\zeta
\\
&\lesssim &\Vert \langle \nabla _{\xi }\rangle ^{0+}\tilde{g}(t,x,\xi )\Vert
_{L_{\xi }^{3-}L_{x}^{3+}}\int_{s=0}^{1}\Vert \hat{\phi}(-\zeta )\hat{\phi}%
(\zeta )\Vert _{L_{\zeta }^{3/2-}}\Vert \tilde{h}(t,x,s\zeta )\Vert
_{L_{\zeta }^{3+}L_{x}^{3+}}ds
\end{eqnarray*}%
Scaling out the $s$ inside the $L_{\zeta }^{3+}$ norm gives $s^{-1+}$:%
\begin{eqnarray*}
&\lesssim &\Vert \langle \nabla _{\xi }\rangle ^{0+}\tilde{g}(t,x,\xi )\Vert
_{L_{\xi }^{3-}L_{x}^{3+}}\Vert \hat{\phi}\Vert
_{L^{3-}}^{2}\int_{s=0}^{1}s^{-1+}\Vert \tilde{h}(t,x,\zeta )\Vert
_{L_{\zeta }^{3+}L_{x}^{3+}}ds \\
&\lesssim &\Vert \hat{\phi}\Vert _{L^{3-}}^{2}\Vert \langle \nabla _{\xi
}\rangle ^{0+}\tilde{g}(t,x,\xi )\Vert _{L_{\xi }^{3-}L_{x}^{3+}}\Vert 
\tilde{h}(t,x,\zeta )\Vert _{L_{\zeta }^{3+}L_{x}^{3+}}
\end{eqnarray*}
For $s>1$ in (\ref{eqn:ineq in uniqueness III pf}), apply H\"{o}lder in $%
\zeta $ as follows%
\begin{eqnarray*}
&&\int_{s=1}^{+\infty }\int_{\zeta }|\hat{\phi}(-\zeta )||\hat{\phi}(\zeta
)|\Vert \langle \nabla _{\xi }\rangle ^{0+}\tilde{g}(t,x,\xi )\Vert _{L_{\xi
}^{3-}L_{x}^{3+}}\Vert \tilde{h}(t,x,s\zeta )\Vert _{L_{x}^{3+}}\,ds\,d\zeta
\\
&\lesssim &\Vert \langle \nabla _{\xi }\rangle ^{0+}\tilde{g}(t,x,\xi )\Vert
_{L_{\xi }^{3-}L_{x}^{3+}}\int_{s=1}^{+\infty }\Vert \hat{\phi}(-\zeta )\hat{%
\phi}(\zeta )\Vert _{L_{\zeta }^{3/2+}}\Vert \tilde{h}(t,x,s\zeta )\Vert
_{L_{\zeta }^{3-}L_{x}^{3+}}
\end{eqnarray*}%
Scaling out the $s$ inside the $L_{\zeta }^{3-}$ norm gives $s^{-1-}$:%
\begin{eqnarray*}
&\lesssim &\Vert \langle \nabla _{\xi }\rangle ^{0+}\tilde{g}(t,x,\xi )\Vert
_{L_{\xi }^{3-}L_{x}^{3+}}\Vert \hat{\phi}\Vert
_{L^{3+}}^{2}\int_{s=1}^{+\infty }s^{-1-}\Vert \tilde{h}(t,x,\zeta )\Vert
_{L_{\zeta }^{3-}L_{x}^{3+}}ds \\
&\lesssim &\Vert \hat{\phi}\Vert _{L^{3+}}^{2}\Vert \langle \nabla _{\xi
}\rangle ^{0+}\tilde{g}(t,x,\xi )\Vert _{L_{\xi }^{3-}L_{x}^{3+}}\Vert 
\tilde{h}(t,x,\zeta )\Vert _{L_{\zeta }^{3-}L_{x}^{3+}}
\end{eqnarray*}

Putting the $0<s<1$ and $1<s<+\infty $ cases together, we obtain 
\begin{align*}
\hspace{0.3in}& \hspace{-0.3in}\Vert \langle \nabla _{\xi }\rangle ^{0+}%
\tilde{Q}_{\alpha ,\sigma }(\tilde{g},\tilde{h})(t,x,\xi )\Vert _{L_{\xi
}^{3-}L_{x}^{3/2+}} \\
& \lesssim \Vert \langle \nabla _{\xi }\rangle ^{0+}\tilde{g}(t,x,\xi )\Vert
_{L_{\xi }^{3-}L_{x}^{3+}}\Vert \tilde{h}(t,x,\xi )\Vert _{(L_{\xi
}^{3+}\cap L_{\xi }^{3-})L_{x}^{3+}} \\
& \lesssim \Vert \langle \nabla _{x}\rangle ^{\frac{1}{2}+}\langle \nabla
_{\xi }\rangle ^{\frac{1}{2}+}\tilde{g}(t,x,\xi )\Vert _{L_{x\xi }^{2}}\Vert
\langle \nabla _{x}\rangle ^{\frac{1}{2}+}\langle \nabla _{\xi }\rangle ^{%
\frac{1}{2}+}\tilde{h}(t,x,\xi )\Vert _{L_{x\xi }^{2}}
\end{align*}%
as claimed.
\end{proof}

\subsection{Iteration of bilinear estimates\label{sec:Iteration Estimates}}

We need 3 estimates from \S \ref{Subsec:Multilinear Estimates}, in which%
\begin{eqnarray}
&&\Vert \langle \nabla _{x}\rangle ^{1+}\langle \nabla _{\xi }\rangle ^{0+}%
\tilde{Q}^{\pm }(e^{it\nabla _{x}\cdot \nabla _{\xi }}\tilde{g},e^{it\nabla
_{x}\cdot \nabla _{\xi }}\tilde{h})\Vert _{L_{(-T,T)}^{2+}L_{x\xi }^{2}}
\label{eqn:UniquenessStrichartz-1} \\
&\lesssim &\Vert \langle \nabla _{x}\rangle ^{1+}\langle \nabla _{\xi
}\rangle ^{0+}\tilde{g}\Vert _{L_{x\xi }^{2}}\Vert \langle \nabla
_{x}\rangle ^{1+}\langle \nabla _{\xi }\rangle ^{0+}\tilde{h}\Vert _{L_{x\xi
}^{2}}  \notag
\end{eqnarray}%
\begin{align}
\hspace{0.3in}& \hspace{-0.3in}\Vert \langle \nabla _{\xi }\rangle ^{0+}%
\check{Q}^{\pm }(\check{g},\check{h})(t,\eta ,\xi )\Vert
_{L_{(-T,T)}^{2+}L_{\xi \eta }^{3-}}  \label{eqn:UniquenessStrichartz-2} \\
& \lesssim \left\{ \begin{aligned} & \| \langle \nabla_\xi \rangle^{0+}
\check g(t,\eta, \xi) \|_{L_{(-T,T)}^\infty L_{\eta\xi}^{3-}} \| \langle
\nabla_x \rangle^{1+} \langle \nabla_\xi \rangle^{0+}\tilde h(0,x, \xi)
\|_{L_{x\xi}^2} \text{ (if $h$ is a linear sol)}\\ &\| \langle \nabla_x
\rangle^{1+} \langle \nabla_\xi \rangle^{0+} \tilde g(0,x, \xi)
\|_{L_{x\xi}^2} \| \langle \nabla_\xi \rangle^{0+} \check h(t,\eta, \xi)
\|_{L_{(-T,T)}^\infty L_{\eta\xi}^{3-}} \text{ (if $g$ is a linear sol)}
\end{aligned}\right.  \notag
\end{align}%
are of Strichartz type (integrating in time is necessary for them to hold.)
and will be used iteratively, and 
\begin{eqnarray}
&&\Vert \langle \nabla _{\xi }\rangle ^{0+}\check{Q}^{\pm }(\check{g},\check{%
h})\Vert _{L_{\eta \xi }^{3-}}  \label{eqn:UniquenessSobolev} \\
&\lesssim &\Vert \langle \nabla _{x}\rangle ^{\frac{1}{2}+}\langle \nabla
_{\xi }\rangle ^{\frac{1}{2}+}\tilde{g}\Vert _{L_{x\xi }^{2}}\Vert \langle
\nabla _{x}\rangle ^{\frac{1}{2}+}\langle \nabla _{\xi }\rangle ^{\frac{1}{2}%
+}\tilde{h}\Vert _{L_{x\xi }^{2}}  \notag
\end{eqnarray}%
which is a fixed time estimate and will be used only once. We illustrate by
the following example on how to use them to estimate $I_{\mu }^{(k+1)}\left(
f^{(k+1)}\right) (t_{1})$. As the role of the collision operator here is to
couple to the next level, we will call the collision operator $Q_{\mu
(j),j}^{\pm }$ the $(j-1)$th coupling to be clear.

\begin{example}
Consider the summand 
\begin{equation*}
I=\int
S_{1,2}^{(1)}Q_{1,2}^{+}S_{2,3}^{(2)}Q_{1,3}^{-}S_{3,4}^{(3)}Q_{3,4}^{-}f^{(4)}(t_{4})d%
\text{\b{t}}_{4}
\end{equation*}%
in $I_{\mu }^{(4)}\left( f^{(4)}\right) (t_{1})$. Plugging in (\ref%
{eqn:H-WRep}), it reads%
\begin{equation*}
I=\int_{\mathcal{P}(\Omega )}d\mu _{t_{4}}(\rho )\int
S_{1,2}^{(1)}Q_{1,2}^{+}S_{2,3}^{(2)}Q_{1,3}^{-}S_{3,4}^{(3)}Q_{3,4}^{-}%
\left( \rho ^{\otimes 4}\right) dt_{2}dt_{3}dt_{4}
\end{equation*}%
where interchanging integration order is allowed as all measures are finite.
Expanding it out, we have%
\begin{equation*}
=\int_{\mathcal{P}(\Omega )}d\mu _{t_{4}}(\rho )\int
S_{1,2}Q^{+}(S_{2,3}Q^{-}(S_{3,4}\rho ,S_{3,4}Q^{-}(\rho ,\rho
)),S_{2,4}\rho )dt_{2}dt_{3}dt_{4}
\end{equation*}%
Notice that, away from the most inner (the 3rd) coupling, every coupling
takes the form $Q^{\pm }(Sf,Sg)$. For the estimates, put $I$ in the $(\eta
,\xi )$-side and apply the $L_{\eta _{1}\xi _{1}}^{3-}$ norm to obtain 
\begin{eqnarray*}
&&\left\Vert \langle \nabla _{\xi _{1}}\rangle ^{0+}\check{I}\right\Vert
_{L_{\eta _{1}\xi _{1}}^{3-}} \\
&\leqslant &\int_{\mathcal{P}(\Omega )}d\left\vert \mu _{t_{4}}\right\vert
(\rho )\int_{[0,T]^{3}}\left\Vert \left( \langle \nabla _{\xi _{1}}\rangle
^{0+}Q^{+}(S_{2,3}Q^{-}(S_{3,4}\rho ,S_{3,4}Q^{-}(\rho ,\rho )),S_{2,4}\rho
)\right) ^{\vee }\right\Vert _{L_{\eta _{1}\xi _{1}}^{3-}}dt_{2}dt_{3}dt_{4}
\end{eqnarray*}%
For the first coupling, Cauchy-Schwarz in $t_{2}$, and apply (\ref%
{eqn:UniquenessStrichartz-2}) to the first coupling, with the bilinear
variable which contains the 3rd coupling put in $L_{\eta _{1}\xi _{1}}^{3-}$%
, that is,%
\begin{eqnarray*}
\left\Vert \langle \nabla _{\xi _{1}}\rangle ^{0+}\check{I}\right\Vert
_{L_{\eta _{1}\xi _{1}}^{3-}} &\leqslant &CT^{\frac{1}{2}-}\int_{\mathcal{P}%
(\Omega )}d\left\vert \mu _{t_{4}}\right\vert (\rho )\left\Vert \langle
\nabla _{x}\rangle ^{1+}\langle \nabla _{\xi _{1}}\rangle ^{0+}\tilde{\rho}%
\right\Vert _{L_{x\xi }^{2}} \\
&&\times \int_{\lbrack 0,T]^{2}}\left\Vert \left( \langle \nabla _{\xi
_{1}}\rangle ^{0+}Q^{-}(S_{3,4}\rho ,S_{3,4}Q^{-}(\rho ,\rho ))\right)
^{\vee }\right\Vert _{L_{\eta \xi }^{3-}}dt_{3}dt_{4}
\end{eqnarray*}%
Doing the same thing for the 2nd coupling,%
\begin{eqnarray*}
\left\Vert \langle \nabla _{\xi _{1}}\rangle ^{0+}\check{I}\right\Vert
_{L_{\eta _{1}\xi _{1}}^{3-}} &\leqslant &\left( CT^{\frac{1}{2}-}\right)
^{2}\int_{\mathcal{P}(\Omega )}d\left\vert \mu _{t_{4}}\right\vert (\rho
)\left\Vert \langle \nabla _{x}\rangle ^{1+}\langle \nabla _{\xi
_{1}}\rangle ^{0+}\tilde{\rho}\right\Vert _{L_{x\xi }^{2}}^{2} \\
&&\times \int_{\lbrack 0,T]}\left\Vert (\langle \nabla _{\xi _{1}}\rangle
^{0+}Q^{-}(\rho ,\rho ))^{\vee }\right\Vert _{L_{\eta \xi }^{3-}}dt_{4}
\end{eqnarray*}%
Apply (\ref{eqn:UniquenessSobolev}) to the 3rd coupling, we get%
\begin{eqnarray*}
&&\left\Vert \langle \nabla _{\xi _{1}}\rangle ^{0+}\check{I}\right\Vert
_{L_{\eta _{1}\xi _{1}}^{3-}} \\
&\leqslant &\left( CT^{\frac{1}{2}-}\right) ^{2}\int_{\mathcal{P}(\Omega
)}d\left\vert \mu _{t_{4}}\right\vert (\rho )\left\Vert \langle \nabla
_{x}\rangle ^{0+}\langle \nabla _{\xi _{1}}\rangle ^{0+}\tilde{\rho}%
\right\Vert _{L_{x\xi }^{2}}^{2}\Vert \langle \nabla _{x}\rangle ^{\frac{1}{2%
}+}\langle \nabla _{\xi }\rangle ^{\frac{1}{2}+}\tilde{\rho}\Vert _{L_{x\xi
}^{2}}^{2}\int_{[0,T]}dt_{4} \\
&\leqslant &\left( CT^{\frac{1}{2}-}\right) ^{2}T\int_{\mathcal{P}(\Omega
)}\left\Vert \langle \nabla _{x}\rangle ^{1+}\langle \nabla _{\xi
_{1}}\rangle ^{0+}\tilde{\rho}\right\Vert _{L_{x\xi }^{2}}^{2}\Vert \langle
\nabla _{x}\rangle ^{\frac{1}{2}+}\langle \nabla _{\xi }\rangle ^{\frac{1}{2}%
+}\tilde{\rho}\Vert _{L_{x\xi }^{2}}^{2}d\left\vert \mu _{t_{4}}\right\vert
(\rho )
\end{eqnarray*}%
as needed.
\end{example}

\subsubsection{Estimate for the general cases}

We handle the general cases by the following algorithm.

\begin{itemize}
\item[Step 1] Put $I_{\mu }^{(k+1)}\left( f^{(k+1)}\right) (t_{1})$ in the $%
L_{\eta _{1}\xi _{1}}^{3-}$ norm on the $(\eta ,\xi )$-side with $\langle
\nabla _{\xi _{1}}\rangle ^{0+}$ applied.

\item[Step 2] Pay a price of $2^{k}$ to expand all the $Q$ inside $I_{\mu
}^{(k+1)}\left( f^{(k+1)}\right) (t_{1})$ into $Q^{\pm }$ so that there is
at most one $Q^{\pm }(\rho ,\rho )$ at the $k$-th coupling in each summand
denoted by $I_{\mu ,sgn}^{(k+1)}$, where $sgn$ means signed. That is%
\begin{equation*}
\left\Vert \langle \nabla _{\xi _{1}}\rangle ^{0+}\left( I_{\mu
}^{(k+1)}\left( f^{(k+1)}\right) (t_{1})\right) ^{\vee }\right\Vert
_{L_{\eta _{1}\xi _{1}}^{3-}}\leqslant 4^{k}\left\Vert \langle \nabla _{\xi
_{1}}\rangle ^{0+}\left( I_{\mu ,sgn}^{(k+1)}\left( f^{(k+1)}\right)
(t_{1})\right) ^{\vee }\right\Vert _{L_{\eta _{1}\xi _{1}}^{3-}}
\end{equation*}

\item[Step 3] Set counter $j=1$, use Minkowski's inequality to put the $%
L_{\eta _{1}\xi _{1}}^{3-}$ norm inside the $d\underline{t}_{k+1}d\left\vert
\mu _{t_{k+1}}\right\vert $ integrals and expand the time integration domain
to $[0,T]^{k}$. That is,%
\begin{eqnarray*}
&&\left\Vert \langle \nabla _{\xi _{1}}\rangle ^{0+}\left( I_{\mu
,sgn}^{(k+1)}\left( f^{(k+1)}\right) (t_{1})\right) ^{\vee }\right\Vert
_{L_{\eta _{1}\xi _{1}}^{3-}} \\
&\leqslant &\int_{\mathcal{P}(\Omega )}d\left\vert \mu _{t_{k+1}}\right\vert
(\rho )\int_{[0,T]^{3}}\left\Vert \left( \langle \nabla _{\xi _{1}}\rangle
^{0+}Q^{\pm }(...)\right) ^{\vee }\right\Vert _{L_{\eta _{1}\xi _{1}}^{3-}}d%
\underline{t}_{k+1}
\end{eqnarray*}

\item[Step 4] If $j<k$, go to Step 5, otherwise go to Step 8.

\item[Step 5] If the $j$-th coupling contains the $k$-th coupling in one of
its two bilinear variables (there can be at most one), then Cauchy-Schwarz
in $t_{j+1}$ and apply estimate (\ref{eqn:UniquenessStrichartz-2}) to the $j$%
-th coupling such that the bilinear variable carrying the $k$-th coupling is
put in $L_{\eta \xi }^{3-}$ and go to Step 7. If not, go to Step 6.

\item[Step 6] Cauchy-Schwarz in $t_{j+1}$ and apply estimate (\ref%
{eqn:UniquenessStrichartz-1}) to the $j$-th coupling.

\item[Step 7] $j=j+1$ and go to Step 4.

\item[Step 8] Apply estimate (\ref{eqn:UniquenessSobolev}) to the $k$-th
coupling, we would have deduced that%
\begin{eqnarray*}
&&\left\Vert \langle \nabla _{\xi _{1}}\rangle ^{0+}\left( I_{\mu
,sgn}^{(k+1)}\left( f^{(k+1)}\right) (t_{1})\right) ^{\vee }\right\Vert
_{L_{\eta _{1}\xi _{1}}^{3-}} \\
&\leqslant &\left( CT^{\frac{1}{2}-}\right) ^{k-1}\int_{\mathcal{P}(\Omega
)}d\left\vert \mu _{t_{k+1}}\right\vert (\rho )\left\Vert \langle \nabla
_{x}\rangle ^{0+}\langle \nabla _{\xi _{1}}\rangle ^{0+}\tilde{\rho}%
\right\Vert _{L_{x\xi }^{2}}^{k-2}\Vert \langle \nabla _{x}\rangle ^{\frac{1%
}{2}+}\langle \nabla _{\xi }\rangle ^{\frac{1}{2}+}\tilde{\rho}\Vert
_{L_{x\xi }^{2}}^{2}\int_{[0,T]}dt_{k+1} \\
&\leqslant &\left( CT^{\frac{1}{2}-}\right) ^{k-1}T\int_{\mathcal{P}(\Omega
)}\left\Vert \langle \nabla _{x}\rangle ^{1+}\langle \nabla _{\xi
_{1}}\rangle ^{0+}\tilde{\rho}\right\Vert _{L_{x\xi }^{2}}^{k-2}\Vert
\langle \nabla _{x}\rangle ^{\frac{1}{2}+}\langle \nabla _{\xi }\rangle ^{%
\frac{1}{2}+}\tilde{\rho}\Vert _{L_{x\xi }^{2}}^{2}d\left\vert \mu
_{t_{k+1}}\right\vert (\rho )
\end{eqnarray*}%
That is,%
\begin{eqnarray*}
&&\sup_{t\in \left[ 0,T\right] }\left\Vert \langle \nabla _{\xi _{1}}\rangle
^{0+}\left( I_{\mu }^{(k+1)}\left( f^{(k+1)}\right) (t_{1})\right) ^{\vee
}\right\Vert _{L_{\eta _{1}\xi _{1}}^{3-}} \\
&\leqslant &2\left( CT^{\frac{1}{2}-}\right) ^{k-1}T\int_{\mathcal{P}(\Omega
)}\left\Vert \langle \nabla _{x}\rangle ^{1+}\langle \nabla _{\xi
_{1}}\rangle ^{0+}\tilde{\rho}\right\Vert _{L_{x\xi }^{2}}^{k-2}\Vert
\langle \nabla _{x}\rangle ^{\frac{1}{2}+}\langle \nabla _{\xi }\rangle ^{%
\frac{1}{2}+}\tilde{\rho}\Vert _{L_{x\xi }^{2}}^{2}d\left\vert \mu
_{t_{k+1}}\right\vert (\rho )
\end{eqnarray*}%
as claimed in Proposition \ref{prop:estimate with boardgame}.
\end{itemize}

\section{Justification of Physicality: Regularity from the Local Maxwellian
Viewpoint\label{S3}}

\emph{The calculations in this section are not rigorous. However, the
content of this section is not needed for the proof of Theorem \ref{thm:main}%
. This section is only intended to motivate the hypotheses of Theorem \ref%
{thm:main}.}

In this section, we give a construction of an $N$-body solution converging
to a local Maxwellian. A simple tensor product of local Maxwellians is not
qualified to be $f_{N}^{(k)}$; this format can only be achieved in the limit 
$f^{(k)}$. By appealing to the law of large numbers to obtain a
representative form of $f_{N}^{(k)}$, we find that $f_{N}^{(k)}$ must
consist not only of the tensor product of local Maxwellians that persist in
the $N\rightarrow \infty $ limit, but also additional quasi-free terms that
should, in some sense, vanish as $N\rightarrow \infty $ while preventing one
from closing the BBGKY hierarchy estimates with only a single Duhamel
iterate.\footnote{%
This is actually conjectured in \cite[p.11]{BCEP4}.} The quasi-free terms,
when measured in the standard Sobolev norms in the $(\boldsymbol{x}_{k},%
\boldsymbol{\xi }_{k})$ reference frame, have growth as $N\rightarrow \infty 
$. In this sense, these terms are irregular (or more precisely, sort of
regular in their own way) and must be isolated in the decomposition of $%
f_{N}^{(k)}$ and represented in their own natural reference frame so that
they can be estimated separately in the BBGKY hierarchy.

\subsection{A $N$-body construction of the local Maxwellian\label%
{SS:construction}}

Let $Y=O(1)$ in spatial $\mathbb{R}^{3}$ and $W=O(1)$ in frequency $\mathbb{%
R }^{3}$. Let $\chi $ be a Schwartz class function on $\mathbb{R}^{3}$. A
wave packet 
\begin{equation*}
u(y,0)=\chi (\frac{y-Y}{\epsilon ^{1/2}})e^{iy\cdot W/\epsilon }
\end{equation*}
is spatially centered at position $Y$ with spatial width $O(\epsilon ^{1/2})$%
, and is frequency centered at $W/\epsilon $ with frequency width $\epsilon
^{-1/2}$. Under the evolution $(i\epsilon \partial _{t}+\epsilon ^{2}\Delta
)u=0$ on a unit time scale, $u(y,t)$ will be spatially centered at position $%
Y-2Wt$ with spatial width $O(\epsilon ^{1/2})$ and frequency centered at $%
W/\epsilon $ with frequency width $O(\epsilon ^{-1/2})$: 
\begin{equation*}
u(y,t)=\chi_t (\frac{y-Y-2Wt}{\epsilon ^{1/2}})e^{iy\cdot W/\epsilon }
\end{equation*}
where $\chi_t(y) = e^{it\Delta}\chi(y)$. In particular, on a unit time
scale, not much decoherence will take place, and this is why we have chosen
spatial width $O(\epsilon ^{1/2})$. \footnote{%
One way to see the lack of decoherence is to let $v(y,t)=u(\epsilon
^{1/2}y,t)$. Then $(i\partial _{t}+\Delta )v=0$ with $v(y,0)=\chi
(y-\epsilon ^{-1/2}Y)e^{iy\cdot W/\epsilon ^{1/2}}$. So the transformed
initial condition $v(y,0)$ solves the normalized Schr\"{o}dinger equation
with $O(1)$ width initial condition on a unit time scale; the oscillatory
phase factor is handled by Galilean invariance.}. Note that 
\begin{equation*}
\func{cov}(W,Y+2tW) = 2t \func{cov}(W,W) = 2t
\end{equation*}
and thus the frequency shift and positional shift, if initially independent,
will have a linearly evolving covariance. It seems reasonable that upon a
collision, this covariance could shift and thus an interacting multiparticle
ansatz should incorporate a shift in time $t_j$ associated with the $j$th
particle.

Let $Y_1, \ldots, Y_N, W_1, \ldots, W_N$ be an independent sample from the
standard normal distribution and consider the wave function 
\begin{equation}  \label{E:LM05}
\begin{aligned} \Psi_N(\boldsymbol{y}_N) = \kappa \sum_{\sigma\in S_N} &
\chi_{\sigma(1)}(\frac{ y_1-Y_{\sigma(1)} - 2t_{\sigma(1)} W_{\sigma(1)}
}{\epsilon^{1/2}}) e^{iy_1\cdot W_{\sigma(1)}/\epsilon} \\ & \cdots
\chi_{\sigma(N)}( \frac{y_N-Y_{\sigma(N)} -2 t_{\sigma(N)} W_{\sigma(N)}
}{\epsilon^{1/2}})e^{iy_N\cdot W_{\sigma(N)}/\epsilon} \end{aligned}
\end{equation}
where $\kappa$ is a suitable normalization, determined below. Note that we
have applied the permutation $\sigma$ to the spatial center indices $%
Y_\bullet$, the frequency center indices $W_\bullet$, and the profile
indices $\chi_\bullet$. The time shifts $t_\bullet$ allow us to consider a
wave-form in which the $(Y_\bullet, W_\bullet)$ pair covariances vary.

When the wave function \eqref{E:LM05} is taken as the initial condition, the
solution along the $N$-body \emph{free} linear flow (no interaction) is of a
similar form with the $Y$'s suitably translated. 
\begin{equation}  \label{E:LM05b}
\begin{aligned} \Psi_N(\boldsymbol{y}_N,t) = \kappa \sum_{\sigma\in S_N} &
\chi_{\sigma(1),t}(\frac{ y_1-Y_{\sigma(1)} - 2(t-t_{\sigma(1)}
)W_{\sigma(1)} }{\epsilon^{1/2}}) e^{iy_1\cdot W_{\sigma(1)}/\epsilon} \\ &
\cdots \chi_{\sigma(N),t}( \frac{y_N-Y_{\sigma(N)} - 2(t-t_{\sigma(N)})
W_{\sigma(N)} }{\epsilon^{1/2}})e^{iy_N\cdot W_{\sigma(N)}/\epsilon}
\end{aligned}
\end{equation}
where $\chi_{j,t}(x) = e^{it\Delta}\chi_j(x)$ (the $\epsilon=1$ free linear
Schr\"odinger propagator). Thus the covariances of the $(Y_\bullet,
W_\bullet)$ pairs evolve linearly from their respective initial values $%
t_\bullet$.

In a collisional model, collisions are expected to occur on average every $%
\epsilon$ increment of time. Although the effect of collisions is weak,
their expected impact over $O(1)$ time is $O(1)$. Upon collision, the
phase/velocity $W_\bullet$ will shift giving rise to a shift in the $%
(Y_\bullet, W_\bullet)$ pair covariance. Thus \eqref{E:LM05} seems to be a
reasonable model of the functional form of the solution at an arbitrary time
and we will perform computations using the form \eqref{E:LM05}.

From \eqref{E:LM05} 
\begin{align*}
\hspace{0.3in}&\hspace{-0.3in} \gamma _{N}(\boldsymbol{y}_{N},\boldsymbol{y}%
_{N}^{\prime }) =\Psi _{N}( \boldsymbol{y}_{N})\bar{\Psi}_{N}(\boldsymbol{y}%
_{N}^{\prime }) \\
& =\kappa ^{2}\sum_{\substack{ \sigma \in S_{N}  \\ \sigma^{\prime} \in
S_{N} }} \chi_{\sigma(1)} ( \frac{y_{1}-Y_{\sigma (1)}-2t_{\sigma(1)}
W_{\sigma(1)}}{\epsilon ^{1/2}}) \\
& \qquad\qquad \bar{\chi}_{\sigma^{\prime}(1)}(\frac{y_{1}^{\prime}-Y_{%
\sigma^{\prime} (1)} - 2t_{\sigma^{\prime}(1)}W_{\sigma^{\prime}(1)} }{%
\epsilon ^{1/2}})e^{i(y_{1}\cdot W_{\sigma (1)}-y_{1}^{\prime }\cdot
W_{\sigma^{\prime} (1)})/\epsilon } \\
& \qquad\qquad \cdots \chi_{\sigma(N)} (\frac{y_{N}-Y_{\sigma (N)} -
2t_{\sigma(N)} W_{\sigma(N)} }{\epsilon^{1/2}}) \\
& \qquad\qquad \bar{\chi}_{\sigma^{\prime}(N)}(\frac{y_{N}^{\prime
}-Y_{\sigma^{\prime} (N)} -2t_{\sigma^{\prime}(N)} W_{\sigma^{\prime}(N)} }{%
\epsilon^{1/2}} )e^{i(y_{N}\cdot W_{\sigma (N)}-y_{N}^{\prime }\cdot
W_{\sigma^{\prime} (N)})/\epsilon }
\end{align*}
Upon setting 
\begin{equation*}
\begin{aligned} &\boldsymbol{y}_N =
\boldsymbol{x}_N+\epsilon\boldsymbol{\xi}_N \\
&\boldsymbol{y}_N^{\prime}=\boldsymbol{x}_N-\epsilon \boldsymbol{\xi}_N
\end{aligned} \qquad \iff \qquad \begin{aligned} &\boldsymbol{x}_N =
\boldsymbol{y}_N+\boldsymbol{y}_N^{\prime} \\ & \boldsymbol{\xi}_N =
\frac{\bds y_N - \bds y_N^{\prime}}{\epsilon} \end{aligned}
\end{equation*}
we have 
\begin{equation*}
\tilde{f}_{N}(\boldsymbol{x}_{N},\boldsymbol{\xi }_{N})=\gamma _{N}( 
\boldsymbol{y}_{N},\boldsymbol{y}_{N}^{\prime })
\end{equation*}
this takes the form 
\begin{equation}  \label{E:LM12}
\begin{aligned} \tilde{f}_{N}(\boldsymbol{x}_{N},\boldsymbol{\xi
}_{N})=\kappa ^{2} \sum_{\substack{ \sigma \in S_{N} \\ \sigma^{\prime} \in
S_{N}}} \prod_{j=1}^{N} \chi_{\sigma(j)} (\frac{ x_{j}-Y_{\sigma
(j)}-2t_{\sigma(j)} W_{\sigma(j)} }{\epsilon^{1/2}}+\epsilon ^{1/2}\xi _{j})
e^{ix_{j}\cdot (W_{\sigma (j)}-W_{\sigma^{\prime} (j)})/\epsilon }& \\
\bar{\chi}_{\sigma^{\prime}(j)}( \frac{x_{j}-Y_{\sigma^{\prime}(j)} -
2t_{\sigma^{\prime}(j)} W_{\sigma^{\prime}(j)} }{\epsilon ^{1/2}}-\epsilon
^{1/2}\xi _{j}) e^{i\xi _{j}\cdot (W_{\sigma (j)}+W_{\sigma^{\prime} (j)})}&
\end{aligned}
\end{equation}

\begin{proposition}
The constant $\kappa$ needed to achieve the normalization $\mathbb{E}
\|\Psi_N \|_{L^2}^2 = 1$, where $\Psi_N$ is defined by \eqref{E:LM05}, is 
\begin{equation}  \label{E:LM11}
\kappa^2 \sim \frac{1}{N! \epsilon^{3N/2}}
\end{equation}
\end{proposition}

\begin{proof}
We assume that $\chi_j = \chi$ and all time shifts $t_j=0$ (the general case
is similar). We have 
\begin{equation}  \label{E:LM06}
\| \Psi_N \|_{L^2}^2 = \int_{\mathbb{R}^{3N}} \tilde f_N(\boldsymbol{x}_N, 
\boldsymbol{0}) \,d\boldsymbol{x}_N = \kappa^2 \sum_{\substack{ \sigma\in
S_N  \\ \sigma^{\prime} \in S_N}} I_{\sigma,\sigma^{\prime}}
\end{equation}
where 
\begin{equation*}
I_{\sigma,\sigma^{\prime}} = \prod_{j=1}^N \int_{x_j} \chi(\frac{x_j
-Y_{\sigma(j)}}{ \epsilon^{1/2}} )\bar \chi(\frac{x_j -Y_{\sigma^{\prime}(j)}%
}{\epsilon^{1/2}}) e^{i x_j\cdot (W_{\sigma(j)} -
W_{\sigma^{\prime}(j)})/\epsilon} dx_j
\end{equation*}
The condition $\mathbb{E} \| \Psi_N \|_{L^2}^2=1$ determines the
normalization constant $\kappa$. We consider the value of $I_{\sigma,
\sigma^{\prime}}$ in various settings, but first let us examine two
representative cases.

\bigskip

\noindent\emph{Case 1}. All $\sigma^{\prime}(j)=\sigma(j)$ for $1\leq j \leq
N$. This case gives the largest expected value. In this case, 
\begin{equation*}
I_{\sigma,\sigma} = \prod_{j=1}^N \int_{x_j} |\chi(\frac{x_j -Y_{\sigma(j)}}{
\epsilon^{1/2}} )|^2 dx_j = \epsilon^{3N/2} (\int |\chi|^2)^N
\end{equation*}
We did not even need to take the expectation; it is constant on the
probability space $\Omega$.

\bigskip

\noindent\emph{Case 2}. All $\sigma^{\prime}(j)=\sigma(j)$ for $k+1 \leq j
\leq N$ but all $\sigma^{\prime}(j) \neq \sigma(j)$ for $1\leq j \leq k$. In
this case, by independence, the expectation is 
\begin{align*}
\mathbb{E} I_{\sigma,\sigma^{\prime}} &= \prod_{j=1}^k \int_{x_j}
\int_{y_{\sigma(j)}, y_{\sigma^{\prime}(j)}, w_{\sigma(j)},
w_{\sigma^{\prime}(j)}} \chi(\frac{x_j -y_{\sigma(j)}}{ \epsilon^{1/2}}
)\bar \chi(\frac{x_j -y_{\sigma^{\prime}(j)}}{\epsilon^{1/2}}) e^{i x_j\cdot
(w_{\sigma(j)} - w_{\sigma^{\prime}(j)})/\epsilon} \\
& \qquad\qquad \qquad \qquad g(y_{\sigma(j)}) g(y_{\sigma^{\prime}(j)})
g(w_{\sigma(j)}) g(w_{\sigma^{\prime}(j)}) dx_j \, dy_{\sigma(j)}
dy_{\sigma^{\prime}(j)}dw_{\sigma(j)} dw_{\sigma^{\prime}(j)} \\
& \quad \prod_{j=k+1}^N \int_{y_{\sigma(j)}} \int_{x_j} |\chi(\frac{x_j
-y_{\sigma(j)}}{\epsilon^{1/2}} )|^2 g(y_{\sigma(j)}) dx_j \, dy_{\sigma(j)}
\end{align*}
The last $N-k$ factors yield $\epsilon^{3(N-k)/2}$, as before. In the first $%
k$ factors, it is easiest to start by carrying out the integrals over $%
w_{\sigma(j)}$ and $w_{\sigma^{\prime}(j)}$ which are just Fourier
transforms of $g$ evaluated at $-x_j/\epsilon$ and $x_j/\epsilon$
respectively. Since we take $g(w)= (2\pi)^{-3/2} e^{-w^2/2}$, the standard
normal distribution, each integral contributes $e^{-|x_j|^2/2\epsilon^2}$.
Thus 
\begin{align*}
\mathbb{E} I_{\sigma,\sigma^{\prime}} = \epsilon^{3(N-k)/2} \prod_{j=1}^k
\int_{x_j} \int_{y_{\sigma(j)}, y_{\sigma^{\prime}(j)}} \chi(\frac{x_j
-y_{\sigma(j)}}{ \epsilon^{1/2}} )\bar \chi(\frac{x_j -y_{\sigma^{\prime}(j)}%
}{\epsilon^{1/2}}) e^{-|x_j|^2/\epsilon^2} & \\
e^{-|y_{\sigma(j)}|^2/2} e^{-|y_{\sigma^{\prime}(j)}|^2/2} dx_j \,
dy_{\sigma(j)} dy_{\sigma^{\prime}(j)} &
\end{align*}
The integral over $x_j$ yields a factor $\epsilon^3$, and the $y_{\sigma(j)}$
and $y_{\sigma^{\prime}(j)}$ integrals yield $\epsilon^{3/2} \int \chi$ and $%
\epsilon^{3/2}\int \bar \chi$, respectively. Thus 
\begin{equation*}
\mathbb{E} I_{\sigma,\sigma^{\prime}} = \epsilon^{3(N-k)/2} \epsilon^{6k}
|\hat\chi(0)|^{2k} (\int|\chi|^2)^{N-k}
\end{equation*}
This is much smaller than Case 1, although the terms in Case 2 occur more
frequently in the sum over permutations.

Now let us return to \eqref{E:LM06}. \emph{Suppose that we fix a permutation 
$\sigma$ of $\{1, \ldots, N\}$}. In terms of $\sigma$, we will categorize
the set of all permutations $\sigma^{\prime}$ of $\{1, \ldots, N\}$.
Specifically, decompose the set of all $\sigma^{\prime}$ into a disjoint
union $E_0 \cup \cdots \cup E_N$, where $E_k$ is the set of all $%
\sigma^{\prime}$ for which the set $S$ of indices $j\in \{1, \ldots, N\}$
for which $\sigma(j) = \sigma^{\prime}(j)$ has cardinality $N-k$. The set $%
E_0$ has cardinality $1$, since $\sigma=\sigma^{\prime}$ on all of $%
\{1,\ldots, N\}$. To determine the cardinality of $E_k$, note first that
there are $\binom{N}{k}$ different ways to select the set $S$. Once $S$ has
been selected, the value of $\sigma^{\prime}$ on $S$ is determined ($%
\sigma^{\prime}=\sigma$ on $S$). On $S^c$ (which has cardinality $k$), we
need to determine the number of possible different selections for the values
of $\sigma^{\prime}$. To do this, consider that for any such $%
\sigma^{\prime} $, $\nu = \sigma^{-1}\circ \sigma^{\prime}$ will have the
property that

\begin{itemize}
\item for each $j\in S$, $\nu(j) = j$ (in other words, $\nu$ fixes $S$).
From this we conclude that $\nu:S^c\to S^c$.

\item for each $j\in S^c$, $\nu(j) \neq j$, but $\nu(j)$ is otherwise
undetermined.
\end{itemize}

Thus $\nu |_{S^{c}}$ is a permutation of $S^{c}$ with no fixed points, and
this type of permutation is called a \emph{derangement\footnote{%
See \url{https://en.wikipedia.org/wiki/Derangement}}}. The number of
derangements of a set of cardinality $k$ is the integer closest of $k!/e$.
We can thus generate all $\sigma^{\prime} $ by allowing $\nu $ to range
through all derangements of $S^{c}$ and for each $\nu $ take $%
\sigma^{\prime} =\sigma \circ \nu $. 
\begin{align*}
|E_{k}|& =(\text{number of ways to select $S$})(\text{number of derangements
on $k$ elements}) \\
& =\binom{N}{k}\left[ \frac{k!}{e}\right] \leq \frac{N!}{(N-k)!}
\end{align*}
where the brackets denote \textquotedblleft integer nearest
to\textquotedblright . From \eqref{E:LM06}, 
\begin{equation*}
\mathbb{E}\Vert \Psi _{N}\Vert _{L^{2}}^{2}=M+R\,,\text{ where }
M=\sum_{\sigma \in S_{N}}\mathbb{E}I_{\sigma ,\sigma }\text{ and }
R=\sum_{\sigma \in S_{N}}\sum_{k=1}^{N}\sum_{\sigma^{\prime} \in E_{k}}%
\mathbb{E} I_{\sigma ,\sigma^{\prime} }
\end{equation*}
Taking all the $\chi $ related integrals to be $1$ for expository
convenience, $M$ is just determined from Case 1 to be 
\begin{equation*}
M=N!\kappa ^{2}\epsilon ^{3N/2}
\end{equation*}
For $R$, there are $N!$ choices for $\sigma $ and for each $\sigma $, there
the set $E_{k}$ has cardinality $|E_{k}|\leq \frac{N!}{(N-k)!}$. Thus from
Case 2, 
\begin{equation*}
|R|\leq N!\kappa ^{2}\sum_{k=1}^{N}\frac{N!}{(N-k)!}\epsilon
^{3(N-k)/2}\epsilon ^{6k}
\end{equation*}
Using the crude bound $\frac{N!}{(N-k)!}\leq N^{k}=\epsilon ^{-3k}$, 
\begin{equation*}
|R|\leq N!\kappa ^{2}\epsilon ^{3N/2}\sum_{k=1}^{N}\epsilon ^{3k/2}\leq M 
\frac{\epsilon ^{3/2}}{1-\epsilon ^{3/2}}
\end{equation*}
Thus $|R|$ is negligible in comparison to $M$. To achieve normalization, we
set 
\begin{equation*}
1=\mathbb{E}\Vert \Psi _{N}\Vert _{L^{2}}^{2}\sim N!\kappa ^{2}\epsilon
^{3N/2}
\end{equation*}
from which it follows that \eqref{E:LM11} holds.
\end{proof}

\subsection{The structure of marginals $\tilde{f}_{N}^{(k)}$}

The following gives the decomposition of $\tilde f_N^{(k)}$ into a \emph{core%
} term plus additional quasi-free terms.

\begin{heuristic}
Let $U_k$ be the set of all subsets of $\{1, \ldots, N\}$ of cardinality $k$%
. Let $S_k$ denote the set of permutations on $\{1,\ldots, k\}$. Then $%
\tilde f_N^{(k)}$ admits a decomposition 
\begin{equation}  \label{E:LM16}
\tilde f_N^{(k)}(\boldsymbol{x}_k, \boldsymbol{\xi}_k) = \sum_{ \pi\in S_k}
\tilde f_{N,\pi}^{(k)}(\boldsymbol{x}_k, \boldsymbol{\xi}_k)
\end{equation}
where 
\begin{equation}  \label{E:LM17}
\tilde f_{N,\pi}^{(k)}(\boldsymbol{x}_k, \boldsymbol{\xi}_k) = \frac{1}{%
\binom{N}{k}} \sum_{ \{\sigma(1), \ldots, \sigma(k)\} \in U_k}
\epsilon^{-3k/2} \tilde f_{N,\sigma,\sigma \circ \pi^{-1}}^{(k)}(\boldsymbol{%
x}_k,\boldsymbol{\xi}_k)
\end{equation}
and 
\begin{equation}  \label{E:LM18}
\begin{aligned} \tilde f_{N,\sigma,\sigma\circ
\pi^{-1}}^{(k)}(\boldsymbol{x}_k, \boldsymbol{\xi}_k) = \prod_{j=1}^{k} &
\chi_{\sigma(j)} (\frac{x_{j}+\epsilon\xi_j-Y_{\sigma(j)}
-2t_{\sigma(j)}W_{\sigma(j)}}{\epsilon ^{1/2}}) e^{ix_{j}\cdot
(W_{\sigma(j)}-W_{\sigma\circ \pi^{-1} (j)})/\epsilon }\\ &
\bar{\chi}_{\sigma\circ \pi^{-1}(j)}(\frac{x_{j}-\epsilon\xi_j
-Y_{\sigma\circ \pi^{-1}(j)} -2 t_{\sigma\circ \pi^{-1}(j)} W_{\sigma\circ
\pi^{-1}(j)}}{\epsilon ^{1/2}}) e^{i\xi _{j}\cdot
(W_{\sigma(j)}+W_{\sigma\circ \pi^{-1} (j)})} \end{aligned}
\end{equation}
\end{heuristic}

\begin{proof}
The marginals are given by 
\begin{equation*}
\tilde f_N^{(k)}(\boldsymbol{x}_k, \boldsymbol{\xi}_k) = \int \tilde f_N(%
\boldsymbol{x}_k, \boldsymbol{x}_{N-k}, \boldsymbol{\xi}_k, \boldsymbol{0})
\, d\boldsymbol{x}_{N-k}
\end{equation*}
By \eqref{E:LM12}, 
\begin{equation}  \label{E:LM13}
\tilde f_N^{(k)}(\boldsymbol{x}_k, \boldsymbol{\xi}_k) =\kappa ^{2}\sum 
_{\substack{ \sigma \in S_{N}  \\ \sigma^{\prime} \in S_{N}}}
I_{\sigma,\sigma^{\prime }}^{N-k} \tilde f_{N,\sigma,\sigma^{\prime }}^{(k)}(%
\boldsymbol{x}_k,\boldsymbol{\xi}_k)
\end{equation}
where 
\begin{equation*}
\begin{aligned} I_{\sigma,\sigma^{\prime }}^{N-k} = \prod_{j=k+1}^N
\int_{x_j} & \chi_{\sigma(j)}(\frac{x_j -Y_{\sigma(j)}-2t_{\sigma(j)}
W_{\sigma(j)}}{\epsilon^{1/2}} ) \\ & \bar
\chi_{\sigma^{\prime}(j)}(\frac{x_j
-Y_{\sigma^{\prime}(j)}-2t_{\sigma^{\prime}(j)}W_{\sigma^{\prime}(j)}}{%
\epsilon^{1/2}}) e^{i x_j\cdot (W_{\sigma(j)} -
W_{\sigma^{\prime}(j)})/\epsilon} dx_j \end{aligned}
\end{equation*}
and 
\begin{equation*}
\begin{aligned} \tilde f_{N,\sigma,\sigma^{\prime}}^{(k)}(\boldsymbol{x}_k,
\boldsymbol{\xi}_k) = \prod_{j=1}^{k} & \chi_{\sigma(j)}
(\frac{x_{j}-Y_{\sigma(j)}-2t_{\sigma(j)} W_{\sigma(j)}}{\epsilon
^{1/2}}+\epsilon ^{1/2}\xi _{j}) e^{ix_{j}\cdot (W_{\sigma
(j)}-W_{\sigma^{\prime} (j)})/\epsilon } \\
&\bar\chi_{\sigma(j)}(\frac{x_{j}-Y_{\sigma^{\prime}(j)}-2t_{\sigma^{%
\prime}(j)} W_{\sigma^{\prime}(j)}}{\epsilon ^{1/2}}-\epsilon ^{1/2}\xi
_{j}) e^{i\xi _{j}\cdot (W_{\sigma (j)}+W_{\sigma^{\prime} (j)})}
\end{aligned}
\end{equation*}
In the sum \eqref{E:LM13}, both factors $I_{\sigma,\sigma^{\prime }}^{N-k}$
and $\tilde f_{N,\sigma,\sigma^{\prime }}^{(k)}(\boldsymbol{x}_k, 
\boldsymbol{\xi}_k)$ are random variables, and for each $(\sigma,\sigma^{%
\prime })$, these two factors are independent. By the arguments in \S \ref%
{SS:construction}, $I_{\sigma,\sigma^{\prime }}^{N-k}$ is dominated by the
case in which $\sigma|_{k+1,\, . \,,N} = \sigma^{\prime }|_{k+1,\, . \,, N}$%
, explicitly 
\begin{equation*}
\sigma(j) = \sigma^{\prime }(j) \text{ for } j=k+1, \ldots, N
\end{equation*}
and in this case, $I_{\sigma,\sigma^{\prime }}^{N-k}$ is a non-random
variable that takes the value $\epsilon^{3(N-k)/2}$. Thus we reduce our
study to 
\begin{equation}  \label{E:LM14}
\tilde f_N^{(k)}(\boldsymbol{x}_k, \boldsymbol{\xi}_k) =\kappa
^{2}\epsilon^{3(N-k)/2}\sum _{\substack{ \sigma, \sigma^{\prime} \in S_{N} 
\\ \sigma|_{k+1,\, . \,,N} = \sigma^{\prime }|_{k+1,\, . \,, N}}} \tilde
f_{N,\sigma,\sigma^{\prime }}^{(k)}(\boldsymbol{x}_k,\boldsymbol{\xi}_k)
\end{equation}
Now, if $\sigma|_{k+1,\, . \,,N} = \sigma^{\prime }|_{k+1,\, . \,, N}$, then
let $\pi = (\sigma^{\prime})^{-1}\circ \sigma$, so that $\pi(j)=j$ for all $%
j\in \{k+1, \ldots, N\}$ and can thus be regarded as an element of $S_k$.
Replacing $\sigma^{\prime }=\sigma \circ \pi^{-1}$, \eqref{E:LM14} becomes 
\begin{equation}  \label{E:LM15}
\tilde f_N^{(k)}(\boldsymbol{x}_k, \boldsymbol{\xi}_k) =\frac{1}{%
N!\epsilon^{3k/2}}\sum_{\sigma \in S_N, \;\pi \in S_k} \tilde
f_{N,\sigma,\sigma \circ \pi^{-1}}^{(k)}(\boldsymbol{x}_k,\boldsymbol{\xi}_k)
\end{equation}
where we have substituted \eqref{E:LM11}. For each $\sigma \in S_N$ and $%
\pi\in S_k$, the definition of $\tilde f_{N,\sigma,\sigma \circ
\pi^{-1}}^{(k)}(\boldsymbol{x}_k,\boldsymbol{\xi}_k)$ depends only on $%
\{\sigma(1), \ldots, \sigma(k)\}$, which is some subset of $\{1, \ldots, N\}$
of cardinality $k$. Let $U_k$ be the set of all subsets of $\{1,\ldots, N\}$
of cardinality $k$. Then, of course, $|U_k| = \binom{N}{k}$, so for a fixed
value of $\{\sigma(1), \ldots, \sigma(k)\}$, there are $N!/ \binom{N}{k}$
terms in the above sum. Thus we obtain \eqref{E:LM16}, \eqref{E:LM17}.
\end{proof}

\begin{heuristic}
\label{H:2} By \eqref{E:LM17}, \eqref{E:LM18} and the law of large numbers 
\begin{equation}  \label{E:LM19}
\begin{aligned} \tilde f_{N, \pi}^{(k)}(\boldsymbol{x}_k,
\boldsymbol{\xi}_k) \sim \mathbb{E} \prod_{j=1}^{k} \epsilon^{-3/2} \chi_{j}
(\frac{x_{j}+\epsilon\xi_j-Y_{j} -2t_{j}W_{j}}{\epsilon ^{1/2}})
e^{ix_{j}\cdot (W_{j}-W_{ \pi^{-1} (j)})/\epsilon } &\\ \bar{\chi}_{
\pi^{-1}(j)}(\frac{x_{j}-\epsilon\xi_j -Y_{ \pi^{-1}(j)} -2 t_{ \pi^{-1}(j)}
W_{ \pi^{-1}(j)}}{\epsilon ^{1/2}}) e^{i\xi _{j}\cdot (W_{j}+W_{ \pi^{-1}
(j)})} & \end{aligned}
\end{equation}
By reindexing the product and using the independence of $\{(Y_1,W_1),
\ldots, (Y_N,W_N)\}$ to bring the expectation inside the product: 
\begin{equation}  \label{E:LM19b}
\begin{aligned} \tilde f_{N, \pi}^{(k)}(\boldsymbol{x}_k,
\boldsymbol{\xi}_k) \sim \prod_{j=1}^{k} \mathbb{E} \; \epsilon^{-3/2}
\chi_{j} (\frac{x_{j}+\epsilon\xi_j-Y_{j} -2t_{j}W_{j}}{\epsilon ^{1/2}})
e^{i(x_j-x_{\pi(j)}) \cdot W_j/\epsilon} &\\
\bar{\chi}_{j}(\frac{x_{\pi(j)}-\epsilon\xi_{\pi(j)} -Y_{ j} -2 t_{ j} W_{
j}}{\epsilon ^{1/2}}) e^{i(\xi_j+\xi_{\pi(j)})\cdot W_j} & \end{aligned}
\end{equation}
Let $G_Y(y)$ and $G_W(w)$ be the pdfs of $Y$ and $W$, which we take to be
standard normal. Then the expectation in \eqref{E:LM19b} can be evaluated to
give 
\begin{equation}  \label{E:LM19c}
\begin{aligned} \tilde f_{N, \pi}^{(k)}(\boldsymbol{x}_k,
\boldsymbol{\xi}_k) \sim \prod_{j=1}^k (1+4t_j^2)^{-3/2} \left(\int
|\chi_j|^2\right) \hat G_W \Big( \frac{2q_j}{\sqrt{1+4t_j^2}} \Big) & \\ G_Y
\Big( \frac{p_j}{\sqrt{1+4t_j^2}} \Big) \exp\left( -4i t_j \frac{q_j\cdot
p_j}{1+4t_j^2} \right) & \end{aligned}
\end{equation}
where 
\begin{equation*}
p_j = \tfrac12 (x_j+x_{\pi(j)}) + \tfrac12\epsilon(\xi_j-\xi_{\pi(j)}) \,,
\qquad q_j = \tfrac12 (\xi_j+\xi_{\pi(j)}) + \tfrac12
(x_j-x_{\pi(j)})/\epsilon
\end{equation*}
\end{heuristic}

\begin{remark}
The core term occurs when $\pi=\text{Identity}$ and in this case, $p_j=x_j$
and $q_j=\xi_j$, and \eqref{E:LM19c} becomes 
\begin{equation}  \label{E:LM19d}
\begin{aligned} \tilde f_{N, \pi}^{(k)}(\boldsymbol{x}_k,
\boldsymbol{\xi}_k) \sim \prod_{j=1}^k (1+4t_j^2)^{-3/2} \left(\int
|\chi_j|^2\right) \hat G_W \Big( \frac{2\xi_j}{\sqrt{1+4t_j^2}} \Big) & \\
G_Y \Big( \frac{x_j}{\sqrt{1+4t_j^2}} \Big) \exp\left( -4i t_j
\frac{\xi_j\cdot x_j}{1+4t_j^2} \right) & \end{aligned}
\end{equation}
Upon taking the Fourier transform $\boldsymbol{\xi}_k \to \boldsymbol{v}_k$,
we obtain 
\begin{equation*}
f_{N,\pi}^{(k)}(\boldsymbol{x}_k, \boldsymbol{v}_k) \sim \prod_{j=1}^k
\left(\int |\chi_j|^2\right) G_Y\left( \frac{x_j}{\sqrt{1+4t_j^2}} \right)
G_W \left[ \sqrt{1+4t_j^2} \left( \frac{v_j}{2} - \frac{2t_jx_j}{1+4t_j^2}%
\right)\right]
\end{equation*}
After completing the square, we obtain 
\begin{equation*}
f_{N,\pi}^{(k)}(\boldsymbol{x}_k, \boldsymbol{v}_k) \sim \prod_{j=1}^k
\left(\int |\chi_j|^2\right) e^{-(x_j-v_jt_j)^2/2} e^{-v_j^2/8}
\end{equation*}
which is the standard form of the local Maxwellian.
\end{remark}

\begin{remark}
At this point, we recall Example \ref{EX:2cycles} an in particular %
\eqref{E:irregularity}, which shows that when $k=2$, $\pi=(12)$, even in the
ideal situation of assuming localization in $\boldsymbol{\xi}_2$, the
function $\tilde f_{N,(12)}^{(2)}(\boldsymbol{x}_2,\boldsymbol{\xi}_2)$
given by \eqref{E:LM19d} only satisfies uniform bounds in $N$ in the space $%
H_{\boldsymbol{x}_2}^s$ for $s\leq \frac34$. From this point of view, such
terms are irregular when measured in the $(\boldsymbol{x}_2,\boldsymbol{\xi}%
_2)$ coordinate frame, since the convergence, compactness, and even the
well-posedness of the limit equation, reside in $H_{\boldsymbol{x}_2}^{1+}$.
\end{remark}

\begin{proof}[Proof of Calculation \protect\ref{H:2}]
Carrying out the expectation in \eqref{E:LM19b}, 
\begin{align*}
\tilde f_{N, \pi}^{(k)}(\boldsymbol{x}_k, \boldsymbol{\xi}_k) \sim &
\prod_{j=1}^k \int_{y_j,w_j} \epsilon^{-3/2} \chi_j\left( \frac{x_j+\epsilon
\xi_j - y_j - 2t_j w_j}{\epsilon^{1/2}} \right) \\
& \qquad \qquad \bar\chi_j \left( \frac{x_{\pi(j)} -\epsilon \xi_{\pi(j)} -
y_j-2t_jw_j}{\epsilon^{1/2}} \right) e^{2iq_j \cdot w_j} G_Y(y_j) G_W(w_j)
\, dy_j \, dw_j
\end{align*}
Shift $y_j \mapsto y_j+2t_jw_j$ and substitute 
\begin{equation*}
G_Y(y_j+2t_jw_j) G_W(w_j) = G_Y \Big( \frac{y_j}{\sqrt{1+4t_j^2}} \Big) G_W %
\left[ \sqrt{1+4t_j^2} \left( w_j+ \frac{2t_jy_j}{1+4t_j^2} \right)\right]
\end{equation*}
Then replace $w_j \mapsto w_j - \frac{2t_jy_j}{1+4t_j^2}$ to obtain 
\begin{align*}
\tilde f_{N, \pi}^{(k)}(\boldsymbol{x}_k, \boldsymbol{\xi}_k) \sim &
\prod_{j=1}^k \int_{y_j,w_j} \epsilon^{-3/2} \chi_j\left( \frac{x_j+\epsilon
\xi_j - y_j }{\epsilon^{1/2}} \right) \bar\chi_j \left( \frac{x_{\pi(j)}
-\epsilon \xi_{\pi(j)}-y_j }{\epsilon^{1/2}} \right) \\
& \qquad \qquad \exp\left( -4i t_j \frac{q_j\cdot y_j}{1+4t_j^2} \right) G_Y %
\Big( \frac{y_j}{\sqrt{1+4t_j^2}} \Big) e^{2iq_j \cdot w_j} G_W \left[ w_j 
\sqrt{1+4t_j^2} \right] \, dy_j \, dw_j
\end{align*}
Carrying out the $w_j$ integral gives 
\begin{align*}
\tilde f_{N, \pi}^{(k)}(\boldsymbol{x}_k, \boldsymbol{\xi}_k) \sim &
\prod_{j=1}^k (1+4t_j^2)^{-3/2} \hat G_W \Big( \frac{2q_j}{\sqrt{1+4t_j^2}} %
\Big)\int_{y_j} \epsilon^{-3/2} \chi_j\left( \frac{x_j+\epsilon \xi_j - y_j 
}{\epsilon^{1/2}} \right) \\
& \qquad \qquad \bar\chi_j \left( \frac{x_{\pi(j)} -\epsilon
\xi_{\pi(j)}-y_j }{\epsilon^{1/2}} \right) \exp\left( -4i t_j \frac{q_j\cdot
y_j}{1+4t_j^2} \right) G_Y \Big( \frac{y_j}{\sqrt{1+4t_j^2}} \Big) \, dy_j
\end{align*}
Replacing $y_j = p_j+\epsilon^{1/2}z_j$ gives 
\begin{align*}
\tilde f_{N, \pi}^{(k)}(\boldsymbol{x}_k, \boldsymbol{\xi}_k) \sim
\prod_{j=1}^k (1+4t_j^2)^{-3/2} \hat G_W \Big( \frac{2q_j}{\sqrt{1+4t_j^2}} %
\Big)\int_{z_j} \chi_j( \epsilon^{1/2} q_j -z_j) \bar\chi_j(
-\epsilon^{1/2}q_j-z_j) & \\
\exp\left( -4i t_j \frac{q_j\cdot (p_j+\epsilon^{1/2}z_j)}{1+4t_j^2} \right)
G_Y \Big( \frac{p_j+\epsilon^{1/2}z_j }{\sqrt{1+4t_j^2}} \Big) \, dz_j &
\end{align*}
This leads to the approximation \eqref{E:LM19c}
\end{proof}

\subsection{Effect of collisions\label{S:CollisionEffects}}

If initial condition \eqref{E:LM05} with all $t_j=0$ evolves \emph{without}
interaction ($\phi=0$), the result is \eqref{E:LM05b}, leading to %
\eqref{E:LM19c}: 
\begin{equation}  \label{E:73-02}
\tilde f^{(k)}_N(\boldsymbol{x}_k, \boldsymbol{\xi}_k) = \sum_{\pi \in S_k}
\tilde f_{N,\pi}^{(k)}(\boldsymbol{x}_k, \boldsymbol{\xi}_k) = \sum_{\pi \in
S_k} \tilde g_{N,\pi}^{(k)}(\boldsymbol{p}_k^\pi, \boldsymbol{q}_k^\pi)
\end{equation}
where 
\begin{equation}  \label{E:73-01}
\tilde g_{N, \pi}^{(k)}(\boldsymbol{p}_k, \boldsymbol{q}_k) \sim
\prod_{j=1}^k (1+4t^2)^{-3/2} \hat G_W \Big( \frac{2q_j}{\sqrt{1+4t^2}} %
\Big) G_Y \Big( \frac{p_j}{\sqrt{1+4t^2}} \Big) \exp\left( -4i t \frac{%
q_j\cdot p_j}{1+4t^2} \right)
\end{equation}
In the collisionless case, \eqref{E:73-02} satisfies the linear BBGKY
hierarchy (with $A=0$, $B=0$). 
\begin{equation*}
i \partial_{t}\tilde f_{N}^{(k)}+ \nabla _{{\boldsymbol{x}}_{k}}\cdot
\nabla_{\boldsymbol{\xi}_k} \tilde f_{N}^{(k)}=0
\end{equation*}
This hierarchy decouples in $k$ and for each $k$, it is just linear
transport. 

We now look for an indication of how the evolution of $\tilde f_{N,\pi}(t, 
\boldsymbol{p}_k, \boldsymbol{q}_k)$ in time will be altered by $\phi \neq 0$%
. We know that $\tilde f_{N}^{(k)}$ satisfies the BBGKY hierarchy (where now 
$A\neq 0$ and $B\neq 0$) as given by \eqref{E:S401}. 
\begin{equation}  \label{E:73-04a}
i\partial _{t}\tilde f_{N}^{(k)}+ \nabla_{{\boldsymbol{x}}_{k}}\cdot \nabla_{%
\boldsymbol{\xi}_k} \tilde f_{N}^{(k)}=\epsilon ^{-1/2}\tilde A_{\epsilon
}^{(k)}\tilde f_{N}^{(k)}+N\epsilon^{-1/2}\tilde B_{\epsilon }^{(k+1)}\tilde
f_{N}^{(k+1)}
\end{equation}

We anticipate that as $N\to \infty$, all $\pi \neq \text{Id}$ terms in %
\eqref{E:73-02} vanish, leaving only the core term with $\pi=\text{Id}$.
Furthermore, the anticipated limiting form of BBGKY is the Boltzmann
hierarchy, in which only the composition of $B$ and $A$ in the Duhamel
expansion survive to give the collision operator: 
\begin{equation}  \label{E:73-03}
i\partial_t \tilde f^{(k)} + \nabla_{\boldsymbol{x}_k}\cdot \nabla_{%
\boldsymbol{\xi}_k} \tilde f^{(k)} = \tilde Q^{(k+1)} \tilde f^{(k+1)}
\end{equation}
where $\tilde Q^{(k+1)}$ is given by \eqref{E:S444}. Using the limiting form %
\eqref{E:73-03} of the equation \eqref{E:73-04a} on the finite $N$
functional form of $\tilde f_N^{(k)}$ as given by \eqref{E:73-02}, we deduce
a type of ``linearization" for the dynamics of $\tilde f_{N,\pi}$ for a
fixed $\pi \in S^k$, as follows. Assuming that $\tilde f_{N,\pi}$ only
interacts with the core term, and the core term can be approximated by its $%
N\to \infty$ limit $\tilde f^{(k)}$, we can write, for fixed $\pi\in S^k$: 
\begin{equation*}
\tilde f_N^{(k)} = \tilde f^{(k)} + \tilde f_{N,\pi}^{(k)}\,, \qquad \tilde
f_N^{(k+1)} = \tilde f^{(k+1)} + \tilde f_{N,\pi}^{(k)} \otimes \tilde
f^{(1)}
\end{equation*}
This leads to the perturbative equation 
\begin{equation}  \label{E:73-04}
i\partial_t \tilde f_{N,\pi}^{(k)} + \nabla_{\boldsymbol{x}_k} \cdot
\nabla_{ \boldsymbol{\xi}_k} \tilde f_{N,\pi}^{(k)} = \tilde Q^{(k+1)}
(\tilde f_{N,\pi}^{(k)} \otimes \tilde f^{(1)})
\end{equation}
Since the limiting collision operator is explicitly given by \eqref{E:S444}
and the form of the local Maxwellian is explicitly given by \eqref{E:LM19d},
we can compute that the effect of the Duhamel operator of the right-side of %
\eqref{E:73-04} on the dynamics of $\tilde f_{N,\pi}^{(k)}$. Written in $(%
\boldsymbol{x}_k, \boldsymbol{v}_k)$ coordinates, the first-order Duhamel
expression is 
\begin{equation}  \label{E:73-05}
\begin{aligned} f_{N,\pi}^{(k)}(t, \boldsymbol{x}_k, \boldsymbol{v}_k) &=
e^{-t \bds v_k \cdot \nabla_{\boldsymbol{x}_k} } f_{N,\pi}^{(k)}(0) \\ &
\qquad + \int_0^t e^{-(t-t') \bds v_k \cdot \nabla_{\boldsymbol{x}_k} }
Q^{(k+1)} ( e^{-t' \bds v_k \cdot \nabla_{\boldsymbol{x}_k} }
f_{N,\pi}^{(k)}(0) \otimes f^{(1)}(t^{\prime })) \, dt^{\prime }
\end{aligned}
\end{equation}
This expression is computable since its components consist of Gaussians. We
are more interested here however in explaining the origin of fluctuations in
the dynamics that give rise to perturbations of the symmetry in coordinates
in $\tilde f_{N,\pi}^{(k)}$. Suppose that instead of substituting %
\eqref{E:LM19d} into the Duhamel term, we use \eqref{E:LM17}-\eqref{E:LM18}
for $k=1$, $\pi=I$ (before the application of averaging in Calculation \ref%
{H:2}). In the case $k=1$, $\pi=I$, \eqref{E:LM17}-\eqref{E:LM18} reduce to
the following 
\begin{equation}  \label{E:73-06}
\begin{aligned} \tilde f_{N}^{(1)}(t', x_1, \xi_1) = \frac{1}{N}
\sum_{\sigma=1}^N \epsilon^{-3/2} \chi_\sigma\left( \frac{x_1+\epsilon \xi_1
- Y_\sigma - 2t_\sigma W_\sigma}{\epsilon^{1/2}} \right) &\\
\bar\chi_\sigma\left( \frac{x_1 - \epsilon \xi_1 - Y_\sigma - 2t_\sigma
W_\sigma}{\epsilon^{1/2}} \right) e^{2i\xi_jW_\sigma} & \end{aligned}
\end{equation}
where $\mathbb{E} t_\sigma = t^{\prime }$. The process is deterministic,
however, we are interested in averages (expected values) which are more
easily extracted from a (pseudo-)random model. Since the Duhamel term in %
\eqref{E:73-05} involves a linear transport propagator, the path of the
integral in time will meet the collection of wave packets in \eqref{E:73-06}
centered at $Y_\sigma + 2t_\sigma W_\sigma$ as $\sigma$ ranges over the full
collection $\{1, \ldots, N\}$, the $j$th particle ($1\leq j \leq k$) will
undergo collisions according to a Poisson process with rate $1/\epsilon$
along its linear path. For expositional simplicity, let us assume these
collisions occur at regularly spaced times -- every $\epsilon$ unit of time.
Writing in terms of characteristics, the linear path of the $j$th particle, $%
1\leq j \leq k$, without perturbation is $x_j(t) = x_j(0) + 2t v_j(0)$, but
with the perturbation (Duhamel term), the path is perturbed. Let us assume
that the effect of each collision on $x_j(t)$ is to randomly either raise or
lower the distance of $x_j(t)$ from $x_j(0) + 2t v_j(0)$, measured
orthogonal to $v_j(0)$, by $\epsilon$. Let 
\begin{equation*}
S_n = \sum_{k=1}^n U_k
\end{equation*}
where $\{U_k\}$ is a collection of independent standard normal random
variables so $S_n$ is a random walk with Gaussian increments. Our model is 
\begin{equation*}
|x_j(t) -x_j(0) - 2tv_j(0)| = |\epsilon S_{\lfloor t/\epsilon \rfloor}|
\end{equation*}
Then 
\begin{equation*}
\func{var} |x_j(t) -x_j(0) - 2tv_j(0)| = \epsilon^2 \lfloor t/\epsilon
\rfloor \approx \epsilon t
\end{equation*}
and thus the standard deviation of these fluctuations, or effective width of
the values around a pure linear trajectory, is $\sqrt{\epsilon t}$.

Moreover, a straightforward calculation shows that the expected number of
zero crossings of $S_{n}$ is, asymptotically $\sim 2\sqrt{n}/\pi $.\footnote{%
See the answer to Question \#1338097 on \url{https://math.stackexchange.com/}%
, for the calculation.}. With $n=\lfloor t/\epsilon \rfloor $, this is $\sim
\epsilon ^{-1/2}$ over a unit time interval. Said differently, the time
steps are of size $\epsilon $, although we cross over $0$ on average every $%
\sim \epsilon ^{1/2}$ units of time.

In this model, the $j$th particle position-velocity covariance fluctuates
around the value $-2t$ with effective width $O(\sqrt{\epsilon })$, but
revisits the exact value $-2t$ every $O(\epsilon ^{1/2})$ time. Recall that
the time shifts $t_{j}$ were inserted into \eqref{E:LM05} to allow the model
to reflect deviations from $C_{j}(t)=-2t$ that could vary from one particle
to the next. Although the appeal to the law of large numbers in Calculation
2, \eqref{E:LM19} should average over the values of $t_{j}$, we need to
account for the fact that the process is \emph{dynamical}. We can interpret
the role of randomness in the particle positions $Y_{j}$ and velocities $%
W_{j}$ to mean that they are randomly selected (sampled) initially (say at
time $0$), and the collection will then evolve in time deterministically
starting from this initial, randomly selected configuration. Then, evolving
forward deterministically in time, each particle suffers collisions
according to a pseudo-random process, such as the simplified one described
above. Thus we have left the $t_{j}$'s in \eqref{E:LM19c} rather than
replace them with an expectation and offer the model above as a way to
suggest that the proper physics could be captured, at the level of particle
densities $\tilde{f}_{N}^{(k)}$, by supposing that, for most times $t$, the
time offsets sastisfy $|t_{i}-t_{j}|=O(\epsilon ^{1/2})$, but for a set of
times $t$ of negligible measure we in fact have $t_{i}=t_{j}$. Moreover,
this set of times of negligible measure is $\epsilon ^{1/2}$ dense on the
timeline. Hence, we conclude condition (\ref{energy bound:high freq H^1
goodness}), along with everything else in \S \ref{s:intro cycle regularity}.

Looking backwards into the proof of Theorem \ref{thm:main}, the above
discussion might be a reason of the emergence of time irreversibility after
everything is finally well-defined and physical. When the whole particle
system returns to its initial state (recurrence) at $t_{r}$, then, as
indication, it is a quasi free symmetry event and $t_{r}\notin
E_{\varepsilon }$ and (\ref{energy bound:high freq H^1 goodness}) does not
happen (though this is not true the other way around), but as $\varepsilon $
tends to zero, the ``jitter" set $E_{\varepsilon }$ becomes dense and the
whole time line are symmetry strengthening events,\footnote{%
In EE, jitters are phase noises in the synchronizing clock, that is exactly
the cause of the $E_{\varepsilon }$ set here. Moreover, jitters in EE indeed
match the prediction here that they never go away, and increases as particle
number increase. (This is one of the reasons for better photolithography.)
One can always observe them directly on oscilloscopes as proof.} hence no
recurrences. Thus, the quantum model (dice) has indeed helped the time
irreversibility and matches \cite[Vol. III, paper 119]{BoltzmannCollection}.
(Of course, this needs more explanation and investigation.)

\section{Proof of Optimality: Well/Ill-posedness Separation of the Limit
Equation\label{sec:illposedness}}

\begin{theorem}
\label{thm:ill-posedness}The quantum Boltzmann equation (\ref{eqn:QBEquation}%
) is locally well-posed in $H_{x}^{s}L_{v}^{2,0+}$ for $s>1$, and ill-posed
in $H_{x}^{s}L_{v}^{2,0+}$ for $s<1.$Moreover, the solution constructed in $%
H_{x}^{s}L_{v}^{2,0+}$, for $s>1$, is nonnegative and in $L_{xv}^{1}$ if the
initial datum has the property.
\end{theorem}

\subsection{Well-posedness}

We prove a $C([0,T];H_{x}^{1+}H_{\xi }^{0+})$ local well-posedness theory
for (\ref{eqn:QBEquation}) on the $\left( x,\xi \right) $ side which is a $%
C([0,T];H_{x}^{1+}L_{v}^{2,0+})$ theory for (\ref{eqn:QBEquation}) on the $%
\left( x,v\right) $ side. That is, we construct a unique solution to (\ref%
{eqn:QBEquation}) in the format of 
\begin{equation*}
i\partial _{t}\tilde{f}+\nabla _{x}\cdot \nabla _{\xi }\tilde{f}=\sum_{\pm
}\pm \tilde{Q}^{\pm }(\tilde{f},\tilde{f})
\end{equation*}%
in the space $C([-T,T];H_{x}^{1+}H_{\xi }^{0+})$ on a time interval whose
length depends on the size of $\Vert \tilde{f}(0)\Vert _{H_{x}^{1+}H_{\xi
}^{0+}}$.

\begin{lemma}
\label{L:10estimates} For given $\tilde{g}$, $\tilde{h}$, consider $\tilde{f}
$ solving 
\begin{equation*}
i\partial _{t}\tilde{f}+\nabla _{x}\cdot \nabla _{\xi }\tilde{f}=\sum_{\pm
}\pm \tilde{Q}^{\pm }(\tilde{g},\tilde{h})
\end{equation*}%
with initial condition $\tilde{f}(0)$. Then 
\begin{align}
\hspace{0.3in}& \hspace{-0.3in}\Vert \langle \nabla _{x}\rangle ^{1+}\langle
\nabla _{\xi }\rangle ^{0+}\tilde{f}\Vert _{C([-T,T];L_{x\xi }^{2})\cap
L_{(-T,T)}^{2+}L_{x\xi }^{3-}}\lesssim \Vert \langle \nabla _{x}\rangle
^{1+}\langle \nabla _{\xi }\rangle ^{0+}\tilde{f}(0)\Vert _{L_{x\xi }^{2}}
\label{estimate:bilinear well-posedness} \\
& +T^{1/2+}\Vert \langle \nabla _{x}\rangle ^{1+}\langle \nabla _{\xi
}\rangle ^{0+}\tilde{g}(t,x,\xi )\Vert _{L_{(-T,T)}^{\infty }L_{x\xi
}^{2}}\Vert \langle \nabla _{x}\rangle ^{1+}\langle \nabla _{\xi }\rangle
^{0+}\tilde{h}(t,x,\xi )\Vert _{L_{(-T,T)}^{2+}L_{x\xi }^{3-}}  \notag
\end{align}%
More precisely, given a choice of $\delta >0$ in the operator $\langle
\nabla _{x}\rangle ^{1+\delta }\langle \nabla _{\xi _{1}}\rangle ^{\delta }$
on the left side, it is possible to select $\delta ^{\prime }>0$, $\delta
^{\prime \prime }>0$ so that the estimate holds with every instance of $%
L_{(-T,T)}^{2+}L_{x\xi }^{3-}$ taken to be $L_{(-T,T)}^{2+\delta ^{\prime
\prime }}L_{x\xi }^{3-\delta ^{\prime }}$ and every instance of the operator 
$\langle \nabla _{x}\rangle ^{1+}\langle \nabla _{\xi _{1}}\rangle ^{0+}$ on
the right side is $\langle \nabla _{x}\rangle ^{1+\delta }\langle \nabla
_{\xi _{1}}\rangle ^{\delta }$ (exactly the same $\delta >0$ as on the left
side). Moreover the pair $(2+\delta ^{\prime \prime },3-\delta ^{\prime })$
is Strichartz admissible.
\end{lemma}

\begin{proof}
The Duhamel form is 
\begin{equation*}
\tilde f(t,x,\xi) = \sum_{\pm} \pm \int_0^t e^{i(t-t^{\prime })
\nabla_x\cdot \nabla_\xi} \tilde Q^{\pm}(\tilde g(t^{\prime }), \tilde
h(t^{\prime })) \, dt^{\prime }
\end{equation*}
By the Strichartz estimate, 
\begin{align*}
\| \langle \nabla_x \rangle^{1+} \langle \nabla_\xi \rangle^{0+} \tilde f
\|_{ C([-T,T]; L_{x\xi}^2) \cap L_{(-T,T)}^{2+} L_{x\xi}^{3-} } &\lesssim \|
\langle \nabla_x \rangle^{1+} \langle \nabla_\xi \rangle^{0+} \tilde f(0)
\|_{L_{x\xi}^2} \\
&\qquad + \sum_\pm \| \langle \nabla_x \rangle^{1+} \langle \nabla_\xi
\rangle^{0+} \tilde Q^{\pm}(\tilde g, \tilde h) \|_{L_{(-T,T)}^1 L_{x\xi}^2}
\end{align*}
By Proposition \ref{P:general-fixed-time-2}, 
\begin{equation*}
\lesssim T^{1/2+} \| \langle \nabla_{x_1}\rangle^{1+} \langle \nabla_{\xi_1}
\rangle^{0+} \tilde g( t, x_1, \xi_1)\tilde h(t,x_1, \xi_2)
\|_{L^{2+}_{(-T,T)} L_{x_1}^2 (L_{\xi_2}^{3+}\cap L_{\xi_2}^{3-})L_{\xi_1}^2}
\end{equation*}
Recall that from the proof of Proposition \ref{P:general-fixed-time-2}, we
have the flexibility to use $L_{\xi_2}^{3+\omega}\cap
L_{\xi_2}^{3-\omega^{\prime }}$ for any $\omega>0$ and $\omega^{\prime }>0$
arbitrarily small (as long as they are both strictly positive). By the
fractional Leibniz rule in $x$, 
\begin{align*}
\lesssim & T^{1/2+}\| \langle \nabla_x \rangle^{1+} \langle \nabla_\xi
\rangle^{0+} \tilde g(t,x,\xi) \|_{ L_{(-T,T)}^\infty L_x^2 L_\xi^2 } \|
\tilde h(t,x,\xi) \|_{ L_{(-T,T)}^{2+} L_x^\infty (L_\xi^{3+}\cap
L_\xi^{3-}) } \\
& \quad + T^{1/2+} \| \langle \nabla_\xi \rangle^{0+} \tilde g(t,x,\xi) \|_{
L_{(-T,T)}^\infty L_x^{6+} L_\xi^2 } \| \langle \nabla_x \rangle^{1+} \tilde
h(t,x,\xi) \|_{ L_{(-T,T)}^{2+} L_x^{3-} (L_\xi^{3+}\cap L_\xi^{3-}) }
\end{align*}
For the two terms $\| \tilde h(t,x,\xi) \|_{ L_{(-T,T)}^{2+} L_x^\infty
(L_\xi^{3+}\cap L_\xi^{3-}) }$ and $\| \langle \nabla_\xi \rangle^{0+}
\tilde g(t,x,\xi) \|_{ L_{(-T,T)}^\infty L_x^{6+} L_\xi^2 }$, we bring the $%
x $-norm to the inside via Minkowski's integral inequality, and then apply
Sobolev in $x$: 
\begin{align*}
\lesssim & T^{1/2+}\| \langle \nabla_x \rangle^{1+} \langle \nabla_\xi
\rangle^{0+} \tilde g(t,x,\xi) \|_{ L_{(-T,T)}^\infty L_x^2 L_\xi^2 } \|
\langle \nabla_x \rangle^{1+} \tilde h(t,x,\xi) \|_{ L_{(-T,T)}^{2+}
(L_\xi^{3+}\cap L_\xi^{3-}) L_x^{3-} } \\
& \quad + T^{1/2+} \| \langle \nabla_x \rangle^{1+} \langle \nabla_\xi
\rangle^{0+} \tilde g(t,x,\xi) \|_{ L_{(-T,T)}^\infty L_\xi^2 L_x^2 } \|
\langle \nabla_x \rangle^{1+} \tilde h(t,x,\xi) \|_{ L_{(-T,T)}^{2+}
L_x^{3-} (L_\xi^{3+}\cap L_\xi^{3-}) }
\end{align*}
In the argument above, the H\"older exponent of $L_x^{3-}$ is chosen to
match exactly the H\"older exponent of $L_\xi^{3-}$. The $L_\xi^{3+}$ norms
are converted to the same $L_\xi^{3-}$ at the expense of $\langle \nabla_\xi
\rangle^{0+}$ via Sobolev. 
\begin{align*}
\lesssim & T^{1/2+}\| \langle \nabla_x \rangle^{1+} \langle \nabla_\xi
\rangle^{0+} \tilde g(t,x,\xi) \|_{ L_{(-T,T)}^\infty L_x^2 L_\xi^2 } \|
\langle \nabla_x \rangle^{1+} \langle \nabla_\xi \rangle^{0+} \tilde
h(t,x,\xi) \|_{ L_{(-T,T)}^{2+} L_{x\xi}^{3-} } \\
& \quad + T^{1/2+} \| \langle \nabla_x \rangle^{1+} \langle \nabla_\xi
\rangle^{0+} \tilde g(t,x,\xi) \|_{ L_{(-T,T)}^\infty L_\xi^2 L_x^2 } \|
\langle \nabla_x \rangle^{1+} \langle \nabla_\xi \rangle^{0+} \tilde
h(t,x,\xi) \|_{ L_{(-T,T)}^{2+} L_{x\xi}^{3-} }
\end{align*}
\end{proof}

Local well-posedness, namely, existence, uniqueness, and uniform continuity
of the datum to solution map, now follows from Lemma \ref{L:10estimates} by
the standard contraction argument. The solution we constructed is also a
strong solution as it is in $C([0,T];H_{x}^{1+}L_{v}^{2,0+})$ and is
nonnegative and in $L_{xv}^{1}$ if the initial datum has the property as we
will prove in \S \ref{sec:proof of L1}. However, it only solves (\ref%
{eqn:QBEquation}) almost everywhere in time in the sense that the
nonlinearity is defined a.e. in time. (An additional $H_{x}^{\frac{1}{2}%
+}L_{v}^{2,\frac{1}{2}+}$ condition will make the solution an everywhere in
time solution.)

\subsubsection{Nonnegativity and persistence of $L_{xv}^{1}\label{sec:proof
of L1}$}

\begin{lemma}[persistence of $H_x^2H_\protect\xi^2$]
\label{L:nn01} Suppose that the initial condition $\tilde f \in H_x^2H_\xi^2$%
. Then the unique solution constructed above in $C([-T,T];
H_x^{1+}H_\xi^{0+})$, where $T>0$ depends on the size of $\| \tilde f(0)
\|_{H_x^{1+}H_\xi^{0+}}$, in fact belongs also to $C([-T,T]; H_x^2 H_\xi^2)$
and this norm is controlled by the corresponding norm of the initial
condition.
\end{lemma}

\begin{proof}
This follows using the same estimates after derivatives are added to the
equation.
\end{proof}

\begin{lemma}[nonnegativity of high regularity solutions]
\label{L:nn02} If $f(0)\in H_x^2 L_v^{2,2}$ and $f(0) \geq 0$ pointwise,
then the corresponding solution satisfies $f(t) \geq 0$ pointwise for all $t$%
.
\end{lemma}

\begin{proof}
This one follows from the same argument in \cite{CDP21}.
\end{proof}

\begin{lemma}[$L^1_{xv}$ bounds of high regularity solutions]
\label{L:nn03} If $f(0) \in H_x^2 L_v^{2,2} \cap L_{x,v}^1$, then the
corresponding solution satisfies $f\in L_{(-T,T)}^\infty L_{x,v}^1$.
\end{lemma}

\begin{proof}
This proof does not need the nonnegativity. We will estimate the solution $%
f(t)$ in the Duhamel form. 
\begin{equation*}
f(t)=S(t)f(0)+\sum_{\pm }\pm \int_{0}^{t}S(t-t^{\prime })Q^{\pm
}(f(t^{\prime }),f(t^{\prime }))\,dt^{\prime }
\end{equation*}%
Applying the $L_{xv}^{1}$ norm, and using that this is preserved by the
linear propagator, 
\begin{equation*}
\Vert f(t)\Vert _{L_{(-T,T)}^{\infty }L_{xv}^{1}}\lesssim \Vert f(0)\Vert
_{L_{xv}^{1}}+\sum_{\pm }T\Vert Q^{\pm }(f(t),f(t))\Vert _{L_{xv}^{1}}
\end{equation*}%
Thus, it suffices to estimate $\Vert Q^{\pm }(f,f)\Vert _{L_{(-T,T)}^{\infty
}L_{xv}^{1}}$. To this end, first note that 
\begin{equation}
\Vert Q^{\pm }(f,f)\Vert _{L_{v}^{1}}\lesssim \Vert Q^{\pm }(f,f)\Vert
_{L_{v}^{2,2}}=\Vert \tilde{Q}^{\pm }(\tilde{f},\tilde{f})\Vert _{H_{\xi
}^{2}}  \label{E:nn06}
\end{equation}%
Recall 
\begin{align*}
\tilde{Q}_{\alpha ,\sigma }(\tilde{g},\tilde{h})(t,x,\xi
)=\int_{s=0}^{\infty }\int_{\zeta }e^{i(\sigma -\alpha )\xi \cdot \zeta
/2}e^{-2is\sigma |\zeta |^{2}/2}\hat{\phi}(-\zeta )\hat{\phi}(\zeta )& \\
\tilde{g}(t,x,\xi -s\zeta )\tilde{h}(t,x,s\zeta )\,ds\,d\zeta &
\end{align*}%
Applying the operator $(1-\Delta _{\xi })$ and differentiating under the
integral sign gives 
\begin{align*}
(1-\Delta _{\xi })\tilde{Q}_{\alpha ,\sigma }(\tilde{g},\tilde{h})(t,x,\xi
)=\int_{s=0}^{\infty }\int_{\zeta }e^{i(\sigma -\alpha )\xi \cdot \zeta
/2}e^{-2is\sigma |\zeta |^{2}/2}\hat{\phi}(-\zeta )\hat{\phi}(\zeta )& \\
(1+\frac{1}{4}(\sigma -\alpha )^{2}|\zeta |^{2}-2(\sigma -\alpha )\zeta
\cdot \nabla _{\xi }-\Delta _{\xi })\tilde{g}(t,x,\xi -s\zeta )\tilde{h}%
(t,x,s\zeta )\,ds\,d\zeta &
\end{align*}%
All of the extra powers of $\zeta $ that have been produced can be absorbed
by $\hat{\phi}$. Thus, Minkowski, we have 
\begin{equation*}
\Vert (1-\Delta _{\xi })\tilde{Q}^{\pm }(\tilde{g},\tilde{h})(t,x,\xi )\Vert
_{L_{\xi }^{2}}\lesssim \Vert \tilde{g}(t,x,\xi )\Vert _{H_{\xi
}^{2}}\int_{s=0}^{\infty }\int_{\zeta }\left\vert \langle \zeta \rangle ^{2}%
\hat{\phi}(-\zeta )\hat{\phi}(\zeta )\right\vert \left\vert \tilde{h}%
(t,x,s\zeta )\right\vert ds\,d\zeta
\end{equation*}%
H\"{o}lder in $\zeta $ like in the proof of Lemma \ref{Lem:Uniqueness III},%
\begin{equation*}
\Vert (1-\Delta _{\xi })\tilde{Q}^{\pm }(\tilde{g},\tilde{h})(t,x,\xi )\Vert
_{L_{\xi }^{2}}\lesssim \Vert \langle \zeta \rangle ^{2}\hat{\phi}(\zeta
)\Vert _{L^{3+}\cap L^{3-}}\Vert \tilde{g}(t,x,\xi )\Vert _{H_{\xi
}^{2}}\Vert \tilde{h}(t,x,\xi )\Vert _{L_{\xi }^{3+}\cap L_{\xi }^{3-}}
\end{equation*}

Applying the $L_{x}^{1}$ norm and using Cauchy-Schwarz in $x$, 
\begin{equation*}
\Vert \tilde{Q}^{\pm }(\tilde{g},\tilde{h})(t,x,\xi )\Vert _{L_{x}^{1}H_{\xi
}^{2}}\lesssim \Vert \tilde{g}(t,x,\xi )\Vert _{L_{x}^{2}H_{\xi }^{2}}\Vert 
\tilde{h}(t,x,\xi )\Vert _{L_{x}^{2}H_{\xi }^{2}}
\end{equation*}%
where we have now absorbed $\Vert \langle \zeta \rangle ^{2}\hat{\phi}(\zeta
)\Vert _{L^{3+}\cap L^{3-}}$ into the implicit constant. Returning to %
\eqref{E:nn06}, 
\begin{equation*}
\Vert Q^{\pm }(f,f)(t,x,v)\Vert _{L_{x}^{1}L_{v}^{1}}\lesssim \Vert
f(t,x,v)\Vert _{L_{x}^{2}L_{v}^{2,2}}\Vert f(t,x,v)\Vert
_{L_{x}^{2}L_{v}^{2,2}}
\end{equation*}
\end{proof}

Combining Lemmas \ref{L:nn02}, \ref{L:nn03}, we obtain

\begin{corollary}
\label{C:nn04} Suppose that $f(0) \in H_x^2L_v^{2,2} \cap L_{x,v}^1$ and $%
f(0) \geq 0$. Then the corresponding solution satisfies $f\in
L_{(-T,T)}^\infty L_{x,v}^1$, $f(t) \geq 0$, and 
\begin{equation}  \label{E:nn07}
\int_{x,v} f(t,x,v) \, dx \, dv = \int_{x,v} f(0,x,v) \, dx \, dv
\end{equation}
so that the $L^1_{xv}$ norm is in fact preserved in time.
\end{corollary}

\begin{proof}
Since the quantity $\iint_{x,v} f(t,x,v) \, dx \, dv$ is defined for all
time by Lemma \ref{L:nn03}, it is meaningful to compute 
\begin{equation*}
\partial_t \iint_{x,v} f(t,x,v) \, dx \, dv =0
\end{equation*}
by substituting the equation and using that the integral of the gain term
matches the integral of the loss term. Thus \eqref{E:nn07} holds.
\end{proof}

Now we can use the continuity of the data-to-solution map and Lemma \ref%
{L:nn01} as follows: Suppose that $f(0)\in H_{x}^{1+}L_{v}^{2,0+}\cap
L_{xv}^{1}$ and $f(0)\geq 0$. Approximate this initial condition in the
space $H_{x}^{1+}L_{v}^{2,0+}\cap L_{xv}^{1}$ by a sequence such that for
each $n$, $f_{n}(0)\in H_{x}^{2}L_{v}^{2,2}$ and $f_{n}(0)\geq 0$. The
continuity of the data-to-solution map implies that $f_{n}\rightarrow f$ in $%
C([-T,T];H_{x}^{1+}L_{v}^{2,0+})$. By Corollary \ref{C:nn04}, applied to
each $f_{n}$, we have each $f_{n}(t)\geq 0$ and 
\begin{equation*}
\forall \;n\,,\qquad \Vert f_{n}\Vert _{L_{(-T,T)}^{\infty
}L_{xv}^{1}}=\Vert f_{n}(0)\Vert _{L_{xv}^{1}}\lesssim \Vert f(0)\Vert
_{L_{xv}^{1}}
\end{equation*}%
Now $f_{n}\rightarrow f$ in $C([-T,T];H_{x}^{1+}L_{v}^{2,0+})$ implies that
for each $t$, $f_{n}(t)\rightarrow f(t)$ in $L_{xv}^{2}$, from which it
follows that there is a subsequence (depending on $t$, although this is not
a problem) such that $f_{n_{k}}(t,x,v)\rightarrow f(t,x,v)$ for pointwise
a.e. $(x,v)\in \mathbb{R}^{6}$. Since this is a nonnegative sequence, it
follows from Fatou's lemma that 
\begin{equation*}
\iint_{x,v}f(t)\,dx\,dv\leq \liminf_{k\rightarrow \infty
}\iint_{x,v}f_{n_{k}}(t)\,dx\,dv\lesssim \iint_{x,v}f(0)\,dx\,dv
\end{equation*}

\subsection{Ill-posedness}

We actually find the following result of ill-posedness.

\begin{lemma}
\label{lem:illposedness}Given any $s_{1}>0$ and $s<1$, the quantum Boltzmann
equation (\ref{eqn:QBEquation}) is ill-posed in $H_{x}^{s}L_{v}^{2,s_{1}}$
i.e. as long as the $x$-derivative is below $1$, ill-posedness persists even
with high $v$-weights.
\end{lemma}

The mechanism of Lemma \ref{lem:illposedness} was first discovered in \cite%
{CH10}. It can be described as the following. While it is universally known
that the gain term is better than the loss term, it was unknown that there
is a regularity gap between the optimal estimates on the gain term and the
loss term such that the gain term cannot cancel the loss term at all. The
``bad" solutions we consider are mainly maximizers of estimate (\ref%
{estimate:bilinear well-posedness}) in the loss term of the collision
operator, while other parts -- the gain term and the free term -- in
estimate (\ref{estimate:bilinear well-posedness}) in fact satisfy better
estimates with lower regularity. That is, in a Duhamel iteration, the loss
term applied to the ``bad" solutions will stay around the same size while
the gain term applied to the ``bad" solutions will become smaller. Hence,
putting in the maximizers of the loss term is like solving (\ref%
{eqn:QBEquation}) with only the loss term which drives down the amplitude of
the solution exponentially fast, and hence creates ill-posedness, in the
sense that, there is a family of norm deflation solutions and thus the datum
to solution map is not uniformly continuous.

We provide a construction of the approximate ``bad" solution $f_{a}$ and a
formal calculation demonstrating the ill-posedness. For the remaining
perturbation argument proving that a small correction $f_{c}$ exists such
that $f=f_{a}+f_{c}$ exactly solves (\ref{eqn:QBEquation}) and still
exhibits ill-posedness behavior, we refer readers to \cite{CH10,CSZ3}.

Fix a $s$ with $s<1$. Let the dyadic parameters satisfy the relationship. 
\begin{equation*}
0<N\ll \max (M_{1},M_{2})^{-1}\ll 1\,,\qquad M_{1}\geq 1\,,\qquad M_{2}\gg 1
\end{equation*}%
We will consider $\left\vert \eta _{2}\right\vert \sim M_{2}$ and $%
\left\vert v_{2}\right\vert \sim N_{2}$, with $M_{2}\gg 1$ and $N_{2}\gg 1$
dyadic. On the unit sphere, lay down a grid of $J\sim M_{2}^{2}N_{2}^{2}$
points $\{e_{j}\}_{j=1}^{J}$, where the points $e_{j}$ are roughly equally
spaced and each have their own neighborhood of unit-sphere surface area $%
\sim M_{2}^{-1}N_{2}^{-1}$. Let $P_{e_{j}}$ denote the orthogonal projection
onto the 1D subspace spanned by $e_{j}$ and $P_{e_{j}}^{\perp }$ denote the
orthogonal projection onto the 2D subspace $\func{span}\{e_{j}\}^{\perp }$.
We write 
\begin{equation*}
f(x,v,t)\approx \frac{M_{1}^{\frac{3}{2}-s}}{N^{3/2}}\chi (M_{1}x)\hat{\chi}(%
\frac{v}{N})
\end{equation*}%
and 
\begin{equation*}
g(x,v_{2},t)\approx \frac{M_{2}^{1-s}}{N_{2}^{2+s_{1}}}\sum_{j=1}^{J}\chi
(M_{2}P_{e_{j}}^{\perp }x)\chi (\frac{P_{e_{j}}x}{N_{2}})\hat{\chi}%
(M_{2}P_{e_{j}}^{\perp }v_{2})\hat{\chi}(\frac{P_{e_{j}}v_{2}}{N_{2}})
\end{equation*}%
whose $H_{x}^{s}L_{v}^{2,s_{1}}$ norms are $O(1)$ and where the $v_{2}$ in
the definition of $g$ is to remind us the $v$ integration in the loss term.
In the $j-$sum over $J\sim M_{2}^{2}N_{2}^{2}$ terms inside $g$, the
velocity supports are almost disjoint and the square of the sum is
approximately the sum of the squares. As mentioned before, $f$ and $g$ are
actually maximizers for the loss term bilinear estimate at critical
regularity while the gain term satisfies better estimates. So we expect a
small gain term minus a large loss term behavior.

For the loss term, we compute $Q^{-}(f,g)$ for which we use the $%
1/\left\langle u\right\rangle $ approximation,%
\begin{equation*}
Q^{-}(f,g)\approx f\int_{\mathbb{R}^{3}}\frac{g(x,v_{2},t)}{\left\langle
v-v_{2}\right\rangle }\,dv_{2}
\end{equation*}%
Notice that $f$'s $v$ support is of size $N$ which is small and hence, $%
\left\langle v-v_{2}\right\rangle \sim N_{2}.$ Carrying out the integral for
the bump functions, 
\begin{eqnarray*}
Q^{-}(f,g) &\approx &f\sum_{j=1}^{J}\chi (M_{2}P_{e_{j}}^{\perp }x)\chi (%
\frac{P_{e_{j}}x}{N_{2}})\frac{M_{2}^{1-s}}{N_{2}^{2+s_{1}}}\frac{1}{N_{2}}%
\frac{N_{2}}{M_{2}^{2}} \\
&=&f\frac{1}{M_{2}^{1+s}N_{2}^{2+s_{1}}}\sum_{j=1}^{J}\chi
(M_{2}P_{e_{j}}^{\perp }x)\chi (\frac{P_{e_{j}}x}{N_{2}}) \\
&=&f\text{ }M_{2}^{1-s}N_{2}^{-s_{1}}\chi (M_{2}x)
\end{eqnarray*}%
A prototype approximate solution suggested by the formal Duhamel iteration
of $f+g$ is then 
\begin{equation*}
f_{a}(x,v,t)\approx \exp (-M_{2}^{1-s}N_{2}^{-s_{1}}\chi (M_{2}x)t)\cdot f+g
\end{equation*}%
which is just $f+g$ above with $f$ preceeded by an exponentially decaying
factor in time. For a fixed $s_{1}$ and $M_{1}>>N_{2}$, when $s<1$, the size
of the exponential term changes substantially on the short time scale $\sim
M_{2}^{s-1}N_{2}^{s_{1}^{-1}}\ll 1$.

Let us now set

\begin{equation*}
M=M_{1}=M_{2}\gg 1\,,\qquad N=M^{-1}\ll 1
\end{equation*}%
then 
\begin{eqnarray*}
f_{a}(x,v,t) &=&M^{3-s}\chi (Mx)\hat{\chi}(Mv)\exp
(-M^{1-s}N_{2}^{-s_{1}}\chi (Mx)t) \\
&&+\frac{M^{1-s}}{N_{2}^{2+s_{1}}}\sum_{j=1}^{J}\chi (MP_{e_{j}}^{\perp
}x)\chi (\frac{P_{e_{j}}x}{N_{2}})\hat{\chi}(MP_{e_{j}}^{\perp }v)\hat{\chi}(%
\frac{P_{e_{j}}v}{N_{2}})
\end{eqnarray*}%
whose $L_{v}^{2,s_{1}}H_{x}^{s_{0}}$ norm for any $0<s\leqslant 1$, $%
s_{0}\geqslant 0$ is%
\begin{equation}
\Vert f_{a}\Vert _{L_{v}^{2,s_{1}}H_{x}^{s_{0}}}\sim M^{s_{0}-s}\exp
[-M^{1-s}N_{2}{}^{-s_{1}}t]\langle M^{1-s}N_{2}{}^{-s_{1}}t\rangle
^{s_{0}}+M^{s_{0}-s}  \label{E:ni01}
\end{equation}%
Thus, if we let 
\begin{equation}
0<s<1\,,\qquad s_{0}=s-\frac{\ln \ln M}{\ln M}\,,\qquad T_{\ast }=-\frac{%
\delta }{M^{1-s}N_{2}{}^{-s_{1}}}\ln M\leq t\leq 0  \label{E:rTdef}
\end{equation}%
then at the endpoints of the interval $\left[ T_{\ast },0\right] $: 
\begin{equation*}
\Vert f_{a}(T_{\ast })\Vert _{L_{v}^{2,s_{1}}H_{x}^{s_{0}}}\sim M^{\delta
}\gg 1\text{, }\Vert f_{a}(0)\Vert _{L_{v}^{2,s_{1}}H_{x}^{s_{0}}}\leq \frac{%
1}{\ln M}\ll 1\,,\qquad
\end{equation*}%
Note that, as $M\nearrow \infty $, $s_{0}\nearrow s$, and this approximate
solution, in $L_{v}^{2,s_{1}}H_{x}^{s_{0}}$, starts very small in $%
L_{v}^{2,s_{1}}H_{x}^{s_{0}}$ at time $0$, and rapidly inflates at time $%
T_{\ast }<0$ to large size in $L_{v}^{2,s_{1}}H_{x}^{s_{0}}$ backwards in
time. By considering the same approximate solution starting at $T_{\ast }<0$
and evolving forward to time $0$, we have an approximate solution that
starts large and deflates to a small size in a very short period of time.


\begin{thebibliography}{99}
\bibitem{AMUXY} R. Alexandre, Y. Morimoto, S. Ukai, C.-J. Xu \& T. Yang, 
\emph{Global existence and full regularity of the Boltzmann equation without
angular cutoff,} Commun. Math. Phys. \textbf{304} (2011) 513--581.

\bibitem{A2017} V. Ardourel, \emph{Irreversibility in the Derivation of the
Boltzmann Equation}, Found. Phys. \textbf{47} (2017), 471--489.

\bibitem{ACI} L. Arkeryd, S. Caprino \& N. Ianiro, \emph{The Homogeneous
Boltzmann Hierarchy and Statistical Solutions to the Homogeneous Boltzmann
Equation}, J. Stat. Phys \textbf{63} (1991), 345--361.

\bibitem{Ar} D. Arsenio, \emph{On the global existence of mild solutions to
the Boltzmann equation for small data in }$L^{D}$\emph{, }Commun. Math.
Phys. \textbf{302} (2011), 453--476.

\bibitem{BE1} M. Beals, \emph{Self-spreading and strength of singularities
for solutions to semilinear wave equations}, Ann. of Math. (2) \textbf{118}
(1983), 187--214.

\bibitem{BCEP1} D. Benedetto, F. Castella, R. Esposito, and M. Pulvirenti, 
\emph{Some considerations on the derivation of the nonlinear quantum
boltzmann equation}, J. Stat. Phys. \textbf{116 }(2004), 381--410.

\bibitem{BCEP2} D. Benedetto, F. Castella, R. Esposito, and M. Pulvirenti, 
\emph{On the weak-coupling limit for bosons and fermions}, Math. Mod. Meth.
Appl. Sci. \textbf{15 }(2005), 1811--1843.

\bibitem{BCEP4} D. Benedetto, F. Castella, R. Esposito, and M. Pulvirenti, 
\emph{From the N-body Schroedinger equation to the quantum Boltzmann
equation: a term-by-term convergence result in the weak coupling regime},
Commun. Math. Phys. \textbf{277} (2008), 1-44.

\bibitem{BoltzmannCollection} L. Boltzmann, Wissenschaftliche Abhandlungen
Vol. I, II, and III. F. Hasen\"{o}hrl (ed.) Leipzig 1909. Reissued New York:
Chelsea, 1969.

\bibitem{B1926} M. Born, \emph{Zur Quantenmechanik der Sto}$\beta $\emph{org%
\"{a}nge}, Zeitschrift f\"{u}r Physik \textbf{37} (1926), 863-867.

\bibitem{B1} J. Bourgain, \emph{Fourier transform restriction phenomena for
certain lattice subsets and applications to nonlinear evolution equations},
Parts I, II, Geometric and Funct. Anal. \textbf{3} (1993), 107--156,
209--262.

\bibitem{BBGL14} J. Bennett, N. Bez, S. Guti\'{e}rrez, \& S. Lee, \emph{On
the Strichartz estimates for the kinetic transport equation, }Comm. Partial
Differential Equations \textbf{39} (2014), 1821--1826.

\bibitem{BP01} N. Bournaveas \& B. Perthame, \emph{Averages over spheres for
kinetic transport equations; hyperbolic Sobolev spaces and Strichartz
inequalities}, J. Math. Pures Appl. \textbf{80} (2001), 517--534.

\bibitem{CC} E. C\'{a}rdenas \& T. Chen, \emph{Quantum Boltzmann dynamics
and bosonized particle-hole interactions in fermion gases},
arXiv:2306.03300, 70pp.

\bibitem{CIP} C. Cercignani, R. Illner, \& M. Pulvirenti, \emph{The
Mathematical Theory of Dilute Gases}, Applied Mathematical Sciences (AMS,
volume 106).

\bibitem{CDP19a} T. Chen, R. Denlinger, \& N. Pavlovi\'{c}, \emph{Local
well-posedness for Boltzmann's equation and the Boltzmann hierarchy via
Wigner transform}, Comm. Math. Phys. \textbf{368} (2019), 427--465.

\bibitem{CDP19b} T. Chen, R. Denlinger, \& N. Pavlovi\'{c}, \emph{Moments
and regularity for a Boltzmann equation via Wigner transform}, Discrete
Contin. Dyn. Syst. \textbf{39} (2019), 4979--5015.

\bibitem{CDP21} T. Chen, R. Denlinger, \& N. Pavlovi\'{c},\emph{\ Small data
global well-posedness for a Boltzmann equation via bilinear spacetime
estimates}. Arch. Ration. Mech. Anal. \textbf{240} (2021), 327--381.

\bibitem{CHPS} T. Chen, C. Hainzl, N. Pavlovi\'{c}, \& R. Seiringer, \emph{%
Unconditional Uniqueness for the Cubic Gross-Pitaevskii Hierarchy via
Quantum de Finetti}, Commun. Pure Appl. Math. \textbf{68} (2015), 1845-1884.

\bibitem{TCMH} T. Chen \& M. Hott, \emph{On the emergence of quantum
Boltzmann fluctuation dynamics near a Bose-Einstein condensate}, J. Stat.
Phys. \textbf{190} (2023), 85.

\bibitem{CP2} T. Chen \& N. Pavlovi\'{c}, \emph{The Quintic NLS as the Mean
Field Limit of a Boson Gas with Three-Body Interactions}, J. Funct. Anal. 
\textbf{260} (2011), 959--997.

\bibitem{CP5} T. Chen \& N. Pavlovi\'{c}, \emph{Derivation of the cubic NLS
and Gross-Pitaevskii hierarchy from manybody dynamics in }$d=3$\emph{\ based
on spacetime norms}, Ann. H. Poincare, \textbf{15} (2014), 543 - 588.

\bibitem{C2} X. Chen, \emph{Collapsing Estimates and the Rigorous Derivation
of the 2d Cubic Nonlinear Schr\"{o}dinger Equation with Anisotropic
Switchable Quadratic Traps}, J. Math. Pures Appl. \textbf{98} (2012),
450--478.

\bibitem{C3} X. Chen, \emph{On the Rigorous Derivation of the 3D Cubic
Nonlinear Schr\"{o}dinger Equation with A Quadratic Trap}, Arch. Rational
Mech. Anal. \textbf{210 }(2013), 365-408.

\bibitem{CY} X. Chen \& Y. Guo, \emph{On the Weak Coupling Limit of Quantum
Many-body Dynamics and the Quantum Boltzmann Equation}, Kinet. Relat. Models 
\textbf{8} (2015), 443-465.

\bibitem{CHe} X. Chen \& L. He, \emph{The longtime dynamics of the quantum
Boltzmann Equation}, in preparation.

\bibitem{CH1} X. Chen \& J. Holmer, \emph{On the Rigorous Derivation of the
2D Cubic Nonlinear Schr\"{o}dinger Equation from 3D Quantum Many-Body
Dynamics}, Arch. Rational Mech. Anal. \textbf{210 }(2013), 909-954.

\bibitem{CH2} X. Chen \& J. Holmer, \emph{On the Klainerman-Machedon
Conjecture of the Quantum BBGKY Hierarchy with Self-interaction}, J. Eur.
Math. Soc. (JEMS) \textbf{18} (2016), 1161-1200.

\bibitem{CH3} X. Chen \& J. Holmer, \emph{Focusing Quantum Many-body
Dynamics: The Rigorous Derivation of the 1D Focusing Cubic Nonlinear Schr%
\"{o}dinger Equation,} Arch. Rational Mech. Anal. \textbf{221} (2016),
631-676.

\bibitem{CH4} X. Chen \& J. Holmer, \emph{Focusing Quantum Many-body
Dynamics II: The Rigorous Derivation of the 1D Focusing Cubic Nonlinear Schr%
\"{o}dinger Equation from 3D,} Analysis \& PDE \textbf{10} (2017), 589-633.

\bibitem{CH5} X. Chen \& J. Holmer, \emph{Correlation structures, Many-body
Scattering Processes and the Derivation of the Gross-Pitaevskii Hierarchy,}
Int. Math. Res. Notices \textbf{2016}, 3051-3110.

\bibitem{CH6} X. Chen \& J. Holmer, \emph{The Rigorous Derivation of the 2D
Cubic Focusing NLS from Quantum Many-body Evolution}, Int. Math. Res.
Notices \textbf{2017}, 4173--4216.

\bibitem{CH7} X. Chen \& J. Holmer, \emph{The Derivation of the
Energy-critical NLS from Quantum Many-body Dynamics}, Invent. Math. \textbf{%
217} (2019), 433-547.

\bibitem{CH8} X. Chen \& J. Holmer, \emph{The Unconditional Uniqueness for
the Energy-critical Nonlinear Schr\"{o}dinger Equation on }$\mathbb{T}^{4}$,
Forum Math. Pi. \textbf{10 }(2022), e3 1-49.

\bibitem{CH9} X. Chen \& J. Holmer, \emph{Quantitative Derivation and
Scattering of the 3D Cubic NLS in the Energy Space}, Ann. PDE \textbf{8}
(2022) Article 11 1-39.

\bibitem{CH10} X. Chen \& J. Holmer, \emph{Well/ill-posedness bifurcation
for the Boltzmann equation with constant collision kernel},
arXiv:2206.11931, 31pp.

\bibitem{CSZ} X. Chen, S. Shen, \& Z. Zhang, \emph{The Unconditional
Uniqueness for the Energy-supercritical NLS}, Ann. PDE \textbf{8} (2022)
Article 14 1-82.

\bibitem{CSZ1} X. Chen, S. Shen, \& Z. Zhang, \emph{Quantitative derivation
of the Euler-Poisson equation from quantum many-body dynamics}, Peking
Mathematical Journal, 69pp. DOI: 10.1007/s42543-023-00065-5.

\bibitem{CSZ2} X. Chen, S. Shen, \& Z. Zhang, \emph{On the mean-field and
semiclassical limit from quantum N-body dynamics}, arXiv:2304.03447,40pp.

\bibitem{CSZ3} X. Chen, S. Shen, \& Z. Zhang, \emph{Well/Ill-posedness of
the Boltzmann Equation with Soft Potential, }arXiv:2310.05042, 47pp.

\bibitem{CSZ4} X. Chen, S. Shen, \& Z. Zhang, \emph{Sharp Global
Well-posedness and Scattering of the Boltzmann Equation, }arXiv:2311.02008,
42pp.

\bibitem{CSWZ} X. Chen, S. Shen, J. Wu, \& Z. Zhang, \emph{The derivation of
the compressible Euler equation from quantum many-body dynamics}, Peking
Mathematical Journal, 56pp. DOI: 10.1007/s42543-023-00066-4.

\bibitem{CPU} X. Chen \& P. Smith, \emph{On the Unconditional Uniqueness of
Solutions to the Infinite Radial Chern-Simons-Schr\"{o}dinger Hierarchy},
Analysis \& PDE \textbf{7} (2014), 1683-1712.



\bibitem{DL89} R. DiPerna \& P.-L. Lions, \emph{On the Cauchy problem for
Boltzmann equations: global existence and weak stability}, Ann. of Math. (2) 
\textbf{130} (1989), 321--366.

\bibitem{DHWY} R. Duan, F. Huang, Y. Wang \& T. Yang, \emph{Global
well-posedness of the Boltzmann equation with large amplitude initial data},
Arch. Ration. Mech. Anal. \textbf{225} (2017), 375--424.

\bibitem{Duan1} R. Duan, S. Liu, \& J. Xu, \emph{Global Well-Posedness in
Spatially Critical Besov Space for the Boltzmann Equation}, Arch. Ration.
Mech. Anal. \textbf{220} (2016) 711--745.

\bibitem{ErdosFock} L. Erd\"{o}s, M. Salmhofer and H. T. Yau, \emph{On the
quantum Boltzmann equation}, J. Stat. Phys. \textbf{116} (2004) 367--380.
MR2083147

\bibitem{E-S-Y2} L. Erd\"{o}s, B. Schlein, and H. T. Yau, \emph{Derivation
of the Cubic non-linear Schr\"{o}dinger Equation from Quantum Dynamics of
Many-body Systems}, Invent. Math. \textbf{167} (2007), 515--614.

\bibitem{E-S-Y5} L. Erd\"{o}s, B. Schlein, and H. T. Yau, \emph{Rigorous
Derivation of the Gross-Pitaevskii Equation with a Large Interaction
Potential}, J. Amer. Math. Soc. \textbf{22} (2009), 1099-1156.

\bibitem{E-S-Y3} L. Erd\"{o}s, B. Schlein, and H. T. Yau, \emph{Derivation
of the Gross-Pitaevskii Equation for the Dynamics of Bose-Einstein Condensate%
}, Annals Math. \textbf{172} (2010), 291-370.{}

\bibitem{EMY} R. Esposito, R. Marra, H. T. Yau, \emph{Navier-Stokes
equations for stochastic particle systems on the lattice}, Comm. Math. Phys. 
\textbf{182} (1996), 395-456.

\bibitem{SaintRaymond} I. Gallagher, L. Saint-Raymond, and B. Texier, \emph{%
"From Newton to Boltzmann: Hard Spheres and Short-range Potentials", }Z\"{u}%
rich Lectures in Advanced Mathematics. European Mathematical Society (EMS), Z%
\"{u}rich, 2013. xii+137 pp. MR3157048

\bibitem{GSS} P. Gressman, V. Sohinger, \& G. Staffilani, \emph{On the
Uniqueness of Solutions to the Periodic 3D Gross-Pitaevskii Hierarchy, }J.
Funct. Anal. \textbf{266} (2014), 4705--4764.

\bibitem{GM1961} S. R. de Groot \& P. Mazur, \emph{Non-equilibrium
thermodynamics}. Amsterdam: North-Holland 1961.

\bibitem{H1} L. He \& J. Jiang, \emph{On the Cauchy problem for the cutoff
Boltzmann equation with small initial data}, arXiv:2203.10756, 25pp.

\bibitem{HLPZ} L. He, X. Lu, M. Pulvirenti, \& Y. Zhou, \emph{On
semi-classical limit of spatially homogeneous quantum Boltzmann equation:
asymptotic expansion}\textbf{, }arXiv:2309.00891, 32pp.

\bibitem{HS1} S. Herr \& V. Sohinger, \emph{The Gross-Pitaevskii Hierarchy
on General Rectangular Tori}, Arch. Rational Mech. Anal., \textbf{220}
(2016), 1119-1158.

\bibitem{HS2} S. Herr \& V. Sohinger, \emph{Unconditional Uniqueness Results
for the Nonlinear Schr\"{o}dinger Equation}, Commun. Contemp. Math. \textbf{%
21} (2019), 1850058.

\bibitem{HTX} Y. Hong, K. Taliaferro, Z. Xie, \emph{Uniqueness of solutions
to the 3D quintic Gross-Pitaevskii hierarchy}, J. Functional Analysis 
\textbf{270} (2016), no. 1, 34--67.

\bibitem{HZ} W. Huang \& W. Zhang, \emph{Nonperturbative renormalization of
quantum thermodynamics from weak to strong couplings}, Phys. Rev. Research 4
(2022), 023141.

\bibitem{Hu1} N. M. Hugenholtz, \emph{Derivation of the Boltzmann Equation
for a Fermi Gas,} J. Stat. Phys. \textbf{32 }(1983), 231--254.

\bibitem{Kato} T. Kato, \emph{On nonlinear Schr\"{o}dinger equations. II. H}$%
^{\emph{s}}$\emph{-solutions and unconditional well-posedness}, J. Anal.
Math. \textbf{67} (1995), 281--306.

\bibitem{KT98} M. Keel \& T. Tao, \emph{Endpoint Strichartz estimates},
Amer. J. Math. \textbf{120} (1998), 955--980.

\bibitem{KSS} K. Kirkpatrick, B. Schlein and G. Staffilani, \emph{Derivation
of the Two Dimensional Nonlinear Schr\"{o}dinger Equation from Many Body
Quantum Dynamics}, Amer. J. Math. \textbf{133} (2011), 91-130.

\bibitem{KMRemarks} S. Klainerman \& M. Machedon, \emph{Remark on
Strichartz-type inequalities} \emph{(Appendices by J. Bourgain and D. Tataru)%
}, IMRN, \textbf{1996}, 201--220.

\bibitem{KM} S. Klainerman \& M. Machedon, \emph{Space-time estimates for
null forms and the local existence theorem}, Comm. Pure Appl. Math. \textbf{%
46} (1993), 1221--1268.

\bibitem{KM1} S. Klainerman \& M. Machedon, \emph{On the Uniqueness of
Solutions to the Gross-Pitaevskii Hierarchy}, Commun. Math. Phys. \textbf{%
279 }(2008), 169-185.

\bibitem{Lanford} O. E. Lanford III, \emph{Time Evolution of Large Classical
Systems}, Lecture Notes in Physics, \textbf{38} (1975), 1--111. MR0479206

\bibitem{LL1987} L.D. Landau \& E.M. Lifshitz, Fluid Mechanics 3rd Edition,
Oxford: Butterworth-Heinemann 1987.

\bibitem{L1999} J. L. Lebowitz, \emph{Statistical mechanics: A selective
review of two central issues}. Reviews of Modern Physics, \textbf{71}
(1999), S346--S357.

\bibitem{MS} Y. Morimoto \& S. Sakamoto, \emph{Global solutions in the
critical Besov space for the non-cutoff Boltzmann equation}, J. Differ. Equ. 
\textbf{261} (2016) 4073-4134.

\bibitem{GS11} P. Gressman \& R. Strain, \emph{Global classical solutions of
the Boltzmann equation without angular cut-off}, J. Amer. Math. Soc. \textbf{%
24} (2011),771--847.

\bibitem{GP1} F. Golse and T. Paul, \emph{Mean-field and classical limit for
the N-body quantum dynamics with Coulomb interaction}, Comm. Pure Appl.
Math., 2021.

\bibitem{Yan1} Y. Guo, \emph{Classical solutions to the Boltzmann equation
for molecules with an angular cutoff,} Arch. Ration. Mech. Anal. \textbf{169}
(2003), 305--353.

\bibitem{Ov11} E. Y. Ovcharov, \emph{Strichartz estimates for the kinetic
transport equation}, SIAM J. Math. Anal. \textbf{43} (2011), 1282--1310.

\bibitem{OVY} S. Olla, S. R. S. Varadhan \& H. T. Yau, \emph{Hydrodynamical
limit for a Hamiltonian system with weak noise}, Comm. Math. Phys. \textbf{%
155} (1993), 523--560.

\bibitem{QY} J. Quastel \& H.-T. Yau, \emph{Lattice gases, large deviations,
and the incompressible Navier--Stokes equations}, Ann. Math. \textbf{148}
(1998) 51--108.

\bibitem{R} G. W. Richmann, Physics Proceedings Moscow, 1956, pp. 69--571.

\bibitem{R1} N. Rourgerie, \emph{De Finetti Theorems, Mean-field Limits and
Bose-Einstein Condensation}, arXiv:1506.05263.

\bibitem{RR1} J. Rauch \& M. Reed \emph{Nonlinear microlocal analysis of
semilinear hyperbolic systems in one space dimension}. Duke Math. J. \textbf{%
49} (1982), 397--475.

\bibitem{SR09} L. Saint-Raymond, \emph{Hydrodynamic limits of the Boltzmann
equation}, volume 1971 of Lecture Notes in Mathematics. Springer-Verlag,
Berlin, 2009.

\bibitem{S1} V.~{Sohinger}, \emph{{A Rigorous Derivation of the Defocusing
Cubic Nonlinear {S}chr\"{o}dinger Equation on }}$\mathbb{T}$\emph{{$^{3}$
from the Dynamics of Many-body Quantum Systems}}, Ann. Inst. H. Poincar\'{e}
Anal. Non Lin\'{e}aire \textbf{32} (2015), 1337--1365.

\bibitem{SS1} V. Sohinger \& R. Strain, \emph{The Boltzmann equation, Besov
spaces, and optimal time decay rates in }$\mathbb{R}_{x}^{n}$, Adv. Math. 
\textbf{261} (2014) 274-332.

\bibitem{Stein} E. M. Stein, \emph{Singular Integrals and Differentiability
Properties of Functions}, Princeton University Press, 1970.

\bibitem{T1} G. Toscani, \emph{Global solution of the initial value problem
for the Boltzmann equation near a local Maxwellian}, Arch. Ration. Mech.
Anal. \textbf{102} (1988), 231--241.

\bibitem{U-U} E. A. Uehling and G. E. Uhlenbeck, \emph{Transport phenomena
in Einstein--Bose and Fermi--Dirac gases}, Phys. Rev. \textbf{43} (1933)
552--561.

\bibitem{UV2015} J. Uffink \& G. Valente, \emph{Lanford's Theorem and the
Emergence of Irreversibility}, Found. Phys. \textbf{45} (2015), 404--438.

\bibitem{U1} S. Ukai, \emph{On the existence of global solutions of mixed
problem for the non-linear Boltzmann equation.} Proc. Jpn. Acad. \textbf{50}
(1974), 179--184.

\bibitem{Z} W. Zhang, P. Lo, H. Xiong, M. Tu, \& F. Nori, \emph{General
Non-Markovian Dynamics of Open Quantum Systems}, Phys. Rev. Lett. \textbf{109%
} (2012), 170402.
\end{thebibliography}
\end{document}